\pdfoutput=1
\documentclass[11pt,journal,draftcls,onecolumn]{IEEEtran}

\usepackage{hyperref}
\usepackage{mwe}
\usepackage{setspace}
\usepackage{chngcntr}
\usepackage{etoolbox}
\usepackage{amsmath,amsfonts,amsthm,amssymb,mathtools,thmtools,amssymb} 
\usepackage{amssymb}
\usepackage{bm}
\usepackage{algorithm}
\usepackage{algpseudocode}
\usepackage{xpatch}
\usepackage{dsfont}
\usepackage{cancel}
\usepackage{overpic}
\usepackage{enumitem}
\usepackage{longtable}
\usepackage[nonumberlist]{glossaries}
\usepackage{tikz,pgfplots}
\usetikzlibrary{calc,arrows,patterns}
\usepgfplotslibrary{groupplots}
\tikzstyle{line}=[draw]
\tikzstyle{arrow}=[draw, -latex]

\algrenewcommand\alglinenumber[1]{{\sffamily\footnotesize#1}}
\algnewcommand{\Initialize}[1]{%
  \State \textbf{Initialize:}
  \Statex \hspace*{\algorithmicindent}\parbox[t]{.8\linewidth}{\raggedright #1}
}

\makeatletter
\xpatchcmd{\algorithmic}{\itemsep\z@}{\itemsep=0.15ex plus1pt}{}{}
\makeatother

\DeclarePairedDelimiter{\ceil}{\lceil}{\rceil}
\DeclareMathOperator*{\argmax}{arg\,max}
\newcommand*{\Scale}[2][4]{\scalebox{#1}{$#2$}}%

\newenvironment{brsm}{
  \bigl[ \begin{smallmatrix} }{%
  \end{smallmatrix} \bigr]}

\newtheorem{theorem}{Theorem}
\newtheorem{corollary}{Corollary}
\newtheorem{lemma}{Lemma}
\newtheorem{definition}{Definition}
\newtheorem{proposition}{Proposition}
\newtheorem{remark}{Remark}
\theoremstyle{definition}
\theoremstyle{remark}

\newcommand\Ccancel[2][black]{\renewcommand\CancelColor{\color{#1}}\xcancel{#2}}

\usepackage{geometry}
\makeglossaries
\newacronym{sc}{SC}{Successive Cancellation}
\newacronym{dmc}{DMC}{Discrete Memoryless Channel}
\newacronym{dms}{DMS}{Discrete Memoryless Source}
\newacronym{bc}{BC}{Broadcast Channel}
\newacronym{dmbc}{DM-BC}{Discrete Memoryless Broadcast Channel}
\newacronym{dbc}{DBC}{Degraded Broadcast Channel}
\newacronym{dmdbc}{DM-DBC}{Discrete Memoryless Degraded Broadcast Channel}
\newacronym{dbcldnls}{DBC-LD-NLS}{DBC with Layered Decoding and Non-Layered Secrecy}
\newacronym{dbcnldls}{DBC-NLD-LS}{DBC with Non-Layered Decoding and Layered Secrecy}
\newacronym{dwtc}{DWTC}{Degraded Wiretap Channel}
\newacronym{wtc}{WTC}{Wiretap Channel}
\newacronym{wtbc}{WTBC}{Wiretap Broadcast Channel}
\newacronym{bcc}{BCC}{Broadcast channel with Confidential Messages}
\newacronym{pcs}{PCS}{Polar Coding Scheme}
\newacronym{bec}{BEC}{Binary Erasure Channel}
\newacronym{bebc}{BE-BC}{Binary Erasure Broadcast Channel}
\newacronym{bsbc}{BS-BC}{Binary Symmetric Broadcast Channel}
\newacronym{bsc}{BSC}{Binary Symmetric Channel}
\newacronym{ciwtbc}{CI-WTBC}{Common Information over the Wiretap Broadcast Channel}
\newacronym{miwtbc}{MI-WTBC}{Multiple Information over the Wiretap Broadcast Channel}

\begin{document}

\title{Polar Coding for the Wiretap Broadcast Channel with multiple messages}

\author{Jaume~del~Olmo~Alos, and
  Javier~R.~Fonollosa,~\IEEEmembership{Senior Member,~IEEE}
  \thanks{This work is supported by ``Ministerio de Ciencia, Innovación y Universidades'' of the Spanish Government, TEC2015-69648-REDC and TEC2016-75067-C4-2-R AEI/FEDER, UE, and the Catalan Government, 2017
SGR 578 AGAUR.
  
  The authors are with the Department of Teoria del Senyal i Comunicacions, Universitat Politècnica de Catalunya, 08034, Barcelona, Spain
  (e-mail: jaume.del.olmo@upc.edu; javier.fonollosa@upc.edu)
 }}

\maketitle

\begin{abstract}
A polar coding scheme is proposed for the Wiretap Broadcast Channel with two legitimate receivers and one eavesdropper. We consider a model in which the transmitter wishes to send different confidential (and non-confidential) information to each legitimate receiver. First, we compare two inner-bounds on the achievable region of this model from the literature, and then we design a polar coding scheme that achieves the stronger one. In the proposed scheme, the encoding uses polar-based Marton's coding, where one inner-layer must be reliably decoded by both legitimate receivers, and each receiver must decode its own corresponding outer-layer. Due to the non-degradedness condition of the broadcast channel, the encoder builds a chaining construction that induces bidirectional dependencies between adjacent blocks. Indeed, these dependencies can occur between different encoding layers of different blocks. The use of two secret-keys, whose length are negligible in terms of rate, are required to prove that the polar code satisfies the strong secrecy condition.
\end{abstract}


\IEEEpeerreviewmaketitle

\section{Introduction}\label{sec:introduction}
This chapter focuses on a channel model over the \gls*{wtbc} where transmitter wants to reliably send different confidential (and non-confidential) messages to different legitimate receivers with the presence of an eavesdropper. This model generalizes the one described in \cite{alos2019polar}, where only common information is intended for both receivers.

There are two different inner-bounds on the achievable region of this model in the literature. On the one hand, the inner-bound found in \cite{7174526} considers confidential information only, while the one in \cite{6420946} consider confidential and non-confidential messages. The random coding techniques used in \cite{7174526} and \cite{6420946} are Marton's coding and rate splitting in conjunction with superposition coding and binning. The only difference between them is that the first uses joint decoding, while the second uses successive decoding. Indeed, if we consider confidential information transmission only, the inner-bound in \cite{7174526} includes the one in \cite{6420946}, but is not straightforward to show that whether the first inner-bound is strictly larger or not: for a given input distribution, the inner-bound in \cite{7174526} is strictly larger; nevertheless, we do not know if the rate points that are only included in this inner-bound for a particular distribution may be in the inner-bound found in \cite{6420946} under another distribution.



We provide a \gls*{pcs} that achieves the inner-bound in \cite{7174526} and, additionally, allows transmitting different non-confidential messages to both legitimate receivers. Our \gls*{pcs} is based in part on the one described in \cite{6949646} that achieves Marton's region of broadcast channels without secrecy constraints, and the one in \cite{alos2019polar} for the \gls*{ciwtbc}. In Marton's coding we have three different layers: one inner-layer that must be reliably decoded by both legitimate receivers, and two outer-layers such that each one conveys information intended only for one receiver. Due to the non-degradedness condition of the channels, the \gls*{pcs} requires the use of a chaining construction which induces bidirectional dependencies between adjacent blocks. Moreover, we show that joint and successive decoding have their counterpart in polar coding, and jointly decoding allows to enlarge the achievable region for a particular input distribution. Indeed, due to the polar-based jointly decoding, our \gls*{pcs} needs to build a chaining construction that introduces dependencies between different encoding layers of adjacent blocks.

The remaining of this chapter is organized as follows. Section~\ref{sec:SMx} introduces the channel model formally as well as an enlarged version of the inner-bound found in \cite{7174526} that considers the transmission of private messages. In Section~\ref{sec:PCSx1} we describe a \gls*{pcs} that achieves this inner-bound. Then, Section~\ref{sec:performancex} evaluates the performance of this \gls*{pcs}. Finally, the concluding remarks are presented in Section~\ref{sec:conclusionx}.

\section{Channel model and achievable region}\label{sec:SMx}
A \gls*{wtbc} $(\mathcal{X}, p_{Y_{(1)}Y_{(2)} Z|X}, \mathcal{Y}_{(1)} \times \mathcal{Y}_{(2)} \times \mathcal{Z})$ with 2 legitimate receivers and an external eavesdropper, as we have seen in the previous chapter, is characterized by the probability transition function $p_{Y_{(1)}Y_{(2)}Z|X}$, where $X \in \mathcal{X}$ denotes the channel input, $Y_{(k)} \in \mathcal{Y}_{(k)}$ denotes the channel output corresponding to the legitimate receiver $k \in [1,2]$, and $Z \in \mathcal{Z}$ denotes the channel output corresponding to the eavesdropper. We consider a model, namely \gls*{miwtbc}, in which the transmitter wishes to send two private messages $W_1$ and $W_2$, and two confidential messages $S_1$ and $S_2$, where $W_1$ and $S_1$ are intended to legitimate Receiver~1, and $W_2$ and $S_2$ are intended to legitimate Receiver~2. A code $\big(\ceil{2^{nR_{W_{(1)}}}},\ceil{2^{nR_{S_{(1)}}}},\ceil{2^{nR_{W_{(2)}}}}, \ceil{2^{nR_{S_{(2)}}}}, n \big)$ for the \gls*{miwtbc} consists of two private message sets $\mathcal{W}_{(1)}$ and $\mathcal{W}_{(2)}$ where $\mathcal{W}_{(k)} \triangleq  \big[1, \ceil{2^{nR_{W_{(k)}}}} \big]$ for $k \in [1,2]$, two confidential message sets $\mathcal{S}_{(1)}$ and $\mathcal{S}_{(2)}$ where $\mathcal{S}_{(k)} \triangleq  \big[1, \ceil{2^{nR_{S_{(k)}}}} \big]$ for $k \in [1,2]$, an encoding function $f: \mathcal{W}_{(1)} \times \mathcal{S}_{(1)} \times \mathcal{W}_{(2)} \times \mathcal{S}_{(2)} \rightarrow \mathcal{X}^n$ that maps $(w_{(1)},w_{(2)},s_{(1)},s_{(2)})$ to a codeword $x^n$, and two decoding functions $g_{(1)}$ and $g_{(2)}$ such that $g_{(k)}:\mathcal{Y}_{(k)}^n \rightarrow \mathcal{W}_{(k)} \times \mathcal{S}_{(k)}$ ($k \in [1,2]$) maps the $k$-th legitimate receiver observations $y_{(k)}^n$ to the estimates $(\hat{w}_{(k)},\hat{s}_{(k)})$. The reliability condition to be satisfied by this code is given by
\begin{IEEEeqnarray}{c}
\lim_{n \rightarrow \infty} \mathbb{P}\left[  (W_{(k)},S_{(k)}) \neq (\hat{W}_{(k)},\hat{S}_{(k)}) \right] = 0,  \quad k \in [1,2]. \IEEEeqnarraynumspace \label{eq:reliabilitycondx}
\end{IEEEeqnarray}
The \emph{strong} secrecy condition is measured in terms of the information leakage and is given by
\begin{IEEEeqnarray}{c}
\lim_{n \rightarrow \infty} I \left( S_{(1)} S_{(2)} ; Z^n \right) = 0.  \label{eq:secrecycondx}
\end{IEEEeqnarray}
This model is plotted in Figure~\ref{fig:chmodelx}. A tuple of rates $(R_{W_{(1)}}, R_{S_{(1)}}, R_{W_{(2)}}, R_{S_{(2)}}) \in \mathbb{R}_{+}^4$ is achievable for the \gls*{miwtbc} if a sequence of $\big(\ceil{2^{nR_{W_{(1)}}}},\ceil{2^{nR_{S_{(1)}}}},\ceil{2^{nR_{W_{(2)}}}}, \ceil{2^{nR_{S_{(2)}}}}, n \big)$ codes that satisfy the reliability and secrecy conditions \eqref{eq:reliabilitycondx} and \eqref{eq:secrecycondx} respectively exists. 
\begin{figure}[h!]
\hspace{1cm}
\centering
\begin{overpic}[width=0.6\linewidth]{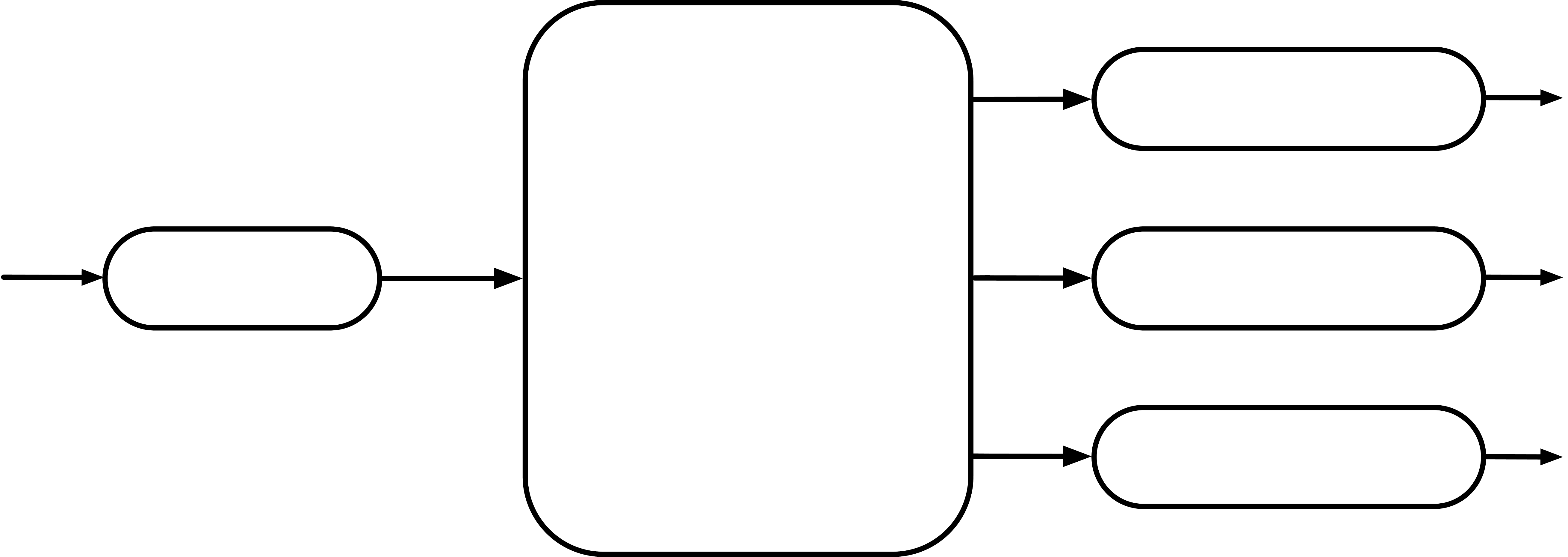}
\put (-34,19.8) {\small $({W}_{(1)},{S}_{(1)},{W}_{(2)},{S}_{(2)})$}
\put (9,16.2) {\small Encoder}
\put (25.5,18.8) {\small $X^n$}
\put (74,27.8) {\small Receiver 1}
\put (74,16.2) {\small Receiver 2}
\put (71,4.8) {\small $\Scale[0.98]{\text{Eavesdropper}}$}
\put (41,18) {\small \gls*{wtbc}}
\put (38,14) {\small $p_{Y_{(1)}Y_{(2)}Z|X}$}
\put (96,31) {\small $(\hat{W}_{(1)},\hat{S}_{(1)})$}
\put (96,19.4) {\small $(\hat{W}_{(2)},\hat{S}_{(2)})$}
\put (96,8.5) {\small $\Ccancel[red]{(S_{(1)},S_{(2)})}$}
\end{overpic}
\caption{Channel model: \gls*{miwtbc}.}\label{fig:chmodelx}
\end{figure}
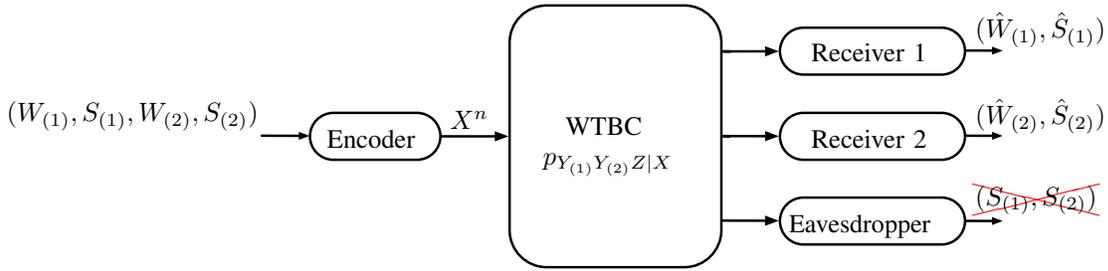

References \cite{7174526} and \cite{6420946} define two different inner-bounds on the capacity region of this model. Indeed, the inner-bound in \cite{7174526} only consider the case where $R_{W_{(1)}}=R_{W_{(2)}} \triangleq 0$, that is, when only confidential messages are transmitted over the broadcast channel. In this situation, \cite{7174526} (Remark~3) points out that this inner-bound includes the one defined in~\cite{6420946}, but it does not specify whether this bound is strictly larger or not.  


\begin{definition}[Adapted from \cite{7174526}]
\label{prop:MIB2}
The region $\mathfrak{R}_{\text{\emph{MI-WTBC}}}^{(\text{\emph{S}})}$ defined by the union over all the pairs of rates $(R_{S_{(1)}},R_{S_{(2)}}) \in \mathbb{R}_{+}^2$ satisfying
\begin{align}
R_{S_{(1)}} & \leq I (V U_{(1)}; Y_{(1)}) - I (V U_{(1)}; Z) \nonumber \\
R_{S_{(2)}} & \leq I (V U_{(2)}; Y_{(2)}) - I (V U_{(2)}; Z) \nonumber \\
R_{S_{(1)}} + R_{S_{(2)}} & \leq I (V U_{(1)}; Y_{(1)}) + I (V U_{(2)}; Y_{(2)}) - I (U_{(1)} ; U_{(2)} | V) \nonumber \\
& \quad - I (V U_{(1)} U_{(2)}; Z) - \max \{ I(V;Y_{(1)}), I(V;Y_{(2)}), I(V;Z) \} \nonumber 
\end{align}
for some distribution $p_{VU_{(1)}U_{(2)}XY_{(1)}Y_{(2)}Z}$ such that $(VU_{(1)}U_{(2)}) - X - (Y_{(1)},Y_{(2)},Z)$ forms a Markov chain and $I(U_{(1)};Y_{(1)}|V) + I(U_{(2)};Y_{(2)}|V) \geq I(U_{(1)};U_{(2)}|V)$, defines an inner-bound on the achievable region of the \gls*{miwtbc} when $R_{W_{(1)}}=R_{W_{(2)}} \triangleq 0$.
\end{definition}

\begin{remark}\label{remark:funcX}
Since the previous inner-bound on the achievable region of the \gls*{miwtbc} cannot be enlarged by considering general distributions $p_{X|VU_{(1)}U_{(2)}}$, the channel input $X$ can be restricted to be any deterministic function of random variables $(VU_{(1)}U_{(2)})$.
\end{remark}

\begin{remark}
If $Z \triangleq \varnothing$ then region $\mathfrak{R}_{\text{\emph{MI-WTBC}}}^{(\text{\emph{S}})}$ reduces to well-known Marton's region for the broadcast channel without secrecy constraints \cite{1056046}.
\end{remark}

\begin{proposition}
When $R_{W_{(1)}}=R_{W_{(2)}} \triangleq 0$, for a particular distribution $p_{VU_{(1)}U_{(2)}XY_{(1)}Y_{(2)}Z}$, $\mathfrak{R}_{\text{\emph{MI-WTBC}}}^{(\text{\emph{S}})}$ in Proposition~\ref{prop:MIB2} is strictly larger than the inner-bound in \cite{6420946} (Theorem~1).
\end{proposition}
\begin{proof}
Consider $\mathfrak{R}_{\text{{MI-WTBC}}}^{(\text{{S}})}$ when $p_{VU_{(1)}U_{(2)}XY_{(1)}Y_{(2)}Z}$ is such that $I(V;Y_{(1)}) < I(V;Y_{(2)})$. In this case, it is easy to check that the inner-bound in~\cite{6420946} (Theorem~1) imposes that $R_{S_{(2)}} \leq I(V;Y_{(1)})+I(U_{(2)};Y_{(2)}|V) - I(VU_{(2)}; Z)$, which is strictly less than the bound on $R_{S_{(2)}}$ in Proposition~\ref{prop:MIB2}. Similarly, the upper-bound on $R_{S_{(1)}}$ in Proposition~\ref{prop:MIB2} is also strictly larger than the one in~\cite{6420946} (Theorem~1) when $I(V;Y_{(2)}) < I(V;Y_{(1)})$. 
\end{proof}

\begin{remark}
In general, we cannot affirm that the inner-bound in Proposition~\ref{prop:MIB2} is strictly larger than the one in \cite{6420946}: the rate tuples that are included only in Proposition~\ref{prop:MIB2} for a particular distribution may be in the inner-bound of \cite{6420946} under another distribution.
\end{remark}

\begin{remark}
The region in~\cite{6420946} imposes that $I(V;Z) \leq I(V;Y_{(k)})$ for any $k \in [1,2]$. Nevertheless, Proposition~\ref{prop:MIB2} does not restrict $p_{VU_{(1)}U_{(2)}XY_{(1)}Y_{(2)}Z}$ to satisfy this condition.
\end{remark}

\begin{remark}
As mentioned in \cite{7174526} (Remark~3), the coding techniques used in \cite{7174526} and  \cite{6420946} to obtain the inner-bounds are almost the same and the only difference are the decoding strategies: joint decoding in \cite{7174526} and successive decoding for \cite{6420946}. Indeed, we will see that these strategies have a connection on the \gls*{pcs} that is described in Section~\ref{sec:PCSx1} and joint decoding enlarges the inner-bound on the achievable region for a particular distribution. 
\end{remark}

For compactness of notation, let $k \in [1,2]$, and $\bar{k} \triangleq [1,2] \setminus k$. 
The following proposition defines an inner-bound on the achievable region for the \gls*{miwtbc} when considering also the transmission of the private messages $W_{(1)}$ and $W_{(2)}$.
\begin{proposition}
\label{prop:MIB2x} 
The region $\mathfrak{R}_{\text{\emph{MI-WTBC}}} \triangleq Conv \big( \mathfrak{R}_{\text{\emph{MI-WTBC}}}^{(1)} \cup \mathfrak{R}_{\text{\emph{MI-WTBC}}}^{(2)} \big)$
defines an inner-bound on the achievable region of the \gls*{miwtbc}, where $Conv(\cdot)$ denotes the convex hull of a set of rate-tuples, for $k \in [1,2]$ we have 
\begin{align*}
\arraycolsep=1.4pt \def\arraystretch{1.1}
& \mathfrak{R}_{\text{\emph{MI-WTBC}}}^{(k)} \\
& \quad \triangleq \bigcup_{\mathcal{P}}
\left\{
\begin{array}{rl}
R_{S_{(k)}} \! \! \! \! \! \! & \leq I(VU_{(k)};Y_{(k)}) - I(VU_{(k)};Z) \\
R_{S_{(\bar{k})}} \! \! \! \! \! \! & \leq  I(VU_{(\bar{k})};Y_{(\bar{k})}) - I(U_{(\bar{k})};U_{(k)}|V) - I(U_{(\bar{k})};Z|VU_{(k)})  \\
\! \! \! \! \! \! & \quad - \max \{ I(V;Y_{(1)}), I(V;Y_{(2)}), I(V;Z) \} \\
R_{S_{(k)}} + R_{W_{(k)}} \! \! \! \! \! \! & \leq I(VU_{(k)};Y_{(k)}) \\
R_{S_{(\bar{k})}} + R_{W_{(\bar{k})}} \! \! \! \! \! \! & \leq  I(VU_{(\bar{k})};Y_{(\bar{k})}) - I(U_{(\bar{k})};U_{(k)}|V) \\
\! \! \! \! \! \! & \quad - \max \{ I(V;Y_{(1)}), I(V;Y_{(2)}), I(V;Z) \} + I(V;Z),
\end{array}
\right\},
\end{align*}
and $\mathcal{P}$ contains all distributions $p_{VU_{(1)}U_{(2)}XY_{(1)}Y_{(2)}Z}$ such that $(VU_{(1)}U_{(2)}) - X - (Y_{(1)},Y_{(2)},Z)$ forms a Markov chain and $I(U_{(1)};Y_{(1)}|V) + I(U_{(2)};Y_{(2)}|V) \geq I(U_{(1)};U_{(2)}|V)$.
\end{proposition}
In Section~\ref{sec:PCSx1} we describe a \gls*{pcs} that achieves the corner point of $\mathfrak{R}_{\text{{MI-WTBC}}}^{(k)}$, where $k \in [1,2]$, and then we discuss how to achieve any rate tuple of the entire region $\mathfrak{R}_{\text{{MI-WTBC}}}$.

\section{Review of Polar Codes}\label{sec:rev}
Let $(\mathcal{X} \times \mathcal{Y}, p_{XY})$ be a Discrete Memoryless Source (DMS), where\footnote{Throughout this paper, we assume binary polarization. Nevertheless, an extension to $q$-ary alphabets is possible \cite{5513555,csasouglu2009polarization}.} ${X} \in \{0,1\}$ and $Y \in \mathcal{Y}$. The polar transform over the $n$-sequence $X^n$, $n$ being any power of $2$, is defined as $U^n \triangleq X^n G_n$, where $G_n \triangleq \begin{brsm}1 & 1\\1 & 0\end{brsm}^{\otimes n}$ is the source polarization matrix \cite{arikan2010source}. Since $G_n = G_n^{-1}$, then $X^n = U^n G_n$.

The polarization theorem for source coding with side information \cite[Theorem~1]{arikan2010source} states that the polar transform extracts the randomness of $X^n$ in the sense that, as $n \rightarrow \infty$, the set of indices $j \in [1,n]$ can be divided practically into two disjoint sets, namely $\smash{\mathcal{H}_{X|Y}^{(n)}}$ and $\smash{\mathcal{L}_{X|Y}^{(n)}}$, such that $U(j)$ for $j \in \mathcal{H}_{X|Y}^{(n)}$ is practically independent of $(U^{1:j-1},Y^n)$ and uniformly distributed, that is, $H ({U(j) | U^{1:j-1}, Y^n} ) \rightarrow 1$, and $U(j)$ for $j \in \smash{\mathcal{L}_{X|Y}^{(n)}}$ is almost determined by $(U^{1:j-1}, Y^n)$, which means that $H ( U(j) | U^{1:j-1}, Y^n ) \rightarrow 0$. Formally, let $\delta_n \triangleq 2^{-n^{\beta}}$, where $\beta \in (0, \frac{1}{2})$, and
\begin{IEEEeqnarray*}{rCl}
\mathcal{H}_{X|Y}^{(n)} & \triangleq & \left\{ j \in [1,n] \! : H \! \left( U(j) \left| U^{1:j-1},Y^n  \right. \! \right) \geq 1-\delta_n \right\}, \\
\mathcal{L}_{X|Y}^{(n)} & \triangleq & \left\{ j \in [1,n] \! : H \! \left( U(j) \left| U^{1:j-1},Y^n  \right. \! \right) \leq \delta_n \right\}. 
\end{IEEEeqnarray*}
Then, by \cite[Theorem~1]{arikan2010source} we have $\smash{\lim_{n \rightarrow \infty} \frac{1}{n} | \mathcal{H}_{X|Y}^{(n)} | }  = H(X|Y)$ and $\smash{\lim_{n \rightarrow \infty} \frac{1}{n} | \mathcal{L}_{X|Y}^{(n)} | } = 1 - H(X|Y)$. Consequently, the number of elements $U(j)$ that \emph{have not polarized} is asymptotically negligible in terms of rate, that is, $\smash{\lim_{n \rightarrow \infty} \frac{1}{n} | ( \mathcal{H}_{X|Y}^{(n)} )^{\text{C}} \setminus  \mathcal{L}_{X|Y}^{(n)}  | } = 0$. 

Furthermore, \cite[Theorem~2]{arikan2010source} states that given the part $U[(\mathcal{L}_{X|Y}^{(n)} )^{\text{C}}]$ and the channel output observations $Y^n$, the remaining part $\smash{U[\mathcal{L}_{X|Y}^{(n)}]}$ can be reconstructed by using SC decoding with error probability in $O(n \delta_n)$.

Similarly to $\mathcal{H}_{X|Y}^{(n)}$ and $\mathcal{L}_{X|Y}^{(n)}$, the sets $\mathcal{H}_{X}^{(n)}$ and $\mathcal{L}_{X}^{(n)}$ can be defined by considering that the observations $Y^n$ are absent. Since conditioning does not increase the entropy, we have $\mathcal{H}_{X}^{(n)} \supseteq \mathcal{H}_{X|Y}^{(n)}$ and $\mathcal{L}_{X}^{(n)} \subseteq \mathcal{L}_{X|Y}^{(n)}$.  A discrete memoryless channel $(\mathcal{X}, p_{Y|X}, \mathcal{Y})$ with some arbitrary $p_X$ can be seen as a DMS $(\mathcal{X} \times \mathcal{Y}, p_{X}p_{Y|X})$. In channel polar coding, first we define the sets of indices $\mathcal{H}_{X|Y}^{(n)}$, $\mathcal{L}_{X|Y}^{(n)}$, $\mathcal{H}_{X}^{(n)}$ and $\mathcal{L}_{X}^{(n)}$ from the target distribution $p_{X}p_{Y|X}$. Then, based on the previous sets, the encoder somehow constructs\footnote{Since the polar-based encoder will construct random variables that must approach the target distribution of the DMS, throughout this paper we use \emph{tilde} above the random variables to emphazise this purpose.} $\tilde{U}^n$ and applies the inverse polar transform $\tilde{X}^n = \tilde{U}^n G_n$. Afterwards, the transmitter sends $\tilde{X}^n$ over the channel, which induces $\tilde{Y}^n$. Let $(\tilde{X}^n,\tilde{Y}^n) \sim \tilde{q}_{X^n}\tilde{q}_{Y^n|X^n}$, if $\mathbb{V} (\tilde{q}_{X^nY^n}, p_{X^nY^n}) \rightarrow 0$ then the receiver can reliably reconstruct $\tilde{U}[\mathcal{L}_{X|Y}^{(n)}]$ from $\tilde{Y}^n$ and $\smash{\tilde{U}[(\mathcal{L}_{X|Y}^{(n)} )^{\text{C}}]}$ by performing SC decoding \cite{korada2010polar}.

\section{Polar coding scheme}\label{sec:PCSx1} 
Notice that $\mathfrak{R}_{\text{MI-WTBC}}$ of Proposition~\ref{prop:MIB2x} is not affected by switching subindices~1~and~2.
Thus, we can assume without loss of generality that $I(V;Y_{(1)}) \leq I(V;Y_{(2)})$. Otherwise, the coding scheme that is described later will also be suitable by simply exchanging the roles of Receiver~1 and Receiver~2. Consequently, the \gls*{pcs} must contemplate three different situations: 
\begin{enumerate}
\setlength\itemsep{-0.10em}
\item[] \emph{Situation~1}: when $I(V;Z) \leq I(V;Y_{(1)}) \leq I(V;Y_{(2)})$,
\item[] \emph{Situation~2}: when $ I(V;Y_{(1)}) < I(V;Z) \leq I(V;Y_{(2)})$,
\item[] \emph{Situation~3}: when $ I(V;Y_{(1)}) \leq I(V;Y_{(2)}) < I(V;Z)$.
\end{enumerate}

Under the previous assumption, in order to proof that $\mathfrak{R}_{\text{{MI-WTBC}}}$ is entirely achievable by using polar codes for any of the situations mentioned before, it suffices to provide a \gls*{pcs} that achieves the corner points $(R_{S_{(1)}}^{\star 1}, R_{S_{(2)}}^{\star 1}, R_{W_{(1)}}^{\star 1}, R_{W_{(2)}}^{\star 1}) \subset \mathfrak{R}_{\text{{MI-WTBC}}}^{(1)}$ and $(R_{S_{(1)}}^{\star 2}, R_{S_{(2)}}^{\star 2}, R_{W_{(1)}}^{\star 2}, R_{W_{(2)}}^{\star 2}) \subset \mathfrak{R}_{\text{{MI-WTBC}}}^{(2)}$, where
\begin{align}
R_{S_{(k)}}^{\star k} & \triangleq H(V U_{(k)} | Z) - H(V U_{(k)} | Y_{(k)}), \label{eq:rsk} \\
R_{S_{(\bar{k})}}^{\star k} & \triangleq H(U_{(\bar{k})} | V U_{(k)} Z) - H(U_{(\bar{k})} | V Y_{(\bar{k})}) \nonumber \\
& \quad - \big( H(V|Y_{(\bar{k})}) -  \min \{ H(V|Y_{(1)}) , H(V|Z) \} \big), \label{eq:rsbk} \\
R_{W_{(k)}}^{\star k} & \triangleq H(VU_{(k)}) - H(VU_{(k)}|Z), \label{eq:rwk} \\
R_{W_{(\bar{k})}}^{\star k} & \triangleq H(U_{(\bar{k})} | VU_{(k)}) - H(U_{(\bar{k})} | VU_{(k)} Z). \label{eq:rwbk}
\end{align}
for any $k \in [1,2]$, and recall that $\bar{k} = [1,2]\setminus k$. We have expressed $(R_{S_{(1)}}^{\star k },R_{S_{(2)}}^{\star k},R_{W_{(1)}}^{\star k },R_{W_{(2)}}^{\star k})$ in terms of entropies for convenience. Indeed, $(R_{S_{(1)}}^{\star k}, R_{S_{(2)}}^{\star k}, R_{W_{(1)}}^{\star k}, R_{W_{(2)}}^{\star k})$ corresponds to the case where, for a given distribution $p_{VU_{(1)}U_{(2)}XY_{(1)}Y_{(2)}Z}$, we set the maximum rate for the confidential and private messages intended for Receiver~$k$, and then we set the maximum possible rates of the remaining messages associated to Receiver~$\bar{k}$.

\begin{remark}
For distributions such that $I(V;Y_{(1)}) < I(V;Y_{(2)})$, the inner-bound in \cite{6420946} does not include the corner point $(R_{S_{(1)}}^{\star 2}, R_{S_{(2)}}^{\star 2}, R_{W_{(1)}}^{\star 2}, R_{W_{(2)}}^{\star 2})$. In this section we will see that polar-based joint decoding is necessary for the \gls*{pcs} to approach this rate tuple.

Moreover, distributions $p_{VU_{(1)}U_{(2)}XY_{(1)}Y_{(2)}Z}$ such that satisfy Situations~$2$ and $3$, that is when $I(V;Y_{(k)}) < I(V;Z)$ for some $k \in [1,2]$, are not considered in the definition of the inner-bound \cite{6420946}. In order to approach $(R_{S_{(1)}}^{\star k}, R_{S_{(2)}}^{\star k}, R_{W_{(1)}}^{\star k}, R_{W_{(2)}}^{\star k})$, for any $k \in [1,2]$, we will see that polar-based joint decoding is also needed in these situations.
\end{remark}



Let $(\mathcal{V} \times \mathcal{U}_{(1)} \times \mathcal{U}_{(2)} \times \mathcal{X} \times \mathcal{Y}_{(1)} \times \mathcal{Y}_{(2)} \times \mathcal{Z} , p_{VU_{(1)}U_{(2)}XY_{(1)}Y_{(2)}Z})$ denote the \gls*{dms} that represents the input $(V,U_{(1)},U_{(2)},X)$ and output $(Y_{(1)},Y_{(2)},Z)$ random variables of the \gls*{miwtbc}, where $|\mathcal{V}|=|\mathcal{U}_{(1)}|=|\mathcal{U}_{(2)}|= |\mathcal{X}| \triangleq 2$. For the input random variable $V$, we define the polar transform $A^n \triangleq V^n G_n$ and the associated sets
\begin{align}
\mathcal{H}_V^{(n)} & \triangleq  \big\{j \in [1,n]  : H  \big( A(j) \big| A^{1:j-1} \big) \geq 1 - \delta_n \big\},   \label{eq:HUx} \\
\mathcal{L}_V^{(n)} & \triangleq  \big\{j \in [1,n]  : H  \big( A(j) \big| A^{1:j-1} \big) \leq  \delta_n \big\},   \label{eq:LUx} \\
\mathcal{H}_{V|Z}^{(n)} &  \triangleq  \big\{j \in [1,n]   :H \big( A(j) \big| A^{1:j-1} Z^n \big) \geq 1 - \delta_n \big\},  \label{eq:HUZx} \\
\mathcal{L}_{V|Y_{(k)}}^{(n)} & \triangleq  \big\{j \in [1,n] : H  \big( A(j) \big| A^{1:j-1} Y_{(k)}^n  \big) \leq  \delta_n \big\}, \quad k = 1, 2. \label{eq:LUY_kx}  
\end{align}
For the random variable $U_{(k)}$, where $k \in [1,2]$, we define the polar transform $T_{(k)}^n \triangleq U_{(k)}^n G_n$. In this model, due to the polar-based Marton's coding strategy similar to the one in \cite{6949646}, the polar code constructions corresponding to $T_{(1)}^n$ and $T_{(2)}^n$ are different depending on the corner point that the \gls*{pcs} must approach. To achieve $(R_{S_{(1)}}^{\star 1}, R_{S_{(2)}}^{\star 1}, R_{W_{(1)}}^{\star 1}, R_{W_{(2)}}^{\star 1})$, the construction of $T_{(1)}^n$ only depends on $V^n$, while $T_{(2)}^n$ depends on $V^n$ and $T_{(1)}^n$. Otherwise, to achieve $(R_{S_{(1)}}^{\star 2}, R_{S_{(2)}}^{\star 2}, R_{W_{(1)}}^{\star 2}, R_{W_{(2)}}^{\star 2})$, the construction of $T_{(2)}^n$ only depends on $V^n$, while the one of $T_{(1)}^n$ depends on $V^n$ and $T_{(2)}^n$. 
Therefore, consider that the \gls*{pcs} must achieve $(R_{S_{(1)}}^{\star k}, R_{S_{(2)}}^{\star k}, R_{W_{(1)}}^{\star k}, R_{W_{(2)}}^{\star k}) \subseteq \mathfrak{R}_{\text{MI-WTBC}}^{(k)}$, where $k \in [1,2]$. Associated to $U_{(k)}$, define 
\begin{align}
\mathcal{H}_{U_{(k)}|V}^{(n)} & \triangleq  \big\{j \in [1,n]  : H  \big( T_{(k)}(j) \big| T_{(k)}^{1:j-1} V^n \big) \geq 1 - \delta_n \big\},   \label{eq:HUVx} \\
\mathcal{L}_{U_{(k)}|V}^{(n)} & \triangleq  \big\{j \in [1,n]  : H  \big( T_{(k)}(j) \big| T_{(k)}^{1:j-1} V^n \big) \leq  \delta_n \big\},   \label{eq:LUVx} \\
\mathcal{H}_{U_{(k)}|VZ}^{(n)} &  \triangleq  \big\{j \in [1,n]   :H \big( T_{(k)}(j) \big| T_{(k)}^{1:j-1} V^n Z^n \big) \geq 1 - \delta_n \big\},  \label{eq:HUVZx} \\
\mathcal{L}_{U_{(k)}|VY_{(k)}}^{(n)} & \triangleq  \big\{j \in [1,n] : H  \big( T_{(k)}(j) \big| T_{(k)}^{1:j-1} V^n Y_{(k)}^n  \big) \leq  \delta_n \big\}. \label{eq:LUVY_kx}  
\end{align}
Also, define the following sets associated to the polar transform $U_{(\bar{k})}$:
\begin{align}
\mathcal{H}_{U_{(\bar{k})}|VU_{(k)}}^{(n)} & \triangleq  \big\{j \in [1,n]  : H  \big( T_{(\bar{k})}(j) \big| T_{(\bar{k})}^{1:j-1} V^n U_{(k)}^n \big) \geq 1 - \delta_n \big\},   \label{eq:HUVUx} \\
\mathcal{L}_{U_{(\bar{k})}|VU_{(k)}}^{(n)} & \triangleq  \big\{j \in [1,n]  : H  \big( T_{(\bar{k})}(j) \big| T_{(\bar{k})}^{1:j-1} V^n U_{(k)}^n \big)  \leq  \delta_n \big\},   \label{eq:LUVUx} \\
\mathcal{H}_{U_{(\bar{k})}|V U_{(k)} Z}^{(n)} &  \triangleq  \big\{j \in [1,n] : H  \big( T_{(\bar{k})}(j) \big| T_{(\bar{k})}^{1:j-1} V^n U_{(k)}^n Z^n \big) \geq 1 - \delta_n \big\},  \label{eq:HUVUZx} \\
\mathcal{L}_{U_{(\bar{k})}|V Y_{(\bar{k})}}^{(n)} &  \triangleq  \big\{j \in [1,n] : H  \big( T_{(\bar{k})}(j) \big| T_{(\bar{k})}^{1:j-1} V^n Y_{(\bar{k})}^n \big) \leq \delta_n \big\}. \label{eq:LUVUY_kx}   
\end{align}
\indent Consider that the encoding takes place over $L$ blocks indexed by $i \in [1,L]$. In order to approach the corner point $(R_{S_{(1)}}^{\star k}, R_{S_{(2)}}^{\star k}, R_{W_{(1)}}^{\star k}, R_{W_{(2)}}^{\star k})$, where $k \in [1,2]$, at Block~$i \in [1,L]$ the encoder will construct $\tilde{A}_i^n$, which will carry part of the confidential and private message that is intended for legitimate Receiver~$k$. Additionally, the encoder will store into $\tilde{A}_i^n$ some elements from $\tilde{A}_{i-1}^n$ (if $i \in [2,L]$) and $\tilde{A}_{i+1}^n$ (if $i \in [1,L-1]$) so that both legitimate receivers are able to reliably reconstruct $\tilde{A}_{1:L}^n$ (chaining construction). Then, the encoder first constructs $\tilde{T}_{(k),i}^n$, which will depend on $\tilde{V}_{i}^n = \tilde{A}_i^n G_n$ and will carry the remaining parts of the confidential and private messages intended for Receiver~$k$. Indeed, it could also depend on $\tilde{V}_{i-1}^n$ and/or $\tilde{V}_{i+1}^n$ if polar-based jointly decoding is considered because some elements of $\tilde{V}_{i-1}^n$ and/or $\tilde{V}_{i+1}^n$ may be stored in $\tilde{T}_{(k),i}^n$. Then, the encoder forms $\tilde{T}_{(\bar{k}),i}^n$, which depends on $\big(\tilde{V}_{i}^n ,\tilde{T}_{(k),i}^n \big)$. In fact, as before, if polar-based jointly decoding is considered then the chaining construction may store some elements of $\tilde{V}_{i-1}^n$ and/or $\tilde{V}_{i+1}^n$ in $\tilde{T}_{(\bar{k}),i}^n$. Finally, it will obtain $\tilde{U}_{(k),i}^n = \tilde{T}_{(k),i}^n G_n$ for $k \in [1,2]$ and deterministically form $\tilde{X}_i^n$ (see Remark~\ref{remark:funcX}). The codeword $\tilde{X}^n$ then is transmitted over the \gls*{wtbc} inducing the channel outputs $(\tilde{Y}_{(1),i}^n,\tilde{Y}_{(2),i}^n,\tilde{Z}_i^n)$.  

For better readability and understanding, the methods of constructing the \emph{inner-layer} $\tilde{A}_{1:L}^n$ and the \emph{outer-layers} $\tilde{T}_{(1),1:L}^n$ and $\tilde{T}_{(2),1:L}^n$ are described independently in the following subsections. Nevertheless, if we consider the encoding moving from Block~1 to Block~$L$, and since dependencies only occur between adjacent blocks, the encoder is able to form $\tilde{T}_{(1),i}^n$ and $\tilde{T}_{(2),i}^n$ once $\tilde{A}_{i+1}^n$ ($i \in [1,L-1]$) is constructed.

\subsection{Construction of the inner-layer}\label{sec:PCSx1il} 
The method of forming $\tilde{A}_{1:L}^n$ is very similar to the one described in \cite{alos2019polar}. Besides the sets in~\eqref{eq:HUx}--\eqref{eq:LUY_kx}, we define
$\mathcal{G}^{(n)} \triangleq \mathcal{H}^{(n)}_{V|Z}$ and $\mathcal{C}^{(n)} \triangleq \mathcal{H}^{(n)}_V \cap \big( \mathcal{H}^{(n)}_{V|Z} \big)^{\text{C}}$, which form a partition of $\mathcal{H}_{V}^{(n)}$.
Moreover, we also define the following partition of the set $\mathcal{G}^{(n)}$:
\begin{IEEEeqnarray}{rCl}
\mathcal{G}^{(n)}_{0} & \triangleq & \mathcal{G}^{(n)} \cap \mathcal{L}^{(n)}_{V|Y_{(1)}}\cap   \mathcal{L}^{(n)}_{V|Y_{(2)}}, \label{eq:sG0x} \\
\mathcal{G}^{(n)}_{1} & \triangleq  & \mathcal{G}^{(n)} \cap \big( \mathcal{L}^{(n)}_{V|Y_{(1)}} \big)^{\text{C}} \cap \mathcal{L}^{(n)}_{V|Y_{(2)}} , \label{eq:sG1x} \\
\mathcal{G}^{(n)}_{2} & \triangleq & \mathcal{G}^{(n)} \cap \mathcal{L}^{(n)}_{V|Y_{(1)}} \cap  \big( \mathcal{L}^{(n)}_{V|Y_{(2)}} \big)^{\text{C}} , \label{eq:sG2x}  \\ 
\mathcal{G}^{(n)}_{1,2} & \triangleq &  \mathcal{G}^{(n)} \cap \big( \mathcal{L}^{(n)}_{V|Y_{(1)}} \big)^{\text{C}}  \cap   \big( \mathcal{L}^{(n)}_{V|Y_{(2)}} \big)^{\text{C}}, \label{eq:sG12x} 
\end{IEEEeqnarray}
and the following partition of the set $\mathcal{C}^{(n)}$:
\begin{IEEEeqnarray}{rCl}
\mathcal{C}^{(n)}_{0} & \triangleq & \mathcal{C}^{(n)} \cap \mathcal{L}^{(n)}_{V|Y_{(1)}} \cap   \mathcal{L}^{(n)}_{V|Y_{(2)}}, \label{eq:sC0x} \\
\mathcal{C}^{(n)}_{1} & \triangleq  & \mathcal{C}^{(n)} \cap  \big( \mathcal{L}^{(n)}_{V|Y_{(1)}} \big)^{\text{C}} \cap \mathcal{L}^{(n)}_{V|Y_{(2)}}, \label{eq:sC1x} \\
\mathcal{C}^{(n)}_{2} & \triangleq & \mathcal{C}^{(n)} \cap \mathcal{L}^{(n)}_{V|Y_{(1)}} \cap  \big( \mathcal{L}^{(n)}_{V|Y_{(2)}} \big)^{\text{C}} , \label{eq:sC2x} \\
\mathcal{C}^{(n)}_{1,2} & \triangleq & \mathcal{C}^{(n)}  \cap  \big( \mathcal{L}^{(n)}_{V|Y_{(1)}} \big)^{\text{C}}  \cap   \big( \mathcal{L}^{(n)}_{V|Y_{(2)}} \big)^{\text{C}}. \label{eq:sC12x} 
\end{IEEEeqnarray}
These sets are graphically represented in \cite{alos2019polar} (Figure~2). 

Recall that $\tilde{A}_i[\mathcal{H}^{(n)}_V ]$, $i \in [1,L]$, is suitable for storing uniformly distributed random sequences, and $\tilde{A}_i [\mathcal{G}^{(n)} ]$ is suitable for storing information to be secured from the eavesdropper. Moreover, sets in \eqref{eq:sG0x}--\eqref{eq:sC12x} with subscript~$k \in [1,2]$ form $\mathcal{H}^{(n)}_V \cap \big( \mathcal{L}^{(n)}_{V|Y_{(k)}} \big)^{\text{C}}$, and the elements of $\tilde{A}_i^n$ corresponding to this set of indices are required by Receiver~$k$ to reliably reconstruct $\tilde{A}_i^n$ entirely by performing \gls*{sc} decoding. 
%

As mentioned before, the \gls*{pcs} must consider three different situations. In Situation~1 we have the condition $I(V;Z) \leq I(V;Y_{1}) \leq I(V;Y_{2})$. As seen in \cite{alos2019polar}, for $n$ sufficiently large this condition imposes the following restriction on the size of previous sets:  
\begin{IEEEeqnarray}{c}
\big| \mathcal{G}^{(n)}_1 \big|  - \big| \mathcal{C}^{(n)}_2 \big| \geq \big| \mathcal{G}^{(n)}_2 \big|  - \big| \mathcal{C}^{(n)}_1 \big| \geq \big| \mathcal{C}^{(n)}_{1,2} \big| - \big| \mathcal{G}^{(n)}_0 \big|; \label{eq:assumpRate1Impl2x1}
\end{IEEEeqnarray}
Similarly, Situation~2, where $I(V;Y_{(1)}) < I(V;Z) \leq I(V;Y_{2})$, imposes that
\begin{IEEEeqnarray}{c}
\big| \mathcal{G}^{(n)}_1 \big|  - \big| \mathcal{C}^{(n)}_2 \big| \geq \big| \mathcal{C}^{(n)}_{1,2} \big|  - \big| \mathcal{G}^{(n)}_0 \big| > \big| \mathcal{G}^{(n)}_{2} \big| - \big| \mathcal{C}^{(n)}_1 \big|; \label{eq:assumpRate1Impl2x2}
\end{IEEEeqnarray}
and Situation~3, where $I(V;Y_{(1)})  \leq I(V;Y_{2}) < I(V;Z)$, imposes that
\begin{IEEEeqnarray}{c}
\big| \mathcal{C}^{(n)}_{1,2} \big|  - \big| \mathcal{G}^{(n)}_0 \big| > \big| \mathcal{G}^{(n)}_{1} \big|  - \big| \mathcal{C}^{(n)}_2 \big| \geq \big| \mathcal{G}^{(n)}_{2} \big| - \big| \mathcal{C}^{(n)}_1 \big|. \label{eq:assumpRate1Impl2x3}
\end{IEEEeqnarray}
Thus, according to \eqref{eq:assumpRate1Impl2x1}--\eqref{eq:assumpRate1Impl2x3}, we must consider six different cases:
\begin{itemize}
\setlength\itemsep{-0.10em}
\item[A.] $| \mathcal{G}^{(n)}_1 | > | \mathcal{C}^{(n)}_2 | $, ${| \mathcal{G}^{(n)}_2 |  > | \mathcal{C}^{(n)}_1 |}$ and ${| \mathcal{G}^{(n)}_0 |  \geq | \mathcal{C}^{(n)}_{1,2} |} \quad$ (only for Situation~1),
\item[B.] $| \mathcal{G}^{(n)}_1 | > | \mathcal{C}^{(n)}_2 | $, ${| \mathcal{G}^{(n)}_2 |  > | \mathcal{C}^{(n)}_1 |}$ and ${| \mathcal{G}^{(n)}_0 |  < | \mathcal{C}^{(n)}_{1,2} |}\quad$ (for all situations),
\item[C.] $| \mathcal{G}^{(n)}_1 | \geq | \mathcal{C}^{(n)}_2 |$, $| \mathcal{G}^{(n)}_2 |  \leq | \mathcal{C}^{(n)}_1 |$ and $| \mathcal{G}^{(n)}_0 |  > | \mathcal{C}^{(n)}_{1,2} |\quad$ (only for Situations~1~and~2),
\item[D.] $| \mathcal{G}^{(n)}_1 | < | \mathcal{C}^{(n)}_2 | $, $| \mathcal{G}^{(n)}_2 |  < | \mathcal{C}^{(n)}_1 |$ and $| \mathcal{G}^{(n)}_0 |  > | \mathcal{C}^{(n)}_{1,2} |\quad$ (for all situations),
\item[E.] $| \mathcal{G}^{(n)}_1 | > | \mathcal{C}^{(n)}_2 | $, $| \mathcal{G}^{(n)}_2 |  < | \mathcal{C}^{(n)}_1 |$ and $| \mathcal{G}^{(n)}_0 |  < | \mathcal{C}^{(n)}_{1,2} |\quad$ (only for Situations~2~and~3),
\item[F.] $| \mathcal{G}^{(n)}_1 | < | \mathcal{C}^{(n)}_2 | $, $| \mathcal{G}^{(n)}_2 |  < | \mathcal{C}^{(n)}_1 |$ and $| \mathcal{G}^{(n)}_0 |  < | \mathcal{C}^{(n)}_{1,2} |\quad$ (only for Situation~3).
\end{itemize}

The inner-layer encoding process to achieve $\smash{(R_{S_{(1)}}^{\star k}, R_{S_{(2)}}^{\star k}, R_{W_{(1)}}^{\star k}, R_{W_{(2)}}^{\star k})}$, for any $k \in [1,2]$, in all cases is summarized in Algorithm~\ref{alg:genericencx}. For $i \in [1,L]$, let $W^{(V)}_{(k),i}$ be a uniformly distributed vector of length $| \mathcal{C}^{(n)} |$ that represents part of the private message intended to legitimate Receiver~$k$. The encoder forms $\tilde{A}_{i}\big[\mathcal{C}^{(n)}\big]$ by simply storing $W^{(V)}_{(k),i}$. Then, from $\tilde{A}_{i}\big[ \mathcal{C}^{(n)}\big]$, we define $\Psi_{i}^{(V)} \triangleq \tilde{A}_{i}\big[ \mathcal{C}_2^{(n)} \big]$, $\Gamma_{i}^{(V)} \triangleq  \tilde{A}_{i}\big[\mathcal{C}_{1,2}^{(n)} \big]$ and $\Theta_{i}^{(V)} \triangleq  \tilde{A}_{i}\big[\mathcal{C}_{1}^{(n)} \big]$.

\begin{algorithm}[t!]
\caption{Inner-layer encoding to achieve $(R_{S_{(1)}}^{\star k}, R_{S_{(2)}}^{\star k}, R_{W_{(1)}}^{\star k}, R_{W_{(2)}}^{\star k}) \subseteq \mathfrak{R}_{\text{MI-WTBC}}^{(k)}$}\label{alg:genericencx}
\begin{algorithmic}[1]
\Require Parts $W^{(V)}_{(k),1:L}$ and $S^{(V)}_{(k),1:L}$ of messages; and secret-key $\{\kappa_{{\Upsilon \Phi}_{(k)}}^{(V)}\}_{k=1}^2$ 
\State $\Psi^{(V)}_{0}$, $\Gamma^{(V)}_{0}$, $\smash{\bar{\Psi}^{(V)}_{0}}$, $\smash{\Pi^{(V)}_{(2),0}}$, $\Lambda^{(V)}_{0}$, $\bar{\Theta}^{(V)}_{L+1}$, $\bar{\Gamma}^{(V)}_{L+1} \leftarrow \varnothing$ \Comment For notation purposes 
\State $\tilde{A}_{1}[\mathcal{C}^{(n)}] \leftarrow W^{(V)}_{(k),1}$
\State $\Psi^{(V)}_{1} \leftarrow \tilde{A}_{1}[\mathcal{C}_{2}^{(n)}]$ \textbf{and} $\Gamma^{(V)}_{1} \leftarrow \tilde{A}_{1}[\mathcal{C}_{1,2}^{(n)}]$
\For{$i = 1$ \textbf{to} $L$}
\If{$i \neq L$}  
\State $\tilde{A}_{i+1}[\mathcal{C}^{(n)}] \leftarrow W^{(V)}_{(k),i+1}$ 
\State $\Psi^{(V)}_{i+1} \leftarrow \tilde{A}_{i+1}[\mathcal{C}_{2}^{(n)}]$ \textbf{and} $\Theta^{(V)}_{i+1} \leftarrow \tilde{A}_{i+1}[\mathcal{C}_{1}^{(n)}]$ \textbf{and} $\Gamma^{(V)}_{i+1} \leftarrow \tilde{A}_{i+1}[\mathcal{C}_{1,2}^{(n)}]$
\EndIf
\State $\tilde{A}_{i} \big[ \mathcal{G}^{(n)} \big]$, $\Pi^{(V)}_{(1),i}$, $\Pi^{(V)}_{(2),i}$, $\Lambda^{(V)}_{i}, \dots $
\State \hspace{0.5cm} $\dots \Delta_{(1),i+1}^{(V)}, \Delta_{(2),i-1}^{(V)} \leftarrow$ \texttt{form2\_A$_{\text{\texttt{G}}}$}$\big(i, S_{(k),i}^{(V)},\Theta^{(V)}_{i+1}, \Gamma^{(V)}_{i+1}, \Psi^{(V)}_{i-1}, \Gamma^{(V)}_{i-1}, \Pi^{(V)}_{(2),i-1}, \Lambda^{(V)}_{i-1}\big)$ 
\For{$j \in \big( \mathcal{H}^{(n)}_V \big)^{\text{C}}$}
\If{$j \in \big( \mathcal{H}^{(n)}_V \big)^{\text{C}}  \setminus \mathcal{L}^{(n)}_V$}
\State $\tilde{A}_i(j) \leftarrow p_{A(j)|A^{1:j-1}}  \big( \tilde{A}_i(j) \big| \tilde{A}_i^{1:j-1}  \big)$
\ElsIf{$j \in \mathcal{L}^{(n)}_V$}
\State $\tilde{A}_i(j) \leftarrow \argmax_{a \in \mathcal{V}} p_{A(j)|A^{1:j-1}} \big( a\big| \tilde{A}_i^{1:j-1}  \big)$
\EndIf
\EndFor
\State $\tilde{V}_i^n = \tilde{A}_i^n G_n$
\State $\Phi_{(k),i}^{(V)} \leftarrow \tilde{A}_i \big[ \big( \mathcal{H}_V^{(n)}  \big)^{\text{C}} \cap \big( \mathcal{L}_{V|Y_{(k)}}^{(n)} \big)^{\text{C}} \big]$ for $k \in [1,2]$
\State \textbf{if} $i=1$ \textbf{then} $\Upsilon_{(1)}^{(V)} \leftarrow \tilde{A}_1 \big[ \mathcal{H}_V^{(n)} \cap \big( \mathcal{L}_{V|Y_{(1)}}^{(n)} \big)^{\text{C}}\big]$
\State \textbf{if} $i=L$ \textbf{then} $\Upsilon_{(2)}^{(V)} \leftarrow \tilde{A}_L \big[ \mathcal{H}_V^{(n)} \cap \big( \mathcal{L}_{V|Y_{(2)}}^{(n)} \big)^{\text{C}}\big]$
\EndFor
\State Send $\Pi_{(1),1:L}$ and $\Delta_{(1),1:L}^{(V)}$ to the encoding responsible for the construction of $\tilde{T}_{(1)}^n$
\State Send $\Delta_{(2),1:L}^{(V)}$ to the encoding responsible for the construction of $\tilde{T}_{(2)}^n$
\State Send $\big( \Phi_{(k),i}^{(V)},\Upsilon_{(k)}^{(V)} \big) \oplus \kappa_{{\Upsilon \Phi}_{(k)}}^{(V)}$ to Receiver~$k \in [1,2]$
\State \textbf{return} $\tilde{V}_{1:L}^n$
\end{algorithmic}
\end{algorithm}

The function \texttt{form2\_A$_{\text{\texttt{G}}}$} in Algorithm~\ref{alg:genericencx} constructs $\tilde{A}_{1:L}\big[\mathcal{G}^{(n)}\big]$ and is explained in detail below. Then, given $\tilde{A}_i \big[ \mathcal{C}^{(n)} \cup \mathcal{G}^{(n)} \big]$, $i \in [1,L]$, the encoder forms the remaining entries of $\tilde{A}_i^n$ by using \gls*{sc} encoding: deterministic for $\tilde{A}_i \big[ \mathcal{L}_{V}^{(n)} \big]$, and random for $\tilde{A}_i \big[ ( \mathcal{H}^{(n)}_V )^{\text{C}} \setminus \mathcal{L}^{(n)}_V \big]$.

Also, from $\tilde{A}_i^n$, where $i \in [1,L]$, the encoder obtains $\Phi_{(k),i}^{(V)} \triangleq \tilde{A}_i \big[ \big( \mathcal{H}_V^{(n)}  \big)^{\text{C}} \cap \big( \mathcal{L}_{V|Y_{(k)}}^{(n)} \big)^{\text{C}} \big]$ for any $k \in [1,2]$. Moreover, from $\tilde{A}_1^n$ and $\tilde{A}_L^n$, it obtains $\Upsilon_{(1)}^{(V)} \triangleq  \tilde{A}_1 \big[ \mathcal{H}_V^{(n)} \cap ( \mathcal{L}_{V|Y_{(1)}}^{(n)} )^{\text{C}} \big]$ and $\Upsilon_{(2)}^{(V)} \triangleq \tilde{A}_L \big[ \mathcal{H}_V^{(n)} \cap ( \mathcal{L}_{V|Y_{(2)}}^{(n)} )^{\text{C}} \big]$, respectively. Recall that $\big( \Upsilon_{(k)}^{(V)}, \Phi_{(k),i}^{(V)} \big)$ is required by legitimate Receiver~$k$ to reliably estimate $\tilde{A}_{1:L}^n$. Hence, the transmitter additionally sends $\big( \Upsilon_{(k)}^{(V)}, \Phi_{(k),i}^{(V)} \big) \oplus \kappa_{{\Upsilon \Phi}_{(k)}}^{(V)}$ to legitimate Receiver~$k$, where $\kappa_{{\Upsilon \Phi}_{(k)}}^{(V)}$ is a uniformly distributed key with size $L  \big| \big( \mathcal{H}_V^{(n)}  \big)^{\text{C}} \cap \big(  \mathcal{L}_{V|Y_{(k)}}^{(n)} \big)^{\text{C}}  \big| +  \big| \mathcal{H}_V^{(n)} \cap \big( \mathcal{L}_{V|Y_{(k)}}^{(n)} \big)^{\text{C}} \big|$.

The function \texttt{form2\_A$_{\text{\texttt{G}}}$} is summarized in Algorithm~\ref{alg:formAx}. For any $i \in [1,L]$, this function stores part of the confidential message intended for Receiver~$k \in [1,2]$, namely $S^{(V)}_{(k),i}$, into $\tilde{A}_i \big[ \mathcal{G}^{(n)} \big]$, as well as different elements of $\tilde{A}_{i-1}^n$ ($i \in [2,L]$) and $\tilde{A}_{i+1}^n$ ($i \in [1,L-1]$) due to the chaining construction: recall that $\big[\Psi_i^{(V)},\Gamma_i^{(V)}\big] = \tilde{A}_i \big[\mathcal{C}^{(n)}_2 \cup \mathcal{C}^{(n)}_{1,2}\big]$ is required by Receiver~2 to reliably estimate $\tilde{A}_i^n$, while $\big[ \Theta_i^{(V)},\Gamma_i^{(V)}\big]= \tilde{A}_i \big[\mathcal{C}^{(n)}_1 \cup \mathcal{C}^{(n)}_{1,2} \big]$ is required by Receiver~1.

Notice in Algorithm~\ref{alg:formAx} that if the \gls*{pcs} operates to achieve the corner point of $\mathfrak{R}_{\text{MI-WTBC}}^{(1)}$ then, 
for $i \in [1,L-1]$, sequences $\Theta_{i+1}^{(V)}$ and $\Gamma_{i+1}^{(V)}$ are not repeated directly in $\tilde{A}_i\big[\mathcal{G}^{(n)} \big]$. Instead, the encoder repeats $\bar{\Theta}_{i+1}^{(V)} \triangleq \Theta_{i+1}^{(V)} \oplus \kappa_{\Theta}^{(V)}$ and $\bar{\Gamma}_{i+1}^{(V)}\triangleq \Gamma_{i+1}^{(V)} \oplus \kappa_{\Gamma}^{(V)}$, where $\kappa_{\Theta}^{(V)}$ and $\kappa_{\Gamma}^{(V)}$ are uniformly distributed keys with length $|\mathcal{C}_{1}^{(n)}|$ and $|\mathcal{C}_{1,2}^{(n)}|$ respectively that are privately shared between transmitter and both receivers. Otherwise, to achieve the corner point of $\mathfrak{R}_{\text{MI-WTBC}}^{(2)}$, instead of $\Psi_{i-1}^{(V)}$ and $\Gamma_{i-1}^{(V)}$, $i \in [2,L]$, the encoder repeats $\bar{\Psi}_{i-1}^{(V)} \triangleq \Psi_{i-1}^{(V)} \oplus \kappa_{\Psi}^{(V)}$ and $\bar{\Gamma}_{i-1}^{(V)} \triangleq \Gamma_{i-1}^{(V)} \oplus \kappa_{\Gamma}^{(V)}$, where $\kappa_{\Psi}^{(V)}$ is a distributed key with length $|\mathcal{C}_{2}^{(n)}|$. Since these keys are reused in all blocks, clearly their size become negligible in terms of rate for $L$ large enough.


As in \cite{alos2019polar}, based on the sets in \eqref{eq:sG0x}--\eqref{eq:sC12x}, let $\mathcal{R}^{(n)}_{1}$, $\mathcal{R}^{\prime (n)}_{1}$,  $\mathcal{R}^{(n)}_{2}$, $\mathcal{R}^{\prime (n)}_{2}$, $\mathcal{R}^{(n)}_{1,2}$, $\mathcal{R}^{\prime (n)}_{1,2}$, $\mathcal{I}^{(n)}$, $\mathcal{R}^{(n)}_{\text{S}}$ and $\mathcal{R}^{(n)}_{\Lambda}$ form an additional partition of $\mathcal{G}^{(n)}$. The definition of $\mathcal{R}^{(n)}_{1}$, $\mathcal{R}^{\prime (n)}_{1}$, $\mathcal{R}^{(n)}_{2}$, $\mathcal{R}^{\prime (n)}_{2}$, $\mathcal{R}^{(n)}_{1,2}$ and $\smash{\mathcal{R}^{\prime (n)}_{1,2}}$ will depend on the particular case (among A to F) and situation (among 1 to 3), as well as on the corner point the \gls*{pcs} must approach. Then, we define
\vspace*{-0.1cm}
\begin{align}
\mathcal{R}^{(n)}_{\text{S}} & \triangleq  \text{any subset of } \mathcal{G}^{(n)}_{1} \setminus \big( \mathcal{R}^{(n)}_{2} \cup  \mathcal{R}^{\prime (n)}_{2} \big) 
\text{ with size } \big| \mathcal{G}^{(n)}_{2} \setminus \big( \mathcal{R}^{(n)}_{1} \cup \mathcal{R}^{\prime (n)}_{1} \big) \big| .  \label{eq:ASetRsx}
\end{align}
Finally, the definition of  $\mathcal{I}^{(n)}$ and $\mathcal{R}_{\Lambda}^{(n)}$ depend on the corner point the \gls*{pcs} must achieve: if the corner point is $(R_{S_{(1)}}^{\star 1}, R_{S_{(2)}}^{\star 1}, R_{W_{(1)}}^{\star 1}, R_{W_{(2)}}^{\star 1}) \subset \mathfrak{R}_{\text{{MI-WTBC}}}^{(1)}$ then
\vspace*{-0.1cm}
\begin{align}
\mathcal{I}^{(n)} & \triangleq \big( \mathcal{G}^{(n)}_{0} \cup \mathcal{G}^{(n)}_{2} \big) \setminus   \big(  \mathcal{R}^{(n)}_{\text{1}} \cup  \mathcal{R}^{\prime (n)}_{1} \cup \mathcal{R}^{(n)}_{1,2} \cup \mathcal{R}^{\prime (n)}_{1,2} \big),   \label{eq:ASetIxk} \\
\mathcal{R}_{\Lambda}^{(n)} & \triangleq  \mathcal{G}_{1,2}^{(n)} \cup \big( \mathcal{G}^{(n)}_{1} \setminus \big( \mathcal{R}^{(n)}_{2} \cup \mathcal{R}^{\prime (n)}_{2} \cup \mathcal{R}^{(n)}_{\text{S}} \big) \big);  \label{eq:ASetLambdaxk}
\end{align}
and if it is $(R_{S_{(1)}}^{\star 2}, R_{S_{(2)}}^{\star 2}, R_{W_{(1)}}^{\star 2}, R_{W_{(2)}}^{\star 2}) \subset \mathfrak{R}_{\text{{MI-WTBC}}}^{(2)}$ then we define
\vspace*{-0.1cm}
\begin{align}
\mathcal{I}^{(n)} & \triangleq \big( \mathcal{G}^{(n)}_{0} \cup \mathcal{G}^{(n)}_{1} \cup \mathcal{G}^{(n)}_{2} \big) \setminus \big(  \mathcal{R}^{(n)}_{\text{1}} \cup  \mathcal{R}^{\prime (n)}_{1} \cup \mathcal{R}^{(n)}_{1,2} \cup \mathcal{R}^{\prime (n)}_{1,2} \cup \mathcal{R}^{(n)}_{2} \cup \mathcal{R}^{\prime (n)}_{2} \cup \mathcal{R}^{(n)}_{\text{S}} \big) \big),   \label{eq:ASetIxbk} \\
\mathcal{R}_{\Lambda}^{(n)} & \triangleq  \mathcal{G}_{1,2}^{(n)}. \label{eq:ASetLambdbx}
\end{align}
Note that $\mathcal{G}^{(n)}_{1} \setminus \big( \mathcal{R}^{(n)}_{2} \cup \mathcal{R}^{\prime (n)}_{2} \cup \mathcal{R}^{(n)}_{\text{S}} \big)$ will belong to $\mathcal{I}^{(n)}$ or $\mathcal{R}_{\Lambda}^{(n)}$ depending on whether the \gls*{pcs} approaches the first or the second corner point, respectively. Based on these sets, define
\vspace*{-0.1cm}
\begin{align}
\Lambda_{i}^{(V)} & \triangleq \tilde{A}_i \big[ \mathcal{R}_{\Lambda}^{(n)} \big], \label{eq:lambdax} \\
\Pi_{(2),i}^{(V)} & \triangleq \tilde{A}_i \big[ \mathcal{G}_{2}^{(n)} \cap \mathcal{I}^{(n)} \big], \label{eq:pi2x} \\
\Pi_{(1),i}^{(V)} & \triangleq \tilde{A}_i \big[ \mathcal{G}_{1}^{(n)} \cap \mathcal{I}^{(n)} \big]. \label{eq:pi1x}
\end{align}
Indeed, note that $\Pi_{(1),i}^{(V)} \neq \varnothing$ only when the \gls*{pcs} must approach the corner point of $\mathfrak{R}_{\text{{MI-WTBC}}}^{(2)}$. 
%

Also, let $S_{(k),i}^{(V)}$ be a uniform random sequence representing the part of the confidential message intended for Receiver~$k \in [1,2]$ that is carried in the inner-layer. Then, $S_{(k),1}^{(V)}$ has size $\big| \mathcal{I}^{(n)} \cup \mathcal{G}_{1,2}^{(n)} \cup \mathcal{G}_{1}^{(n)} \big|$; for $i \in [2,L-1]$, $S_{(k),i}^{(V)}$ has size $\big| \mathcal{I}^{(n)}\big|$; and $S_{(k),L}^{(V)}$ has size $\big| \mathcal{I}^{(n)} \cup \mathcal{G}_{2}^{(n)} \big|$.
%

Moreover, for $i \in [1,L]$ we write $\smash{\Psi_i^{(V)} \triangleq \big[ \Psi_{1,i}^{(V)}, \Psi_{2,i}^{(V)} , \Psi_{3,i}^{(V)} \big]}$ and $\smash{\bar{\Psi}_i^{(V)} \triangleq \big[ \bar{\Psi}_{1,i}^{(V)}, \bar{\Psi}_{2,i}^{(V)} , \bar{\Psi}_{3,i}^{(V)} \big]}$, we write $\smash{\Theta_i^{(V)}  \triangleq  \big[\Theta_{1,i}^{(V)}, \Theta_{2,i}^{(V)}, \Theta_{3,i}^{(V)}  \big]}$ and $\smash{\bar{\Theta}_i^{(V)}  \triangleq  \big[ \bar{\Theta}_{1,i}^{(V)}, \bar{\Theta}_{2,i}^{(V)},\bar{\Theta}_{3,i}^{(V)}  \big]}$, and we write $\smash{\Gamma_i^{(V)} \triangleq \big[ \Gamma_{1,i}^{(V)}, \Gamma_{2,i}^{(V)} , \Gamma_{3,i}^{(V)} \big]}$ and $\smash{\bar{\Gamma}_i^{(V)}  \triangleq  \big[ \bar{\Gamma}_{1,i}^{(V)}, \bar{\Gamma}_{2,i}^{(V)} ,\bar{\Gamma}_{3,i}^{(V)} \big]}$, where we will define $\smash{\Psi_{p,i}^{(V)}}$, $\smash{\bar{\Psi}^{(V)}_{p,i}}$, $\smash{\Theta^{(V)}_{p,i}}$, $\smash{\bar{\Theta}^{(V)}_{p,i}}$, $\smash{\Gamma^{(V)}_{p,i}}$ and $\smash{\bar{\Gamma}^{(V)}_{p,i}}$, for any $p \in [1,2,3]$, accordingly in each case. For notation purposes, let 
\begin{align}
\Delta_{(1),i}^{(V)} & \triangleq \big[ \bar{\Theta}^{(V)}_{3,i}, \bar{\Gamma}^{(V)}_{3,i} \big] \quad \text{and} \quad \Delta_{(2),i}^{(V)} \triangleq \big[ \Psi^{(V)}_{3,i}, \Gamma^{(V)}_{3,i} \big]  \quad \text{(if }k=1\text{)}, \label{eq:delta1x} \\
\Delta_{(1),i}^{(V)} & \triangleq \big[ \Theta^{(V)}_{3,i}, \Gamma^{(V)}_{3,i} \big] \quad \text{and} \quad \Delta_{(2),i}^{(V)} \triangleq \big[ \bar{\Psi}^{(V)}_{3,i}, \bar{\Gamma}^{(V)}_{3,i} \big] \quad \text{(if }k=2\text{)}. \label{eq:delta2x}
\end{align}
According to Algorithm~\ref{alg:genericencx}, recall that $\smash{\Pi_{(1),1:L}^{(V)}}$ and $\smash{\Delta_{(1),1:L}^{(V)}}$ are sent to the outer-layer associated to Receiver~1, while $\smash{\Delta_{(2),1:L}^{(V)}}$ is sent to the outer-layer associated to Receiver~2.

\begin{algorithm}[h!]
\caption{Function \texttt{form2\_A$_{\text{\texttt{G}}}$} to achieve $(R_{S_{(1)}}^{\star k}, R_{S_{(2)}}^{\star k}, R_{W_{(1)}}^{\star k}, R_{W_{(2)}}^{\star k}) \subseteq \mathfrak{R}_{\text{MI-WTBC}}^{(k)}$}\label{alg:formAx}
\begin{algorithmic}[1]
\Require $i$, $S_{(k),i}^{(V)}$, $\Theta^{(V)}_{i+1}$, $\Gamma^{(V)}_{i+1}$, $\Psi^{(V)}_{i-1}$, $\Gamma^{(V)}_{i-1}$, $\Pi^{(V)}_{(2),i-1}$, $\Lambda^{(V)}_{i-1}$; secret-keys $\kappa_{\Theta}^{(V)}$, $\kappa_{\Psi}^{(V)}$ and $\kappa_{\Gamma}^{(V)}$
\State Define $\mathcal{R}^{(n)}_{1}$, $\mathcal{R}^{\prime (n)}_{1}$,  $\mathcal{R}^{(n)}_{2}$, $\mathcal{R}^{\prime (n)}_{2}$, $\mathcal{R}^{(n)}_{1,2}$, $\mathcal{R}^{\prime (n)}_{1,2}$, $\mathcal{I}^{(n)}$, $\mathcal{R}^{(n)}_{\text{S}}$, $\mathcal{R}^{(n)}_{\Lambda}$ 
\State \textbf{if} $i=1$ \textbf{then} $\tilde{A}_{1}[\mathcal{I}^{(n)} \cup \mathcal{G}^{(n)}_{1} \cup \mathcal{G}^{(n)}_{1,2}] \leftarrow S_{(k),1}^{(V)}$
\State \textbf{if} $i \in [2,L-1]$ \textbf{then} $\tilde{A}_{i}[\mathcal{I}^{(n)}] \leftarrow S_{(k),i}^{(V)}$
\State \textbf{if} $i = L$ \textbf{then} $\tilde{A}_{L}[\mathcal{I}^{(n)} \cup \mathcal{G}^{(n)}_{2}] \leftarrow S_{(k),L}^{(V)}$
\State $\bar{\Psi}^{(V)}_{i-1} \leftarrow \Psi^{(V)}_{i-1} \oplus \kappa^{(V)}_{\Psi}$ \textbf{and} $\bar{\Gamma}^{(V)}_{i-1} \leftarrow \Gamma^{(V)}_{i-1} \oplus \kappa^{(V)}_{\Gamma}$ 
\State $\bar{\Theta}^{(V)}_{i+1} \leftarrow \Theta^{(V)}_{i+1} \oplus \kappa^{(V)}_{\Theta}$ \textbf{and} $\bar{\Gamma}^{(V)}_{i+1} \leftarrow \Gamma^{(V)}_{i+1} \oplus \kappa^{(V)}_{\Gamma}$
\State \textbf{if} $k=1$ \textbf{then} $\tilde{A}_{i}[\mathcal{R}_{1,2}^{(n)}] \leftarrow \Gamma^{(V)}_{1,i-1} \oplus \bar{\Gamma}^{(V)}_{1,i+1}$ \textbf{else} $\tilde{A}_{i}[\mathcal{R}_{1,2}^{(n)}] \leftarrow \bar{\Gamma}^{(V)}_{1,i-1} \oplus \Gamma^{(V)}_{1,i+1}$ 
\State \textbf{if} $k=1$ \textbf{then} $\tilde{A}_{i}[\mathcal{R}_{1,2}^{\prime (n)}] \leftarrow \Psi^{(V)}_{2,i-1} \oplus \bar{\Theta}^{(V)}_{2,i+1}$ \textbf{else}  $\tilde{A}_{i}[\mathcal{R}_{1,2}^{\prime (n)}] \leftarrow \bar{\Psi}^{(V)}_{2,i-1} \oplus \Theta^{(V)}_{2,i+1}$ 
\If{$i \in [1,L-1]$}
\State \textbf{if} $k=1$ \textbf{then} $\tilde{A}_{i}[\mathcal{R}_{1}^{(n)}] \leftarrow \bar{\Theta}^{(V)}_{1,i+1}$ \textbf{else} $\tilde{A}_{i}[\mathcal{R}_{1}^{(n)}] \leftarrow \Theta^{(V)}_{1,i+1}$
\State \textbf{if} $k=1$ \textbf{then} $\tilde{A}_{i}[\mathcal{R}_{1}^{\prime (n)}] \leftarrow \bar{\Gamma}^{(V)}_{2,i+1}$ \textbf{else} $\tilde{A}_{i}[\mathcal{R}_{1}^{\prime (n)}] \leftarrow \Gamma^{(V)}_{2,i+1}$
\EndIf
\If{$i \in [2,L]$}
\State \textbf{if} $k=1$ \textbf{then} $\tilde{A}_{i}[\mathcal{R}_{2}^{(n)}] \leftarrow \Psi^{(V)}_{1,i-1}$ \textbf{else} $\tilde{A}_{i}[\mathcal{R}_{2}^{(n)}] \leftarrow \bar{\Psi}^{(V)}_{1,i-1}$
\State \textbf{if} $k=1$ \textbf{then} $\tilde{A}_{i}[\mathcal{R}_{2}^{\prime (n)}] \leftarrow \Gamma^{(V)}_{2,i-1}$ \textbf{else} $\tilde{A}_{i}[\mathcal{R}_{2}^{\prime (n)}] \leftarrow \bar{\Gamma}^{(V)}_{2,i-1}$
\State $\tilde{A}_{i}[\mathcal{R}_{\text{S}}^{(n)}] \leftarrow \Pi^{(V)}_{(2),i-1}$
\State $\tilde{A}_{i}[\mathcal{R}_{\Lambda}^{(n)}] \leftarrow \Lambda^{(V)}_{i-1}$
\EndIf
\State $\Pi_{(1),i}^{(V)} \leftarrow \tilde{A}_i [\mathcal{I}^{(n)} \cap \mathcal{G}^{(n)}_{1}]$ \textbf{and} $\Pi_{(2),i}^{(V)} \leftarrow \tilde{A}_i [\mathcal{I}^{(n)} \cap \mathcal{G}^{(n)}_{2}]$
\State $\Lambda_i^{(V)} \leftarrow \tilde{A}_i [\mathcal{R}^{(n)}_{\Lambda}]$
\State \textbf{if} $k=1$ \textbf{then} $\Delta^{(V)}_{(1),i} \leftarrow \big(\bar{\Theta}^{(V)}_{3,i}, \bar{\Gamma}^{(V)}_{3,i} \big)$ \textbf{else} $\Delta^{(V)}_{(1),i} \leftarrow \big(\Theta^{(V)}_{3,i}, \Gamma^{(V)}_{3,i} \big)$
\State \textbf{if} $k=1$ \textbf{then} $\Delta^{(V)}_{(2),i} \leftarrow \big( \Psi^{(V)}_{3,i}, \Gamma^{(V)}_{3,i} \big)$ \textbf{else} $\Delta^{(V)}_{(2),i} \leftarrow \big( \bar{\Psi}^{(V)}_{3,i}, \bar{\Gamma}^{(V)}_{3,i} \big)$
\State \textbf{return} $\tilde{A}_{i} \big[ \mathcal{G}^{(n)} \big]$, $\Pi_{(1),i}^{(V)}$, $\Pi_{(2),i}^{(V)}$, $\Lambda_i^{(V)}$, $\Delta^{(V)}_{(1),i}$ and $\Delta^{(V)}_{(2),i}$
\end{algorithmic}
\end{algorithm}

\vspace{0.15cm}
\noindent \textbf{Case A when} $\bm{I(V;Z)\leq I(V;Y_{(1)}) \leq I(V;Y_{(2)})}$

In Case A, we have $| \mathcal{G}^{(n)}_1 | > | \mathcal{C}^{(n)}_2 |$, ${| \mathcal{G}^{(n)}_2 |  > | \mathcal{C}^{(n)}_1 |}$ and ${| \mathcal{G}^{(n)}_0 |  \geq | \mathcal{C}^{(n)}_{1,2} |}$. 
\begin{enumerate}
\item \emph{Achievability of} $(R_{S_{(1)}}^{\star 1}, R_{S_{(2)}}^{\star 1}, R_{W_{(1)}}^{\star 1}, R_{W_{(2)}}^{\star 1}) \subset \mathfrak{R}_{\text{MI-WTBC}}^{(1)}$.
In this case, the construction of $\tilde{A}_{1:L}\big[ \mathcal{G}^{(n)} \big]$ is the same as the one in \cite{alos2019polar} (Section~IV.B.1, Case A). Hence, define
\begin{align}
\mathcal{R}^{(n)}_{1} & \triangleq  \text{any subset of } \mathcal{G}^{(n)}_2 \text{ with size } \big| \mathcal{C}^{(n)}_1 \big|, \nonumber \\ 
\mathcal{R}^{(n)}_{2} & \triangleq  \text{any subset of } \mathcal{G}^{(n)}_{1} \text{ with size } \big| \mathcal{C}^{(n)}_2 \big|, \nonumber \\ 
\mathcal{R}^{(n)}_{1,2} & \triangleq \text{any subset of } \mathcal{G}^{(n)}_0 \text{ with size } \big| \mathcal{C}^{(n)}_{1,2} \big|, \nonumber 
\end{align}
and $\mathcal{R}^{\prime(n)}_{1} = \mathcal{R}^{\prime(n)}_{2} = \mathcal{R}^{\prime(n)}_{1,2} \triangleq \emptyset$. Therefore, according to \eqref{eq:ASetRsx}--\eqref{eq:ASetLambdaxk}, we have
\begin{align}
\mathcal{R}^{(n)}_{\text{S}} & =  \text{any subset of } \mathcal{G}^{(n)}_{1} \setminus  \mathcal{R}^{(n)}_{2} \text{ with size } \big| \mathcal{G}^{(n)}_{2} \big| - \big| \mathcal{C}^{(n)}_{1}  \big|, \nonumber  \\
\mathcal{I}^{(n)} & = \big( \mathcal{G}^{(n)}_{0} \cup \mathcal{G}^{(n)}_{2} \big) \setminus   \big(  \mathcal{R}^{(n)}_{\text{1}} \cup \mathcal{R}^{(n)}_{1,2}  \big),  \nonumber \\
\mathcal{R}_{\Lambda}^{(n)} & =  \mathcal{G}_{1,2}^{(n)} \cup \big( \mathcal{G}^{(n)}_{1} \setminus \big( \mathcal{R}^{(n)}_{2} \cup \mathcal{R}^{(n)}_{\text{S}} \big) \big).\nonumber 
\end{align}
From condition \eqref{eq:assumpRate1Impl2x1}, all previous sets exist. Also, for $i \in [1,L]$, define $\Psi_{1,i}^{(V)} \triangleq \Psi_{i}^{(V)}$, $\Gamma_{1,i}^{(V)} \triangleq \Gamma_{i}^{(V)}$, $\bar{\Theta}_{1,i}^{(V)} \triangleq \bar{\Theta}_{i}^{(V)}$, $\bar{\Gamma}_{1,i}^{(V)} \triangleq \bar{\Gamma}_{i}^{(V)}$ and $\Psi_{p,i}^{(V)} = \Gamma_{p,i}^{(V)} = \bar{\Theta}_{p,i}^{(V)} = \bar{\Gamma}_{p,i}^{(V)} \triangleq \varnothing$, where $p \in [2,3]$. Then, according to \eqref{eq:pi2x}--\eqref{eq:delta2x}, for $i \in [1,L]$ we have $\Pi_{(2),i}^{(V)} = \tilde{A}_{i}\big[ \mathcal{I}^{(n)} \cap \mathcal{G}_2^{(n)} \big]$ with size $\big| \mathcal{G}^{(n)}_{2} \big| - \big| \mathcal{C}^{(n)}_{1}  \big|$, $\Pi_{(1),i}^{(V)} = \tilde{A}_{i}\big[ \mathcal{I}^{(n)} \cap \mathcal{G}_1^{(n)} \big] = \varnothing$, and $\Delta_{(1),i}^{(V)} = \Delta_{(2),i}^{(V)} = \varnothing$.

According to Algorithm~\ref{alg:formAx}, recall that, for $i \in [2,L]$, the chaining construction repeats $\Psi_{i-1}^{(V)}$ and $\Gamma_{i-1}^{(V)}$ entirely in $\tilde{A}_{i}\big[ \mathcal{R}_{2}^{(n)} \big] \subseteq \tilde{A}_{i}\big[ \mathcal{G}_{1}^{(n)} \big]$ and $\tilde{A}_{i}\big[ \mathcal{R}_{1,2}^{(n)} \big] \subseteq \tilde{A}_{i}\big[ \mathcal{G}_{0}^{(n)} \big]$ respectively. Also, for $i \in [1,L-1]$, it repeats the sequences $\bar{\Theta}_{i+1}^{(V)}$ and $\bar{\Gamma}_{i+1}^{(V)}$ in $\tilde{A}_{i}\big[ \mathcal{R}_{1}^{(n)} \big] \subseteq \tilde{A}_{i}\big[ \mathcal{G}_{2}^{(n)} \big]$ and $\tilde{A}_{i}\big[ \mathcal{G}_{0}^{(n)} \big]$ respectively. Indeed, for $i \in [2,L-1]$, recall that the encoder stores $\Gamma_{i-1}^{(V)} \oplus \bar{\Gamma}_{i+1}^{(V)}$ in $\tilde{A}_i\big[\mathcal{R}_{1,2}^{(n)}\big]$. The inner-layer carries confidential information $S_{(1),1:L}^{(V)}$ intended for Receiver~1, and for $i \in [2,L]$ the encoder repeats $\Pi_{(2),i-1}^{(V)}$ in $\tilde{A}_{i}[\mathcal{R}_{\text{S}}^{(n)}]$. Finally, for $i \in [2,L]$, it repeats $\Lambda_{i-1}^{(V)}$ in $\tilde{A}_{i}[\mathcal{R}_{\Lambda}^{(n)}]$ and, therefore, notice that $\Lambda_{1}^{(V)}$ is replicated in all blocks. This particular encoding procedure is graphically represented in \cite{alos2019polar}~(Figure~3). 

\item \emph{Achievability of} $(R_{S_{(1)}}^{\star 2}, R_{S_{(2)}}^{\star 2}, R_{W_{(1)}}^{\star 2}, R_{W_{(2)}}^{\star 2}) \subset \mathfrak{R}_{\text{MI-WTBC}}^{(2)}$.
Define $\mathcal{R}^{(n)}_{1}$, $\mathcal{R}^{(n)}_{2}$, $\mathcal{R}^{(n)}_{1,2}$, $\mathcal{R}^{\prime (n)}_{1}$, $\mathcal{R}^{\prime (n)}_{2}$ and $\mathcal{R}^{\prime (n)}_{1,2}$ as for the previous corner point. According to \eqref{eq:ASetRsx}, \eqref{eq:ASetIxbk} and \eqref{eq:ASetLambdbx}: 
\begin{align}
\mathcal{R}^{(n)}_{\text{S}} & =  \text{any subset of } \mathcal{G}^{(n)}_{1} \setminus  \mathcal{R}^{(n)}_{2} \text{ with size } \big| \mathcal{G}^{(n)}_{2} \big| - \big| \mathcal{C}^{(n)}_{1}  \big|, \nonumber \\
\mathcal{I}^{(n)} & = \big( \mathcal{G}^{(n)}_{0} \cup \mathcal{G}^{(n)}_{1} \cup \mathcal{G}^{(n)}_{2} \big) \setminus   \big(  \mathcal{R}^{(n)}_{\text{1}} \cup \mathcal{R}^{(n)}_{\text{2}} \cup \mathcal{R}^{(n)}_{1,2}  \cup \mathcal{R}^{(n)}_{\text{S}} \big),  \nonumber \\
\mathcal{R}_{\Lambda}^{(n)} & =  \mathcal{G}_{1,2}^{(n)}. \nonumber
\end{align}
Also, for $i \in [1,L]$, we define $\bar{\Psi}_{1,i}^{(V)} \triangleq \bar{\Psi}_{i}^{(V)}$, $\bar{\Gamma}_{1,i}^{(V)} \triangleq \bar{\Gamma}_{i}^{(V)}$, $\Theta_{1,i}^{(V)} \triangleq \Theta_{i}^{(V)}$, $\Gamma_{1,i}^{(V)} \triangleq \Gamma_{i}^{(V)}$ and, therefore, $\bar{\Psi}_{p,i}^{(V)} = \bar{\Gamma}_{p,i}^{(V)} = \Theta_{p,i}^{(V)} = \Gamma_{p,i}^{(V)} \triangleq \varnothing$, where $p \in [2,3]$. Then, according to \eqref{eq:pi2x}--\eqref{eq:delta2x}, for $i \in [1,L]$ we have $\Pi_{(2),i}^{(V)} = \tilde{A}_{i}\big[ \mathcal{I}^{(n)} \cap \mathcal{G}_2^{(n)} \big]$ with size $\big| \mathcal{G}^{(n)}_{2} \big| - \big| \mathcal{C}^{(n)}_{1}  \big|$, $\Pi_{(1),i}^{(V)} = \tilde{A}_{i}\big[ \mathcal{I}^{(n)} \cap \mathcal{G}_1^{(n)} \big]$ with size $\big| \mathcal{G}^{(n)}_{1} \big| - \big| \mathcal{C}^{(n)}_{2} \big| - \big( \big| \mathcal{G}^{(n)}_{2} \big| - \big| \mathcal{C}^{(n)}_{1} \big| \big)$, and $\Delta_{(1),i}^{(V)} = \Delta_{(2),i}^{(V)} = \varnothing$.

According to Algorithm~\ref{alg:formAx}, now the inner-layer carries confidential information $S_{(2),1:L}^{(V)}$ intended for Receiver~2. Indeed, for $i \in [1,L]$, $\tilde{A}_i\big[ \mathcal{G}_{1}^{(n)} \setminus \big( \mathcal{R}^{(n)}_{2} \cup \mathcal{R}^{(n)}_{\text{S}} \big)\big]$, which previously contained part of $\Lambda^{(V)}_i$, now store part of $S_{(2),i}^{(V)}$. As before, for $i \in [2,L]$, $\Pi_{(2),i-1}^{(V)}$ is repeated in $\tilde{A}_{i}[\mathcal{R}_{\text{S}}^{(n)}]$ and $\Lambda_{i-1}^{(V)}$ in $\tilde{A}_{i}[\mathcal{R}_{\Lambda}^{(n)}]$. Now, as will be seen in Section~\ref{sec:PCSx1ol}, for $i \in [1,L-1]$ the sequence $\Pi_{(1),i+1}^{(V)}$ will be repeated in the outer-layer $\tilde{T}_{(1),i}$ associated to Receiver~1. This particular encoding procedure is graphically represented in Figure~\ref{fig:EncCasAx}. 
\end{enumerate}

\begin{figure}[h!]
\centering
\begin{overpic}[width=0.9\linewidth]{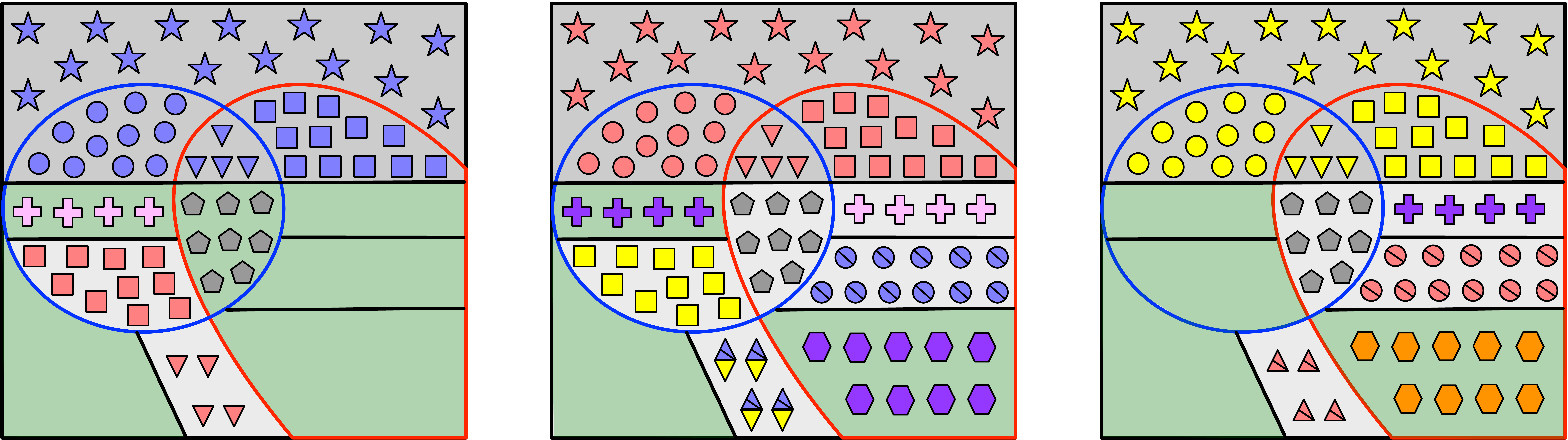}
\put (10.5,-3) {\small Block 1}
\put (45.75,-3) {\small Block 2}
\put (80.75,-3) {\small Block 3}
\end{overpic}
\vspace{0.5cm}
\caption{\setstretch{1.35} Case A when $\smash{I(V;Z) \leq I(V;Y_{(1)}) \leq I(V;Y_{(2)})}$: inner-layer encoding to achieve the corner point of $\smash{\mathfrak{R}_{\text{MI-WTBC}}^{(2)}}$ that leads to the construction of $\smash{\tilde{A}_{1:L}[ \mathcal{H}_V^{(n)} ]}$ when $L=3$. Consider Block~2: the sets $\smash{\mathcal{R}^{(n)}_{1}}$, $\smash{\mathcal{R}^{(n)}_{2}}$, $\smash{\mathcal{R}^{(n)}_{1,2}}$, $\smash{\mathcal{R}^{(n)}_{\text{S}}}$ and $\smash{\mathcal{R}^{(n)}_{\Lambda}}$ are those areas filled with yellow squares, blue circles, blue and yellow diamonds, pink crosses, and gray pentagons, respectively; and $\mathcal{I}^{(n)}$ is the green filled area. At Block~$i \in [1,L]$, $\smash{W_{(2),i}^{(V)}}$ is represented by symbols of the same color (e.g., red symbols at Block~2), and $\smash{\Theta_i^{(V)}}$, $\smash{\Psi_i^{(V)}}$ and $\smash{\Gamma_i^{(V)}}$ are represented by squares, circles and triangles respectively. Also, $\smash{\bar{\Psi}_i^{(V)}}$ and $\smash{\bar{\Gamma}_i^{(V)}}$ are denoted by circles and triangles, respectively, with a line through them. At Block~$i \in [2,L-1]$, the diamonds denote $\smash{\bar{\Gamma}_{i-1}^{(V)} \oplus \Gamma_{i+1}^{(V)}}$. In Block~$i \in [1,L]$, $\smash{S_{(2),i}^{(V)}}$ is stored into those entries whose indices belong to the green filled area. For $i \in [2,L]$, $\smash{\Pi_{(2),i-1}^{(V)}}$ is denoted by crosses (e.g., purple crosses at Block~2), and is repeated in $\tilde{A}_{i}[ \mathcal{R}^{(n)}_{\text{S}} ]$. For $i \in [1,L-1]$, the sequence $\smash{\Pi_{(1),i+1}^{(V)}}$ is denoted by hexagons, and it will be send to the outer-layer $\smash{\tilde{T}_{(1),i}}$ associated to Receiver~1. At Block~1, $\smash{\Lambda_1^{(V)}}$ is denoted by gray pentagons, and is repeated in all blocks. Finally, $\smash{\Upsilon_{(1)}^{(V)}}$ and $\smash{\Upsilon_{(2)}^{(V)}}$ are those entries inside the red curve at Block~1 and the blue curve at Block~$L$,~respectively.}\label{fig:EncCasAx} 
\end{figure}

\vspace{0.15cm}
\noindent \textbf{Case B when} $\bm{I(V;Z)\leq I(V;Y_{(1)}) \leq I(V;Y_{(2)})}$ 

In Case~B, we have $| \mathcal{G}^{(n)}_1 | > | \mathcal{C}^{(n)}_2 |$, ${| \mathcal{G}^{(n)}_2 |  > | \mathcal{C}^{(n)}_1 |}$ and ${| \mathcal{G}^{(n)}_0 |  < | \mathcal{C}^{(n)}_{1,2} |}$.

\begin{enumerate}
\item \emph{Achievability of} $(R_{S_{(1)}}^{\star 1}, R_{S_{(2)}}^{\star 1}, R_{W_{(1)}}^{\star 1}, R_{W_{(2)}}^{\star 1}) \subset \mathfrak{R}_{\text{MI-WTBC}}^{(1)}$.
In this case, the construction of $\tilde{A}_{1:L}\big[ \mathcal{G}^{(n)} \big]$ is the same as the one described in \cite{alos2019polar}~(Section~IV.B.2, Case B). Now, since $\big| \mathcal{G}^{(n)}_0 \big|  < \big| \mathcal{C}^{(n)}_{1,2} \big|$, only a part of $\Gamma_{i-1}^{(V)}$ ($i \in [2,L]$) and $\bar{\Gamma}_{i+1}^{(V)}$ ($i \in [1,L-1]$) can be repeated in $\tilde{A}_i \big[ \mathcal{G}_0^{(n)} \big]$. Thus, define
\begin{align}
\mathcal{R}^{(n)}_{1} & \triangleq  \text{any subset of } \mathcal{G}^{(n)}_2 \text{ with size } \big| \mathcal{C}^{(n)}_1 \big|,  \nonumber \\
\mathcal{R}^{(n)}_{2} & \triangleq  \text{any subset of } \mathcal{G}^{(n)}_{1} \text{ with size } \big| \mathcal{C}^{(n)}_2 \big|, \nonumber \\
\mathcal{R}^{(n)}_{1,2} & \triangleq \mathcal{G}^{(n)}_0, \nonumber \\
\mathcal{R}^{\prime (n)}_{1} & \triangleq \text{any subset of } \mathcal{G}^{(n)}_2 \setminus \mathcal{R}^{(n)}_{1} \text{ with size } \big| \mathcal{C}^{(n)}_{1,2} \big| - \big| \mathcal{G}^{(n)}_{0} \big|,  \nonumber \\ 
\mathcal{R}^{\prime (n)}_{2} & \triangleq \text{any subset of } \mathcal{G}^{(n)}_1 \setminus \mathcal{R}^{(n)}_{2} \text{ with size } \big| \mathcal{C}^{(n)}_{1,2} \big| - \big| \mathcal{G}^{(n)}_{0} \big|, \nonumber 
\end{align}
and $\mathcal{R}^{\prime(n)}_{1,2} \triangleq \emptyset$. Therefore, according to \eqref{eq:ASetRsx}--\eqref{eq:ASetLambdaxk}, we have
\begin{align}
\mathcal{R}^{(n)}_{\text{S}} & =  \text{any subset of } \mathcal{G}^{(n)}_{1} \setminus 	\big( \mathcal{R}^{(n)}_{2} \cup \mathcal{R}^{\prime (n)}_{2} \big) \nonumber  \\
& \quad \text{ with size } \big| \mathcal{G}^{(n)}_{2} \big| - \big| \mathcal{C}^{(n)}_{1}  \big| - \big( \big| \mathcal{C}^{(n)}_{1,2} \big| - \big| \mathcal{G}^{(n)}_{0}  \big| \big) , \nonumber \\
\mathcal{I}^{(n)} & = \mathcal{G}^{(n)}_{2}  \setminus   \big(  \mathcal{R}^{(n)}_{\text{1}} \cup \mathcal{R}^{\prime (n)}_{\text{1}}  \big),  \nonumber \\
\mathcal{R}_{\Lambda}^{(n)} & \triangleq  \mathcal{G}_{1,2}^{(n)} \cup \big( \mathcal{G}^{(n)}_{1} \setminus \big( \mathcal{R}^{(n)}_{2} \cup \mathcal{R}^{\prime (n)}_{2} \cup \mathcal{R}^{(n)}_{\text{S}} \big) \big). \nonumber
\end{align}
From \eqref{eq:assumpRate1Impl2x1}, all previous sets exist. For $i \in [1,L]$, we define $\Psi_{1,i}^{(V)} \triangleq \Psi_{i}^{(V)}$, $\bar{\Theta}_{1,i}^{(V)} \triangleq \bar{\Theta}_{i}^{(V)}$ and $\Psi_{p,i}^{(V)} = \bar{\Theta}_{p,i}^{(V)} \triangleq  \varnothing$ for $p \in [2,3]$; and $\Gamma_{1,i}^{(V)}$ and $\bar{\Gamma}_{1,i}^{(V)}$ as any part of $\Gamma_{i}^{(V)}$ and $\bar{\Gamma}_{i}^{(V)}$, respectively, with size $\big| \mathcal{G}^{(n)}_{0} \big|$, $\Gamma_{2,i}^{(V)}$ and $\bar{\Gamma}_{2,i}^{(V)}$ as the remaining parts with size $\big| \mathcal{C}^{(n)}_{1,2} \big| - \big| \mathcal{G}^{(n)}_{0} \big|$, and $\Gamma_{3,i}^{(V)} = \bar{\Gamma}_{3,i}^{(V)} \triangleq \varnothing$. From \eqref{eq:pi2x}--\eqref{eq:delta2x}, for $i \in [1,L]$ we have $\Pi_{(2),i}^{(V)} = \tilde{A}_{i}\big[ \mathcal{I}^{(n)} \cap \mathcal{G}_2^{(n)} \big]$ with size $\big| \mathcal{G}^{(n)}_{2} \big| - \big| \mathcal{C}^{(n)}_{1}  \big| - \big( \big| \mathcal{C}^{(n)}_{1,2} \big| - \big| \mathcal{G}^{(n)}_{0}  \big| \big)$, $\Pi_{(1),i}^{(V)} = \tilde{A}_{i}\big[ \mathcal{I}^{(n)} \cap \mathcal{G}_1^{(n)} \big] = \varnothing$, and we have $\Delta_{(1),i}^{(V)} = \Delta_{(2),i}^{(V)} = \varnothing$.

According to Algorithm~\ref{alg:formAx}, for $i \in [1,L]$, $\Gamma_{1,i-1}^{(V)} \oplus \bar{\Gamma}_{1,i+1}^{(V)}$ is repeated in $\tilde{A}_i \big[\mathcal{R}_{1,2}^{(n)} \big]$, where $\Gamma_{1,0}^{(V)} = \bar{\Gamma}_{1,L}^{(V)} = \varnothing$. On the other hand, now $\Gamma_{2,i-1}^{(V)}$ ($i \in [2,L]$) and $\bar{\Gamma}_{2,i+1}^{(V)}$ ($i \in [1,L-1]$) are repeated in $\tilde{A}_i \big[\mathcal{R}_{2}^{\prime (n)} \big]$ and $\tilde{A}_i \big[\mathcal{R}_{1}^{\prime (n)} \big]$ respectively. The inner-layer carries confidential information $S_{(1),1:L}^{(V)}$ intended for Receiver~1, and for $i \in [2,L]$ the encoder repeats $\Pi_{(2),i-1}^{(V)}$ in $\tilde{A}_{i}[\mathcal{R}_{\text{S}}^{(n)}]$. Indeed, since $\mathcal{I}^{(n)} \subseteq \mathcal{G}_2^{(n)}$, we have $\Pi_{(2),i}^{(V)} = S_{(1),i}^{(V)}$ for $i \in [1,L]$. Finally, for $i \in [2,L]$, the encoder repeats $\Lambda_{i-1}^{(V)}$ in $\tilde{A}_{i}[\mathcal{R}_{\Lambda}^{(n)}]$ and, hence, $\Lambda_{1}^{(V)}$ is replicated in all blocks. This particular encoding procedure is graphically represented in \cite{alos2019polar}~(Figure~4).

\item \emph{Achievability of} $(R_{S_{(1)}}^{\star 2}, R_{S_{(2)}}^{\star 2}, R_{W_{(1)}}^{\star 2}, R_{W_{(2)}}^{\star 2}) \subset \mathfrak{R}_{\text{MI-WTBC}}^{(2)}$.
Define $\mathcal{R}^{(n)}_{1}$, $\mathcal{R}^{(n)}_{2}$, $\mathcal{R}^{(n)}_{1,2}$, $\mathcal{R}^{\prime (n)}_{2}$ and $\mathcal{R}^{\prime (n)}_{1,2}$ as for the previous corner point, and $\mathcal{R}^{\prime (n)}_{1} \triangleq \emptyset$. From \eqref{eq:ASetRsx}, \eqref{eq:ASetIxbk} and \eqref{eq:ASetLambdbx}: 
\begin{align}
\mathcal{R}^{(n)}_{\text{S}} & =  \text{any subset of } \mathcal{G}^{(n)}_{1} \setminus 	\big( \mathcal{R}^{(n)}_{2} \cup \mathcal{R}^{\prime (n)}_{2} \big) \text{ with size } \big| \mathcal{G}^{(n)}_{2} \big| - \big| \mathcal{C}^{(n)}_{1}  \big|, \nonumber \\
\mathcal{I}^{(n)} & = \big( \mathcal{G}^{(n)}_{1} \cup \mathcal{G}^{(n)}_{2} \big) \setminus   \big(  \mathcal{R}^{(n)}_{\text{1}} \cup \mathcal{R}^{(n)}_{\text{2}} \cup \mathcal{R}^{\prime (n)}_{\text{2}} \cup \mathcal{R}^{(n)}_{\text{S}} \big),  \nonumber \\
\mathcal{R}_{\Lambda}^{(n)} & =  \mathcal{G}_{1,2}^{(n)}. \nonumber
\end{align}
For $i \in [1,L]$, define $\bar{\Psi}_{1,i}^{(V)} \triangleq \bar{\Psi}_{i}^{(V)}$, $\Theta_{1,i}^{(V)} \triangleq \Theta_{i}^{(V)}$ and $\bar{\Psi}_{p,i}^{(V)} = \Theta_{p,i}^{(V)} \triangleq  \varnothing$ for $p \in [2,3]$. Also, we define $\bar{\Gamma}_{1,i}^{(V)}$ as any part of $\bar{\Gamma}_{i}^{(V)}$ with size $\big| \mathcal{G}^{(n)}_{0} \big|$, $\bar{\Gamma}_{2,i}^{(V)}$ as the remaining part with size $\big| \mathcal{C}^{(n)}_{1,2} \big| - \big| \mathcal{G}^{(n)}_{0} \big|$, and $\Gamma_{3,i}^{(V)} \triangleq \varnothing$. On the other hand, now we define $\Gamma_{1,i}^{(V)}$ as any part of $\Gamma_{i}^{(V)}$ with size $\big| \mathcal{G}^{(n)}_{0} \big|$, $\Gamma_{2,i}^{(V)} \triangleq \varnothing$ and $\Gamma_{3,i}^{(V)}$ as the remaining part with size $\big| \mathcal{C}^{(n)}_{1,2} \big| - \big| \mathcal{G}^{(n)}_{0} \big|$. According to \eqref{eq:pi2x}--\eqref{eq:delta2x}, for $i \in [1,L]$ we have $\Pi_{(2),i}^{(V)} = \tilde{A}_{i}\big[ \mathcal{I}^{(n)} \cap \mathcal{G}_2^{(n)} \big]$ with size $\big| \mathcal{G}^{(n)}_{2} \big| - \big| \mathcal{C}^{(n)}_{1}  \big| - \big( \big| \mathcal{C}^{(n)}_{1,2} \big| - \big| \mathcal{G}^{(n)}_{0}  \big| \big)$, sequence $\Pi_{(1),i}^{(V)} = \tilde{A}_{i}\big[ \mathcal{I}^{(n)} \cap \mathcal{G}_1^{(n)} \big]$ with size $\big| \mathcal{G}^{(n)}_{1} \big| - \big| \mathcal{C}^{(n)}_{2}  \big| - \big( \big| \mathcal{G}^{(n)}_{2} \big| - \big| \mathcal{C}^{(n)}_{1}  \big| \big)$, $\Delta_{(1),i}^{(V)} = \Gamma_{3,i}^{(V)}$ with size $\big| \mathcal{C}^{(n)}_{1,2} \big| - \big| \mathcal{G}^{(n)}_{0} \big|$ and $\Delta_{(2),i}^{(V)} = \varnothing$.

According to Algorithm~\ref{alg:formAx}, the inner-layer carries confidential information $S_{(2),1:L}^{(V)}$ intended for Receiver~2. For $i \in [2,L]$, $\Pi_{(2),i-1}^{(V)}$ is repeated in $\tilde{A}_{i}[\mathcal{R}_{\text{S}}^{(n)}]$ and $\Lambda_{i-1}^{(V)}$ in $\tilde{A}_{i}[\mathcal{R}_{\Lambda}^{(n)}]$. Now, for $i \in [1,L-1]$ both $\Pi_{(1),i+1}^{(V)}$ and $\Delta_{(1),i+1}^{(V)}$ will be repeated in outer-layer $\tilde{T}_{(1),i}$ associated to Receiver~1. This particular encoding is represented in Figure~\ref{fig:EncCasB1x}.
\end{enumerate}

\begin{figure}[h!]
\centering
\vspace*{-0cm}
\begin{overpic}[width=0.9\linewidth]{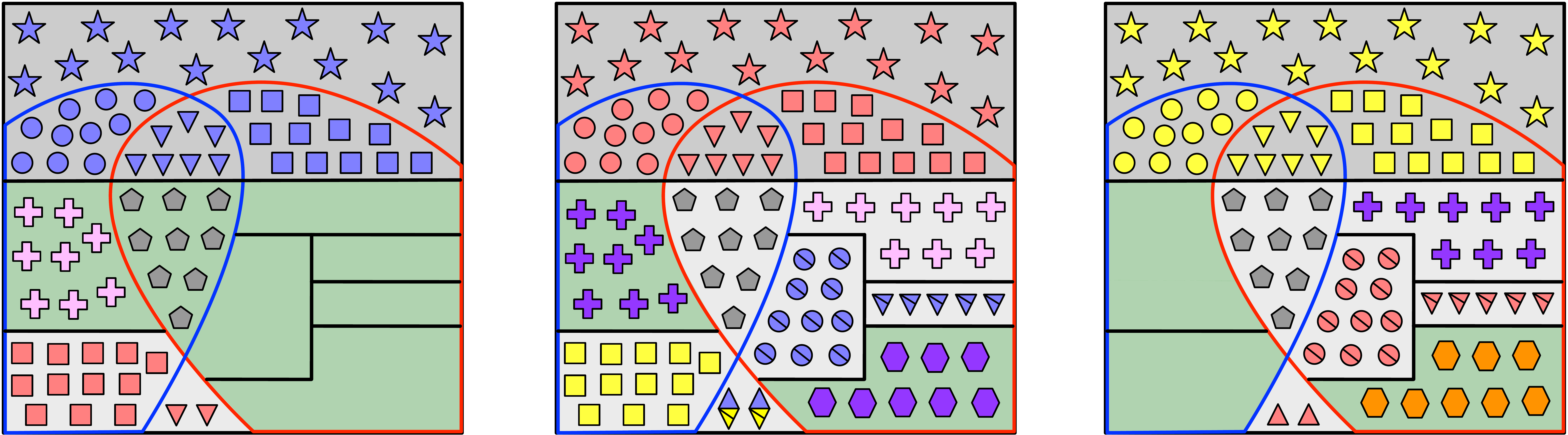}
\put (10.5,-3) {\small Block 1}
\put (45.75,-3) {\small Block 2}
\put (80.75,-3) {\small Block 3}
\end{overpic}
\vspace{0.5cm}
\caption{\setstretch{1.35}Case B when $\smash{ I(V;Z) \leq I(V;Y_{(1)}) \leq I(V;Y_{(2)})}$: inner-layer encoding to achieve the corner point of $\smash{\mathfrak{R}_{\text{MI-WTBC}}^{(2)}}$ that leads to the construction of $\smash{\tilde{A}_{1:L}[ \mathcal{H}_V^{(n)}]}$ when $L=3$. Consider Block~2: the sets $\smash{\mathcal{R}^{(n)}_{1}}$, $\smash{\mathcal{R}^{(n)}_{2}}$, $\smash{\mathcal{R}^{\prime (n)}_{2}}$, $\smash{\mathcal{R}^{(n)}_{1,2}}$, $\smash{\mathcal{R}^{(n)}_{\text{S}}}$ and $\smash{\mathcal{R}^{(n)}_{\Lambda}}$ are those areas filled with yellow squares, blue circles, blue triangles, blue and yellow diamonds, pink crosses, and gray pentagons, respectively; and $\mathcal{I}^{(n)}$ is the green filled area. At Block~$i \in [1,L]$, $\smash{W_{(2),i}^{(V)}}$ is represented by symbols of the same color (e.g., red symbols at Block~2), and $\smash{\Theta_i^{(V)}}$, $\smash{\Psi_i^{(V)}}$ and $\smash{\Gamma_i^{(V)}}$ are represented by squares, circles and triangles respectively. Also, $\smash{\bar{\Psi}_i^{(V)}}$ and $\smash{\bar{\Gamma}_i^{(V)}}$ are denoted by circles and triangles, respectively, with a line through them. At Block~$i \in [2,L-1]$, the diamonds denote $\smash{\bar{\Gamma}_{1,i-1}^{(V)} \oplus \Gamma_{1,i+1}^{(V)}}$. For $i \in [1,L-1]$, the elements of $\smash{\Gamma_{i+1}^{(V)}}$ that do not belong to $\smash{\Gamma_{1,i+1}^{(V)}}$ are not repeated in $\smash{\tilde{A}_i[\mathcal{G}^{(n)}]}$, but $\smash{\Delta_{(1),i+1}^{(V)} = \Gamma_{3,i+1}^{(V)}}$ will be sent to the outer-layer $\smash{\tilde{T}_{(1),i}}$. In Block~$i \in [1,L]$, $\smash{S_{(2),i}^{(V)}}$ is stored into those entries whose indices belong to the green filled area. For $i \in [2,L]$, $\smash{\Pi_{(2),i-1}^{(V)}}$ is denoted by crosses and is repeated in $\tilde{A}_{i}[ \mathcal{R}^{(n)}_{\text{S}} ]$. For $i \in [1,L-1]$, the sequence $\smash{\Pi_{(1),i+1}^{(V)}}$ is denoted by hexagons, and it will be send also to $\smash{\tilde{T}_{(1),i}}$. At Block~1, $\smash{\Lambda_1^{(V)}}$ is denoted by gray pentagons and is repeated in all blocks. Finally, $\smash{\Upsilon_{(1)}^{(V)}}$ and $\smash{\Upsilon_{(2)}^{(V)}}$ are those entries inside the red curve at Block~1 and the blue curve at Block~$L$,~respectively.
}\label{fig:EncCasB1x} 
\end{figure}

\newpage
\vspace{0.15cm}
\noindent \textbf{Case B when} $\bm{I(V;Y_{(1)}) < I(V;Z)\leq  I(V;Y_{(2)})}$

\begin{enumerate}
\item \emph{Achievability of} $(R_{S_{(1)}}^{\star 1}, R_{S_{(2)}}^{\star 1}, R_{W_{(1)}}^{\star 1}, R_{W_{(2)}}^{\star 1}) \subset \mathfrak{R}_{\text{MI-WTBC}}^{(1)}$. In this situation, according to condition~\eqref{eq:assumpRate1Impl2x2}, $\big| \mathcal{G}^{(n)}_{2} \big| - \big| \mathcal{C}^{(n)}_{1} \big| <  \big| \mathcal{C}^{(n)}_{1,2} \big| - \big| \mathcal{G}^{(n)}_{0} \big|$. Therefore, for $i \in [1,L-1]$, sequence $\bar{\Gamma}_{i+1}^{(V)}$ cannot be repeated entirely in $\tilde{A}_i\big[ \mathcal{G}^{(n)}_{0} \cup \big( \mathcal{G}^{(n)}_{2} \setminus \mathcal{R}^{(n)}_{1} \big) \big]$. Thus, we define 
\begin{align}
\mathcal{R}^{(n)}_{1} & \triangleq  \text{any subset of } \mathcal{G}^{(n)}_2 \text{ with size } \big| \mathcal{C}^{(n)}_1 \big|, \nonumber  \\
\mathcal{R}^{(n)}_{2} & \triangleq  \text{any subset of } \mathcal{G}^{(n)}_{1} \text{ with size } \big| \mathcal{C}^{(n)}_2 \big|, \nonumber \\
\mathcal{R}^{(n)}_{1,2} & \triangleq \mathcal{G}^{(n)}_0, \nonumber \\
\mathcal{R}^{\prime (n)}_{1} & \triangleq  \mathcal{G}^{(n)}_2 \setminus \mathcal{R}^{(n)}_{1}, \nonumber  \\
\mathcal{R}^{\prime (n)}_{2} & \triangleq \text{any subset of } \mathcal{G}^{(n)}_1 \setminus \mathcal{R}^{(n)}_{2} \text{ with size } \big| \mathcal{C}^{(n)}_{1,2} \big| - \big| \mathcal{G}^{(n)}_{0} \big|, \nonumber
\end{align}
and $\mathcal{R}^{\prime(n)}_{1,2} \triangleq \emptyset$. Therefore, according to \eqref{eq:ASetRsx}--\eqref{eq:ASetLambdaxk}, we have $\mathcal{R}^{(n)}_{\text{S}} = \mathcal{I}^{(n)} = \emptyset$ and
\begin{align}
\mathcal{R}_{\Lambda}^{(n)} & \triangleq  \mathcal{G}_{1,2}^{(n)} \cup \big( \mathcal{G}^{(n)}_{1} \setminus \big( \mathcal{R}^{(n)}_{2} \cup \mathcal{R}^{\prime (n)}_{2} \big) \big). \nonumber
\end{align}
For $i \in [1,L]$, we define $\smash{\Psi_{1,i}^{(V)} \triangleq \Psi_{i}^{(V)}}$, $\smash{\bar{\Theta}_{1,i}^{(V)} \triangleq \bar{\Theta}_{i}^{(V)}}$ and $\smash{\Psi_{p,i}^{(V)} = \bar{\Theta}_{p,i}^{(V)} \triangleq  \varnothing}$ for $p \in [2,3]$. Also, we define $\Gamma_{1,i}^{(V)}$ as any part of $\Gamma_{i}^{(V)}$ with size $\big| \mathcal{G}^{(n)}_{0} \big|$, $\Gamma_{2,i}^{(V)}$ as the remaining part with size $\big| \mathcal{C}^{(n)}_{1,2} \big| - \big| \mathcal{G}^{(n)}_{0} \big|$ and $\Gamma_{3,i}^{(V)} \triangleq \varnothing$; and $\bar{\Gamma}_{1,i}^{(V)}$ as any part of $\bar{\Gamma}_{i}^{(V)}$ with size $\big| \mathcal{G}^{(n)}_{0} \big|$, $\bar{\Gamma}_{2,i}^{(V)}$ as any part of $\bar{\Gamma}_{i}^{(V)}$ that is not included in $\bar{\Gamma}_{1,i}^{(V)}$ with size $\big| \mathcal{G}^{(n)}_{2} \big| - \big| \mathcal{C}^{(n)}_{1} \big|$, and $\bar{\Gamma}_{3,i}^{(V)}$ as the remaining part of $\bar{\Gamma}_{i}^{(V)}$ with size $\big| \mathcal{C}^{(n)}_{1,2} \big| - \big| \mathcal{G}^{(n)}_{0} \big| - \big( \big| \mathcal{G}^{(n)}_{2} \big| - \big| \mathcal{C}^{(n)}_{1} \big| \big)$. Thus, according to \mbox{\eqref{eq:pi2x}--\eqref{eq:delta2x}}, for $i \in [1,L]$ we have $\Pi_{(2),i}^{(V)} = \tilde{A}_{i}\big[ \mathcal{I}^{(n)} \cap \mathcal{G}_2^{(n)} \big] = \varnothing$, $\Pi_{(1),i}^{(V)} = \tilde{A}_{i}\big[ \mathcal{I}^{(n)} \cap \mathcal{G}_1^{(n)} \big] = \varnothing$, $\Delta_{(1),i}^{(V)}= \bar{\Gamma}_{3,i}^{(V)}$ and $\Delta_{(2),i}^{(V)}= \varnothing$. According to Algorithm~\ref{alg:formAx}, since $\mathcal{I}^{(n)} = \emptyset$ then $\tilde{A}_{2:L-1}^n \big[ \mathcal{G}^{(n)} \big]$ does not carry confidential information $S_{(1),2:L-1}^{(V)}$. For $i \in [1,L-1]$, $\Delta_{(1),i+1}^{(V)}$ will be repeated in outer-layer $\tilde{T}_{(1),i}^n$. This particular encoding is graphically represented in Figure~\ref{fig:EncCasB2x}.

\item \emph{Achievability of} $(R_{S_{(1)}}^{\star 2}, R_{S_{(2)}}^{\star 2}, R_{W_{(1)}}^{\star 2}, R_{W_{(2)}}^{\star 2}) \subset \mathfrak{R}_{\text{MI-WTBC}}^{(2)}$.
Construction of $\tilde{A}_{1:L}\big[\mathcal{G}^{(n)}\big]$ is the same as the one to achieve this rate tuple when $I(V;Z)\leq I(V;Y_{(1)})\leq I(V;Y_{(2)})$.
\end{enumerate}

\begin{figure}[h!]
\centering
\begin{overpic}[width=0.9\linewidth]{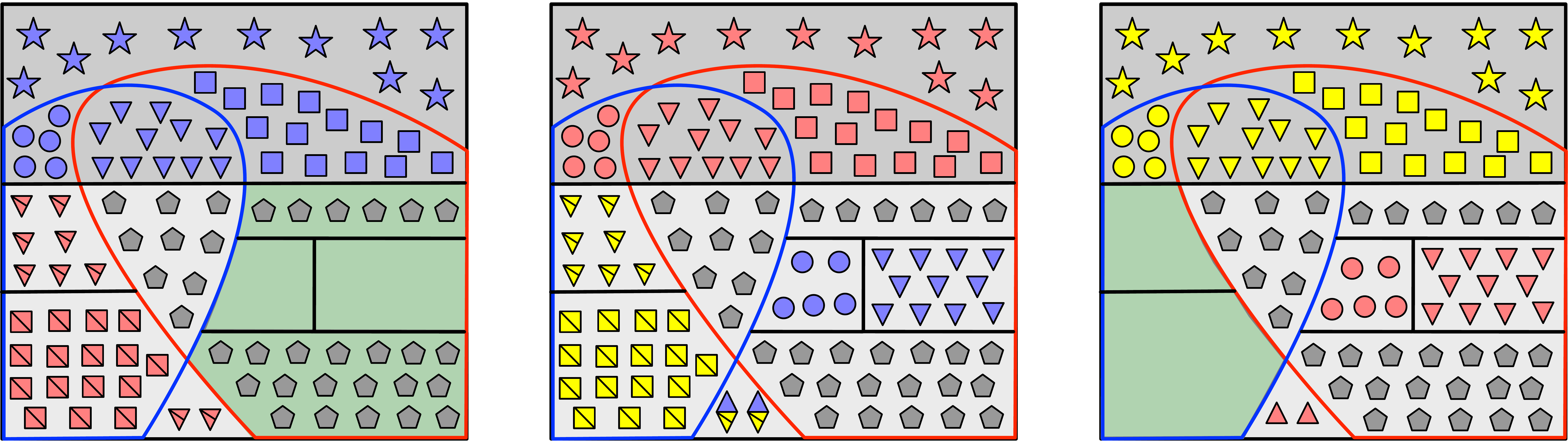}
\put (10.5,-3) {\small Block 1}
\put (45.75,-3) {\small Block 2}
\put (80.75,-3) {\small Block 3}
\end{overpic}
\vspace{0.5cm}
\caption{\setstretch{1.35} Case B when $\smash{I(V;Y_{(1)}) < I(V;Z) \leq I(V;Y_{(2)})}$: inner-layer encoding to achieve the corner point of $\smash{\mathfrak{R}_{\text{MI-WTBC}}^{(1)}}$ that leads to the construction of $\smash{\tilde{A}_{1:L}[ \mathcal{H}_V^{(n)} ]}$ when $L=3$. Consider Block~2: the sets $\smash{\mathcal{R}^{(n)}_{1}}$, $\smash{\mathcal{R}^{\prime (n)}_{1}}$, $\smash{\mathcal{R}^{(n)}_{2}}$, $\smash{\mathcal{R}^{\prime (n)}_{2}}$, $\smash{\mathcal{R}^{(n)}_{1,2}}$ and $\smash{\mathcal{R}^{(n)}_{\Lambda}}$ are those areas filled with yellow squares, yellow triangles, blue circles, blue triangles, blue and yellow diamonds, and gray pentagons, respectively. At Block~$i \in [1,L]$, $\smash{W_{(1),i}^{(V)}}$ is represented by symbols of the same color (e.g., red symbols at Block~2), and $\smash{\Theta_i^{(V)}}$, $\smash{\Psi_i^{(V)}}$ and $\smash{\Gamma_i^{(V)}}$ are represented by squares, circles and triangles respectively. Also, $\smash{\bar{\Theta}_i^{(V)}}$ and $\smash{\bar{\Gamma}_i^{(V)}}$ are denoted by squares and triangles, respectively, with a line through them. At Block~$i \in [2,L-1]$, the diamonds denote $\smash{\Gamma_{i-1}^{(V)} \oplus \bar{\Gamma}_{i+1}^{(V)}}$. For $i \in [2,L-1]$, $\smash{\tilde{A}_i [ \mathcal{G}^{(n)} ]}$ does not carry confidential information $\smash{S_{(1),i}^{(V)}}$, but only $\smash{\tilde{A}_1 [ \mathcal{G}^{(n)} ]}$ and $\smash{\tilde{A}_L [ \mathcal{G}^{(n)} ]}$ does (into the green area). The elements of $\smash{\bar{\Gamma}_{2,i+1}^{(V)}}$ that do not fit in $\smash{\tilde{A}_{i}[\mathcal{R}^{\prime (n)}_1]}$ will be sent to outer-layer $\smash{\tilde{T}_{(1),i}^n}$. At Block~1, $\smash{\Lambda_1^{(V)}}$ is denoted by gray pentagons, and is repeated in all blocks. Finally, sequences $\smash{\Upsilon_{(1)}^{(V)}}$ and $\smash{\Upsilon_{(2)}^{(V)}}$ are those entries inside the red curve at Block~1 and the blue curve at Block~$L$,~respectively.
}\label{fig:EncCasB2x} 
\end{figure}

\vspace{0.15cm}
\noindent \textbf{Case B when} \bm{$I(V;Y_{(1)}) \leq  I(V;Y_{(2)}) < I(V;Z)$}

\begin{enumerate}
\item \emph{Achievability of} $(R_{S_{(1)}}^{\star 1}, R_{S_{(2)}}^{\star 1}, R_{W_{(1)}}^{\star 1}, R_{W_{(2)}}^{\star 1}) \subset \mathfrak{R}_{\text{MI-WTBC}}^{(1)}$.
In this situation, define
\begin{align}
\mathcal{R}^{(n)}_{1} & \triangleq  \text{any subset of } \mathcal{G}^{(n)}_2 \text{ with size } \big| \mathcal{C}^{(n)}_1 \big|, \nonumber  \\
\mathcal{R}^{(n)}_{2} & \triangleq  \text{any subset of } \mathcal{G}^{(n)}_{1} \text{ with size } \big| \mathcal{C}^{(n)}_2 \big|, \nonumber \\
\mathcal{R}^{(n)}_{1,2} & \triangleq \mathcal{G}^{(n)}_0, \nonumber \\
\mathcal{R}^{\prime (n)}_{1} & \triangleq  \mathcal{G}^{(n)}_2 \setminus \mathcal{R}^{(n)}_{1},  \nonumber \\
\mathcal{R}^{\prime (n)}_{2} & \triangleq \mathcal{G}^{(n)}_1 \setminus \mathcal{R}^{(n)}_{2},  \nonumber
\end{align}
and $\mathcal{R}^{\prime(n)}_{1,2} \triangleq \emptyset$. Thus, according to \eqref{eq:ASetRsx}--\eqref{eq:ASetLambdaxk}, $\mathcal{R}^{(n)}_{\text{S}} = \mathcal{I}^{(n)} = \emptyset$ and $\mathcal{R}_{\Lambda}^{(n)}=\mathcal{G}_{1,2}^{(n)}$. In this situation we have defined $\mathcal{R}^{\prime (n)}_{1}$ and $\mathcal{R}^{\prime (n)}_{2}$ as above because, according to condition~\eqref{eq:assumpRate1Impl2x3}, $\big| \mathcal{G}^{(n)}_{2} \big| - \big| \mathcal{C}^{(n)}_{1} \big| <  \big| \mathcal{C}^{(n)}_{1,2} \big| - \big| \mathcal{G}^{(n)}_{0} \big|$ and $\big| \mathcal{G}^{(n)}_{1} \big| - \big| \mathcal{C}^{(n)}_{2} \big| <  \big| \mathcal{C}^{(n)}_{1,2} \big| - \big| \mathcal{G}^{(n)}_{0} \big|$. Therefore, neither $\Gamma_{i-1}^{(V)}$ (for $i \in [2,L]$) nor $\bar{\Gamma}_{i+1}^{(V)}$ (for $i \in [1,L-1]$) can be repeated entirely in $\tilde{A}_i\big[ \mathcal{G}^{(n)}_{0} \cup \big( \mathcal{G}^{(n)}_{1} \setminus \mathcal{R}^{(n)}_{2} \big) \big]$ or $\tilde{A}_i\big[ \big( \mathcal{G}^{(n)}_{0} \cup \big( \mathcal{G}^{(n)}_{2} \setminus \mathcal{R}^{(n)}_{1} \big) \big]$, respectively.

For $i \in [1,L]$, we define $\Psi_{1,i}^{(V)} \triangleq \Psi_{i}^{(V)}$, $\bar{\Theta}_{1,i}^{(V)} \triangleq \bar{\Theta}_{i}^{(V)}$ and $\Psi_{p,i}^{(V)} = \bar{\Theta}_{p,i}^{(V)} \triangleq  \varnothing$ for $p \in [2,3]$; we define $\Gamma_{1,i}^{(V)}$ as any part of $\Gamma_{i}^{(V)}$ with size $\big| \mathcal{G}^{(n)}_{0} \big|$, $\Gamma_{2,i}^{(V)}$ as any part of $\Gamma_{i}^{(V)}$ that is not included in $\Gamma_{1,i}^{(V)}$ with size $\big| \mathcal{G}^{(n)}_{1} \big| - \big| \mathcal{C}^{(n)}_{2} \big|$, and $\Gamma_{3,i}^{(V)}$ as the remaining part of $\Gamma_{i}^{(V)}$ with size $\big| \mathcal{C}^{(n)}_{1,2} \big| - \big| \mathcal{G}^{(n)}_{0} \big| - \big( \big| \mathcal{G}^{(n)}_{1} \big| - \big| \mathcal{C}^{(n)}_{2} \big| \big)$; and we define $\bar{\Gamma}_{1,i}^{(V)}$ as any part of $\bar{\Gamma}_{i}^{(V)}$ with size $\big| \mathcal{G}^{(n)}_{0} \big|$, $\bar{\Gamma}_{2,i}^{(V)}$ as any part of $\bar{\Gamma}_{i}^{(V)}$ that is not included in $\bar{\Gamma}_{1,i}^{(V)}$ with size $\big| \mathcal{G}^{(n)}_{2} \big| - \big| \mathcal{C}^{(n)}_{1} \big|$, and $\bar{\Gamma}_{3,i}^{(V)}$ as the remaining part with size $\big| \mathcal{C}^{(n)}_{1,2} \big| - \big| \mathcal{G}^{(n)}_{0} \big| - \big( \big| \mathcal{G}^{(n)}_{2} \big| - \big| \mathcal{C}^{(n)}_{1} \big| \big)$. Hence, according to \mbox{\eqref{eq:pi2x}--\eqref{eq:delta2x}}, for $i \in [1,L]$ we have $\Pi_{(2),i}^{(V)} = \tilde{A}_{i}\big[ \mathcal{I}^{(n)} \cap \mathcal{G}_2^{(n)} \big] = \varnothing$, $\Pi_{(1),i}^{(V)} = \tilde{A}_{i}\big[ \mathcal{I}^{(n)} \cap \mathcal{G}_1^{(n)} \big] = \varnothing$, and we have $\Delta_{(1),i}^{(V)}= \bar{\Gamma}_{3,i}^{(V)}$ and $\Delta_{(2),i}^{(V)}= \Gamma_{3,i}^{(V)}$. According to Algorithm~\ref{alg:formAx}, since $\mathcal{I}^{(n)} = \emptyset$ then $\tilde{A}_{2:L-1}^n \big[ \mathcal{G}^{(n)} \big]$ does not carry confidential information $S_{(1),2:L-1}^{(V)}$. For $i \in [1,L-1]$, $\Delta_{(1),i+1}^{(V)}$ will be repeated in $\tilde{T}_{(1),i}^n$, while $\Delta_{(2),i-1}^{(V)}$, for $i \in [2,L]$, will be repeated in $\tilde{T}_{(2),i}^n$. This particular encoding is represented in Figure~\ref{fig:EncCasB3x}.

\vspace{-0.1cm}
\item \emph{Achievability of} $(R_{S_{(1)}}^{\star 2}, R_{S_{(2)}}^{\star 2}, R_{W_{(1)}}^{\star 2}, R_{W_{(2)}}^{\star 2}) \subset \mathfrak{R}_{\text{MI-WTBC}}^{(2)}$.
Construction of $\tilde{A}_{1:L}\big[\mathcal{G}^{(n)}\big]$ is almost the same as that to achieve the previous corner point of region $\mathfrak{R}_{\text{MI-WTBC}}^{(1)}$: to approach this rate tuple the encoder repeats $\big[\bar{\Psi}_{i-1}^{(V)},\bar{\Gamma}_{i-1}^{(V)} \big]$ and $\big[\Theta_{i+1}^{(V)},\Gamma_{i+1}^{(V)} \big]$ in Block~$i$. 
\end{enumerate}

\begin{figure}[h!]
\centering
\vspace*{-0cm}
\begin{overpic}[width=0.9\linewidth]{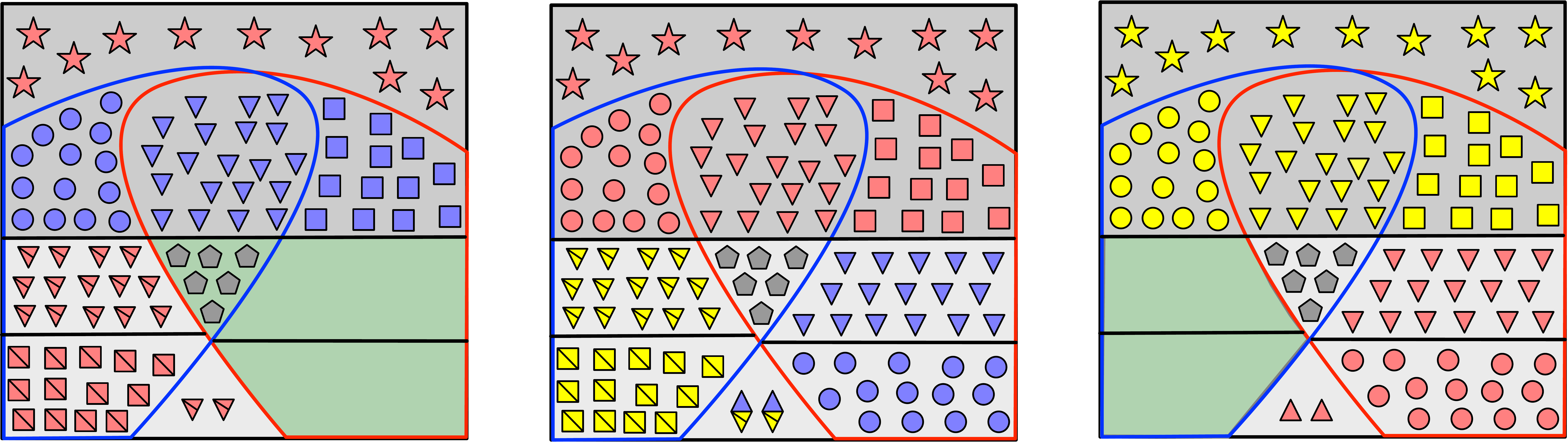}
\put (10.5,-3) {\small Block 1}
\put (45.75,-3) {\small Block 2}
\put (80.75,-3) {\small Block 3}
\end{overpic}
\vspace{0.5cm}
\caption{\setstretch{1.35}Case B when $\smash{ I(V;Y_{(1)}) \leq I(V;Y_{(2)}) < I(V;Z)}$: inner-layer encoding to achieve the corner point of $\smash{\mathfrak{R}_{\text{MI-WTBC}}^{(1)}}$ that leads to the construction of $\smash{\tilde{A}_{1:L}[ \mathcal{H}_V^{(n)}]}$ when $L=3$. Consider Block~2: the sets $\smash{\mathcal{R}^{(n)}_{1}}$, $\smash{\mathcal{R}^{\prime (n)}_{1}}$, $\smash{\mathcal{R}^{(n)}_{2}}$, $\smash{\mathcal{R}^{\prime (n)}_{2}}$, $\smash{\mathcal{R}^{(n)}_{1,2}}$ and $\smash{\mathcal{R}^{(n)}_{\Lambda}}$ are those areas filled with yellow squares, yellow triangles, blue circles, blue triangles, blue and yellow diamonds, and gray pentagons, respectively. At Block~$i \in [1,L]$, $\smash{W_{(1),i}^{(V)}}$ is represented by symbols of the same color (e.g., red symbols at Block~2), and $\smash{\Theta_i^{(V)}}$, $\smash{\Psi_i^{(V)}}$ and $\smash{\Gamma_i^{(V)}}$ are represented by squares, circles and triangles respectively. Also, $\smash{\bar{\Theta}_i^{(V)}}$ and $\smash{\bar{\Gamma}_i^{(V)}}$ are denoted by squares and triangles, respectively, with a line through them. The diamonds at Block~$i \in [2,L-1]$ denote $\smash{\Gamma_{1,i-1}^{(V)} \oplus \bar{\Gamma}_{1,i+1}^{(V)}}$. For $i \in [2,L-1]$, $\smash{\tilde{A}_i [ \mathcal{G}^{(n)} ]}$ does not carry confidential information $\smash{S_{(1),i}^{(V)}}$, but only $\smash{\tilde{A}_1 [ \mathcal{G}^{(n)} ]}$ and $\smash{\tilde{A}_L [ \mathcal{G}^{(n)} ]}$ does (into the green area). The elements of $\smash{\Gamma_{i-1}^{(V)}}$, or $\smash{\bar{\Gamma}_{i+1}^{(V)}}$, which do not fit in $\smash{\tilde{A}_{i}[\mathcal{R}^{(n)}_{1,2} \cup \mathcal{R}^{\prime (n)}_2]}$, or $\smash{\tilde{A}_{i}[\mathcal{R}^{(n)}_{1,2} \cup \mathcal{R}^{\prime (n)}_1]}$, will be sent to the outer-layer $\smash{\tilde{T}_{(2),i}^n}$, or $\smash{\tilde{T}_{(1),i}^n}$, respectively. At Block~1, $\smash{\Lambda_1^{(V)}}$ is denoted by gray pentagons and is replicated in all blocks. Finally, sequences $\smash{\Upsilon_{(1)}^{(V)}}$ and $\smash{\Upsilon_{(2)}^{(V)}}$ are those entries inside the red curve at Block~1 and the blue curve at Block~$L$,~respectively.
}\label{fig:EncCasB3x} 
\end{figure}

\vspace{0.15cm}
\noindent \textbf{Case C when} $\bm{I(V;Z)\leq I(V;Y_{(1)}) \leq I(V;Y_{(2)})}$

In Case~C, we have $| \mathcal{G}^{(n)}_1 | \geq | \mathcal{C}^{(n)}_2 |$, $| \mathcal{G}^{(n)}_2 |  \leq | \mathcal{C}^{(n)}_1 |$ and $| \mathcal{G}^{(n)}_0 |  > | \mathcal{C}^{(n)}_{1,2} |$.

\begin{enumerate}
\item \emph{Achievability of} $(R_{S_{(1)}}^{\star 1}, R_{S_{(2)}}^{\star 1}, R_{W_{(1)}}^{\star 1}, R_{W_{(2)}}^{\star 1}) \subset \mathfrak{R}_{\text{MI-WTBC}}^{(1)}$. 
The construction of $\tilde{A}_{1:L}\big[ \mathcal{G}^{(n)} \big]$ in this case is the same as that in \cite{alos2019polar} (Section~IV.3, Case C). Since $\big| \mathcal{G}^{(n)}_2 \big|  \leq \big| \mathcal{C}^{(n)}_1 \big|$, now for $i \in [1,L-1]$ we have that $\bar{\Theta}_{i+1}^{(V)}$ fills all the elements of $\tilde{A}_i \big[\mathcal{G}_2^{(n)} \big]$. Therefore:
\begin{align}
\mathcal{R}^{(n)}_{1,2} & \triangleq \text{any subset of } \mathcal{G}^{(n)}_{0} \text{ with size } \big| \mathcal{C}^{(n)}_{1,2} \big|,  \nonumber \\
\mathcal{R}^{(n)}_{1} & \triangleq \Scale[0.92]{\text{the union of }} \mathcal{G}^{(n)}_{2} \Scale[0.92]{\text{ with any subset of }} \mathcal{G}^{(n)}_{0} \setminus  \mathcal{R}_{1,2}^{(n)} \Scale[0.92]{\text{ with size }} \big| \mathcal{C}^{(n)}_{1} \big| - \big| \mathcal{G}^{(n)}_{2} \big|, \nonumber \\ 
\mathcal{R}^{(n)}_{2} & \triangleq  \text{any subset of } \mathcal{G}^{(n)}_{1} \text{ with size } \big| \mathcal{C}^{(n)}_2 \big|, \nonumber 
\end{align}
and $\mathcal{R}^{\prime(n)}_{1} = \mathcal{R}^{\prime(n)}_{2} = \mathcal{R}^{\prime(n)}_{1,2} \triangleq \emptyset$. Hence, according to \eqref{eq:ASetRsx}--\eqref{eq:ASetLambdaxk}, we have $\mathcal{R}^{(n)}_{\text{S}}= \emptyset$ and
\begin{align}
\mathcal{I}^{(n)} & = \mathcal{G}^{(n)}_{0}  \setminus  \mathcal{R}^{(n)}_{1,2} ,  \nonumber \\
\mathcal{R}_{\Lambda}^{(n)} & =  \mathcal{G}_{1,2}^{(n)} \cup \big( \mathcal{G}^{(n)}_{1} \setminus  \mathcal{R}^{(n)}_{2}  \big). \nonumber
\end{align}
 From condition \eqref{eq:assumpRate1Impl2x1}, all previous sets exist. For $i \in [1,L]$, we define $\Psi_{1,i}^{(V)} \triangleq \Psi_{i}^{(V)}$, $\Gamma_{1,i}^{(V)} \triangleq \Gamma_{i}^{(V)}$, $\bar{\Theta}_{1,i}^{(V)} \triangleq \bar{\Theta}_{i}^{(V)}$, $\bar{\Gamma}_{1,i}^{(V)} \triangleq \bar{\Gamma}_{i}^{(V)}$, and $\Psi_{p,i}^{(V)} = \Gamma_{p,i}^{(V)} = \bar{\Theta}_{p,i}^{(V)} = \bar{\Gamma}_{p,i}^{(V)} \triangleq \varnothing$ for $p \in [2,3]$. According to \eqref{eq:pi2x}--\eqref{eq:delta2x}, for $i \in [1,L]$ we have $\Pi_{(2),i}^{(V)} = \tilde{A}_{i}\big[ \mathcal{I}^{(n)} \cap \mathcal{G}_2^{(n)} \big] = \varnothing$, $\Pi_{(1),i}^{(V)} = \tilde{A}_{i}\big[ \mathcal{I}^{(n)} \cap \mathcal{G}_1^{(n)} \big] = \varnothing$, and $\Delta_{(1),i}^{(V)} = \Delta_{(2),i}^{(V)} = \varnothing$. According to Algorithm~\ref{alg:formAx}, the inner-layer carries confidential information $S_{(1),1:L}^{(V)}$ intended for Receiver~1. For $i \in [2,L]$, the encoder repeats $\Lambda_{i-1}^{(V)}$ in $\tilde{A}_{i}[\mathcal{R}_{\Lambda}^{(n)}]$ and, therefore, $\Lambda_{1}^{(V)}$ is replicated in all blocks. This particular encoding procedure is represented in \cite{alos2019polar}~(Figure~5).

\item \emph{Achievability of} $(R_{S_{(1)}}^{\star 2}, R_{S_{(2)}}^{\star 2}, R_{W_{(1)}}^{\star 2}, R_{W_{(2)}}^{\star 2}) \subset \mathfrak{R}_{\text{MI-WTBC}}^{(2)}$. 
Define $\mathcal{R}^{(n)}_{2}$, $\mathcal{R}^{(n)}_{1,2}$, $\mathcal{R}^{\prime (n)}_{1}$, $\mathcal{R}^{\prime (n)}_{2}$ and $\mathcal{R}^{\prime (n)}_{1,2}$ as for the previous corner point, and define $\mathcal{R}^{(n)}_{1} \triangleq \mathcal{G}^{(n)}_{2}$. Thus, according to \eqref{eq:ASetRsx}, \eqref{eq:ASetIxbk} and \eqref{eq:ASetLambdbx}, we have $\mathcal{R}^{(n)}_{\text{S}} = \emptyset$ and
\begin{align}
\mathcal{I}^{(n)} & = \big( \mathcal{G}^{(n)}_{0} \cup \mathcal{G}^{(n)}_{1}  \big) \setminus   \big(   \mathcal{R}^{(n)}_{\text{2}} \cup \mathcal{R}^{(n)}_{1,2} \big),  \nonumber \\
\mathcal{R}_{\Lambda}^{(n)} & =  \mathcal{G}_{1,2}^{(n)}. \nonumber
\end{align}
For $i \in [1,L]$, let $\bar{\Psi}_{1,i}^{(V)} \triangleq \bar{\Psi}_{i}^{(V)}$, $\Gamma_{1,i}^{(V)} \triangleq \Gamma_{i}^{(V)}$, $\bar{\Gamma}_{1,i}^{(V)} \triangleq \bar{\Gamma}_{i}^{(V)}$, and $\bar{\Psi}_{p,i}^{(V)} = \Gamma_{p,i}^{(V)} = \bar{\Gamma}_{p,i}^{(V)} \triangleq \varnothing$ for $p \in [2,3]$. Now, we define $\Theta_{1,i}^{(V)}$ as any part of $\Theta_{i}^{(V)}$ with size $\big| \mathcal{G}_2^{(n)}\big|$, $\Theta_{2,i}^{(V)} \triangleq \varnothing$, and $\Theta_{3,i}^{(V)}$ as the remaining part of $\Theta_{i}^{(V)}$ with size $\big| \mathcal{C}^{(n)}_{1} \big| - \big| \mathcal{G}^{(n)}_{2} \big|$. According to \eqref{eq:pi2x}--\eqref{eq:delta2x}, for $i \in [1,L]$ we have $\Pi_{(2),i}^{(V)} = \tilde{A}_{i}\big[ \mathcal{I}^{(n)} \cap \mathcal{G}_2^{(n)} \big] = \varnothing$, $\Pi_{(1),i}^{(V)} = \tilde{A}_{i}\big[ \mathcal{I}^{(n)} \cap \mathcal{G}_1^{(n)} \big]$ with size $\big| \mathcal{G}_1^{(n)} \big| - \big| \mathcal{C}_2^{(n)} \big|$, $\Delta_{(1),i}^{(V)} = \Theta_{3,i}^{(V)}$, and $\Delta_{(2),i}^{(V)} = \varnothing$. Now, the inner-layer carries confidential information $S_{(2),1:L}^{(V)}$ intended for Receiver~2. For $i \in [1,L-1]$, both $\Delta_{(1),i+1}^{(V)}$ and $\Pi_{(1),i+1}^{(V)}$ will be repeated in outer-layer $\tilde{T}_{(1),i}^{n}$. This particular encoding is represented in Figure~\ref{fig:EncCasC1x}.
\end{enumerate}

\vspace{0.15cm}
\noindent  \textbf{Case C when} $\bm{I(V;Y_{(1)})  < I(V;Z)\leq I(V;Y_{(2)})}$

\begin{enumerate}
\item \emph{Achievability of} $(R_{S_{(1)}}^{\star 1}, R_{S_{(2)}}^{\star 1}, R_{W_{(1)}}^{\star 1}, R_{W_{(2)}}^{\star 1}) \subset \mathfrak{R}_{\text{MI-WTBC}}^{(1)}$.
In this situation, according to~\eqref{eq:assumpRate1Impl2x2}, $\big| \mathcal{G}^{(n)}_{2} \big| - \big| \mathcal{C}^{(n)}_{1} \big| <  \big| \mathcal{C}^{(n)}_{1,2} \big| - \big| \mathcal{G}^{(n)}_{0} \big|$. Hence, for $i \in [1,L-1]$, $\bar{\Theta}_{i+1}^{(V)}$ cannot be repeated entirely in $\tilde{A}_i\big[ \mathcal{G}^{(n)}_{2} \cup \big( \mathcal{G}^{(n)}_{0} \setminus  \mathcal{R}_{1,2}^{(n)} \big)  \big]$. Therefore, define
\begin{align} 
\mathcal{R}^{(n)}_{1,2} & \triangleq \text{any subset of } \mathcal{G}^{(n)}_{0} \text{ with size } \big| \mathcal{C}^{(n)}_{1,2} \big|,  \nonumber \\
\mathcal{R}^{(n)}_{1} & \triangleq \mathcal{G}^{(n)}_{2} \cup \big( \mathcal{G}^{(n)}_{0} \setminus  \mathcal{R}_{1,2}^{(n)} \big), \nonumber \\
\mathcal{R}^{(n)}_{2} & \triangleq  \text{any subset of } \mathcal{G}^{(n)}_{1} \text{ with size } \big| \mathcal{C}^{(n)}_2 \big|, \nonumber
\end{align}
and $\mathcal{R}^{\prime(n)}_{1}= \mathcal{R}^{\prime(n)}_{2}= \mathcal{R}^{\prime(n)}_{1,2} \triangleq \emptyset$. Then, according to \eqref{eq:ASetRsx}--\eqref{eq:ASetLambdaxk}, $\mathcal{R}^{(n)}_{\text{S}} = \mathcal{I}^{(n)} = \emptyset$ and
\begin{align}
\mathcal{R}_{\Lambda}^{(n)} & =  \mathcal{G}_{1,2}^{(n)} \cup \big( \mathcal{G}^{(n)}_{1} \setminus  \mathcal{R}^{(n)}_{2}  \big). \nonumber
\end{align}
For $i \in [1,L]$, let $\Psi_{1,i}^{(V)} \triangleq \Psi_{i}^{(V)}$, $\Gamma_{1,i}^{(V)} \triangleq \Gamma_{i}^{(V)}$, $\bar{\Gamma}_{1,i}^{(V)} \triangleq \bar{\Gamma}_{i}^{(V)}$, and $\Psi_{p,i}^{(V)} = \Gamma_{p,i}^{(V)} = \bar{\Gamma}_{p,i}^{(V)} \triangleq \varnothing$ for $p \in [2,3]$. Also, we define $\bar{\Theta}_{1,i}^{(V)}$ as any part of $\bar{\Theta}_{i}^{(V)}$ with size $\big| \mathcal{G}^{(n)}_{2} \big| + \big(\big| \mathcal{G}^{(n)}_{0} \big| - \big| \mathcal{C}^{(n)}_{1,2} \big|  \big)$, $\bar{\Theta}_{2,i}^{(V)} = \varnothing$ and $\bar{\Theta}_{3,i}^{(V)}$ as the remaining part with size $\big| \mathcal{C}^{(n)}_{1} \big| - \big| \mathcal{G}^{(n)}_{2} \big| - \big(\big| \mathcal{G}^{(n)}_{0} \big| - \big| \mathcal{C}^{(n)}_{1,2} \big|  \big)$. According to \eqref{eq:pi2x}--\eqref{eq:delta2x}, for $i \in [1,L]$ we have $\Pi_{(2),i}^{(V)} = \tilde{A}_{i}\big[ \mathcal{I}^{(n)} \cap \mathcal{G}_2^{(n)} \big] = \varnothing$, $\Pi_{(1),i}^{(V)} = \tilde{A}_{i}\big[ \mathcal{I}^{(n)} \cap \mathcal{G}_1^{(n)} \big] = \varnothing$, $\Delta_{(1),i}^{(V)} = \bar{\Theta}_{3,i}^{(V)}$, and $\Delta_{(2),i}^{(V)} = \varnothing$. Since $\mathcal{I}^{(n)} = \emptyset$ then $\tilde{A}_{2:L-1}^n \big[ \mathcal{G}^{(n)} \big]$ does not carry confidential information $S_{(1),2:L-1}^{(V)}$. For $i \in [1,L-1]$, sequence $\Delta_{(1),i+1}^{(V)}$ will be repeated in outer-layer $\tilde{T}_{(1),i}^{n}$ associated to Receiver~1. 

\item \emph{Achievability of} $(R_{S_{(1)}}^{\star 2}, R_{S_{(2)}}^{\star 2}, R_{W_{(1)}}^{\star 2}, R_{W_{(2)}}^{\star 2}) \subset \mathfrak{R}_{\text{MI-WTBC}}^{(2)}$. 
Construction of $\tilde{A}_{1:L}\big[\mathcal{G}^{(n)}\big]$ is the same as that to achieve this rate tuple when $I(V;Z) \leq I(V;Y_{(1)})\leq I(V;Y_{(2)})$.
\end{enumerate}

\begin{figure}[h!]
\centering
\vspace*{-0cm}
\begin{overpic}[width=0.9\linewidth]{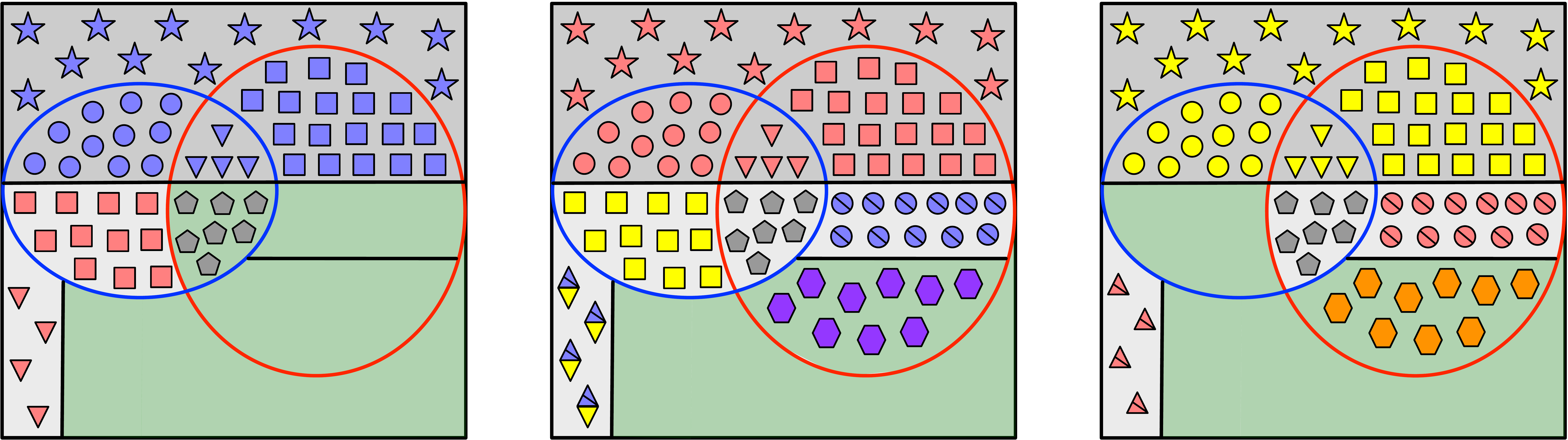}
\put (10.5,-3) {\small Block 1}
\put (45.75,-3) {\small Block 2}
\put (80.75,-3) {\small Block 3}
\end{overpic}
\vspace{0.5cm}
\caption{\setstretch{1.35} Case C when $\smash{ I(V;Z) \leq I(V;Y_{(1)}) \leq I(V;Y_{(2)})}$: inner-layer encoding to achieve the corner point of $\smash{\mathfrak{R}_{\text{MI-WTBC}}^{(2)}}$ that leads to the construction of $\smash{\tilde{A}_{1:L}[ \mathcal{H}_V^{(n)}]}$ when $L=3$. Consider the Block~2: $\smash{\mathcal{R}^{(n)}_{1}}$, $\smash{\mathcal{R}^{(n)}_{2}}$, $\smash{\mathcal{R}^{(n)}_{1,2}}$ and $\smash{\mathcal{R}^{(n)}_{\Lambda}}$ are those areas filled with yellow squares, blue circles, blue and yellow diamonds, and gray pentagons, respectively, and $\smash{\mathcal{I}^{(n)}}$ is the green filled area. At Block~$i \in [1,L]$, $\smash{W_{(2),i}^{(V)}}$ is represented by symbols of the same color (e.g., red symbols at Block~2), and $\smash{\Theta_i^{(V)}}$, $\smash{\Psi_i^{(V)}}$ and $\smash{\Gamma_i^{(V)}}$ are represented by squares, circles and triangles respectively. Also, $\smash{\bar{\Psi}_i^{(V)}}$ and $\smash{\bar{\Gamma}_i^{(V)}}$ are denoted by circles and triangles, respectively, with a line through them. At Block~$i \in [2,L-1]$, the diamonds denote $\smash{\bar{\Gamma}_{1,i-1}^{(V)} \oplus \Gamma_{1,i+1}^{(V)}}$. For $i \in [1,L-1]$, the elements of $\smash{\Theta_{i+1}^{(V)}}$ that do not belong to $\smash{\Theta_{1,i+1}^{(V)}}$ are not repeated in $\smash{\tilde{A}_i[\mathcal{G}^{(n)}]}$, but $\smash{\Delta_{(1),i+1}^{(V)} = \Theta_{3,i+1}^{(V)}}$ will be sent to the outer-layer $\smash{\tilde{T}_{(1),i}}$. For $i \in [1,L]$, confidential information $\smash{S_{(2),i}^{(V)}}$ is stored into those entries belonging to the green area. For $i \in [1,L-1]$, sequence $\smash{\Pi_{(1),i+1}^{(V)}}$ is denoted by hexagons, and it will be send also to $\smash{\tilde{T}_{(1),i}}$. At Block~1, $\smash{\Lambda_1^{(V)}}$ is denoted by gray pentagons and is replicated in all blocks. Finally, $\smash{\Upsilon_{(1)}^{(V)}}$ and $\smash{\Upsilon_{(2)}^{(V)}}$ are those entries inside the red curve at Block~1 and the blue curve at Block~$L$,~respectively.
}\label{fig:EncCasC1x} 
\end{figure}

\vspace{0.15cm}
\noindent \textbf{Case D when} $\bm{I(V;Z) \leq I(V;Y_{(1)}) \leq I(V;Y_{(2)})}$

In Case~D, we have $| \mathcal{G}^{(n)}_1 | < | \mathcal{C}^{(n)}_2 |$, $| \mathcal{G}^{(n)}_2 |  < | \mathcal{C}^{(n)}_1 |$ and $| \mathcal{G}^{(n)}_0 |  > | \mathcal{C}^{(n)}_{1,2} |$.

\begin{enumerate}
\item \emph{Achievability of} $(R_{S_{(1)}}^{\star 1}, R_{S_{(2)}}^{\star 1}, R_{W_{(1)}}^{\star 1}, R_{W_{(2)}}^{\star 1}) \subset \mathfrak{R}_{\text{MI-WTBC}}^{(1)}$. In this case, the construction of $\tilde{A}_{1:L}\big[ \mathcal{G}^{(n)} \big]$ is the same as the one described in \cite{alos2019polar}~(Section~IV.4, Case D). Since $| \mathcal{G}^{(n)}_1 | < | \mathcal{C}^{(n)}_2 |$, now for $i \in [2,L]$ only a part of $\Psi_{i-1}^{(V)}$ can be repeated entirely in $\tilde{A}_i[\mathcal{G}_1^{(n)}]$. Consequently, we define $\mathcal{R}^{(n)}_{2} \triangleq \mathcal{G}_{1}^{(n)}$,
\begin{align} 
\mathcal{R}^{(n)}_{1,2} & \triangleq \text{any subset of } \mathcal{G}^{(n)}_{0} \text{ with size } \big| \mathcal{C}^{(n)}_{1,2} \big|,  \nonumber \\
\mathcal{R}_{1,2}^{\prime (n)}  & \triangleq \text{any subset of } \mathcal{G}_0^{(n)} \setminus \mathcal{R}_{1,2}^{(n)} \text{ with size } \big| \mathcal{C}_{2}^{(n)} \big| - \big| \mathcal{G}_{1}^{(n)} \big|, \nonumber \\
\mathcal{R}^{(n)}_{1} & \triangleq \text{the union of } \mathcal{G}^{(n)}_{2} \text{ with any subset of }  \nonumber \\
& \quad \mathcal{G}^{(n)}_{0} \setminus  \big( \mathcal{R}_{1,2}^{(n)} \cup \mathcal{R}_{1,2}^{\prime (n)}  \big) \text{ with size } \big| \mathcal{C}^{(n)}_{1} \big| - \big| \mathcal{G}^{(n)}_{2} \big| -  \big( \big| \mathcal{C}^{(n)}_{2} \big|  - \big| \mathcal{G}^{(n)}_{1} \big|\big); \nonumber
\end{align}
and $\mathcal{R}^{\prime(n)}_{1} = \mathcal{R}^{\prime(n)}_{2} \triangleq \emptyset$. Hence, according to \eqref{eq:ASetRsx}--\eqref{eq:ASetLambdaxk}, we have $\mathcal{R}^{(n)}_{\text{S}}= \emptyset$ and
\begin{align}
\mathcal{I}^{(n)} & = \mathcal{G}^{(n)}_{0}  \setminus \big( \mathcal{R}^{(n)}_{1,2} \cup   \mathcal{R}^{\prime (n)}_{1,2} \big), \nonumber \\
\mathcal{R}_{\Lambda}^{(n)} & =  \mathcal{G}_{1,2}^{(n)}. \nonumber 
\end{align}
From condition \eqref{eq:assumpRate1Impl2x1}, all previous sets exist. For $i \in [1,L]$, let $\Gamma_{1,i}^{(V)} \triangleq \Gamma_{i}^{(V)}$, $\bar{\Gamma}_{1,i}^{(V)} \triangleq \bar{\Gamma}_{i}^{(V)}$ and $\bar{\Gamma}_{p,i}^{(V)} = \Gamma_{p,i}^{(V)} \triangleq \varnothing$ for $p \in [2,3]$. On the other hand, define $\Psi_{1,i}^{(V)}$ as any part $\Psi_{i}^{(V)}$ with size $\big| \mathcal{G}_{1}^{(n)} \big|$, $\Psi_{2,i}^{(V)}$ as the remaining part with size $\big| \mathcal{C}_{2}^{(n)} \big| - \big| \mathcal{G}_{1}^{(n)} \big|$, and $\Psi_{3,i}^{(V)} \triangleq \varnothing$; and $\bar{\Theta}_{1,i}^{(V)}$ as any part of $\bar{\Theta}_{i}^{(V)}$ with size $\big| \mathcal{C}^{(n)}_{1} \big| -  \big( \big| \mathcal{C}^{(n)}_{2} \big|  - \big| \mathcal{G}^{(n)}_{1} \big| \big)$, $\bar{\Theta}_{2,i}^{(V)}$ as the remaining part with size $\big| \mathcal{C}_{2}^{(n)} \big| - \big| \mathcal{G}_{1}^{(n)} \big|$, and $\bar{\Theta}_{3,i}^{(V)} \triangleq \varnothing$. Thus, from \eqref{eq:pi2x}--\eqref{eq:delta2x}, for $i \in [1,L]$ we have $\Pi_{(2),i}^{(V)} = \tilde{A}_{i}\big[ \mathcal{I}^{(n)} \cap \mathcal{G}_2^{(n)} \big] = \varnothing$, $\Pi_{(1),i}^{(V)} = \tilde{A}_{i}\big[ \mathcal{I}^{(n)} \cap \mathcal{G}_1^{(n)} \big] = \varnothing$, and $\Delta_{(1),i}^{(V)} = \Delta_{(2),i}^{(V)} = \varnothing$. 

According to Algorithm~\ref{alg:formAx}, the inner-layer carries confidential information $S_{(1),1:L}^{(V)}$ intended for Receiver~1. Also, instead of repeating $\Psi_{2,i-1}^{(V)}$ (the part of $\Psi_{i-1}^{(V)}$ that does not fit in $\tilde{A}_{i}^n \big[\mathcal{G}_{1}^{(n)}\big]$) in a specific part of $\tilde{A}_{i} \big[\mathcal{G}_{0}^{(n)} \big]$, the encoder stores $\Psi_{2,i-1}^{(V)} \oplus \bar{\Theta}_{2,i+1}^{(V)}$ into $\tilde{A}_{i} \big[\mathcal{R}_{1,2}^{\prime (n)} \big] \subseteq \tilde{A}_{i}\big[ \mathcal{G}_{0}^{(n)}\big]$, where $\bar{\Theta}_{2,i+1}^{(V)}$ denotes the elements of $\bar{\Theta}_{i+1}^{(V)}$ that do not fit in $\tilde{A}_{i} \big[\mathcal{G}_{2}^{(n)}\big]$. Finally, for $i \in [2,L]$, $\Lambda_{i-1}^{(V)}$ is repeated in $\tilde{A}_{i}[\mathcal{R}_{\Lambda}^{(n)}]$ and, hence, $\Lambda_{1}^{(V)}$ is replicated in all blocks. This particular encoding is graphically represented in~\cite{alos2019polar}~(Figure~5).

\item \emph{Achievability of} $(R_{S_{(1)}}^{\star 2}, R_{S_{(2)}}^{\star 2}, R_{W_{(1)}}^{\star 2}, R_{W_{(2)}}^{\star 2}) \subset \mathfrak{R}_{\text{MI-WTBC}}^{(2)}$. 
Define $\mathcal{R}^{(n)}_{2}$, $\mathcal{R}^{(n)}_{1,2}$, $\mathcal{R}^{\prime (n)}_{1}$, $\mathcal{R}^{\prime (n)}_{2}$ and $\mathcal{R}^{\prime (n)}_{1,2}$ as for the previous corner point, and define $\mathcal{R}^{(n)}_{1} \triangleq \mathcal{G}^{(n)}_{2}$. Thus, according to \eqref{eq:ASetRsx}, \eqref{eq:ASetIxbk} and \eqref{eq:ASetLambdbx}, we have $\mathcal{R}^{(n)}_{\text{S}} = \emptyset$ and
\begin{align}
\mathcal{I}^{(n)} & =  \mathcal{G}^{(n)}_{0} \setminus   \big( \mathcal{R}^{(n)}_{1,2}  \cup \mathcal{R}^{\prime (n)}_{1,2} \big), \nonumber \\
\mathcal{R}_{\Lambda}^{(n)} & =  \mathcal{G}_{1,2}^{(n)}. \nonumber
\end{align}
For $i \in [1,L]$, let $\Gamma_{1,i}^{(V)} \triangleq \Gamma_{i}^{(V)}$, $\bar{\Gamma}_{1,i}^{(V)} \triangleq \bar{\Gamma}_{i}^{(V)}$ and $\bar{\Gamma}_{p,i}^{(V)} = \Gamma_{p,i}^{(V)} \triangleq \varnothing$ for $p \in [2,3]$; and define $\bar{\Psi}_{1,i}^{(V)}$ as any part $\bar{\Psi}_{i}^{(V)}$ with size $\big| \mathcal{G}_{1}^{(n)} \big|$, $\bar{\Psi}_{2,i}^{(V)}$ as the remaining part with size $\big| \mathcal{C}_{2}^{(n)} \big| - \big| \mathcal{G}_{1}^{(n)} \big|$, and $\bar{\Psi}_{3,i}^{(V)} \triangleq \varnothing$. On the other hand, define $\Theta_{1,i}^{(V)}$ as any part of $\Theta_{i}^{(V)}$ with size $\big| \mathcal{G}^{(n)}_{2} \big|$, $\Theta_{2,i}^{(V)}$ as any part of $\Theta_{i}^{(V)}$ that is not included in $\Theta_{1,i}^{(V)}$ with size $\big| \mathcal{C}^{(n)}_{2} \big|  - \big| \mathcal{G}^{(n)}_{1} \big|$, and $\Theta_{3,i}^{(V)}$ as the remaining part of $\Theta_{i}^{(V)}$ with size $\big| \mathcal{C}^{(n)}_{1} \big| - \big| \mathcal{G}^{(n)}_{2} \big| -  \big( \big| \mathcal{C}^{(n)}_{2} \big|  - \big| \mathcal{G}^{(n)}_{1} \big|\big)$. Therefore, according to \eqref{eq:pi2x}--\eqref{eq:delta2x}, for $i \in [1,L]$ we have $\Pi_{(2),i}^{(V)} = \tilde{A}_{i}\big[ \mathcal{I}^{(n)} \cap \mathcal{G}_2^{(n)} \big] = \varnothing$, $\Pi_{(1),i}^{(V)} = \tilde{A}_{i}\big[ \mathcal{I}^{(n)} \cap \mathcal{G}_1^{(n)} \big] = \varnothing$, $\Delta_{(1),i}^{(V)} = \Theta_{3,i}^{(V)}$, and $\Delta_{(2),i}^{(V)} = \varnothing$. Now, the inner-layer carries confidential information $S_{(2),1:L}^{(V)}$ intended for Receiver~2. For $i \in [1,L-1]$, sequence $\Delta_{(1),i+1}^{(V)}$ will be repeated in outer-layer $\tilde{T}_{(1),i}^{n}$. This particular encoding procedure is graphically represented in Figure~\ref{fig:EncCasD1x}.
\end{enumerate}

\begin{figure}[h!]
\centering
\vspace*{-0cm}
\begin{overpic}[width=0.9\linewidth]{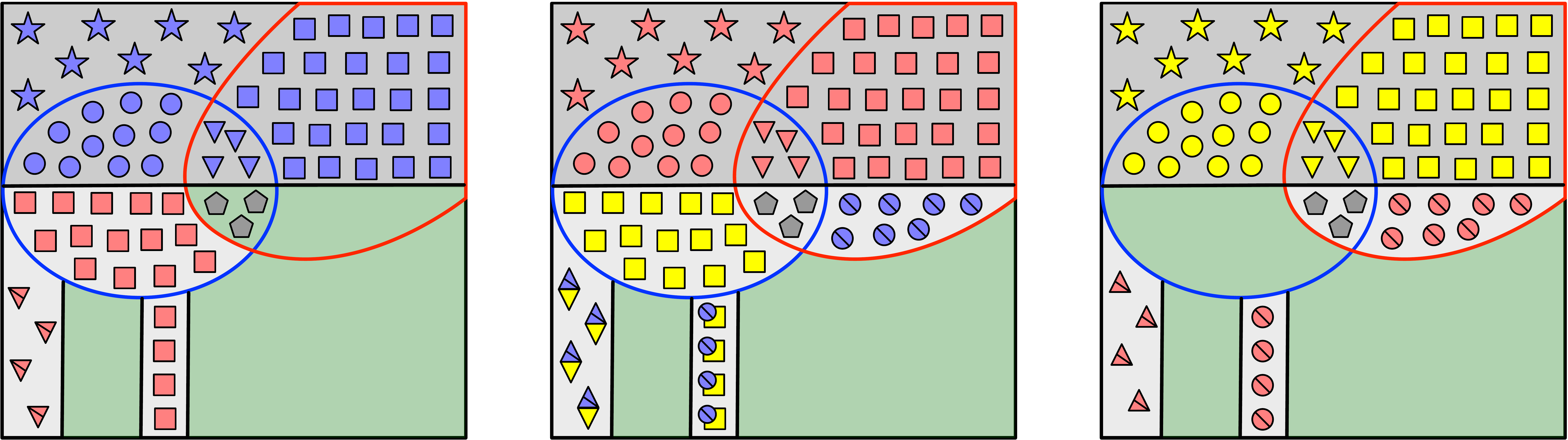}
\put (10.5,-3) {\small Block 1}
\put (45.75,-3) {\small Block 2}
\put (80.75,-3) {\small Block 3}
\end{overpic}
\vspace{0.5cm}
\caption{\setstretch{1.35} Case D when $\smash{ I(V;Z) \leq I(V;Y_{(1)}) \leq I(V;Y_{(2)})}$: inner-layer encoding to achieve the corner point of $\smash{\mathfrak{R}_{\text{MI-WTBC}}^{(2)}}$ that leads to the construction of $\smash{\tilde{A}_{1:L}[ \mathcal{H}_V^{(n)}]}$ when $L=3$. Consider Block~2: $\smash{\mathcal{R}^{(n)}_{1}}$, $\smash{\mathcal{R}^{(n)}_{2}}$, $\smash{\mathcal{R}^{(n)}_{1,2}}$, $\smash{\mathcal{R}^{\prime (n)}_{1,2}}$ and $\smash{\mathcal{R}^{(n)}_{\Lambda}}$ are those areas filled with yellow squares, blue circles, blue and yellow diamonds, yellow squares overlapped by blue circles, and gray pentagons, respectively, and $\smash{\mathcal{I}^{(n)}}$ is the green filled area. At Block~$i \in [1,L]$, $\smash{W_{(2),i}^{(V)}}$ is represented by symbols of the same color, and $\smash{\Theta_i^{(V)}}$, $\smash{\Psi_i^{(V)}}$ and $\smash{\Gamma_i^{(V)}}$ are represented by squares, circles and triangles respectively. Also, $\smash{\bar{\Psi}_i^{(V)}}$ and $\smash{\bar{\Gamma}_i^{(V)}}$ are denoted by circles and triangles, respectively, with a line through them. At Block~$i \in [2,L-1]$, the diamonds denote $\smash{\bar{\Gamma}_{1,i-1}^{(V)} \oplus \Gamma_{1,i+1}^{(V)}}$, while the yellow squares overlapped by blue circles denote $\smash{\bar{\Psi}_{2,i-1}^{(V)} \oplus \Theta_{2,i+1}^{(V)}}$. For $i \in [1,L-1]$, the elements of $\smash{\Theta_{i+1}^{(V)}}$ that are included neither in $\smash{\Theta_{1,i+1}^{(V)}}$ nor $\smash{\Theta_{2,i+1}^{(V)}}$ are not repeated in $\smash{\tilde{A}_i[\mathcal{G}^{(n)}]}$, but $\smash{\Delta_{(1),i+1}^{(V)} = \Theta_{3,i+1}^{(V)}}$ will be sent to outer-layer $\smash{\tilde{T}_{(1),i}}$. For $i \in [1,L]$, $\smash{S_{(2),i}^{(V)}}$ is stored into those entries belonging to the green area. At Block~1, sequence $\smash{\Lambda_1^{(V)}}$ is denoted by gray pentagons and is replicated in all blocks. Finally, $\smash{\Upsilon_{(1)}^{(V)}}$ and $\smash{\Upsilon_{(2)}^{(V)}}$ are those entries inside the red curve at Block~1 and the blue curve at Block~$L$,~respectively.
}\label{fig:EncCasD1x} 
\end{figure}

\vspace{0.15cm}
\noindent \textbf{Case D when} $\bm{I(V;Y_{(1)})  < I(V;Z) \leq I(V;Y_{(2)})}$

\begin{enumerate}
\item \emph{Achievability of} $(R_{S_{(1)}}^{\star 1}, R_{S_{(2)}}^{\star 1}, R_{W_{(1)}}^{\star 1}, R_{W_{(2)}}^{\star 1}) \subset \mathfrak{R}_{\text{MI-WTBC}}^{(1)}$. 
In this situation, according to~\eqref{eq:assumpRate1Impl2x2}, we have $\big| \mathcal{G}^{(n)}_{2} \big| - \big| \mathcal{C}^{(n)}_{1} \big| <  \big| \mathcal{C}^{(n)}_{1,2} \big| - \big| \mathcal{G}^{(n)}_{0} \big|$. Hence, for $i \in [1,L-1]$, $\bar{\Theta}_{i+1}^{(V)}$ cannot be repeated entirely in $\tilde{A}_i\big[ \mathcal{G}^{(n)}_{2} \cup \big( \mathcal{G}^{(n)}_{0} \setminus  \mathcal{R}_{1,2}^{(n)}  \big)  \big]$. Therefore, define $\mathcal{R}^{(n)}_{2} \triangleq \mathcal{G}^{(n)}_{1}$,
\begin{align} 
\mathcal{R}^{(n)}_{1,2} & \triangleq \text{any subset of } \mathcal{G}^{(n)}_{0} \text{ with size } \big| \mathcal{C}^{(n)}_{1,2} \big|, \nonumber \\
\mathcal{R}_{1,2}^{\prime (n)}  & \triangleq \text{any subset of } \mathcal{G}_0^{(n)} \setminus \mathcal{R}_{1,2}^{(n)} \text{ with size } \big| \mathcal{C}_{2}^{(n)} \big| - \big| \mathcal{G}_{1}^{(n)} \big|, \nonumber \\
\mathcal{R}^{(n)}_{1} & \triangleq \mathcal{G}^{(n)}_{2} \cup \big( \mathcal{G}^{(n)}_{0} \setminus \big( \mathcal{R}_{1,2}^{(n)} \cup \mathcal{R}_{1,2}^{\prime (n)}  \big) \big), \nonumber
\end{align}
and $\mathcal{R}^{\prime(n)}_{1}= \mathcal{R}^{\prime(n)}_{2} \triangleq \emptyset$. Then, from \eqref{eq:ASetRsx}--\eqref{eq:ASetLambdaxk}, $\mathcal{R}^{(n)}_{\text{S}} = \mathcal{I}^{(n)} = \emptyset$ and $\mathcal{R}_{\Lambda}^{(n)} =  \mathcal{G}_{1,2}^{(n)}$. For $i \in [1,L]$, let $\Gamma_{1,i}^{(V)} \triangleq \Gamma_{i}^{(V)}$, $\bar{\Gamma}_{1,i}^{(V)} \triangleq \bar{\Gamma}_{i}^{(V)}$ and $\bar{\Gamma}_{p,i}^{(V)} = \Gamma_{p,i}^{(V)} \triangleq \varnothing$ for $p \in [2,3]$; and define $\Psi_{1,i}^{(V)}$ as any part $\Psi_{i}^{(V)}$ with size $\big| \mathcal{G}_{1}^{(n)} \big|$, $\Psi_{2,i}^{(V)}$ as the remaining part with size $\big| \mathcal{C}_{2}^{(n)} \big| - \big| \mathcal{G}_{1}^{(n)} \big|$, and $\Psi_{3,i}^{(V)} \triangleq \varnothing$. On the other hand, we define $\bar{\Theta}_{1,i}^{(V)}$ as any part of $\bar{\Theta}_{i}^{(V)}$ with size $\big| \mathcal{G}^{(n)}_{2}\big| + \big| \mathcal{G}^{(n)}_{0}\big|  -  \big|  \mathcal{C}^{(n)}_{1,2} \big| -  \big( \big| \mathcal{C}^{(n)}_{2} \big| + \big|  \mathcal{G}^{(n)}_{1}\big|  \big)$, $\bar{\Theta}_{2,i}^{(V)}$ as any part of $\bar{\Theta}_{i}^{(V)}$ that is not included in $\bar{\Theta}_{1,i}^{(V)}$ with size $\big| \mathcal{C}^{(n)}_{2} \big| - \big|  \mathcal{G}^{(n)}_{1}\big|$, and $\bar{\Theta}_{3,i}^{(V)}$ as the remaining part with size $\big|\mathcal{C}^{(n)}_{1}\big| - \big( \big| \mathcal{G}^{(n)}_{2} \big| + \big|  \mathcal{G}^{(n)}_{0}\big| - \big|  \mathcal{C}^{(n)}_{1,2}\big| \big)$. According to \eqref{eq:pi2x}--\eqref{eq:delta2x}, for $i \in [1,L]$ we have $\Pi_{(2),i}^{(V)} = \tilde{A}_{i}\big[ \mathcal{I}^{(n)} \cap \mathcal{G}_2^{(n)} \big] = \varnothing$, $\Pi_{(1),i}^{(V)} = \tilde{A}_{i}\big[ \mathcal{I}^{(n)} \cap \mathcal{G}_1^{(n)} \big] = \varnothing$, $\Delta_{(1),i}^{(V)} = \bar{\Theta}_{3,i}^{(V)}$, and $\Delta_{(2),i}^{(V)} = \varnothing$. Since $\mathcal{I}^{(n)} = \emptyset$, $\tilde{A}_{2:L-1}^n \big[ \mathcal{G}^{(n)} \big]$ does not carry confidential information $S_{(1),2:L-1}^{(V)}$. For $i \in [1,L-1]$, sequence $\Delta_{(1),i+1}^{(V)}$ will be repeated in outer-layer $\tilde{T}_{(1),i}^{n}$. 

\item \emph{Achievability of} $(R_{S_{(1)}}^{\star 2}, R_{S_{(2)}}^{\star 2}, R_{W_{(1)}}^{\star 2}, R_{W_{(2)}}^{\star 2}) \subset \mathfrak{R}_{\text{MI-WTBC}}^{(2)}$. 
Construction of $\tilde{A}_{1:L}\big[\mathcal{G}^{(n)}\big]$ is the same as that to achieve this rate tuple when $I(V;Z) \leq I(V;Y_{(1)})\leq I(V;Y_{(2)})$.
\end{enumerate}

\vspace{0.15cm} 
\noindent \textbf{Case D when} $\bm{I(V;Y_{(1)})  \leq I(V;Y_{(2)}) < I(V;Z)}$

\begin{enumerate}
\item \emph{Achievability of} $(R_{S_{(1)}}^{\star 1}, R_{S_{(2)}}^{\star 1}, R_{W_{(1)}}^{\star 1}, R_{W_{(2)}}^{\star 1}) \subset \mathfrak{R}_{\text{MI-WTBC}}^{(1)}$. In this situation, from \eqref{eq:assumpRate1Impl2x3}, we have $\big| \mathcal{G}_0^{(n)} \big| - \big| \mathcal{C}_{1,2}^{(n)} \big| < \big| \mathcal{C}_1^{(n)} \big| - \big| \mathcal{G}_2^{(n)} \big|$ and $\big| \mathcal{G}_0^{(n)} \big| - \big| \mathcal{C}_{1,2}^{(n)} \big| < \big| \mathcal{C}_2^{(n)} \big| - \big| \mathcal{G}_1^{(n)} \big|$. Therefore:
\begin{align} 
\mathcal{R}^{(n)}_{1,2} & \triangleq \text{any subset of } \mathcal{G}^{(n)}_{0} \text{ with size } \big| \mathcal{C}^{(n)}_{1,2} \big|, \nonumber \\
\mathcal{R}_{1,2}^{\prime (n)}  & \triangleq \mathcal{G}_0^{(n)} \setminus \mathcal{R}_{1,2}^{(n)}, \nonumber
\end{align}
$\mathcal{R}^{(n)}_{1} \triangleq \mathcal{G}^{(n)}_{2}$, $\mathcal{R}^{(n)}_{2} \triangleq \mathcal{G}^{(n)}_{1}$ and $\mathcal{R}^{\prime(n)}_{1}= \mathcal{R}^{\prime(n)}_{2} \triangleq \emptyset$. Then, from \eqref{eq:ASetRsx}--\eqref{eq:ASetLambdaxk}, we have $\mathcal{R}^{(n)}_{\text{S}} = \mathcal{I}^{(n)} = \emptyset$ and $\mathcal{R}_{\Lambda}^{(n)} =  \mathcal{G}_{1,2}^{(n)}$. For $i \in [1,L]$, let $\Gamma_{1,i}^{(V)} \triangleq \Gamma_{i}^{(V)}$, $\bar{\Gamma}_{1,i}^{(V)} \triangleq \bar{\Gamma}_{i}^{(V)}$ and $\bar{\Gamma}_{p,i}^{(V)} = \Gamma_{p,i}^{(V)} \triangleq \varnothing$ for $p \in [2,3]$. Also, we define $\bar{\Theta}_{1,i}^{(V)}$ as any part of $\bar{\Theta}_{i}^{(V)}$ with size $\big| \mathcal{G}_2^{(n)} \big|$, $\bar{\Theta}_{2,i}^{(V)}$ as any part of $\bar{\Theta}_{i}^{(V)}$ that is not included in $\bar{\Theta}_{1,i}^{(V)}$ with size $\big| \mathcal{G}_0^{(n)} \big| - \big| \mathcal{C}_{1,2}^{(n)} \big|$, and $\bar{\Theta}_{3,i}^{(V)}$ as the remaining part of $\bar{\Theta}_{i}^{(V)}$ with size $\big| \mathcal{C}_1^{(n)} \big| - \big| \mathcal{G}_2^{(n)} \big| - \big( \big| \mathcal{G}_0^{(n)} \big| - \big| \mathcal{C}_{1,2}^{(n)} \big| \big)$. On the other hand, define $\Psi_{1,i}^{(V)}$ as any part of $\Psi_{i}^{(V)}$ with size $\big| \mathcal{G}_1^{(n)} \big|$, $\Psi_{2,i}^{(V)}$ as any part of $\Psi_{i}^{(V)}$ that is not included in $\Psi_{1,i}^{(V)}$ with size $\big| \mathcal{G}_0^{(n)} \big| - \big| \mathcal{C}_{1,2}^{(n)} \big|$, and $\Psi_{3,i}^{(V)}$ as the remaining part with size $\big| \mathcal{C}_2^{(n)} \big| - \big| \mathcal{G}_1^{(n)} \big| -  \big( \big| \mathcal{G}_0^{(n)} \big| - \big| \mathcal{C}_{1,2}^{(n)} \big| \big)$. According to \eqref{eq:pi2x}--\eqref{eq:delta2x}, for $i \in [1,L]$ we have $\Pi_{(2),i}^{(V)} = \tilde{A}_{i}\big[ \mathcal{I}^{(n)} \cap \mathcal{G}_2^{(n)} \big] = \varnothing$, $\Pi_{(1),i}^{(V)} = \tilde{A}_{i}\big[ \mathcal{I}^{(n)} \cap \mathcal{G}_1^{(n)} \big] = \varnothing$, $\Delta_{(1),i}^{(V)} = \bar{\Theta}_{3,i}^{(V)}$, and $\Delta_{(2),i}^{(V)} = \Psi_{3,i}^{(V)}$. The sequence $\Delta_{(1),i+1}^{(V)}$ ($i \in [1,L-1]$) will be repeated in $\tilde{T}_{(1),i}^n$, while $\Delta_{(2),i-1}^{(V)}$ ($i \in [2,L]$) will be stored in $\tilde{T}_{(2),i}^n$. Since $\mathcal{I}^{(n)} = \emptyset$, then $\tilde{A}_{2:L-1}^n \big[ \mathcal{G}^{(n)} \big]$ does not carry confidential information. This particular encoding is represented in Figure~\ref{fig:EncCasD2x}.

\item \emph{Achievability of} $(R_{S_{(1)}}^{\star 2}, R_{S_{(2)}}^{\star 2}, R_{W_{(1)}}^{\star 2}, R_{W_{(2)}}^{\star 2}) \subset \mathfrak{R}_{\text{MI-WTBC}}^{(2)}$.
Construction of $\tilde{A}_{1:L}\big[\mathcal{G}^{(n)}\big]$ is almost the same as that to achieve the previous corner point of region $\mathfrak{R}_{\text{MI-WTBC}}^{(1)}$: to approach this rate tuple the encoder repeats $\big[\bar{\Psi}_{i-1}^{(V)},\bar{\Gamma}_{i-1}^{(V)} \big]$ and $\big[\Theta_{i+1}^{(V)},\Gamma_{i+1}^{(V)} \big]$ in Block~$i$. 
\end{enumerate}

\begin{figure}[h!]
\centering
\vspace*{-0cm}
\begin{overpic}[width=0.9\linewidth]{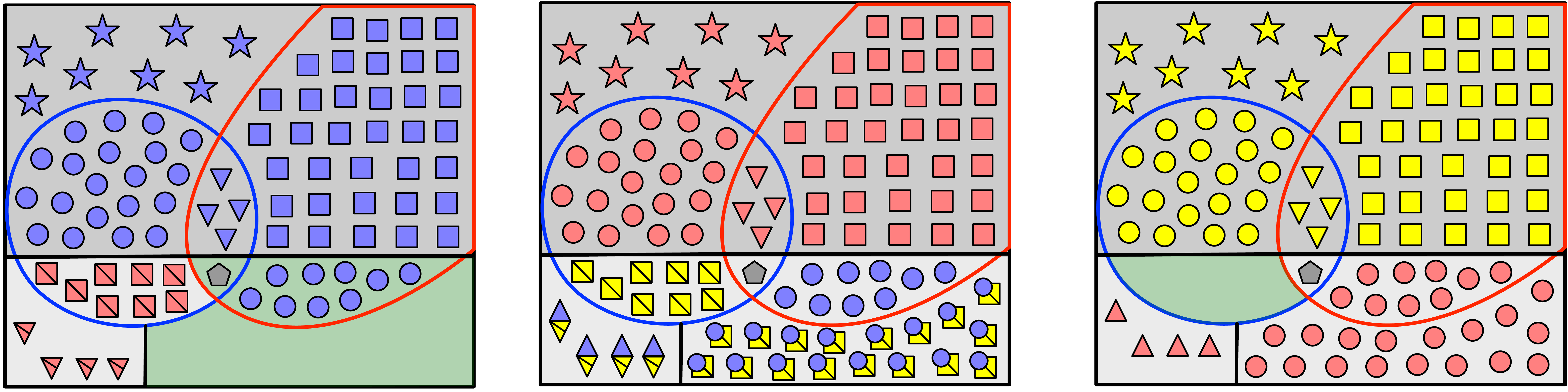}
\put (10.5,-3) {\small Block 1}
\put (45.75,-3) {\small Block 2}
\put (80.75,-3) {\small Block 3}
\end{overpic}
\vspace{0.5cm}
\caption{\setstretch{1.35}Case D when $\smash{ I(V;Y_{(1)}) \leq I(V;Y_{(2)}) < I(V;Z)}$: inner-layer encoding to achieve the corner point of $\smash{\mathfrak{R}_{\text{MI-WTBC}}^{(1)}}$ that leads to the construction of $\smash{\tilde{A}_{1:L}[ \mathcal{H}_V^{(n)}]}$ when $L=3$. Consider Block~2: the sets $\smash{\mathcal{R}^{(n)}_{1}}$, $\smash{\mathcal{R}^{(n)}_{2}}$, $\smash{\mathcal{R}^{(n)}_{1,2}}$, $\smash{\mathcal{R}^{\prime (n)}_{1,2}}$ and $\smash{\mathcal{R}^{(n)}_{\Lambda}}$ are those areas filled with yellow squares, blue circles, blue and yellow diamonds, yellow squares overlapped by blue circles, and gray pentagons, respectively. At Block~$i \in [1,L]$, $\smash{W_{(1),i}^{(V)}}$ is represented by symbols of the same color (red symbols at Block~2), and $\smash{\Theta_i^{(V)}}$, $\smash{\Psi_i^{(V)}}$ and $\smash{\Gamma_i^{(V)}}$ are represented by squares, circles and triangles respectively. Also, $\smash{\bar{\Theta}_i^{(V)}}$ and $\smash{\bar{\Gamma}_i^{(V)}}$ are denoted by squares and triangles, respectively, with a line through them. The diamonds at Block~$i \in [2,L-1]$ represent $\smash{\Gamma_{1,i-1}^{(V)} \oplus \bar{\Gamma}_{1,i+1}^{(V)}}$, while the squares overlapped by circles denote $\smash{\Psi_{2,i-1}^{(V)} \oplus \bar{\Theta}_{2,i+1}^{(V)}}$. For $i \in [2,L-1]$, $\smash{\tilde{A}_i [ \mathcal{G}^{(n)} ]}$ does not carry confidential information $\smash{S_{(1),i}^{(V)}}$, but only $\smash{\tilde{A}_1 [ \mathcal{G}^{(n)} ]}$ and $\smash{\tilde{A}_L [ \mathcal{G}^{(n)} ]}$ does (into the green area). The elements of $\smash{\Psi_{i-1}^{(V)}}$, or $\smash{\bar{\Theta}_{i+1}^{(V)}}$, which do not fit in $\smash{\tilde{A}_{i}[\mathcal{R}^{\prime (n)}_{1,2} \cup \mathcal{R}^{(n)}_2]}$, or $\smash{\tilde{A}_{i}[\mathcal{R}^{\prime (n)}_{1,2} \cup \mathcal{R}^{(n)}_1]}$, will be sent to $\smash{\tilde{T}_{(2),i}^n}$, or $\smash{\tilde{T}_{(1),i}^n}$, respectively. At Block~1, $\smash{\Lambda_1^{(V)}}$ is denoted by gray pentagons and is repeated in all blocks. Finally, sequences $\smash{\Upsilon_{(1)}^{(V)}}$ and $\smash{\Upsilon_{(2)}^{(V)}}$ are those entries inside the red curve at Block~1 and the blue curve at Block~$L$,~respectively.
}\label{fig:EncCasD2x} 
\end{figure}

\vspace{0.15cm}
\noindent \textbf{Case E when} $\bm{I(V;Y_{(1)}) <  I(V;Z) \leq I(V;Y_{(2)})}$

In Case~E, we have $| \mathcal{G}^{(n)}_1 | > | \mathcal{C}^{(n)}_2 |$, $| \mathcal{G}^{(n)}_2 |  < | \mathcal{C}^{(n)}_1 |$ and $| \mathcal{G}^{(n)}_0 |  < | \mathcal{C}^{(n)}_{1,2} |$.

\begin{enumerate}
\item \emph{Achievability of} $(R_{S_{(1)}}^{\star 1}, R_{S_{(2)}}^{\star 1}, R_{W_{(1)}}^{\star 1}, R_{W_{(2)}}^{\star 1}) \subset \mathfrak{R}_{\text{MI-WTBC}}^{(1)}$. 
In this case, we define 
\begin{align} 
\mathcal{R}^{(n)}_{1} & \triangleq \mathcal{G}^{(n)}_{2}, \nonumber \\
\mathcal{R}^{(n)}_{2} & \triangleq \text{any subset of } \mathcal{G}^{(n)}_{1} \text{ with size } \big| \mathcal{C}^{(n)}_{2} \big|, \nonumber \\
\mathcal{R}^{(n)}_{1,2} & \triangleq \mathcal{G}^{(n)}_{0}, \nonumber \\
\mathcal{R}^{\prime (n)}_{2} & \triangleq \text{any subset of } \mathcal{G}^{(n)}_1 \setminus \mathcal{R}^{(n)}_{2} \text{ with size } \big| \mathcal{C}^{(n)}_{1,2} \big| - \big| \mathcal{G}^{(n)}_{0} \big|, \nonumber
\end{align}
and $\mathcal{R}^{\prime (n)}_{1} = \mathcal{R}^{\prime (n)}_{1,2} \triangleq \varnothing$. Then, from \eqref{eq:ASetRsx}--\eqref{eq:ASetLambdaxk}, we have $\mathcal{R}^{(n)}_{\text{S}} = \mathcal{I}^{(n)} = \emptyset$ and
\begin{align}
\mathcal{R}_{\Lambda}^{(n)} =  \mathcal{G}_{1,2}^{(n)} \cup \big(\mathcal{G}_{1}^{(n)} \setminus \mathcal{R}_{2}^{(n)} \big). \nonumber
\end{align}
From condition \eqref{eq:assumpRate1Impl2x2}, all previous sets exist. For $i \in [1,L]$, we define $\smash{\Psi_{1,i}^{(V)} \triangleq \Psi_{i}^{(V)}}$ and $\smash{\Psi_{p,i}^{(V)}  \triangleq  \varnothing}$ for $p \in [2,3]$. Also, we define $\bar{\Theta}_{1,i}^{(V)}$ as any part of $\bar{\Theta}_{i}^{(V)}$ with size $\big| \mathcal{G}^{(n)}_{2} \big|$, $\bar{\Theta}_{2,i}^{(V)} \triangleq \varnothing$, and $\bar{\Theta}_{3,i}^{(V)}$ as the remaining part with size $\big| \mathcal{C}^{(n)}_{1} \big| - \big| \mathcal{G}^{(n)}_{2}\big|$. We define $\Gamma_{1,i}^{(V)}$ as any part of $\Gamma_{i}^{(V)}$ with size $\big| \mathcal{G}^{(n)}_{0} \big|$, $\Gamma_{2,i}^{(V)}$ as the remaining part with size $\big| \mathcal{C}^{(n)}_{1,2} \big| - \big| \mathcal{G}^{(n)}_{0}\big|$, and $\Gamma_{3,i}^{(V)} \triangleq \varnothing$. Finally, $\bar{\Gamma}_{1,i}^{(V)}$ is defined as any part of $\bar{\Gamma}_{i}^{(V)}$ with size $\big| \mathcal{G}^{(n)}_{0} \big|$, $\bar{\Gamma}_{2,i}^{(V)} \triangleq \varnothing$, and $\bar{\Gamma}_{3,i}^{(V)}$ as the remaining part of $\bar{\Gamma}_{i}^{(V)}$ with size $\big| \mathcal{C}^{(n)}_{1,2} \big| - \big| \mathcal{G}^{(n)}_{0}\big|$. According to \eqref{eq:pi2x}--\eqref{eq:delta2x}, for $i \in [1,L]$ we have $\Pi_{(2),i}^{(V)} = \tilde{A}_{i}\big[ \mathcal{I}^{(n)} \cap \mathcal{G}_2^{(n)} \big] = \varnothing$, $\Pi_{(1),i}^{(V)} = \tilde{A}_{i}\big[ \mathcal{I}^{(n)} \cap \mathcal{G}_1^{(n)} \big] = \varnothing$, $\Delta_{(1),i}^{(V)} = \big[ \bar{\Theta}_{3,i}^{(V)}, \bar{\Gamma}_{3,i}^{(V)} \big]$ with size $\big| \mathcal{C}^{(n)}_{1} \big| - \big| \mathcal{G}^{(n)}_{2}\big| + \big| \mathcal{C}^{(n)}_{1,2} \big| - \big| \mathcal{G}^{(n)}_{0}\big|$, and $\Delta_{(2),i}^{(V)} = \varnothing$. Since $\mathcal{I}^{(n)} = \emptyset$, $\tilde{A}_{2:L-1}^n \big[ \mathcal{G}^{(n)} \big]$ does not carry confidential information $S_{(1),2:L-1}^{(V)}$. For $i \in [1,L-1]$, sequence $\Delta_{(1),i+1}^{(V)}$ will be repeated in outer-layer $\tilde{T}_{(1),i}^{n}$. 

\vspace{-0.1cm}
\item \emph{Achievability of} $(R_{S_{(1)}}^{\star 2}, R_{S_{(2)}}^{\star 2}, R_{W_{(1)}}^{\star 2}, R_{W_{(2)}}^{\star 2}) \subset \mathfrak{R}_{\text{MI-WTBC}}^{(2)}$.
Construction of $\tilde{A}_{1:L}\big[\mathcal{G}^{(n)}\big]$ is similar to the previous one to achieve the previous corner point of region $\mathfrak{R}_{\text{MI-WTBC}}^{(1)}$. We define $\mathcal{R}^{(n)}_{1}$, $\mathcal{R}^{(n)}_{2}$, $\mathcal{R}^{(n)}_{1,2}$, $\mathcal{R}^{\prime (n)}_{1}$, $\mathcal{R}^{\prime (n)}_{2}$ and $\mathcal{R}^{\prime (n)}_{1,2}$ as for the previous corner point. Thus, according to \eqref{eq:ASetRsx}, \eqref{eq:ASetIxbk} and \eqref{eq:ASetLambdbx}, we have $\mathcal{R}^{(n)}_{\text{S}} = \emptyset$ and
\vspace*{-0.1cm}
\begin{align}
\mathcal{I}^{(n)} & =  \mathcal{G}^{(n)}_{1} \setminus   \big( \mathcal{R}^{(n)}_{2}  \cup \mathcal{R}^{\prime (n)}_{2} \big), \nonumber \\
\mathcal{R}_{\Lambda}^{(n)} & =  \mathcal{G}_{1,2}^{(n)}. \nonumber
\end{align}
Now, for $i \in [1,L]$ we define $\smash{\bar{\Psi}_{1,i}^{(V)} \triangleq \bar{\Psi}_{i}^{(V)}}$ and $\smash{\bar{\Psi}_{p,i}^{(V)}  \triangleq  \varnothing}$ for $p \in [2,3]$. Also, $\Theta_{1,i}^{(V)}$ is defined as any part of $\Theta_{i}^{(V)}$ with size $\big| \mathcal{G}^{(n)}_{2} \big|$, $\Theta_{2,i}^{(V)} \triangleq \varnothing$, and $\Theta_{3,i}^{(V)}$ as the remaining part with size $\big| \mathcal{C}^{(n)}_{1} \big| - \big| \mathcal{G}^{(n)}_{2}\big|$. We define $\bar{\Gamma}_{1,i}^{(V)}$ as any part of $\bar{\Gamma}_{i}^{(V)}$ with size $\big| \mathcal{G}^{(n)}_{0} \big|$, $\bar{\Gamma}_{2,i}^{(V)}$ as the remaining part with size $\big| \mathcal{C}^{(n)}_{1,2} \big| - \big| \mathcal{G}^{(n)}_{0}\big|$, and $\bar{\Gamma}_{3,i}^{(V)} \triangleq \varnothing$. Finally, we define $\Gamma_{1,i}^{(V)}$ as any part of $\Gamma_{i}^{(V)}$ with size $\big| \mathcal{G}^{(n)}_{0} \big|$, $\bar{\Gamma}_{2,i}^{(V)} \triangleq \varnothing$, and $\Gamma_{3,i}^{(V)}$ as the remaining part with size $\big| \mathcal{C}^{(n)}_{1,2} \big| - \big| \mathcal{G}^{(n)}_{0}\big|$. From \eqref{eq:pi2x}--\eqref{eq:delta2x}, for $i \in [1,L]$ we have \mbox{$\Pi_{(2),i}^{(V)} = \tilde{A}_{i}\big[ \mathcal{I}^{(n)} \cap \mathcal{G}_2^{(n)} \big] = \varnothing$}, \mbox{$\Pi_{(1),i}^{(V)} = \tilde{A}_{i}\big[ \mathcal{I}^{(n)} \cap \mathcal{G}_1^{(n)} \big]$} with size \mbox{$\big| \mathcal{G}^{(n)}_{1} \big| - \big| \mathcal{C}^{(n)}_{2}\big| - \big( \big| \mathcal{C}^{(n)}_{1,2} \big| - \big| \mathcal{G}^{(n)}_{0}\big| \big)$}, $\Delta_{(2),i}^{(V)} = \varnothing$ and the sequence $\Delta_{(1),i}^{(V)} = \big[ \Theta_{3,i}^{(V)}, \Gamma_{3,i}^{(V)} \big]$ with size \mbox{$\big| \mathcal{C}^{(n)}_{1} \big| - \big| \mathcal{G}^{(n)}_{2}\big| + \big| \mathcal{C}^{(n)}_{1,2} \big| - \big| \mathcal{G}^{(n)}_{0}\big|$}. According to Algorithm~\ref{alg:formAx}, now $\tilde{A}_{1:L}^n \big[ \mathcal{G}^{(n)} \big]$ carries confidential information $S_{(2),1:L}^{(V)}$. For $i \in [1,L-1]$, $\Delta_{(1),i+1}^{(V)}$ will be repeated in $\tilde{T}_{(1),i}^{n}$. 
\end{enumerate}

\vspace{0.15cm}
\noindent \textbf{Case E when} $\bm{I(V;Y_{(1)}) \leq I(V;Y_{(2)}) <  I(V;Z)}$

\begin{enumerate}
\item \emph{Achievability of} $(R_{S_{(1)}}^{\star 1}, R_{S_{(2)}}^{\star 1}, R_{W_{(1)}}^{\star 1}, R_{W_{(2)}}^{\star 1}) \subset \mathfrak{R}_{\text{MI-WTBC}}^{(1)}$. 
According to condition \eqref{eq:assumpRate1Impl2x3}, now we have $\big| \mathcal{G}_1^{(n)} \big| - \big| \mathcal{C}_2^{(n)} \big| < \big| \mathcal{C}_{1,2}^{(n)} \big| - \big| \mathcal{G}_0^{(n)} \big|$ and $\big| \mathcal{G}_2^{(n)} \big| - \big| \mathcal{C}_1^{(n)} \big| < \big| \mathcal{C}_{1,2}^{(n)} \big| - \big| \mathcal{G}_0^{(n)} \big|$. Now, all the elements of $\Gamma_{i-1}^{(V)}$, $i \in [2,L]$, cannot be repeated in $\tilde{A}_{i}\big[ \mathcal{G}_{1}^{(n)} \setminus \mathcal{R}_{2}^{(n)} \big]$ as before. Thus:
\begin{align} 
\mathcal{R}^{(n)}_{1} & \triangleq \mathcal{G}^{(n)}_{2}, \nonumber \\
\mathcal{R}^{(n)}_{2} & \triangleq \text{any subset of } \mathcal{G}^{(n)}_{1} \text{ with size } \big| \mathcal{C}^{(n)}_{2} \big|, \nonumber \\
\mathcal{R}^{(n)}_{1,2} & \triangleq \mathcal{G}^{(n)}_{0}, \nonumber \\
\mathcal{R}^{\prime (n)}_{2} & \triangleq \mathcal{G}^{(n)}_1 \setminus \mathcal{R}^{(n)}_{2}, \nonumber
\end{align}
and $\mathcal{R}^{\prime (n)}_{1} = \mathcal{R}^{\prime (n)}_{1,2} \triangleq \varnothing$. Then, from \eqref{eq:ASetRsx}--\eqref{eq:ASetLambdaxk}, we have $\mathcal{R}^{(n)}_{\text{S}} = \mathcal{I}^{(n)} = \emptyset$ and $\mathcal{R}_{\Lambda}^{(n)} =  \mathcal{G}_{1,2}^{(n)}$.
For $i \in [1,L]$ we define $\smash{\Psi_{1,i}^{(V)} \triangleq \Psi_{i}^{(V)}}$ and $\smash{\Psi_{p,i}^{(V)}  \triangleq  \varnothing}$ for $p \in [2,3]$. Also, $\bar{\Theta}_{1,i}^{(V)}$ is defined as any part of $\bar{\Theta}_{i}^{(V)}$ with size $\big| \mathcal{G}^{(n)}_{2} \big|$, $\bar{\Theta}_{2,i}^{(V)} \triangleq \varnothing$, and $\bar{\Theta}_{3,i}^{(V)}$ as the remaining part with size $\big| \mathcal{C}^{(n)}_{1} \big| - \big| \mathcal{G}^{(n)}_{2}\big|$. We define $\bar{\Gamma}_{1,i}^{(V)}$ as any part of $\bar{\Gamma}_{i}^{(V)}$ with size $\big| \mathcal{G}^{(n)}_{0} \big|$, $\bar{\Gamma}_{2,i}^{(V)} \triangleq \varnothing$, and $\bar{\Gamma}_{3,i}^{(V)}$ as the remaining part with size $\big| \mathcal{C}^{(n)}_{1,2} \big| - \big| \mathcal{G}^{(n)}_{0}\big|$. Finally, we define $\Gamma_{1,i}^{(V)}$ as any part of $\Gamma_{i}^{(V)}$ with size $\big| \mathcal{G}^{(n)}_{0} \big|$, $\Gamma_{2,i}^{(V)}$ as any part of $\Gamma_{i}^{(V)}$ that is not included in $\Gamma_{1,i}^{(V)}$ with size $\big| \mathcal{G}^{(n)}_{1} \big| - \big| \mathcal{C}^{(n)}_{2} \big|$, and $\Gamma_{3,i}^{(V)}$ as the remaining part with size $\big|\mathcal{C}^{(n)}_{1,2} \big| - \big| \mathcal{G}^{(n)}_{0} \big| - \big( \big| \mathcal{G}^{(n)}_{1} \big| - \big| \mathcal{C}^{(n)}_{2} \big| \big)$. According to \eqref{eq:pi2x}--\eqref{eq:delta2x}, for $i \in [1,L]$ we have $\Pi_{(2),i}^{(V)} = \tilde{A}_{i}\big[ \mathcal{I}^{(n)} \cap \mathcal{G}_2^{(n)} \big] = \varnothing$, \mbox{$\Pi_{(1),i}^{(V)} \! = \tilde{A}_{i}\big[ \mathcal{I}^{(n)} \cap \mathcal{G}_1^{(n)} \big] \! = \varnothing$}, $\Delta_{(1),i}^{(V)} \! = \big[ \bar{\Theta}_{3,i}^{(V)}, \bar{\Gamma}_{3,i}^{(V)} \big]$ with size \mbox{$\big| \mathcal{C}^{(n)}_{1} \big| - \big| \mathcal{G}^{(n)}_{2}\big| + \big| \mathcal{C}^{(n)}_{1,2} \big| - \big| \mathcal{G}^{(n)}_{0}\big|$}, and $\Delta_{(2),i}^{(V)} \! = \Gamma_{3,i}^{(V)}$ with size $\big|\mathcal{C}^{(n)}_{1,2} \big| - \big| \mathcal{G}^{(n)}_{0} \big| - \big| \mathcal{G}^{(n)}_{1} \big| + \big| \mathcal{C}^{(n)}_{2} \big|$. Since $\mathcal{I}^{(n)}=\emptyset$, $\tilde{A}_{2:L-1}^n$ does not carry confidential information $S_{(1),2:L-1}^{(V)}$. Lastly, $\Delta_{(1),i+1}^{(V)}$ ($i \in [1,L-1]$) will be repeated in $\tilde{T}_{(1),i}^{n}$, and $\Delta_{(2),i-1}^{(V)}$ ($i \in [2,L]$) will be repeated in $\tilde{T}_{(2),i}^{n}$.

\item \emph{Achievability of} $(R_{S_{(1)}}^{\star 2}, R_{S_{(2)}}^{\star 2}, R_{W_{(1)}}^{\star 2}, R_{W_{(2)}}^{\star 2}) \subset \mathfrak{R}_{\text{MI-WTBC}}^{(2)}$.
Construction of $\tilde{A}_{1:L}\big[\mathcal{G}^{(n)}\big]$ is exactly the same as the one to achieve the previous corner point of region $\mathfrak{R}_{\text{MI-WTBC}}^{(1)}$. 
\end{enumerate}

\vspace{0.15cm}
\textbf{Case F when} $\bm{I(V;Y_{(1)} \leq I(V;Y_{(2)})  <  I(V;Z)}$

In Case~F, we have $| \mathcal{G}^{(n)}_1 | < | \mathcal{C}^{(n)}_2 |$, $| \mathcal{G}^{(n)}_2 |  < | \mathcal{C}^{(n)}_1 |$ and $| \mathcal{G}^{(n)}_0 |  < | \mathcal{C}^{(n)}_{1,2} |$.

\begin{enumerate}
\item \emph{Achievability of} $(R_{S_{(1)}}^{\star 1}, R_{S_{(2)}}^{\star 1}, R_{W_{(1)}}^{\star 1}, R_{W_{(2)}}^{\star 1}) \subset \mathfrak{R}_{\text{MI-WTBC}}^{(1)}$. 
In this case, for $i \in [2,L]$ the \gls*{pcs} can repeat entirely neither $\Gamma_{i-1}^{(V)}$ in $\tilde{A}_{i}\big[ \mathcal{G}_{0}^{(n)} \big]$ nor $\Psi_{i-1}^{(V)}$ in $\tilde{A}_{i}\big[ \mathcal{G}_{1}^{(n)} \big]$. Also, for $i \in [1,L-1]$, it can repeat entirely neither $\bar{\Gamma}_{i+1}^{(V)}$ in $\tilde{A}_{i}\big[ \mathcal{G}_{0}^{(n)} \big]$ nor $\bar{\Theta}_{i+1}^{(V)}$ in $\tilde{A}_{i}\big[ \mathcal{G}_{2}^{(n)} \big]$. Thus, we define $\mathcal{R}_{1,2}^{(n)} \triangleq \mathcal{G}_0^{(n)}$, $\mathcal{R}_1^{(n)} \triangleq \mathcal{G}_2^{(n)}$, $\mathcal{R}_{2}^{(n)} \triangleq \mathcal{G}_1^{(n)}$, and \mbox{$\mathcal{R}_1^{\prime (n)} = \mathcal{R}_2^{\prime (n)} = \mathcal{R}_{1,2}^{\prime (n)} \triangleq \varnothing$}. For $i \in [1,L]$, we define $\smash{\Psi_{1,i}^{(V)}}$ as any part of  $\Psi_{i}^{(V)}$ with size $\big| \mathcal{G}_1^{(n)} \big|$, $\Psi_{1,i}^{(V)} \triangleq \varnothing$, and $\Psi_{3,i}^{(V)}$ as the remaining part with size $\big| \mathcal{C}_2^{(n)} \big| - \big| \mathcal{G}_1^{(n)} \big|$. Also, we define $\bar{\Theta}_{1,i}^{(V)}$ as any part of  $\bar{\Theta}_{i}^{(V)}$ with size $\big| \mathcal{G}_2^{(n)} \big|$, $\bar{\Theta}_{2,i}^{(V)} \triangleq \varnothing$, and $\bar{\Theta}_{3,i}^{(V)}$ as the remaining part with size $\big| \mathcal{C}_1^{(n)} \big| - \big| \mathcal{G}_2^{(n)} \big|$. Finally, we define $\Gamma_{1,i}^{(V)}$ and $\bar{\Gamma}_{1,i}^{(V)}$ as any part of $\Gamma_{i}^{(V)}$ and $\bar{\Gamma}_{i}^{(V)}$, respectively, with size $\big| \mathcal{G}^{(n)}_{0} \big|$, $\Gamma_{2,i}^{(V)} = \bar{\Gamma}_{2,i}^{(V)} \triangleq \varnothing$, and $\Gamma_{3,i}^{(V)}$ and $\bar{\Gamma}_{3,i}^{(V)}$ as the remaining parts with size $\big| \mathcal{C}^{(n)}_{1,2} \big| - \big| \mathcal{G}^{(n)}_{0} \big|$. According to \eqref{eq:pi2x}--\eqref{eq:delta2x}, for $i \in [1,L]$ we have $\Pi_{(2),i}^{(V)} = \tilde{A}_{i}\big[ \mathcal{I}^{(n)} \cap \mathcal{G}_2^{(n)} \big] = \varnothing$, \mbox{$\Pi_{(1),i}^{(V)} \! = \tilde{A}_{i}\big[ \mathcal{I}^{(n)} \cap \mathcal{G}_1^{(n)} \big] \! = \varnothing$}, $\Delta_{(1),i}^{(V)} \! = \big[ \bar{\Theta}_{3,i}^{(V)}, \bar{\Gamma}_{3,i}^{(V)} \big]$ with size \mbox{$\big| \mathcal{C}^{(n)}_{1} \big| - \big| \mathcal{G}^{(n)}_{2}\big| + \big| \mathcal{C}^{(n)}_{1,2} \big| - \big| \mathcal{G}^{(n)}_{0}\big|$}, and $\Delta_{(2),i}^{(V)} \! = \big[\Psi_{3,i}^{(V)}, \Gamma_{3,i}^{(V)} \big]$ with size $\big| \mathcal{C}^{(n)}_{2} \big| - \big| \mathcal{G}^{(n)}_{1} \big| + \big|\mathcal{C}^{(n)}_{1,2} \big| - \big| \mathcal{G}^{(n)}_{0} \big|$. Since $\mathcal{I}^{(n)}=\emptyset$, $\tilde{A}_{2:L-1}^n$ does not carry confidential information $S_{(1),2:L-1}^{(V)}$. Lastly, $\Delta_{(1),i+1}^{(V)}$ ($i \in [1,L-1]$) will be repeated in $\tilde{T}_{(1),i}^{n}$, and $\Delta_{(2),i-1}^{(V)}$ ($i \in [2,L]$) will be repeated in $\tilde{T}_{(2),i}^{n}$.

\item \emph{Achievability of} $(R_{S_{(1)}}^{\star 2}, R_{S_{(2)}}^{\star 2}, R_{W_{(1)}}^{\star 2}, R_{W_{(2)}}^{\star 2}) \subset \mathfrak{R}_{\text{MI-WTBC}}^{(2)}$.
Construction of $\tilde{A}_{1:L}\big[\mathcal{G}^{(n)}\big]$ is almost the same as that to achieve the previous corner point of region $\mathfrak{R}_{\text{MI-WTBC}}^{(1)}$: to approach this rate tuple the encoder repeats $\big[\bar{\Psi}_{i-1}^{(V)},\bar{\Gamma}_{i-1}^{(V)} \big]$ and $\big[\Theta_{i+1}^{(V)},\Gamma_{i+1}^{(V)} \big]$ in Block~$i$.
\end{enumerate}

\noindent \textbf{Summary of the construction of} $\bm{\tilde{A}_{1:L}\big[\mathcal{G}^{(n)}\big]}$

From \eqref{eq:ASetRsx}--\eqref{eq:delta2x}; and from the definition of $\mathcal{R}_{1}^{(n)}$, $\mathcal{R}_{2}^{(n)}$, $\mathcal{R}_{1,2}^{(n)}$, $\mathcal{R}_{1}^{\prime (n)}$, $\mathcal{R}_{2}^{\prime (n)}$ and $\mathcal{R}_{1,2}^{\prime (n)}$, and $\Psi_{p,i}^{(V)}$, $\bar{\Psi}_{p,i}^{(V)}$, $\Gamma_{p,i}^{(V)}$, $\bar{\Gamma}_{p,i}^{(V)}$, $\Theta_{p,i}^{(V)}$ and $\bar{\Theta}_{p,i}^{(V)}$ for $p \in [1,3]$ and $i \in [1,L]$ in each case, we have:
\begin{enumerate}
\item when $I(V;Z) \leq I(V;Y_{(1)}) \leq I(V;Y_{(2)})$,
\begin{itemize}
\item if the \gls*{pcs} operates to achieve $(R_{S_{(1)}}^{\star 1}, R_{S_{(2)}}^{\star 1}, R_{W_{(1)}}^{\star 1}, R_{W_{(2)}}^{\star 1}) \subset \mathfrak{R}_{\text{MI-WTBC}}^{(1)}$:
\begin{itemize}
\item the inner-layer carries $S_{(1),1:L}^{(V)}$, and $\big| \mathcal{I}^{(n)} \big| = \big| \mathcal{G}_0^{(n)} \big| + \big| \mathcal{G}_2^{(n)} \big| - \big| \mathcal{C}_1^{(n)} \big| - \big| \mathcal{C}_{1,2}^{(n)} \big|$;
\item we have $\Pi_{(1),i}^{(V)}= \varnothing$ and $\Delta_{(1),i}^{(V)} = \varnothing$;
\item we have $\Delta_{(2),i}^{(V)} = \varnothing$.
\end{itemize}
\item if the \gls*{pcs} operates to achieve $(R_{S_{(1)}}^{\star 2}, R_{S_{(2)}}^{\star 2}, R_{W_{(1)}}^{\star 2}, R_{W_{(2)}}^{\star 2}) \subset \mathfrak{R}_{\text{MI-WTBC}}^{(2)}$:
\begin{itemize}
\item the inner-layer carries $S_{(2),1:L}^{(V)}$, and $\big| \mathcal{I}^{(n)} \big| = \big| \mathcal{G}_0^{(n)} \big| + \big| \mathcal{G}_1^{(n)} \big| - \big| \mathcal{C}_2^{(n)} \big| - \big| \mathcal{C}_{1,2}^{(n)} \big|$;
\item the overall size of $\big[ \Pi_{(1),i}^{(V)},\Delta_{(1),i}^{(V)} \big]$ is $\big| \mathcal{G}_1^{(n)} \big| + \big| \mathcal{C}_1^{(n)} \big| - \big| \mathcal{G}_2^{(n)} \big| - \big| \mathcal{C}_2^{(n)} \big|$;
\item we have $\Delta_{(2),i}^{(V)} = \varnothing$.
\end{itemize}
\end{itemize}
\item when $ I(V;Y_{(1)}) < I(V;Z) \leq I(V;Y_{(2)})$,
\begin{itemize}
\item if the \gls*{pcs} operates to achieve $(R_{S_{(1)}}^{\star 1}, R_{S_{(2)}}^{\star 1}, R_{W_{(1)}}^{\star 1}, R_{W_{(2)}}^{\star 1}) \subset \mathfrak{R}_{\text{MI-WTBC}}^{(1)}$:
\begin{itemize}
\item the inner-layer carries $S_{(1),1:L}^{(V)}$, but $\big| \mathcal{I}^{(n)} \big| = 0$;
\item we have $\Pi_{(1),i}^{(V)}= \varnothing$, and the size of $\Delta_{(1),i}^{(V)}$ is $\big| \mathcal{C}_1^{(n)} \big| + \big| \mathcal{C}_{1,2}^{(n)} \big| - \big| \mathcal{G}_0^{(n)} \big| - \big| \mathcal{G}_2^{(n)} \big|$;
\item we have $\Delta_{(2),i}^{(V)} = \varnothing$.
\end{itemize}
\item if the \gls*{pcs} operates to achieve $(R_{S_{(1)}}^{\star 2}, R_{S_{(2)}}^{\star 2}, R_{W_{(1)}}^{\star 2}, R_{W_{(2)}}^{\star 2}) \subset \mathfrak{R}_{\text{MI-WTBC}}^{(2)}$:
\begin{itemize}
\item the inner-layer carries $S_{(2),1:L}^{(V)}$, and $\big| \mathcal{I}^{(n)} \big| = \big| \mathcal{G}_0^{(n)} \big| + \big| \mathcal{G}_1^{(n)} \big| - \big| \mathcal{C}_2^{(n)} \big| - \big| \mathcal{C}_{1,2}^{(n)} \big|$;
\item the overall size of $\big[ \Pi_{(1),i}^{(V)},\Delta_{(1),i}^{(V)} \big]$ is $\big| \mathcal{G}_1^{(n)} \big| + \big| \mathcal{C}_1^{(n)} \big| - \big| \mathcal{G}_2^{(n)} \big| - \big| \mathcal{C}_2^{(n)} \big|$;
\item we have $\Delta_{(2),i}^{(V)}= \varnothing$.
\end{itemize}
\end{itemize}
\item when $ I(V;Y_{(1)}) \leq I(V;Y_{(2)}) < I(V;Z)$,
\begin{itemize}
\item if the \gls*{pcs} operates to achieve $(R_{S_{(1)}}^{\star 1}, R_{S_{(2)}}^{\star 1}, R_{W_{(1)}}^{\star 1}, R_{W_{(2)}}^{\star 1}) \subset \mathfrak{R}_{\text{MI-WTBC}}^{(1)}$:
\begin{itemize}
\item the inner-layer carries $S_{(1),1:L}^{(V)}$, but $\big| \mathcal{I}^{(n)} \big| = 0$;
\item we have $\Pi_{(1),i}^{(V)}= \varnothing$, and the size of $\Delta_{(1),i}^{(V)}$ is $\big| \mathcal{C}_1^{(n)} \big| + \big| \mathcal{C}_{1,2}^{(n)} \big| - \big| \mathcal{G}_0^{(n)} \big| - \big| \mathcal{G}_2^{(n)} \big|$;
\item the length of $\Delta_{(2),i}^{(V)}$ is $\big| \mathcal{C}_2^{(n)} \big| + \big| \mathcal{C}_{1,2}^{(n)} \big| - \big| \mathcal{G}_0^{(n)} \big| - \big| \mathcal{G}_1^{(n)} \big|$.
\end{itemize}
\item if the \gls*{pcs} operates to achieve $(R_{S_{(1)}}^{\star 2}, R_{S_{(2)}}^{\star 2}, R_{W_{(1)}}^{\star 2}, R_{W_{(2)}}^{\star 2}) \subset \mathfrak{R}_{\text{MI-WTBC}}^{(2)}$:
\begin{itemize}
\item the inner-layer carries $S_{(2),1:L}^{(V)}$, but $\big| \mathcal{I}^{(n)} \big| = 0$;
\item we have $\Pi_{(1),i}^{(V)}= \varnothing$, and the size of $\Delta_{(1),i}^{(V)}$ is $\big| \mathcal{C}_1^{(n)} \big| + \big| \mathcal{C}_{1,2}^{(n)} \big| - \big| \mathcal{G}_0^{(n)} \big| - \big| \mathcal{G}_2^{(n)} \big|$;
\item the length of $\Delta_{(2),i}^{(V)}$ is $\big| \mathcal{C}_2^{(n)} \big| + \big| \mathcal{C}_{1,2}^{(n)} \big| - \big| \mathcal{G}_0^{(n)} \big| - \big| \mathcal{G}_1^{(n)} \big|$.
\end{itemize}
\end{itemize}
\end{enumerate}

\subsection{Construction of the outer-layers}\label{sec:PCSx1ol}

Consider that the \gls*{pcs} must achieve $(R_{S_{(1)}}^{\star k}, R_{S_{(2)}}^{\star k}, R_{W_{(1)}}^{\star k}, R_{W_{(2)}}^{\star k}) \! \subseteq \! \mathfrak{R}_{\text{MI-WTBC}}^{(k)}$, where~\mbox{$k \in [1,2]$}. In order to achieve this corner point, for $i \in [1,L]$ the \gls*{pcs} first constructs $\tilde{T}_{(k),i}^n$ associated to Receiver~$k$, and then forms $\tilde{T}_{(\bar{k}),i}^n$ associated to Receiver $\bar{k}$, where recall that $\bar{k} = [1,2]\setminus k$.

\begin{enumerate}
\item New sets associated to ${T}_{(k),1:L}^n$.
The sets $\mathcal{H}_{U_{(k)}|V}^{(n)}$, $\mathcal{L}_{U_{(k)}|V}^{(n)}$, $\mathcal{H}_{U_{(k)}|VZ}^{(n)}$ and $\mathcal{L}_{U_{(k)}|VY_{(k)}}^{(n)}$ associated to outer-layer $T_{(k)}^n = U_{(k)}^n G_n$ are defined as in \eqref{eq:HUVx}--\eqref{eq:LUVY_kx}. Besides the previous sets, define the following partition of $\mathcal{H}_{U_{(k)}|V}^{(n)}$:
\begin{align}
\mathcal{F}_0^{(n)} & \triangleq \mathcal{H}_{U_{(k)}|VZ}^{(n)} \cap \mathcal{L}_{U_{(k)}|VY_{(k)}}^{(n)}, \label{eq:f0} \\
\mathcal{F}_k^{(n)} & \triangleq \mathcal{H}_{U_{(k)}|VZ}^{(n)} \setminus \mathcal{L}_{U_{(k)}|VY_{(k)}}^{(n)}, \label{eq:fk} \\
\mathcal{J}_0^{(n)} & \triangleq \mathcal{H}_{U_{(k)}|V}^{(n)} \cap \big( \mathcal{H}_{U_{(k)}|VZ}^{(n)} \big)^{\text{C}} \cap \mathcal{L}_{U_{(k)}|VY_{(k)}}^{(n)}, \label{eq:j0} \\
\mathcal{J}_k^{(n)} & \triangleq \mathcal{H}_{U_{(k)}|V}^{(n)} \cap \big( \mathcal{H}_{U_{(k)}|VZ}^{(n)} \big)^{\text{C}} \setminus \mathcal{L}_{U_{(k)}|VY_{(k)}}^{(n)}. \label{eq:jk}
\end{align}
For $i \in [1,L]$, $\tilde{T}_{(k),i}\big[\mathcal{H}_{U_{(k)}|V}^{(n)}\big]$ will be suitable for storing uniformly distributed random sequences that are independent of $\tilde{V}_i^n$, and $\tilde{T}_{(k),i}\big[ \mathcal{F}_0^{(n)} \cup \mathcal{F}_k^{(n)} \big] = \tilde{T}_{(k),i}\big[\mathcal{H}_{U_{(k)}|VZ}^{(n)} \big]$ is suitable for storing information to be secured from the eavesdropper. Moreover, $\tilde{T}_{(k),i}\big[ \mathcal{F}_k^{(n)} \cup \mathcal{J}_k^{(n)} \big] = \tilde{T}_{(k),i}\big[\mathcal{H}_{U_{(k)}|V}^{(n)} \setminus \mathcal{L}_{U_{(k)}|VY_{(k)}}^{(n)} \big]$ is the uniformly distributed part independent of $\tilde{V}_i^n$ that is needed by Receiver~$k$ to reliably reconstruct $\tilde{T}_{(k),i}^n$ from observations~$\tilde{Y}_{(k),i}^n$ and sequence $\tilde{V}_{i}^n$ by performing \gls*{sc} decoding. 

\vspace{-0.1cm}
We consider that $I(U_{(k)};Y_{(k)}|V) \geq I(U_{(k)};Z|V)$ (see Remark~\ref{remark:assumptionU} and Remark~\ref{remark:assumptionU2}). Therefore, besides the partition defined in \eqref{eq:f0}--\eqref{eq:jk}, we define
\begin{align}
\mathcal{D}_{k}^{(n)} & \triangleq \! \text{ any subset of } \mathcal{F}_{0}^{(n)} \text{ with size } \big| \mathcal{J}_{k}^{(n)} \big|, \label{eq:dk} \\
\mathcal{L}_{k}^{(n)} & \triangleq \! \text{ any subset of } \mathcal{F}_{0}^{(n)} \setminus \mathcal{D}_{k}^{(n)} \text{ with size } \big\{ \big| \mathcal{C}_{k}^{(n)} \big| + \big| \mathcal{C}_{1,2}^{(n)} \big| - \big| \mathcal{G}_0^{(n)} \big| - \big| \mathcal{G}_{\bar{k}}^{(n)} \big| \big\}^{\! +}. \label{eq:lk}
\end{align}
The set $\mathcal{D}_{k}^{(n)}$ exists because we have
\begin{align}
\big| \mathcal{F}_{0}^{(n)} \big| - \big| \mathcal{J}_{k}^{(n)} \big| & = \Big| \mathcal{H}_{U_{(k)}|VZ}^{(n)} \cap \mathcal{L}_{U_{(k)}|VY_{(k)}}^{(n)} \Big| - \Big| \mathcal{H}_{U_{(k)}|V}^{(n)} \cap \big( \mathcal{H}_{U_{(k)}|VZ}^{(n)} \big)^{\text{C}} \setminus \mathcal{L}_{U_{(k)}|VY_{(k)}}^{(n)} \Big| \nonumber \\
& \geq \Big| \mathcal{H}_{U_{(k)}|VZ}^{(n)} \cap \mathcal{L}_{U_{(k)}|VY_{(k)}}^{(n)} \Big| - \Big|  \big( \mathcal{H}_{U_{(k)}|VZ}^{(n)} \big)^{\text{C}} \setminus \mathcal{L}_{U_{(k)}|VY_{(k)}}^{(n)} \Big|  \nonumber \\
& = \Big| \mathcal{H}_{U_{(k)}|VZ}^{(n)}  \Big| - \Big|  \big( \mathcal{L}_{U_{(k)}|VY_{(k)}}^{(n)} \big)^{\text{C}} \Big| \geq 0 , \label{eq:xnose}
\end{align} 
where the positivity holds by assumption and from applying source polarization \cite{arikan2010source} because
\begin{align}
\frac{1}{n} \Big( \Big| \mathcal{H}_{U_{(k)}|VZ}^{(n)}  \Big| -  \Big|  \big( \mathcal{L}_{U_{(k)}|VY_{(k)}}^{(n)} \big)^{\text{C}} \Big| \Big) \xrightarrow{n \rightarrow \infty} H(U_{(k)}|VZ) - H(U_{(k)}|VY_{(k)}). \label{eq:Dtheta}
\end{align}
On the other hand, according to \eqref{eq:assumpRate1Impl2x1}--\eqref{eq:assumpRate1Impl2x3}, if $k \! = 1$ and $\bar{k} =2$ then $\mathcal{L}_{k}^{(n)} = \emptyset$ in Situation~1, where $I(V;Z) \leq I (V;Y_{(1)}) \leq I(V;Y_{(2)})$; and $\mathcal{L}_{k}^{(n)} \neq \emptyset$ in Situation~2, where $I (V;Y_{(1)}) \! < \! I(V;Z) \! \leq \! I(V;Y_{(2)})$, and Situation~3, where $I (V;Y_{(1)}) \! \leq I(V;Y_{(2)}) \! < I(V;Z)$. Otherwise, if $k = 2$ and $\bar{k}=1$, then we have $\mathcal{L}_{k}^{(n)} = \emptyset$ in Situation 1 and  Situation 2, while $\mathcal{L}_{k}^{(n)} \neq \emptyset$ in Situation~3.
In situations where $\mathcal{L}_{k}^{(n)} \neq \emptyset$, if we consider only input distributions that imply $(R_{S_{(1)}}^{\star k}, R_{S_{(2)}}^{\star k}, R_{W_{(1)}}^{\star k}, R_{W_{(2)}}^{\star k}) \in \mathbb{R}^{4}_{+}$, set $\mathcal{L}_{k}^{(n)}$ exists because for $i \in [1,L]$ the entries $\tilde{A}_i(j)$ such that $j \in \mathcal{F}_{0}^{(n)} \setminus \big( \mathcal{D}_{k}^{(n)} \cup \mathcal{J}_{k}^{(n)} \big)$ will be intended for storing $S_{(k)}$, and the rate of $S_{(k)}$ carried in the inner-layer is is negligible (see Section~\ref{sec:performance_ratesx}).  

\vspace{-0.1cm}
\item New sets associated to ${T}_{(\bar{k}),1:L}^n$.
Sets $\smash{\mathcal{H}_{U_{(\bar{k})}|VU_{(k)}}^{(n)}}$, $\smash{\mathcal{L}_{U_{(\bar{k})}|VU_{(k)}}^{(n)}}$, $\smash{\mathcal{H}_{U_{(\bar{k})}|V U_{(k)} Z}^{(n)}}$ and $\smash{\mathcal{L}_{U_{(\bar{k})}|V Y_{(k)}}^{(n)}}$ associated to $\smash{T_{(\bar{k})}^n = U_{(\bar{k})}^n G_n}$ are defined in \eqref{eq:HUVUx}--\eqref{eq:LUVUY_kx}. Besides the previous sets, define:
\begin{align}
\mathcal{Q}_0^{(n)} & \triangleq \mathcal{H}_{U_{(\bar{k})}|VU_{(k)}Z}^{(n)} \cap \mathcal{L}_{U_{(\bar{k})}|VY_{(\bar{k})}}^{(n)}, \label{eq:w0} \\
\mathcal{Q}_{\bar{k}}^{(n)} & \triangleq \mathcal{H}_{U_{(\bar{k})}|VU_{(k)}Z}^{(n)} \setminus \mathcal{L}_{U_{(\bar{k})}|VY_{(\bar{k})}}^{(n)}, \label{eq:wk} \\
\mathcal{B}_0^{(n)} & \triangleq \mathcal{H}_{U_{(\bar{k})}|VU_{(k)}}^{(n)} \cap \big( \mathcal{H}_{U_{(\bar{k})}|VU_{(k)}Z}^{(n)} \big)^{\text{C}} \cap \mathcal{L}_{U_{(\bar{k})}|VY_{(\bar{k})}}^{(n)}, \label{eq:b0} \\
\mathcal{B}_{\bar{k}}^{(n)} & \triangleq \mathcal{H}_{U_{(\bar{k})}|VU_{(k)}}^{(n)} \cap \big(\mathcal{H}_{U_{(\bar{k})}|VU_{(k)}Z}^{(n)} \big)^{\text{C}} \setminus \mathcal{L}_{U_{(\bar{k})}|V Y_{(\bar{k})}}^{(n)}. \label{eq:bk}
\end{align}
For $i \in [1,L]$, the entries of $\tilde{T}_{(\bar{k}),i}\big[\mathcal{H}_{U_{(\bar{k})}|VU_{(k)}}^{(n)}\big]$ will be suitable for storing uniformly distributed random sequences that are independent of $(\tilde{V}_i^n,\tilde{U}_{(k),i}^n)$, and $\smash{\tilde{T}_{(\bar{k}),i}\big[\mathcal{Q}_{0}^{(n)} \cup \mathcal{Q}_{\bar{k}}^{(n)}\big]}$ will be suitable for storing information to be secured from the eavesdropper. Also,  the elements of $\tilde{T}_{(\bar{k}),i}\big[ \mathcal{B}_{\bar{k}}^{(n)} \cup \mathcal{Q}_{\bar{k}}^{(n)} \big]$ are required by Receiver~$\bar{k}$ to reliably construct the entire sequence $\tilde{T}_{(\bar{k}),i}^n$ from $(\tilde{V}_{i}^n,\tilde{Y}_{(\bar{k}),i}^n)$ by using \gls*{sc} decoding. Additionally, define
\vspace*{-0.1cm}
\begin{align}
\mathcal{O}_{\bar{k}}^{(n)} & \triangleq \text{ any subset of } \mathcal{Q}_0^{(n)} \text{ with size } \Big| \big( \mathcal{H}_{U_{(\bar{k})}|VU_{(k)}}^{(n)} \big) ^{\text{C}} \cap  \mathcal{H}_{U_{(\bar{k})}|V}^{(n)} \setminus \mathcal{L}_{U_{(\bar{k})}|VY_{(\bar{k})}}^{(n)} \Big|, \label{eq:ok} \\
\mathcal{N}_{\bar{k}}^{(n)} & \triangleq \text{ any subset of } \mathcal{Q}_0^{(n)} \setminus \mathcal{O}_{\bar{k}}^{(n)}  \text{ with size } \big| \mathcal{B}_{\bar{k}}^{(n)} \big|, \label{eq:qk}
\end{align}
\vspace*{-0.1cm}
and $\mathcal{M}_{\bar{k}}^{(n)}$, which is defined as follows. If $k=1$ and $\bar{k}=2$, then
\begin{align}
\mathcal{M}_{2}^{(n)}  & \! \triangleq \Scale[0.97]{\text{any subset of } \mathcal{Q}_0^{(n)} \! \setminus \! (\mathcal{O}_{2}^{(n)} \cup \mathcal{N}_{2}^{(n)}) \text{ with size } \big\{ \big| \mathcal{C}_2^{(n)} \big| + \big| \mathcal{C}_{1,2}^{(n)} \big| - \big| \mathcal{G}_0^{(n)} \big| - \big| \mathcal{G}_1^{(n)} \big| \big\}^{\! +}}. \label{eq:m2}
\end{align} 

\vspace{-0.35cm}
\noindent Consequently, $\mathcal{M}_{2}^{(n)} \neq \emptyset$ only in Situation 3, where $I(V;Y_{(1)}) \leq I(V;Y_{(2)}) < I(V;Z)$.
On the other hand, if $k = 2$ and $\bar{k}=1$, then
\vspace*{-0.16cm}
\begin{align}
\mathcal{M}_{1}^{(n)}  & \triangleq \text{any subset of } \mathcal{Q}_0^{(n)}  \setminus (\mathcal{O}_{1}^{(n)} \cup \mathcal{N}_{1}^{(n)}) \nonumber \\
& \quad \, \, \,  \text{ with size} 
\left\{ 
\begin{array}{ll}
\! \big| \mathcal{G}_1^{(n)} \big| + \big| \mathcal{C}_{1}^{(n)} \big| - \big| \mathcal{G}_2^{(n)} \big| - \big| \mathcal{C}_2^{(n)} \big| & \text{if } I(V;Z) \leq I(V;Y_{(2)}), \\
\! \big| \mathcal{C}_1^{(n)} \big| + \big| \mathcal{C}_{1,2}^{(n)} \big| - \big| \mathcal{G}_0^{(n)} \big| - \big| \mathcal{G}_2^{(n)} \big| & \text{otherwise}.
\end{array} \right. \label{eq:m1}
\end{align}
Recall that $I(V;Z) \leq I(V;Y_{(2)})$ in Situation 1, where $I(V;Z) < I(V;Y_{(1)}) \leq I(V;Y_{(2)})$, and in Situation 2, where $I(V;Y_{(1)}) < I(V;Z) \leq I(V;Y_{(2)})$.

If we consider only distributions implying $(R_{S_{(1)}}^{\star k}, R_{S_{(2)}}^{\star k}, R_{W_{(1)}}^{\star k}, R_{W_{(2)}}^{\star k}) \in \mathbb{R}^{4}_{+}$, then $\mathcal{O}_{\bar{k}}^{(n)}$, $\mathcal{N}_{\bar{k}}^{(n)}$ and $\mathcal{M}_{\bar{k}}^{(n)}$ must exist because, for $i \in [1,L]$, $\tilde{T}_{(\bar{k}),i} \big[ \mathcal{Q}_{0}^{(n)} \setminus \big( \mathcal{O}_{\bar{k}}^{(n)} \cup \mathcal{N}_{\bar{k}}^{(n)} \cup \mathcal{M}_{\bar{k}}^{(n)} \big) \big]$ is the only part that will be intended for storing confidential information~$S_{(\bar{k})}$.

\end{enumerate}

\vspace{0.15cm}
\noindent \textbf{Construction of} $\bm{\tilde{T}_{(1),1:L}^n}$ and $\bm{\tilde{T}_{(2),1:L}^n}$ for $\bm{(R_{S_{(1)}}^{\star 1}, R_{S_{(2)}}^{\star 1}, R_{W_{(1)}}^{\star 1}, R_{W_{(2)}}^{\star 1}) \in \mathfrak{R}_{\text{MI-WTBC}}^{(1)}}$

In this case ($k=1$ and $\bar{k}=2$), given $\tilde{V}_{i}^n$, for $i \in [1,L]$ the encoder first constructs $\tilde{T}_{(1),i}^n$ associated to Receiver~1. Then, given $\tilde{V}_{i}^n$ and $\tilde{T}_{(1),i}^n$, it forms $\tilde{T}_{(2),i}^n$ associated to Receiver~2.   

\begin{enumerate}
\item Construction of $\tilde{T}_{(1),i}^n$. 
Associated to $\tilde{T}_{(1),i}^n$, we have defined the sets $\mathcal{F}_0^{(n)}$, $\mathcal{F}_1^{(n)}$, $\mathcal{J}_0^{(n)}$ and $\mathcal{J}_1^{(n)}$ as in \eqref{eq:f0}--\eqref{eq:jk}, and $\mathcal{D}_1^{(n)}$ and $\mathcal{L}_1^{(n)}$ as in \eqref{eq:dk} and \eqref{eq:lk} respectively. 

For $i \in [1,L]$, let $W_{(1),i}^{(U)}$ be a uniformly distributed vector of length $\big| \mathcal{J}_{0}^{(n)} \cup \mathcal{J}_{1}^{(n)} \big|$ that represents part of the private message intended for Receiver~$1$. The encoder forms $\tilde{T}_{(1),i}\big[ \mathcal{J}_{0}^{(n)} \cup \mathcal{J}_{1}^{(n)} \big]$ by simply storing $W^{(U)}_{(1),i}$. We define $\Theta^{(U)}_{(1),i} \triangleq \tilde{T}_{(1),i} \big[ \mathcal{J}_{1}^{(n)} \big]$, which is required by Receiver~1 to reliably estimate $\tilde{T}_{(1),i}^n$. Hence, for $i \in [1,L-1]$, sequence~$\Theta^{(U)}_{(1),i+1}$ is repeated in $\tilde{T}_{(1),i} \big[ \mathcal{D}_{1}^{(n)} \big] \subseteq \tilde{T}_{(1),i} \big[ \mathcal{F}_{0}^{(n)} \big]$. This sequence is not repeated directly, but the encoder copies instead $\bar{\Theta}_{(1),i+1}^{(U)}$ that is obtained as follows. Let $\kappa_{\Theta}^{(U)}$ be a uniformly distributed key with length $\big| \mathcal{J}_{1}^{(n)} \big|$ that is privately shared between transmitter and Receiver~1. Then, for $i \in [1,L]$, we obtain $\bar{\Theta}^{(U)}_{(1),i} \triangleq \Theta_{(1),i}^{(U)} \oplus \kappa_{\Theta}^{(U)}$. Since $\kappa_{\Theta}^{(U)}$ is reused in all blocks, it is clear that its size becomes negligible in terms of rate for $L$ large enough.

For $i \in [1,L]$, let $S_{(1),i}^{(U)}$ be a uniformly distributed vector that represents part of the confidential message intended for Receiver~$1$. At Block~1, $S_{(1),1}^{(U)}$ has size $\big|  \big( \mathcal{F}_{0}^{(n)} \cup \mathcal{F}_{1}^{(n)} \big) \setminus \big( \mathcal{D}_{1}^{(n)} \cup \mathcal{L}_{1}^{(n)} \big) \big|$ and is stored in $\tilde{T}_{(1),1}\big[  \big( \mathcal{F}_{0}^{(n)} \cup \mathcal{F}_{1}^{(n)} \big) \setminus \big( \mathcal{D}_{1}^{(n)} \cup \mathcal{L}_{1}^{(n)} \big) \big]$; for $i \in [2,L-1]$, $S_{(1),i}^{(U)}$ has size $\big|  \mathcal{F}_{0}^{(n)} \setminus \big( \mathcal{D}_{1}^{(n)} \cup \mathcal{L}_{1}^{(n)} \big) \big|$ and is stored in $\tilde{T}_{(1),i}\big[  \mathcal{F}_{0}^{(n)} \setminus \big( \mathcal{D}_{1}^{(n)} \cup \mathcal{L}_{1}^{(n)} \big) \big]$; and at Block~$L$, $S_{(1),L}^{(U)}$ has size $\big|  \mathcal{F}_{0}^{(n)} \big|$ and is stored into $\tilde{T}_{(1),L}\big[  \mathcal{F}_{0}^{(n)} \big]$. Moreover, for $i \in [1,L]$, we define $\Lambda^{(U)}_{(1),i} \triangleq \tilde{T}_{(1),i}\big[ \mathcal{F}_{1}^{(n)} \big]$. For $i \in [2,L]$, $\Lambda^{(U)}_{(1),i-1}$ is repeated in $\tilde{T}_{(1),i}\big[ \mathcal{F}_{1}^{(n)} \big]$ and, therefore, $\Lambda^{(U)}_{(1),1}$, which contains part of $S_{(1),1}^{(U)}$, is replicated in all blocks. 

Furthermore, for $i \in [1,L-1]$, the encoder repeats\footnote{
From Section~\ref{sec:PCSx1il}, $\hat{\Pi}_{(1),1:L}^{(V)}=\varnothing$ when the \gls*{pcs} operates to achieve the corner point of $\mathfrak{R}_{\text{MI-WTBC}}^{(1)}$.}
$\big[ \Pi_{(1),i+1}^{(V)} \Delta_{(1),i+1}^{(V)}\big]$, which contain part of $\tilde{A}_{i+1}^n$, in $\tilde{T}_{(1),i}\big[ \mathcal{L}_{1}^{(n)} \big]$. According to the summary of the construction of $\tilde{A}_{1:L}\big[ \mathcal{G}^{(n)} \big]$ in the last part of Section~\ref{sec:PCSx1il}, notice that the length of $\Delta_{(1),i+1}^{(V)}$ is $\big| \mathcal{L}_{1}^{(n)} \big|$.

Then, for $i \in [1,L]$, given $\tilde{T}_{(1),i} \big[ \mathcal{H}_{U_{(1)}|V}^{(n)} \big]$ and $\tilde{V}_i^n$ the encoder forms the remaining entries of~$\tilde{T}_{(1),i}^n$ by using \gls*{sc} encoding: deterministic \gls*{sc} encoding for the elements of $\tilde{T}_{(1),i} \big[ \mathcal{L}_{U_{(1)}|V}^{(n)} \big]$ and random \gls*{sc} encoding for the entries of $\tilde{T}_{(1),i} \big[ \big( \mathcal{H}_{U_{(1)}|V}^{(n)} \big)^{\text{C}} \setminus \mathcal{L}_{U_{(1)}|V}^{(n)} \big]$.


For $i \in [1,L]$, the encoder obtains $\Phi_{(1),i}^{(U)} \triangleq \tilde{T}_{(1),i}  \big[ \big( \mathcal{H}_{U_{(1)}|V}^{(n)}  \big)^{\text{C}} \setminus \mathcal{L}_{U_{(1)}|VY_{(1)}}^{(n)} \big]$. Also, it obtains $\Upsilon_{(1)}^{(U)} \triangleq  \tilde{T}_{(1),1}  \big[ \mathcal{H}_{U_{(1)}|V}^{(n)} \setminus \mathcal{L}_{U_{(1)}|VY_{(1)}}^{(n)}  \big]$ from Block~1. The transmitter additionally sends $\big( \Upsilon_{(1)}^{(U)}, \Phi_{(1),1:L}^{(U)} \big) \oplus \kappa_{{\Upsilon \Phi}_{(1)}}^{(U)}$ to Receiver~$1$, where $\kappa_{{\Upsilon \Phi}_{(1)}}^{(U)}$ is a uniformly distributed key with size $L  \big| \big( \mathcal{H}_{U_{(1)}|V}^{(n)}  \big)^{\text{C}} \setminus  \mathcal{L}_{U_{(1)}|VY_{(1)}}^{(n)}  \big| +  \big| \mathcal{H}_{U_{(1)}|V}^{(n)} \setminus \mathcal{L}_{U_{(1)}|VY_{(1)}}^{(n)}  \big|$ that is privately shared between transmitter and Receiver~1.

Figure~\ref{fig:enc_outerR2_1} may graphically represent this construction of $\tilde{T}_{(1),1:L}^n$ if we do the following substitutions: $\mathcal{F}_0^{(n)} \leftarrow \mathcal{Q}_0^{(n)}$, $\mathcal{F}_1^{(n)} \leftarrow \mathcal{Q}_1^{(n)}$, $\mathcal{J}_0^{(n)} \leftarrow \mathcal{B}_0^{(n)}$, $\mathcal{J}_1^{(n)} \leftarrow \mathcal{B}_1^{(n)}$, $\mathcal{D}_1^{(n)} \leftarrow \mathcal{N}_1^{(n)}$, $\mathcal{L}_1^{(n)} \leftarrow \mathcal{M}_1^{(n)}$, $\emptyset \leftarrow \mathcal{O}_1^{(n)}$, $\big(\mathcal{H}_{U_{(1)}|V}^{(n)}\big)^{\text{C}} \leftarrow \big(\mathcal{H}_{U_{(1)}|VU_{(2)}}^{(n)}\big)^{\text{C}}$, $O_{(1),1:L}^{(U)} \leftarrow \varnothing$. Moreover, at Block~$i \in [1,L-1]$, the encoder repeats ${\Theta^{(U)}_{(1),i+1} \oplus \kappa^{(U)}_{\Theta}}$ instead of ${\Theta^{(U)}_{(1),i+1}}$.


\end{enumerate}

\begin{enumerate}
\setcounter{enumi}{1}
\item Construction of $\tilde{T}_{(2),i}^n$.
Associated to $\tilde{T}_{(2),i}^n$, we have defined the sets $\mathcal{Q}_0^{(n)}$, $\mathcal{Q}_2^{(n)}$, $\mathcal{B}_0^{(n)}$ and $\mathcal{B}_2^{(n)}$ as in \eqref{eq:w0}--\eqref{eq:bk}, and $\mathcal{O}_{2}^{(n)}$, $\mathcal{N}_2^{(n)}$ and $\mathcal{M}_2^{(n)}$ as in \eqref{eq:ok}--\eqref{eq:m2} respectively.

The construction of $\tilde{T}_{(2),1:L}^n$ is graphically summarized in Figure~\ref{fig:enc_outerR1_2}. For $i \in [1,L]$, let $W_{(2),i}^{(U)}$ be a uniformly distributed vector of length $\big| \mathcal{B}_{0}^{(n)} \cup \mathcal{B}_{2}^{(n)} \big|$ that represents the entire private message intended for Receiver~$2$. The encoder forms $\tilde{T}_{(2),i} \big[ \mathcal{B}_{0}^{(n)} \cup \mathcal{B}_{2}^{(n)} \big]$ by simply storing $W^{(U)}_{(2),i}$. We define $\Psi^{(U)}_{(2),i} \triangleq \tilde{T}_{(2),i} \big[ \mathcal{B}_{2}^{(n)} \big]$, which is required by Receiver~2 to reliably estimate $\tilde{T}_{(2),i}^n$. Thus, for $i \in [2,L]$, $\Psi^{(U)}_{(2),i-1}$ is repeated in $\tilde{T}_{(2),i} \big[ \mathcal{N}_{2}^{(n)} \big] \subseteq \tilde{T}_{(2),i} \big[ \mathcal{Q}_{0}^{(n)} \big]$. 

For $i \in [1,L]$, let $S_{(2),i}^{(U)}$ be a uniformly distributed vector that represents part of the confidential message intended for legitimate Receiver~$2$. At Block~1, $S_{(2),1}^{(U)}$ has size $\big|  \mathcal{Q}_{0}^{(n)} \cup \mathcal{Q}_{2}^{(n)} \big|$ and is stored in $\tilde{T}_{(2),1}\big[ \mathcal{Q}_{0}^{(n)} \cup \mathcal{Q}_{2}^{(n)} \big]$; and for $i \in [2,L]$, $S_{(2),i}^{(U)}$ has size $\big|  \mathcal{Q}_{0}^{(n)} \setminus \big( \mathcal{O}_{2}^{(n)} \cup \mathcal{N}_{2}^{(n)} \cup \mathcal{M}_{2}^{(n)} \big) \big|$ and is stored in $\tilde{T}_{(2),i}\big[  \mathcal{Q}_{0}^{(n)} \setminus \big( \mathcal{O}_{2}^{(n)} \cup \mathcal{N}_{2}^{(n)} \cup \mathcal{M}_{2}^{(n)} \big) \big]$. Moreover, for $i \in [1,L]$ we define $\Lambda^{(U)}_{(2),i} \triangleq \tilde{T}_{(2),i}\big[ \mathcal{Q}_{2}^{(n)} \big]$. For $i \in [2,L]$, $\Lambda^{(U)}_{(2),i-1}$ is repeated in $\tilde{T}_{(2),i}\big[ \mathcal{Q}_{2}^{(n)} \big]$ and, hence, $\Lambda^{(U)}_{(2),1}$, which contains part of $S_{(2),1}^{(U)}$, is replicated in all blocks. 

Furthermore, for $i \in [2,L]$ the encoder repeats $\Delta_{(2),i-1}^{(V)}$, which recall that contains part of $\tilde{A}_{i-1}^n$, in $\tilde{T}_{(2),i}\big[ \mathcal{M}_{2}^{(n)} \big]$. According to the summary of the construction of $\tilde{A}_{1:L}\big[ \mathcal{G}^{(n)} \big]$ in the last part of Section~\ref{sec:PCSx1il}, the length of $\Delta_{(2),i-1}^{(V)}$ matches with $\big| \mathcal{M}_{2}^{(n)} \big|$.

Then, for $i \in [1,L]$, given $\tilde{T}_{(2),i} \big[ \mathcal{H}_{U_{(2)}|VU_{(1)}}^{(n)} \big]$, $\tilde{V}_i^n$ and $\tilde{T}_{(1),i}^n$, the encoder forms the remaining entries of~$\tilde{T}_{(2),i}^n$ by using \gls*{sc} encoding. Now, notice that $\smash{\tilde{T}_{(2),i} \big[ \big( \mathcal{H}_{U_{(2)}|VU_{(1)}}^{(n)} \big)^{\text{C}} \big]}$ must depend not only on sequence $\tilde{V}_{i}^n$, but also on sequence $\tilde{U}_{(1),i}^n$ that was constructed before. Moreover, $\tilde{T}_{(2),i} \big[ \mathcal{L}_{U_{(2)}|VU_{(1)}}^{(n)} \big]$ is formed by performing deterministic \gls*{sc} encoding, while $\tilde{T}_{(2),i} \big[ \big( \mathcal{H}_{U_{(2)}|VU_{(1)}}^{(n)} \big)^{\text{C}} \setminus \mathcal{L}_{U_{(2)}|VU_{(1)}}^{(n)} \big]$ is drawn randomly.

\begin{figure}[h!]
\vspace{0.1cm}
\hspace{0.4cm}
\centering
\includegraphics[width=0.97\linewidth]{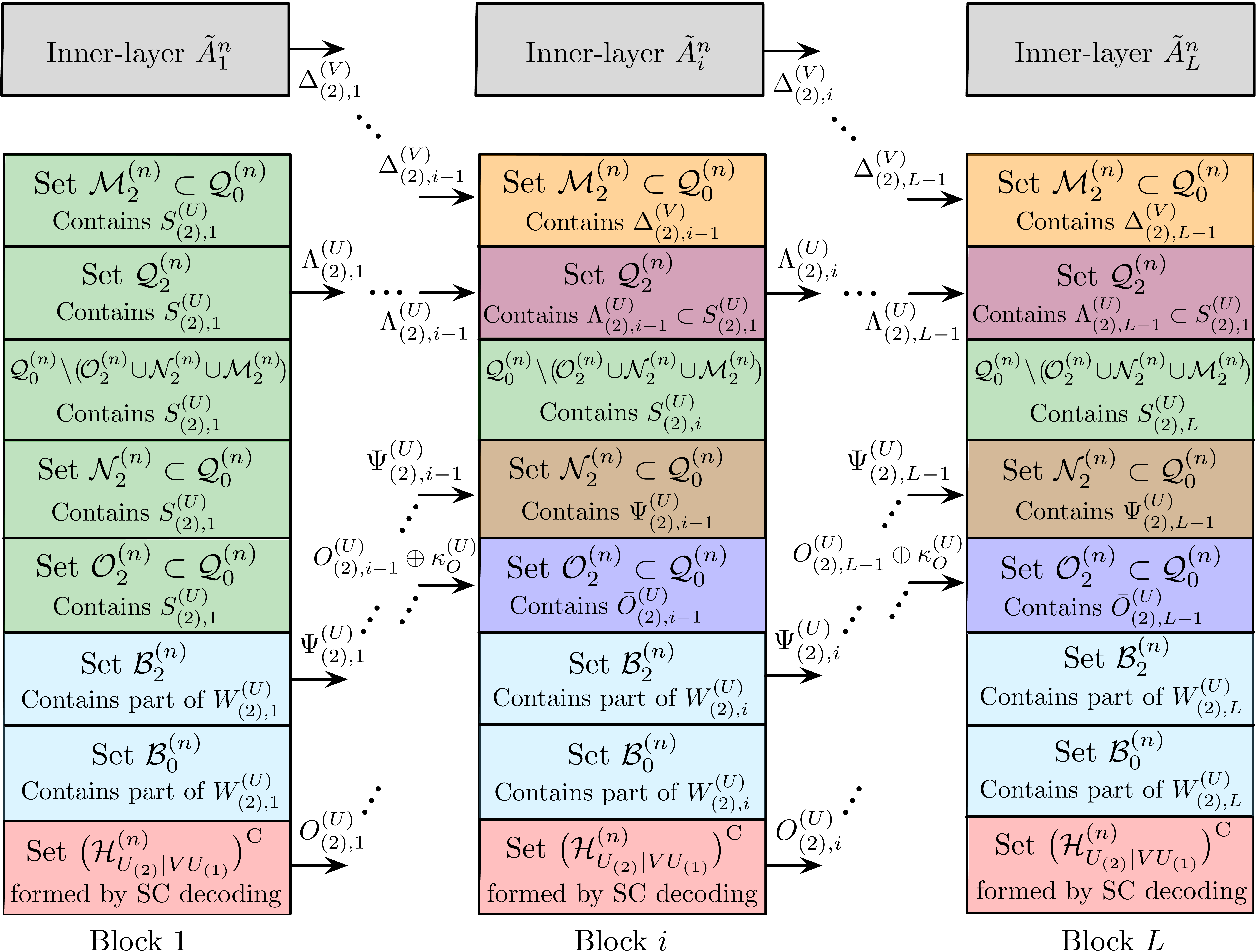}
\caption{\setstretch{1.35} Construction of outer-layer $\smash{\tilde{T}_{(2),1:L}}$ associated to Receiver~2 when the \gls*{pcs} must aproach the corner point $\smash{(R_{S_{(1)}}^{\star 1}, R_{S_{(2)}}^{\star 1}, R_{W_{(1)}}^{\star 1}, R_{W_{(2)}}^{\star 1}) \subseteq \mathfrak{R}_{\text{MI-WTBC}}^{(1)}}$. For any Block~$i \in [1,L]$, blue and green colors are used to represent the elements of $\smash{\tilde{T}_{(2),i}^n}$ that contain independent private and confidential information, respectively. For $i \in [2,L]$, orange, brown, red-purple and blue-purple colors represent those entries that contain information repeated from Block~$i-1$: $\tilde{T}_{(2),i}[\mathcal{M}_{2}^{(n)}]$ (in orange) repeats information from $\tilde{A}_{i-1}^n$, $\tilde{T}_{(2),i}[\mathcal{N}_{2}^{(n)}]$ (in brown) repeats $\Psi_{(2),i-1}^{(U)}$, $\tilde{T}_{(2),i}[\mathcal{O}_{2}^{(n)}]$ (in blue-purple) repeats $\bar{O}_{(2),i-1}^{(U)}$, and $\tilde{T}_{(2),i}[\mathcal{Q}_{2}^{(n)}]$ (in red-purple) repeats $\Lambda_{(2),i-1}^{(U)}$. Finally, for $i \in [1,L]$, $\tilde{T}_{(2),i}[(\mathcal{H}_{U_{(2)}|VU_{(1)}}^{(n)})^{\text{C}}]$ (in red) is drawn by \gls*{sc} encoding.}\label{fig:enc_outerR1_2}
\end{figure}

For $i \in [1,L]$, we define sequences $O_{(2),i}^{(U)} \triangleq \tilde{T}_{(2),i}  \big[ \big( \mathcal{H}_{U_{(2)}|VU_{(1)}}^{(n)}  \big)^{\text{C}} \cap\mathcal{H}_{U_{(2)}|V}^{(n)} \setminus \mathcal{L}_{U_{(2)}|VY_{(2)}}^{(n)} \big]$ and $\Phi_{(2),i}^{(U)} \triangleq \tilde{T}_{(2),i}  \big[ \big( \mathcal{H}_{U_{(2)}|V}^{(n)}  \big)^{\text{C}} \setminus \mathcal{L}_{U_{(2)}|VY_{(2)}}^{(n)} \big]$, where notice that $\big[O_{(2),i}^{(U)},\Phi_{(2),i}^{(U)}\big]$ contains those entries of $\tilde{T}_{(2),i}  \big[ \big( \mathcal{H}_{U_{(2)}|VU_{(1)}}^{(n)}  \big)^{\text{C}}\big]$ that are needed by Receiver~2 to reliably estimate $\tilde{T}_{(2),i}^n$, that is, $\big[O_{(2),i}^{(U)},\Phi_{(2),i}^{(U)}\big] = \tilde{T}_{(2),i}  \big[ \big( \mathcal{H}_{U_{(2)}|VU_{(1)}}^{(n)}  \big)^{\text{C}} \setminus \mathcal{L}_{U_{(2)}|VY_{(2)}}^{(n)} \big]$. Let $\kappa_{O}^{(U)}$ be a uniformly distributed key with size $\big| \big(\mathcal{H}_{U_{(2)}|VU_{(1)}}^{(n)}  \big)^{\text{C}} \cap\mathcal{H}_{U_{(2)}|V}^{(n)} \setminus \mathcal{L}_{U_{(2)}|VY_{(2)}}^{(n)} \big|$ that is privately-shared between transmitter and Receiver~2. For $i \in [1,L]$, we define $\bar{O}_{(2),i}^{(U)} \triangleq O_{(2),i}^{(U)} \oplus \kappa_O^{(U)}$. Since $O_{(2),i}^{(U)}$ is required by Receiver~2 to estimate $\tilde{T}_{(2),i}^n$, for $i \in [2,L]$ the encoder repeats $\bar{O}_{(2),i-1}^{(U)}$ in $\tilde{T}_{(2),i}\big[ \mathcal{O}^{(n)}_{2} \big]$. Notice that $\kappa_O^{(U)}$ is reused in all blocks, so it is clear that its size becomes negligible in terms of rate for $L$ large enough.
Furthermore, the encoder obtains $\Upsilon_{(2)}^{(U)} \triangleq  \tilde{T}_{(2),L}  \big[ \mathcal{H}_{U_{(2)}|V}^{(n)} \setminus \mathcal{L}_{U_{(2)}|VY_{(2)}}^{(n)}  \big]$ from Block~$L$, and the transmitter additionally sends $\big( \Upsilon_{(2)}^{(U)}, \Phi_{(2),1:L}^{(U)} \big) \oplus \kappa_{{\Upsilon \Phi}_{(2)}}^{(U)}$ to Receiver~$2$, $\kappa_{{\Upsilon \Phi}_{(2)}}^{(U)}$ being a uniformly distributed key with size $L  \big| \big( \mathcal{H}_{U_{(2)}|V}^{(n)}  \big)^{\text{C}} \setminus  \mathcal{L}_{U_{(2)}|VY_{(2)}}^{(n)}  \big| +  \big| \mathcal{H}_{U_{(2)}|V}^{(n)} \setminus \mathcal{L}_{U_{(2)}|VY_{(2)}}^{(n)}  \big|$ that is privately shared between transmitter and Receiver~2.

Finally, for $i \in [1,L]$, the encoder obtains $\tilde{X}_i^n \triangleq f\big(\tilde{V}_i^n, \tilde{U}_{(1),i}^n, \tilde{U}_{(2),i}^n \big)$, where recall that $f(\cdot)$ may be any deterministic one-to-one function. The transmitter sends $\tilde{X}_i^n$ over the \gls*{wtbc}, which induces the channel outputs $(\tilde{Y}_{(1),i}^n,\tilde{Y}_{(2),i}^n,\tilde{Z}_i^n)$.  

\end{enumerate}

\begin{remark}
For $i \in [2,L]$, notice that $O_{(2),i-1}^{(U)}$, which is not negligible in terms of rate, is almost uniformly distributed and independent of $\tilde{V}_{i-1}^n$, but is dependent on $(\tilde{V}_{i-1}^n, \tilde{T}_{(1),i-1}^n)$. Since $\tilde{T}_{(2),i}\big[ \mathcal{O}^{(n)}_{2} \big] \subset \tilde{T}_{(2),i}\big[ \mathcal{Q}^{(n)}_{0} \big]$ is suitable for storing sequences that are uniform and independent of $(\tilde{V}_{i}^n, \tilde{T}_{(1),i}^n)$, the secret-key $\kappa_{O}^{(U)}$ is used to ensure that $\bar{O}_{(2),i-1}^{(U)}$ is totally random (see Section~\ref{sec:distributionDMSx}).
\end{remark}

\vspace{0.15cm}
\noindent \textbf{Construction of} $\bm{\tilde{T}_{(1),1:L}^n}$ and $\bm{\tilde{T}_{(2),1:L}^n}$ for $\bm{(R_{S_{(1)}}^{\star 2}, R_{S_{(2)}}^{\star 2}, R_{W_{(1)}}^{\star 2}, R_{W_{(2)}}^{\star 2}) \in \mathfrak{R}_{\text{MI-WTBC}}^{(2)}}$

In this case ($k=2$ and $\bar{k}=1$), given $\tilde{V}_{i}^n$, for $i \in [1,L]$ the encoder first constructs $\tilde{T}_{(2),i}^n$ associated to Receiver~2. Then, given $\tilde{V}_{i}^n$ and $\tilde{T}_{(2),i}^n$, it forms $\tilde{T}_{(1),i}^n$ associated to Receiver~1.

\begin{enumerate}[leftmargin=0.6cm]
\item Construction of $\tilde{T}_{(2),i}^n$.
Associated to $\tilde{T}_{(2),i}^n$, we have defined the sets $\mathcal{F}_0^{(n)}$, $\mathcal{F}_2^{(n)}$, $\mathcal{J}_0^{(n)}$ and $\mathcal{J}_2^{(n)}$ as in \eqref{eq:f0}--\eqref{eq:jk}, and $\mathcal{D}_2^{(n)}$ and $\mathcal{L}_2^{(n)}$ as in \eqref{eq:dk} and \eqref{eq:lk} respectively.

For $i \in [1,L]$, let $W_{(2),i}^{(U)}$ be a uniformly distributed vector of length $\big| \mathcal{J}_{0}^{(n)} \cup \mathcal{J}_{2}^{(n)} \big|$ that represents part of the private message intended for Receiver~$2$. The encoder forms $\tilde{T}_{(2),i}\big[ \mathcal{J}_{0}^{(n)} \cup \mathcal{J}_{2}^{(n)} \big]$ by simply storing $W^{(U)}_{(2),i}$. Then, now we define $\Psi^{(U)}_{(2),i} \triangleq \tilde{T}_{(2),i} \big[ \mathcal{J}_{2}^{(n)} \big]$, which is required by Receiver~2 to reliably estimate $\tilde{T}_{(2),i}^n$. Thus, for $i \in [2,L]$, $\Psi^{(U)}_{(2),i-1}$ is repeated in $\tilde{T}_{(2),i} \big[ \mathcal{D}_{2}^{(n)} \big] \subseteq \tilde{T}_{(2),i} \big[ \mathcal{F}_{0}^{(n)} \big]$. This sequence is not repeated directly, but the encoder copies instead $\bar{\Psi}_{(2),i-1}^{(U)}$ that is obtained as follows. Let $\kappa_{\Psi}^{(U)}$ be a uniformly distributed key with length $\big| \mathcal{J}_{2}^{(n)} \big|$. Then, for $i \in [1,L]$, we obtain $\bar{\Psi}^{(U)}_{(2),i} \triangleq \Psi_{(2),i}^{(U)} \oplus \kappa_{\Psi}^{(U)}$. Since $\kappa_{\Psi}^{(U)}$ is reused in all blocks, its size is negligible in terms of rate for $L$ large enough.

For $i \in [1,L]$, let $S_{(2),i}^{(U)}$ be a uniformly distributed vector that represents the confidential message intended for Receiver~$2$. At Block~1, $S_{(2),1}^{(U)}$ has size $\big|  \mathcal{F}_{0}^{(n)} \cup \mathcal{F}_{2}^{(n)} \big|$ and is stored in $\tilde{T}_{(2),1}\big[ \mathcal{F}_{0}^{(n)} \cup \mathcal{F}_{2}^{(n)} \big]$; and for $i \in [2,L]$, $S_{(2),i}^{(U)}$ has size $\big|  \mathcal{F}_{0}^{(n)} \setminus \big( \mathcal{D}_2 \cup \mathcal{L}_{2}^{(n)} \big) \big|$ and is stored into $\tilde{T}_{(2),i}\big[\mathcal{F}_{0}^{(n)} \setminus \big( \mathcal{D}_2 \cup \mathcal{L}_{2}^{(n)} \big) \big]$. Moreover, for $i \in [1,L]$ we define $\Lambda^{(U)}_{(2),i} \triangleq \tilde{T}_{(2),i}\big[ \mathcal{F}_{2}^{(n)} \big]$. For $i \in [2,L]$, $\Lambda^{(U)}_{(2),i-1}$ is repeated in $\tilde{F}_{(2),i}\big[ \mathcal{F}_{2}^{(n)} \big]$ and, therefore, $\Lambda^{(U)}_{(2),1}$, which contains part of the confidential message $S_{(2),1}^{(U)}$, is replicated in all blocks. 

Furthermore, for $i \in [2,L]$, the encoder repeats $\Delta_{(2),i-1}^{(V)}$, which recall that contains part of $\tilde{A}_{i-1}^n$, in $\tilde{T}_{(2),i}\big[ \mathcal{L}_{2}^{(n)} \big]$. According to the summary of the construction of $\tilde{A}_{1:L}\big[ \mathcal{G}^{(n)} \big]$ in the last part of Section~\ref{sec:PCSx1il}, the length of sequence $\Delta_{(2),i-1}^{(V)}$ is $\big| \mathcal{L}_{2}^{(n)} \big|$.

For $i \in [1,L]$, given $\tilde{T}_{(2),i} \big[ \mathcal{H}_{U_{(2)}|V}^{(n)} \big]$ and $\tilde{V}_i^n$ the encoder forms $\tilde{T}_{(2),i} \big[ \big( \mathcal{H}_{U_{(2)}|V}^{(n)} \big)^{\text{C}} \big]$ by using \gls*{sc} encoding: deterministic for $\tilde{T}_{(2),i} \big[ \mathcal{L}_{U_{(2)}|V}^{(n)} \big]$, and random for $\tilde{T}_{(2),i} \big[ \big( \mathcal{H}_{U_{(2)}|V}^{(n)} \big)^{\text{C}} \setminus \mathcal{L}_{U_{(2)}|V}^{(n)} \big]$.

For $i \in [1,L]$, the encoder obtains $\Phi_{(2),i}^{(U)} \triangleq \tilde{T}_{(2),i}  \big[ \big( \mathcal{H}_{U_{(2)}|V}^{(n)}  \big)^{\text{C}} \setminus \mathcal{L}_{U_{(2)}|VY_{(1)}}^{(n)} \big]$. Also, it obtains $\Upsilon_{(2)}^{(U)} \triangleq  \tilde{T}_{(2),L}  \big[ \mathcal{H}_{U_{(2)}|V}^{(n)} \setminus \mathcal{L}_{U_{(2)}|VY_{(2)}}^{(n)}  \big]$ from Block~$L$. The transmitter additionally sends $\big( \Upsilon_{(2)}^{(U)}, \Phi_{(2),1:L}^{(U)} \big) \oplus \kappa_{{\Upsilon \Phi}_{(2)}}^{(U)}$ to Receiver~$2$, where $\kappa_{{\Upsilon \Phi}_{(2)}}^{(U)}$ now is a uniformly distributed key with size $L  \big| \big( \mathcal{H}_{U_{(2)}|V}^{(n)}  \big)^{\text{C}} \setminus  \mathcal{L}_{U_{(2)}|VY_{(2)}}^{(n)}  \big| +  \big| \mathcal{H}_{U_{(2)}|V}^{(n)} \setminus \mathcal{L}_{U_{(2)}|VY_{(2)}}^{(n)}  \big|$ that is privately shared between transmitter and Receiver~2.

Figure~\ref{fig:enc_outerR1_2} may graphically represent this construction of $\tilde{T}_{(2),1:L}^n$ if we do the following substitutions: $\mathcal{F}_0^{(n)} \leftarrow \mathcal{Q}_0^{(n)}$, $\mathcal{F}_2^{(n)} \leftarrow \mathcal{Q}_2^{(n)}$, $\mathcal{J}_0^{(n)} \leftarrow \mathcal{B}_0^{(n)}$, $\mathcal{J}_2^{(n)} \leftarrow \mathcal{B}_2^{(n)}$, $\mathcal{D}_2^{(n)} \leftarrow \mathcal{N}_2^{(n)}$, $\mathcal{L}_2^{(n)} \leftarrow \mathcal{M}_2^{(n)}$, $\emptyset \leftarrow \mathcal{O}_2^{(n)}$, $\big(\mathcal{H}_{U_{(2)}|V}^{(n)}\big)^{\text{C}} \leftarrow \big(\mathcal{H}_{U_{(2)}|VU_{(1)}}^{(n)}\big)^{\text{C}}$, $O_{(2),1:L}^{(U)} \leftarrow \varnothing$. Moreover, at Block~$i \in [2,L]$, the encoder repeats $\Psi^{(U)}_{(2),i-1} \oplus \kappa^{(U)}_{\Psi}$ instead of $\Psi^{(U)}_{(2),i-1}$.

\item Construction of $\tilde{T}_{(1),i}^n$.
Associated to $\tilde{T}_{(1),i}^n$, now we have defined $\mathcal{Q}_0^{(n)}$, $\mathcal{Q}_1^{(n)}$, $\mathcal{B}_0^{(n)}$ and $\mathcal{B}_1^{(n)}$ as in \eqref{eq:w0}--\eqref{eq:bk}, and $\mathcal{O}_1^{(n)}$, $\mathcal{N}_1^{(n)}$ and $\mathcal{M}_1^{(n)}$ as in \eqref{eq:ok}, \eqref{eq:qk} and \eqref{eq:m1} respectively.

The construction of $\tilde{T}_{(2),1:L}^n$ is graphically summarized in Figure~\ref{fig:enc_outerR2_1}. For $i \in [1,L]$, let $W_{(1),i}^{(U)}$ be a uniformly distributed vector of length $\big| \mathcal{B}_{0}^{(n)} \cup \mathcal{B}_{1}^{(n)} \big|$ that represents the entire private message intended for legitimate Receiver~$1$. The encoder forms $\tilde{T}_{(1),i} \big[ \mathcal{B}_{0}^{(n)} \cup \mathcal{B}_{1}^{(n)} \big]$ by simply storing $W^{(U)}_{(1),i}$. Then, we define $\Theta^{(U)}_{(1),i} \triangleq \tilde{T}_{(1),i} \big[ \mathcal{B}_{1}^{(n)} \big]$, which is required by Receiver~1 to reliably estimate $\tilde{T}_{(1),i}^n$. Hence, for $i \in [1,L-1]$, sequence $\Theta^{(U)}_{(1),i+1}$ is repeated in $\tilde{T}_{(1),i} \big[ \mathcal{N}_{1}^{(n)} \big] \subseteq \tilde{T}_{(1),i} \big[ \mathcal{Q}_{0}^{(n)} \big]$. 

For $i \in [1,L]$, let $S_{(1),i}^{(U)}$ be a uniform vector that represents the confidential message intended for Receiver~$1$. At Block~1, $S_{(1),1}^{(U)}$ has size $\big|  \big( \mathcal{Q}_{0}^{(n)} \cup \mathcal{Q}_{1}^{(n)} \big) \setminus \big( \mathcal{O}_{1}^{(n)} \cup \mathcal{N}_{1}^{(n)} \cup \mathcal{M}_{1}^{(n)} \big) \big|$ and is stored in $\tilde{T}_{(1),1}\big[  \big( \mathcal{Q}_{0}^{(n)} \cup \mathcal{Q}_{1}^{(n)} \big) \setminus \big( \mathcal{O}_{1}^{(n)} \cup \mathcal{N}_{1}^{(n)} \cup \mathcal{M}_{1}^{(n)} \big) \big]$; for $i \in [2,L-1]$, $S_{(1),i}^{(U)}$ has size $\big|  \mathcal{Q}_{0}^{(n)} \setminus \big( \mathcal{O}_{1}^{(n)} \cup \mathcal{N}_{1}^{(n)} \cup \mathcal{M}_{1}^{(n)} \big) \big|$ and is stored in $\tilde{T}_{(1),i}\big[  \mathcal{Q}_{0}^{(n)} \setminus \big( \mathcal{O}_{1}^{(n)} \cup \mathcal{N}_{1}^{(n)} \cup \mathcal{M}_{1}^{(n)} \big) \big]$; and at Block~$L$, $S_{(1),L}^{(U)}$ has size $\big|  \mathcal{Q}_{0}^{(n)} \big|$ and is stored into $\tilde{T}_{(1),L}\big[  \mathcal{Q}_{0}^{(n)} \big]$. Moreover, for $i \in [1,L]$ we define $\Lambda^{(U)}_{(1),i} \triangleq \tilde{T}_{(1),i}\big[ \mathcal{Q}_{1}^{(n)} \big]$. For $i \in [2,L]$, $\Lambda^{(U)}_{(1),i-1}$ is repeated in $\tilde{T}_{(1),i}\big[ \mathcal{Q}_{1}^{(n)} \big]$. Hence, $\Lambda^{(U)}_{(1),1}$, which contains part of $S_{(1),1}^{(U)}$, is replicated in all blocks. 

Furthermore, for $i \in [1,L-1]$, the encoder repeats $\big[ \Pi_{(1),i+1}^{(V)} \Delta_{(1),i+1}^{(V)}\big]$, which contains part of $\tilde{A}_{i+1}^n$, in $\tilde{T}_{(1),i}\big[ \mathcal{M}_{1}^{(n)} \big]$. According to the summary of the construction of $\tilde{A}_{1:L}\big[ \mathcal{G}^{(n)} \big]$ in the last part of Section~\ref{sec:PCSx1il}, the overall length of $\big[ \Pi_{(1),i+1}^{(V)}, \Delta_{(1),i+1}^{(V)}\big]$ is $\big| \mathcal{M}_{1}^{(n)} \big|$.

Then, for $i \in [1,L]$, given $\tilde{T}_{(1),i} \big[ \mathcal{H}_{U_{(1)}|VU_{(2)}}^{(n)} \big]$, $\tilde{V}_i^n$ and $\tilde{U}_{(2),i}^n$, the encoder forms the remaining entries of~$\tilde{T}_{(1),i}^n$ by using \gls*{sc} encoding. Now, $\tilde{T}_{(1),i} \big[ \big( \mathcal{H}_{U_{(1)}|VU_{(2)}}^{(n)} \big)^{\text{C}} \big]$ must depend not only on $\tilde{V}_{i}^n$, but also on $\tilde{U}_{(2),i}^n$. Moreover, $\tilde{T}_{(1),i} \big[ \mathcal{L}_{U_{(1)}|VU_{(2)}}^{(n)} \big]$ is formed by performing deterministic \gls*{sc} encoding, while $\tilde{T}_{(1),i} \big[ \big( \mathcal{H}_{U_{(1)}|VU_{(2)}}^{(n)} \big)^{\text{C}} \setminus \mathcal{L}_{U_{(1)}|VU_{(2)}}^{(n)} \big]$ is drawn randomly.

\begin{figure}[h!]
\vspace{0.1cm}
\hspace{0.4cm}
\centering
\includegraphics[width=0.97\linewidth]{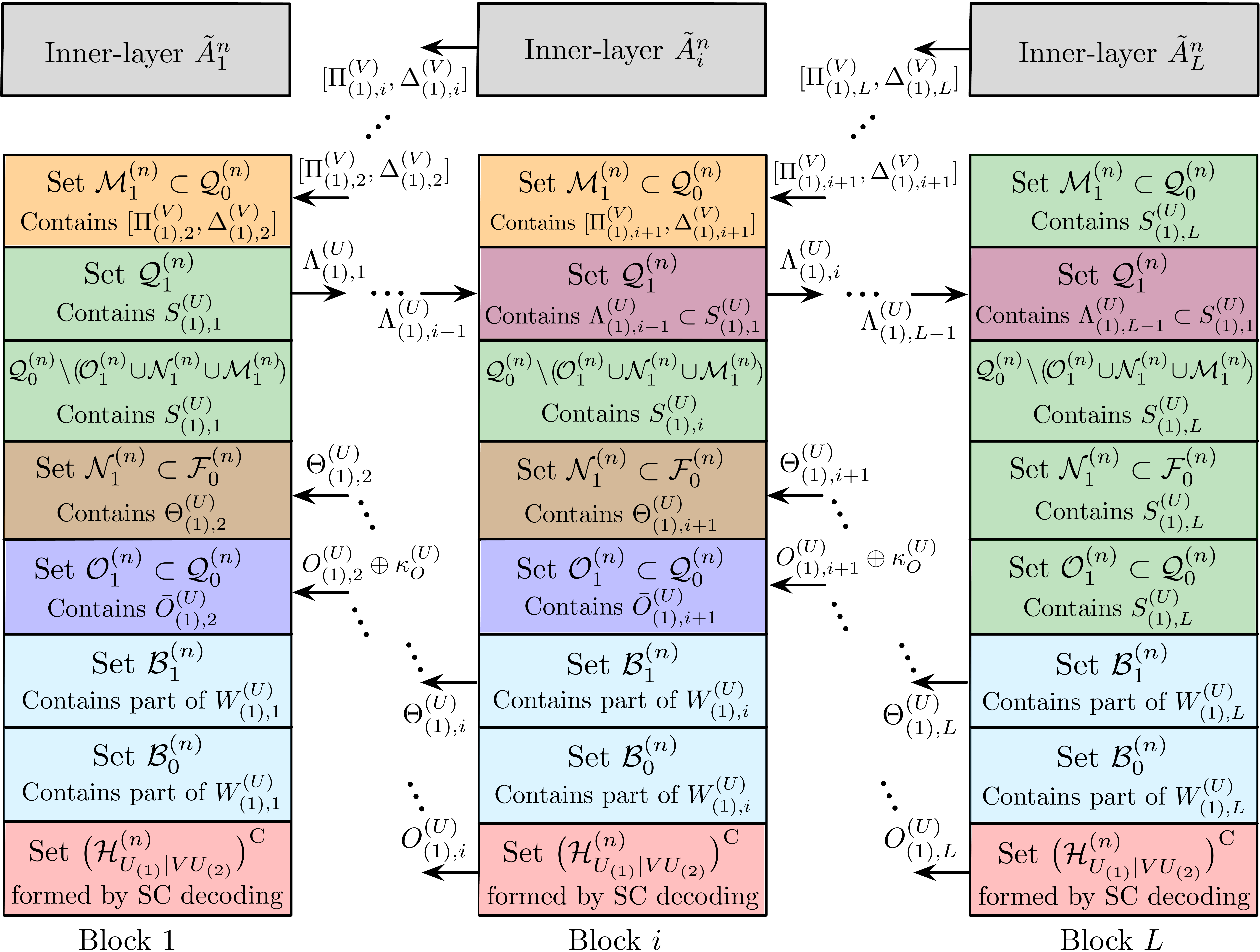}
\caption{\setstretch{1.35} Construction of outer-layer $\smash{\tilde{T}_{(1),1:L}}$ associated to Receiver~1 when the \gls*{pcs} must aproach the corner point $\smash{(R_{S_{(1)}}^{\star 2}, R_{S_{(2)}}^{\star 2}, R_{W_{(1)}}^{\star 2}, R_{W_{(2)}}^{\star 2}) \subseteq \mathfrak{R}_{\text{MI-WTBC}}^{(2)}}$. For any Block~$i \in [1,L]$, blue and green colors are used to represent the elements of $\smash{\tilde{T}_{(1),i}^n}$ that contain independent private and confidential information, respectively. For $i \in [1,L-1]$, orange, brown and blue-purple colors represent those entries that contain information repeated from Block~$i+1$: $\tilde{T}_{(1),i}[\mathcal{M}_{1}^{(n)}]$ (in orange) repeats information from $\tilde{A}_{i+1}^n$, $\tilde{T}_{(1),i}[\mathcal{N}_{1}^{(n)}]$ (in brown) repeats $\Theta_{(1),i+1}^{(U)}$, and $\tilde{T}_{(1),i}[\mathcal{O}_{1}^{(n)}]$ (in blue-purple) repeats $\bar{O}_{(1),i+1}^{(U)}$. Recall that $\Lambda_{(1),1}^{(U)}$, which contain part of the confidential information of Block~1, is replicated in $\tilde{T}_{(1),2:L}[\mathcal{Q}_{1}^{(n)}]$ (in red-purple). Finally, for $i \in [1,L]$, the elements of $\tilde{T}_{(1),i}[(\mathcal{H}_{U_{(1)}|VU_{(2)}}^{(n)})^{\text{C}}]$ (in red) are drawn by \gls*{sc} encoding.}\label{fig:enc_outerR2_1}
\end{figure}

For $i \in [1,L]$, we define sequences $O_{(1),i}^{(U)} \triangleq \tilde{T}_{(1),i}  \big[ \big( \mathcal{H}_{U_{(1)}|VU_{(2)}}^{(n)}  \big)^{\text{C}} \cap \mathcal{H}_{U_{(1)}|V}^{(n)} \setminus \mathcal{L}_{U_{(1)}|VY_{(1)}}^{(n)} \big]$ and $\Phi_{(1),i}^{(U)} \triangleq \tilde{T}_{(1),i}  \big[ \big( \mathcal{H}_{U_{(1)}|V}^{(n)}  \big)^{\text{C}} \setminus \mathcal{L}_{U_{(1)}|VY_{(1)}}^{(n)} \big]$, where notice that $\big[O_{(1),i}^{(U)},\Phi_{(1),i}^{(U)}\big]$ contains those entries of $\tilde{T}_{(1),i}  \big[ \big( \mathcal{H}_{U_{(1)}|VU_{(2)}}^{(n)}  \big)^{\text{C}}\big]$ that are needed by Receiver~1 to reliably estimate $\tilde{T}_{(1),i}^n$, that is, $\big[O_{(1),i}^{(U)},\Phi_{(1),i}^{(U)}\big] = \tilde{T}_{(1),i}  \big[ \big( \mathcal{H}_{U_{(1)}|VU_{(2)}}^{(n)}  \big)^{\text{C}} \setminus \mathcal{L}_{U_{(1)}|VY_{(1)}}^{(n)} \big]$. Let $\kappa_{O}^{(U)}$ be a uniformly distributed key with size $\big| \big(\mathcal{H}_{U_{(1)}|VU_{(2)}}^{(n)}  \big)^{\text{C}} \cap\mathcal{H}_{U_{(1)}|V}^{(n)} \setminus \mathcal{L}_{U_{(1)}|VY_{(1)}}^{(n)} \big|$ that is privately-shared between transmitter and Receiver~1. For $i \in [1,L]$, we define $\bar{O}_{(1),i}^{(U)} \triangleq O_{(1),i}^{(U)} \oplus \kappa_O^{(U)}$. Since $O_{(1),i}^{(U)}$ is required by Receiver~1 to estimate $\tilde{T}_{(1),i}^n$, for $i \in [1,L-1]$ the encoder repeats $\bar{O}_{(1),i+1}^{(U)}$ in $\tilde{T}_{(1),i}\big[ \mathcal{O}^{(n)}_{1} \big]$. Notice that $\kappa_O^{(U)}$ is reused in all blocks, so it is clear that its size becomes negligible in terms of rate for $L$ large enough.
Furthermore, the encoder obtains $\Upsilon_{(1)}^{(U)} \triangleq  \tilde{T}_{(1),1}  \big[ \mathcal{H}_{U_{(1)}|V}^{(n)} \setminus \mathcal{L}_{U_{(1)}|VY_{(1)}}^{(n)}  \big]$ from Block~$1$. The transmitter additionally sends $\big( \Upsilon_{(1)}^{(U)}, \Phi_{(1),1:L}^{(U)} \big) \oplus \kappa_{{\Upsilon \Phi}_{(1)}}^{(U)}$ to Receiver~$1$, where $\kappa_{{\Upsilon \Phi}_{(1)}}^{(U)}$ is a uniformly distributed key with size $L  \big| \big( \mathcal{H}_{U_{(1)}|V}^{(n)}  \big)^{\text{C}} \setminus  \mathcal{L}_{U_{(1)}|VY_{(1)}}^{(n)}  \big| +  \big| \mathcal{H}_{U_{(1)}|V}^{(n)} \setminus \mathcal{L}_{U_{(1)}|VY_{(1)}}^{(n)}  \big|$ that is privately shared between transmitter and Receiver~1.


Finally, for $i \in [1,L]$, the encoder obtains $\tilde{X}_i^n \triangleq f\big(\tilde{V}_i^n, \tilde{U}_{(1),i}^n, \tilde{U}_{(2),i}^n \big)$, where recall that $f(\cdot)$ may be any deterministic one-to-one function. The transmitter sends $\tilde{X}_i^n$ over the \gls*{wtbc}, which induces the channel outputs $(\tilde{Y}_{(1),i}^n,\tilde{Y}_{(2),i}^n,\tilde{Z}_i^n)$.

\end{enumerate}

\begin{remark}\label{remark:assumptionU}
Consider the construction of $\tilde{T}_{(k),1:L}^n$ to achieve $(R_{S_{(1)}}^{\star k}, R_{S_{(2)}}^{\star k}, R_{W_{(1)}}^{\star k}, R_{W_{(2)}}^{\star k}) \in \mathfrak{R}_{\text{MI-WTBC}}^{(k)}$. If we consider that $I(U_{(k)};Y_{(k)}|V) < I(U_{(k)};Z|V)$, according to \eqref{eq:xnose} and \eqref{eq:Dtheta}, notice that $\big| \mathcal{F}_{0}^{(n)} \big| - \big| \mathcal{J}_{k}^{(n)} \big| < 0$. Consequently, if $k=1$, for $i \in [1,L-1]$ the encoder cannot repeat entirely the sequence $\bar{\Theta}_{(1),i+1}^{(U)}$ of length $\big| \mathcal{J}_{1}^{(n)} \big|$ in some elements of $\tilde{T}_{(1),i}\big[ \mathcal{F}_{0}^{(n)} \big]$. Similarly, if $k=2$, for $i \in [2,L]$ the encoder cannot repeat the sequence $\bar{\Psi}_{(2),i-1}^{(U)}$ of length $\big| \mathcal{J}_{2}^{(n)} \big|$ in $\tilde{T}_{(2),i}\big[ \mathcal{F}_{0}^{(n)} \big]$. 

Therefore, under this assumption, the encoding strategy will be as follows. If $k=1$, for $i \in [1,L]$ we will define $\Delta_{(1),i}^{(U)}$ as any part of $\bar{\Theta}_{(1),i}^{(U)}$ with size $\big| \mathcal{J}_{1}^{(n)} \big| - \big| \mathcal{F}_{0}^{(n)} \big|$. For $i \in [1,L-1]$, the sequence $\Delta_{(1),i+1}^{(U)}$ will be repeated in some part of the inner-layer $\tilde{A}_{i}\big[ \mathcal{I}^{(n)} \big]$, whereas the remaining elements of $\bar{\Theta}_{(1),i+1}^{(U)}$ will be stored in $\tilde{T}_{(1),i}\big[ \mathcal{F}_{0}^{(n)} \big]$. Similarly, if $k=2$, for $i \in [1,L]$ we will define $\Delta_{(2),i}^{(U)}$ as any part of $\bar{\Psi}_{(2),i}^{(U)}$ with size $\big| \mathcal{J}_{2}^{(n)} \big| - \big| \mathcal{F}_{0}^{(n)} \big|$. For $i \in [2,L]$, the sequence $\Delta_{(2),i-1}^{(U)}$ will be repeated in some part of the inner-layer $\tilde{A}_{i}\big[ \mathcal{I}^{(n)} \big]$, whereas the remaining elements of $\bar{\Psi}_{(2),i-1}^{(U)}$ will be repeated in $\tilde{T}_{(2),i}\big[ \mathcal{F}_{0}^{(n)} \big]$.
%
%
\end{remark}

\begin{remark}\label{remark:assumptionU2}
If we consider input distributions that imply $(R_{S_{(1)}}^{\star k}, R_{S_{(2)}}^{\star k}, R_{W_{(1)}}^{\star k}, R_{W_{(2)}}^{\star k}) \in \mathbb{R}^{4}_{+}$ for some $k \in [1,2]$, it is clear that if $I(U_{(k)};Y_{(k)}|V) < I(U_{(k)};Z|V)$ then we must only consider those Situations (among 1 to 3) in the inner-layer where $I(V;Y_{(k)}) \geq I(V;Z)$. In this case, part of the elements of the inner-layer that previously carried confidential information may be used now to carry $\Delta_{(1),2:L}^{(U)}$ or $\Delta_{(2),1:L-1}^{(U)}$.

We will not go into the details of the encoding/decoding when $I(U_{(k)};Y_{(k)}|V) \! < \! I(U_{(k)};Z|V)$ because both the construction and the performance analysis will be very similar to those of the the contemplated cases such that the outer-layer must repeat some elements of the inner-layer.
\end{remark}

\subsection{Decoding}\label{sec:decodingx}
By using $\kappa_{\Upsilon\Phi_{(k)}}^{(V)}$ and $\kappa_{\Upsilon\Phi_{(k)}}^{(U)}$, consider that $\big( \Phi_{(k),1:L}^{(V)}, \Upsilon_{(k)}^{(V)} \big)$ and $\big( \Phi_{(k),1:L}^{(U)}, \Upsilon_{(k)}^{(U)} \big)$ have been reliably obtained by Receiver~$k \in [1,2]$ before starting the decoding process.


\newpage
\noindent \textbf{Decoding at Receiver~1}

This receiver forms the estimates $\hat{A}_{1:L}^n$ and $\hat{T}_{(1),1:L}^n$ of $\tilde{A}_{1:L}^n$ and $\tilde{T}_{(1),1:L}^n$ respectively by going forward, i.e., from $(\hat{A}_{1}^n, \hat{T}_{(1),1}^n)$ to $(\hat{A}_{L}^n, \hat{T}_{(1),L}^n)$. For~$i \in [1,L]$, it forms $\hat{A}_{i}^n$ first, and then $\hat{T}_{(1),i}^n$.

The decoding process at Receiver~1 when the \gls*{pcs} must achieve the corner point of $\mathfrak{R}_{\text{MI-WTBC}}^{(1)}$ or $\mathfrak{R}_{\text{MI-WTBC}}^{(2)}$ is summarized in Algorithm~\ref{alg:decoding1x}. Despite $\tilde{A}_{1:L}^n$ does not carry information intended for Receiver~$1$ when the \gls*{pcs} operates to achieve the corner point of $\mathfrak{R}_{\text{MI-WTBC}}^{(2)}$, recall that this receiver needs $\tilde{A}_{1:L}^n$ to reliably reconstruct $\tilde{T}_{(1),1:L}^n$.

In both cases, Receiver~1 constructs $\hat{A}^n_{1}$ and $\hat{T}^n_{(1),1}$ as follows. Given $\smash{\big(\Upsilon_{(1)}^{(V)},\Phi_{(1),1}^{(V)}\big)}$ and $\big(\Upsilon_{(1)}^{(U)}, \Phi_{(1),1}^{(U)} \big)$, it knows $\smash{\tilde{A}_{1}\big[ \big(\mathcal{L}_{V|Y_{(1)}}^{(n)} \big)^{\text{C}}\big]}$ and $\smash{\tilde{T}_{(1),1}\big[ \big(\mathcal{L}_{U_{(1)}|VY_{(1)}}^{(n)} \big)^{\text{C}}\big]}$, respectively. Therefore, from observations $\tilde{Y}_{(1),1}^n$ and $\smash{\tilde{A}_{1}\big[ \big(\mathcal{L}_{V|Y_{(1)}}^{(n)} \big)^{\text{C}}\big]}$, it entirely constructs $\hat{A}^n_{1}$ by performing \gls*{sc} decoding. Then, from sequences $\tilde{Y}_{(1),1}^n$, $\smash{\tilde{T}_{(1),1}\big[ \big(\mathcal{L}_{U_{(1)}|VY_{(1)}}^{(n)} \big)^{\text{C}}\big]}$ and $\hat{V}_{1}^n = \hat{A}_1^n G_n$, this receiver form $\hat{T}^n_{(1),1}$. 

Recall that $\Lambda_1^{(V)}$ and $\Lambda_{(1),1}^{(U)}$ have been replicated in all blocks. Thus, Receiver~1 obtains $\smash{\hat{\Lambda}_{1:L}^{(V)}} = \hat{A}_{1}\big[\mathcal{R}_{\Lambda}^{(n)}\big]$, while it obtains $\smash{\hat{\Lambda}_{(1),1:L}^{(U)}} = \hat{T}_{(1),1}\big[\mathcal{F}_{1}^{(n)}\big]$ or $\smash{\hat{\Lambda}_{(1),1:L}^{(U)}} = \hat{T}_{(1),1}\big[\mathcal{Q}_{1}^{(n)}\big]$ depending on whether the \gls*{pcs} must achieve the corner point of $\mathfrak{R}_{\text{MI-WTBC}}^{(1)}$ or $\mathfrak{R}_{\text{MI-WTBC}}^{(2)}$ respectively.

For $i \in [1,L-1]$, consider that $\hat{A}_{1:i}^n$ and $\hat{T}_{(1),1:i}^n$ have already been constructed. Then, the construction of $\hat{A}_{i+1}^n$ and $\hat{T}_{(1),i+1}^n$ is slightly different depending on whether the \gls*{pcs} must achieve the corner point of $\mathfrak{R}_{\text{MI-WTBC}}^{(1)}$ ($k=1$) or $\mathfrak{R}_{\text{MI-WTBC}}^{(2)}$ ($k=2$):
\begin{enumerate}
\item[] $k=1)$ From $\hat{A}_i^n$ and $\hat{T}_{(1),i}^n$, it obtains $\Upsilon_{(1),i+1}^{\prime (V)} \triangleq \big(\hat{{\Psi}}_{1,i}^{(V)} , \hat{{\Gamma}}_{2,i}^{(V)} , \hat{\Pi}_{(2),i}^{(V)}, \hat{\Theta}_{i+1}^{(V)} ,\hat{\Gamma}_{i+1}^{(V)} , \hat{\Lambda}_{i+1}^{(V)} \big)$. Notice in Algorithm~\ref{alg:decoding1x} that secret-keys $\kappa_{\Theta}^{(V)}$ and $\kappa_{\Gamma}^{(V)}$ are needed to obtain $\hat{\Theta}_{i+1}^{(V)}$ and $\hat{\Gamma}_{i+1}^{(V)}$, respectively. Moreover, $\Psi_{i-1}^{(V)}$ and $\Gamma_{i-1}^{(V)}$ may also be necessary for this purpose, and notice that they are available because $\big(\tilde{A}_{i-1}^n, \hat{T}_{(1),i-1}^n\big)$ has already been constructed. Recall that part of $\hat{\Theta}_{i+1}^{(V)}$ and $\hat{\Gamma}_{i+1}^{(V)}$ are obtained from $\hat{T}_{(1),i}^n$. On the other hand, from $\hat{T}_{(1),i}^n$, it obtains $\Upsilon_{(1),i+1}^{\prime (U)} \triangleq \big( \hat{{\Theta}}_{(1),i+1}^{(U)}, \hat{\Lambda}_{(1),i+1}^{(U)} \big)$, and $\kappa_{\Theta}^{(U)}$ is needed to get $\hat{{\Theta}}_{(1),i+1}^{(U)}$. 

\item[] $k=2$) From $\hat{A}_i^n$ and $\hat{T}_{(1),i}^n$, it obtains $\hat{\Upsilon}_{(1),i+1}^{\prime (V)} \triangleq  \big(\hat{\bar{\Psi}}_{1,i}^{(V)} , \hat{\bar{\Gamma}}_{2,i}^{(V)} , \hat{\Pi}_{(2),i}^{(V)}, \hat{\Theta}_{i+1}^{(V)} ,\hat{\Gamma}_{i+1}^{(V)} , \hat{\Lambda}_{i+1}^{(V)} \big)$ and $\Upsilon_{(1),i+1}^{\prime (U)} \triangleq \big( \hat{{\Theta}}_{(1),i+1}^{(U)}, \hat{\Lambda}_{(1),i+1}^{(U)}, \hat{{O}}_{(1),i+1}^{(U)} \big)$. Now, for $i \in [1,L-1]$ the encoder have repeated an \emph{encrypted} version of ${\Psi}_{i}^{(V)}$ and ${\Gamma}_{i}^{(V)}$ at Block~$i+1$, whereas ${\Theta}_{i+1}^{(V)}$, ${\Gamma}_{i+1}^{(V)}$ and ${\Theta}_{(1),i+1}^{(U)}$ have been repeated in Block~$i$ directly. Hence, notice in Algorithm~\ref{alg:decoding1x} that secret-keys $\kappa_{\Psi}^{(V)}$ and $\kappa_{\Gamma}^{(V)}$ are needed, whereas neither $\kappa_{\Theta}^{(V)}$ nor $\kappa_{\Theta}^{(U)}$ are used in this case. Moreover, notice that Receiver~1 now obtains $\hat{{O}}_{(1),i+1}^{(U)} = \hat{T}_{(1),i}[\mathcal{O}_{1}^{(n)}] \oplus \kappa_{O}^{(U)}$, which recall that contains part of the elements of $\hat{T}_{(1),i+1}^n$ that have been drawn by performing \gls*{sc} encoding and are needed by Receiver~1 to reliably reconstruct $\hat{T}_{(1),i+1}^n$. 
\end{enumerate}

Finally, given $\smash{\big(\Upsilon_{(1),i+1}^{\prime (V)},\Phi_{(1),i+1}^{(V)}\big)} \supseteq \hat{A}_{i+1}\big[ \big(  \mathcal{L}_{V|Y_{(1)}}^{(n)}\big)^{\text{C}}\big]$ and $\tilde{Y}_{(1),i+1}^n$, it performs \gls*{sc} decoding to construct $\hat{A}_{i+1}^n$. Then, given $\smash{\big(\Upsilon_{(1),i+1}^{\prime (U)},\Phi_{(1),i+1}^{(U)}\big)} \supseteq \hat{T}_{(1),i+1}\big[ \big(  \mathcal{L}_{U_{(1)}|VY_{(1)}}^{(n)}\big)^{\text{C}}\big]$, $\hat{V}_{i+1}^n = \hat{A}_{i+1}^n G_n$ and $\tilde{Y}_{(1),i+1}^n$, it performs \gls*{sc} decoding to construct $\hat{T}_{(1),i+1}^n$.

\begin{algorithm}[h!]
\caption{\small{Decoding at Receiver~1 when \gls*{pcs} operates to achieve the corner point of $\mathfrak{R}_{\text{MI-WTBC}}^{(k)}$}}\label{alg:decoding1x}
\begin{algorithmic}[1]
\Require $\Upsilon_{(1)}^{(V)}$, $\Phi_{(1),1:L}^{(V)}$, $\Upsilon_{(1)}^{(U)}$, $\Phi_{(1),1:L}^{(U)}$, and $\tilde{Y}_{(1),1:L}^n$.
\Require \textbf{if} $k=1$ \textbf{then} $\{\kappa_{\Gamma}^{(V)},\kappa_{\Theta}^{(V)},\kappa_{\Theta}^{(U)}\}$ \textbf{else} $\{\kappa_{\Gamma}^{(V)},\kappa_{\Psi}^{(V)},\kappa_{O}^{(U)}\}$
\State $\hat{A}_1^n \leftarrow \big( \Upsilon_{(1)}^{(V)}, \Phi_{(1),1}^{(V)}, \tilde{Y}_{(1),1}^n \big)$
\Comment by using \gls*{sc} decoding
\State $\hat{T}_{(1),1}^n \leftarrow \big( \Upsilon_{(1)}^{(U)}, \Phi_{(1),1}^{(U)}, \hat{A}^n_{1}G_n, \tilde{Y}_{(1),1}^n \big)$
\Comment by using \gls*{sc} decoding
\State $\hat{\Lambda}_{2:L}^{(V)} \leftarrow \hat{A}_1\big[\mathcal{R}_{\Lambda}^{(n)} \big]$
\State \textbf{if} $k=1$ \textbf{then} $\hat{\Lambda}_{2:L}^{(U)} \leftarrow \hat{T}_{(1),1}\big[\mathcal{F}_{1}^{(n)} \big]$ \textbf{else} $\hat{\Lambda}_{2:L}^{(U)} \leftarrow \hat{T}_{(1),1}\big[\mathcal{Q}_{1}^{(n)} \big]$
\For{$i = 1$ \text{to} $L-1$} 
\State $\hat{\Pi}_{(2),i}^{(V)} \leftarrow \hat{A}_{i}[ \mathcal{I}^{(n)} \cap \mathcal{G}_{2}^{(n)}]$
\If{$k=1$}
\State $\hat{\Psi}_{i}^{(V)} \leftarrow \hat{A}_i [\mathcal{C}_{2}^{(n)}]$ \textbf{and} $\hat{\Gamma}_{i}^{(V)} \leftarrow \hat{A}_i[\mathcal{C}_{1,2}^{(n)}]$ 
\State $\hat{\bar{\Theta}}_{1,i+1}^{(V)} \leftarrow \hat{A}_i[\mathcal{R}_{1}^{(n)}]$ \textbf{and} 
$\hat{\bar{\Theta}}_{2,i+1}^{(V)} \leftarrow \hat{A}_i[\mathcal{R}_{1,2}^{\prime (n)}] \oplus \hat{\Psi}_{2,i-1}^{(V)}$ 
\State $\hat{\bar{\Gamma}}_{1,i+1}^{(V)} \leftarrow \hat{A}_i[\mathcal{R}_{1,2}^{(n)}]  \oplus \hat{\Gamma}_{1,i-1}^{(V)}$ \textbf{and} $\hat{\bar{\Gamma}}_{2,i+1}^{(V)} \leftarrow \hat{A}_i[\mathcal{R}_{1}^{\prime (n)}]$
\State $\Delta_{(1),i+1}^{(U)} \leftarrow \hat{T}_{(1),i}[\mathcal{L}_{1}^{(n)}]$
\Comment $\Delta_{(1),i+1}^{(U)}=\big[\hat{\bar{\Theta}}_{3,i+1}^{(V)},\hat{\bar{\Gamma}}_{3,i+1}^{(V)}\big]$
\State $\hat{\Theta}_{i+1}^{(V)} \leftarrow \big[\hat{\bar{\Theta}}_{1,i+1}^{(V)}, \hat{\bar{\Theta}}_{2,i+1}^{(V)}, \hat{\bar{\Theta}}_{3,i+1}^{(V)} \big] \oplus \kappa_{\Theta}^{(V)}$
\State $\hat{\Gamma}_{i+1}^{(V)} \leftarrow \big[\hat{\bar{\Gamma}}_{1,i+1}^{(V)},\hat{\bar{\Gamma}}_{2,i+1}^{(V)},\hat{\bar{\Gamma}}_{3,i+1}^{(V)}\big] \oplus \kappa_{\Gamma}^{(V)}$
\State $\hat{\Upsilon}_{(1),i+1}^{\prime (V)} \leftarrow  \big(\hat{\Psi}_{1,i}^{(V)} , \hat{\Gamma}_{2,i}^{(V)} ,\hat{\Pi}_{(2),i}^{(V)} , \hat{\Theta}_{i+1}^{(V)} ,\hat{\Gamma}_{i+1}^{(V)} , \hat{\Lambda}_{i+1}^{(V)} \big)$
\State $\hat{{\Theta}}_{(1),i+1}^{(U)} \leftarrow \hat{T}_{(1),i}[\mathcal{D}_{1}^{(n)}] \oplus \kappa_{\Theta}^{(U)}$
\State $\hat{\Upsilon}_{(1),i+1}^{\prime (U)} \leftarrow  \big( \hat{{\Theta}}_{(1),i+1}^{(U)}, \hat{\Lambda}_{(1),i+1}^{(U)} \big)$
\Else
\State $\hat{\Psi}_{i}^{(V)} \leftarrow \hat{A}_i [\mathcal{C}_{2}^{(n)}]$ \textbf{and} $\hat{\Gamma}_{i}^{(V)} \leftarrow \hat{A}_i[\mathcal{C}_{1,2}^{(n)}]$
\State $\hat{\bar{\Psi}}_{i}^{(V)} \leftarrow \hat{\Psi}_{i}^{(V)} \oplus \kappa_{\Psi}^{(V)}$ \textbf{and} $\hat{\bar{\Gamma}}_{i}^{(V)} \leftarrow \hat{\Gamma}_{i}^{(V)} \oplus \kappa_{\Gamma}^{(V)}$
\State $\hat{{\Theta}}_{1,i+1}^{(V)} \leftarrow \hat{A}_i[\mathcal{R}_{1}^{(n)}]$ \textbf{and} 
$\hat{{\Theta}}_{2,i+1}^{(V)} \leftarrow \hat{A}_i[\mathcal{R}_{1,2}^{\prime (n)}] \oplus \hat{\bar{\Psi}}_{2,i-1}^{(V)}$ 
\State $\hat{{\Gamma}}_{1,i+1}^{(V)} \leftarrow \hat{A}_i[\mathcal{R}_{1,2}^{(n)}]  \oplus \hat{\bar{\Gamma}}_{1,i-1}^{(V)}$ \textbf{and} $\hat{{\Gamma}}_{2,i+1}^{(V)} \leftarrow \hat{A}_i[\mathcal{R}_{1}^{\prime (n)}]$
\State $\big(\hat{\Pi}_{(1),i+1}^{(V)}, \Delta_{(1),i+1}^{(U)} \big) \leftarrow \hat{T}_{(1),i}[\mathcal{M}_{1}^{(n)}]$
\Comment $\Delta_{(1),i+1}^{(U)}=\big[\hat{{\Theta}}_{3,i+1}^{(V)},\hat{{\Gamma}}_{3,i+1}^{(V)}\big]$
\State $\hat{\Upsilon}_{(1),i+1}^{\prime (V)} \leftarrow  \big(\hat{\bar{\Psi}}_{1,i}^{(V)} , \hat{\bar{\Gamma}}_{2,i}^{(V)} , \hat{\Pi}_{(2),i}^{(V)}, \hat{\Pi}_{(1),i+1}^{(V)}, \hat{\Theta}_{i+1}^{(V)} ,\hat{\Gamma}_{i+1}^{(V)} , \hat{\Lambda}_{i+1}^{(V)} \big)$
\State $\hat{{\Theta}}_{(1),i+1}^{(U)} \leftarrow \hat{T}_{(1),i}[\mathcal{N}_{1}^{(n)}]$
\State $\hat{{O}}_{(1),i+1}^{(U)} \leftarrow \hat{T}_{(1),i}[\mathcal{O}_{1}^{(n)}] \oplus \kappa_{O}^{(U)}$
\State $\hat{\Upsilon}_{(1),i+1}^{\prime (U)} \leftarrow  \big( \hat{{\Theta}}_{(1),i+1}^{(U)}, \hat{\Lambda}_{(1),i+1}^{(U)}, \hat{{O}}_{(1),i+1}^{(U)} \big)$
\EndIf
\State $\hat{A}_{i+1}^n \leftarrow \big( \hat{\Upsilon}_{(1),i+1}^{\prime (V)}, \Phi_{(1),i+1}^{(V)}, \tilde{Y}_{(1),i+1}^n \big)$
\State $\hat{T}_{(1),i+1}^n \leftarrow \big( \hat{\Upsilon}_{(1),i+1}^{\prime (U)}, \Phi_{(1),i+1}^{(U)}, \hat{A}^n_{i+1}G_n, \tilde{Y}_{(1),i+1}^n \big)$
\EndFor \\
\textbf{Return} $\big(\hat{W}_{(1),1:L},\hat{S}_{(1),1:L} \big) \leftarrow \big(\hat{A}_{i+1}^n \hat{T}_{(1),i+1}^n\big)$
\end{algorithmic}
\end{algorithm}

\newpage
\noindent \textbf{Decoding at Receiver~2}

This receiver forms the estimates $\hat{A}_{1:L}^n$ and $\hat{T}_{(2),1:L}^n$ of $\tilde{A}_{1:L}^n$ and $\tilde{T}_{(2),1:L}^n$ respectively by going backward, that is, from $(\hat{A}_{L}^n, \hat{T}_{(2),L}^n)$ to $(\hat{A}_{1}^n, \hat{T}_{(2),1}^n)$. 

The decoding process at Receiver~2 when the \gls*{pcs} must achieve the corner point of $\mathfrak{R}_{\text{MI-WTBC}}^{(1)}$ or $\mathfrak{R}_{\text{MI-WTBC}}^{(2)}$ is summarized in Algorithm~\ref{alg:decoding2x}. Despite $\tilde{A}_{1:L}^n$ does not carry information intended for Receiver~$2$ when the \gls*{pcs} operates to achieve the corner point of $\mathfrak{R}_{\text{MI-WTBC}}^{(1)}$, recall that this receiver needs $\tilde{A}_{1:L}^n$ to reliably reconstruct $\tilde{T}_{(2),1:L}^n$.

In both cases, Receiver~2 constructs $\hat{A}^n_{L}$ and $\hat{T}^n_{(2),L}$ as follows. Given $\smash{\big(\Upsilon_{(2)}^{(V)},\Phi_{(2),L}^{(V)}\big)}$ and $\big(\Upsilon_{(2)}^{(U)}, \Phi_{(2),L}^{(U)} \big)$, it knows $\smash{\tilde{A}_{L}\big[ \big(\mathcal{L}_{V|Y_{(2)}}^{(n)} \big)^{\text{C}}\big]}$ and $\smash{\tilde{T}_{(2),L}\big[ \big(\mathcal{L}_{U_{(2)}|VY_{(2)}}^{(n)} \big)^{\text{C}}\big]}$, respectively. Hence, from observations $\tilde{Y}_{(2),L}^n$ and $\smash{\tilde{A}_{L}\big[ \big(\mathcal{L}_{V|Y_{(2)}}^{(n)} \big)^{\text{C}}\big]}$, it entirely constructs $\hat{A}^n_{L}$ by performing \gls*{sc} decoding. Then, from sequences $\tilde{Y}_{(2),L}^n$, $\smash{\tilde{T}_{(2),L}\big[ \big(\mathcal{L}_{U_{(2)}|VY_{(2)}}^{(n)} \big)^{\text{C}}\big]}$ and $\hat{V}_{L}^n = \hat{A}_L^n G_n$, this receiver form $\hat{T}^n_{(2),L}$. 

Recall that $\Lambda_1^{(V)}$ and $\Lambda_{(2),1}^{(U)}$ have been replicated in all blocks. Thus, Receiver~2 obtains $\smash{\hat{\Lambda}_{1:L}^{(V)}} = \hat{A}_{L}\big[\mathcal{R}_{\Lambda}^{(n)}\big]$, while it gets $\smash{\hat{\Lambda}_{(2),1:L}^{(U)}} = \hat{T}_{(2),L}\big[\mathcal{F}_{2}^{(n)}\big]$ or $\smash{\hat{\Lambda}_{(2),1:L}^{(U)}} = \hat{T}_{(2),L}\big[\mathcal{Q}_{2}^{(n)}\big]$ depending on whether the \gls*{pcs} must achieve the corner point of $\mathfrak{R}_{\text{MI-WTBC}}^{(1)}$ or $\mathfrak{R}_{\text{MI-WTBC}}^{(2)}$ respectively.

For $i \in [2,L]$, consider that $\hat{A}_{i:L}^n$ and $\hat{T}_{(2),i:L}^n$ have already been formed. Then, the construction of $\hat{A}_{i-1}^n$ and $\hat{T}_{(2),i-1}^n$ is slightly different depending on whether the \gls*{pcs} must achieve the corner point of $\mathfrak{R}_{\text{MI-WTBC}}^{(1)}$ ($k=1$) or $\mathfrak{R}_{\text{MI-WTBC}}^{(2)}$ ($k=2$):
\begin{enumerate}
\item[] $k=1)$ From $\hat{A}_i^n$ and $\hat{T}_{(2),i}^n$, it obtains $\hat{\Upsilon}_{(2),i-1}^{\prime (V)} \triangleq  \big(\hat{\bar{\Theta}}_{1,i}^{(V)} , \hat{\bar{\Gamma}}_{2,i}^{(V)} , \hat{\Pi}_{(2),i-1}^{(V)}, \hat{\Psi}_{i-1}^{(V)} ,\hat{\Gamma}_{i-1}^{(V)} , \hat{\Lambda}_{i-1}^{(V)} \big)$ and $\hat{\Upsilon}_{(2),i-1}^{\prime (U)} \triangleq  \big( \hat{{\Psi}}_{(2),i-1}^{(U)}, \hat{\Lambda}_{(2),i-1}^{(U)}, \hat{{O}}_{(2),i-1}^{(U)} \big)$. Notice in Algorithm~\ref{alg:decoding2x} that $\kappa_{\Theta}^{(V)}$ and $\kappa_{\Gamma}^{(V)}$ are needed because the encoder have repeated an \emph{encrypted} version of $\hat{\Theta}_{i}^{(V)}$ and $\hat{\Gamma}_{i}^{(V)}$ at Block~$i-1$. In order to obtain $\hat{\Psi}_{i-1}^{(V)}$ and $\hat{\Gamma}_{i-1}^{(V)}$, recall that part of these sequences may have been repeated in $\tilde{T}_{(2),i}^n$. Moreover, notice that part of the encrypted versions of $\hat{\Theta}_{i+1}^{(V)}$ and $\hat{\Gamma}_{i+1}^{(V)}$ may also be needed to obtain $\hat{\Psi}_{i-1}^{(V)}$ and $\hat{\Gamma}_{i-1}^{(V)}$, which are available because $\big(\tilde{A}_{i+1}^n, \hat{T}_{(2),i+1}^n\big)$ has already been constructed. Now Receiver~2 obtains $\Psi_{(2),i-1}^{(U)}$ and $\hat{{O}}_{(2),i-1}^{(U)} = \hat{T}_{(2),i}[\mathcal{O}_{2}^{(n)}] \oplus \kappa_{O}^{(U)}$ from $\tilde{T}_{(2),i}^n$, where recall that $\hat{{O}}_{(2),i-1}^{(U)}$ contains part of the elements of $\hat{T}_{(2),i-1}^n$ that have been drawn by performing \gls*{sc} encoding. 

\item[] $k=2$) From $\hat{A}_i^n$ and $\hat{T}_{(2),i}^n$, it obtains $\hat{\Upsilon}_{(2),i-1}^{\prime (V)} \triangleq  \big(\hat{\Theta}_{1,i}^{(V)} , \hat{\Gamma}_{2,i}^{(V)} ,\hat{\Pi}_{(2),i-1}^{(V)} , \hat{\Psi}_{i-1}^{(V)} ,\hat{\Gamma}_{i-1}^{(V)} , \hat{\Lambda}_{i-1}^{(V)} \big)$ and $\hat{\Upsilon}_{(2),i-1}^{\prime (U)} \triangleq  \big( \hat{{\Psi}}_{(2),i-1}^{(U)}, \hat{\Lambda}_{(2),i-1}^{(U)} \big)$. Now, the Receiver~2 uses $\kappa_{\Psi}^{(V)}$, $\kappa_{\Gamma}^{(V)}$ and $\kappa_{\Psi}^{(U)}$ because the encoder have repeated an encrypted version of $\hat{\Psi}_{i-1}^{(V)}$, $\hat{\Gamma}_{i-1}^{(V)}$ and $\hat{\Psi}_{i-1}^{(U)}$ in Block~$i$.
\end{enumerate}

Finally, given $\smash{\big(\Upsilon_{(2),i-1}^{\prime (V)},\Phi_{(2),i-1}^{(V)}\big)} \supseteq \hat{A}_{i-1}\big[ \big(  \mathcal{L}_{V|Y_{(2)}}^{(n)}\big)^{\text{C}}\big]$ and $\tilde{Y}_{(2),i-1}^n$, it performs \gls*{sc} decoding to construct $\hat{A}_{i-1}^n$. Then, given $\smash{\big(\Upsilon_{(2),i-1}^{\prime (U)},\Phi_{(2),i-1}^{(U)}\big)} \supseteq \hat{T}_{(2),i-1}\big[ \big(  \mathcal{L}_{U_{(2)}|VY_{(2)}}^{(n)}\big)^{\text{C}}\big]$, $\hat{V}_{i-1}^n = \hat{A}_{i-1}^n G_n$ and $\tilde{Y}_{(2),i-1}^n$, it performs \gls*{sc} decoding to construct $\hat{T}_{(2),i-1}^n$. 

\newpage

\begin{algorithm}[h!]
\caption{\small{Decoding at Receiver~2 when \gls*{pcs} operates to achieve the corner point of $\mathfrak{R}_{\text{MI-WTBC}}^{(k)}$}}\label{alg:decoding2x}
\begin{algorithmic}[1]
\Require $\Upsilon_{(2)}^{(V)}$, $\Phi_{(2),1:L}^{(V)}$, $\Upsilon_{(2)}^{(U)}$, $\Phi_{(2),1:L}^{(U)}$, and $\tilde{Y}_{(2),1:L}^n$.
\Require \textbf{if} $k=1$ \textbf{then} $\{\kappa_{\Gamma}^{(V)},\kappa_{\Theta}^{(V)}, \kappa_{O}^{(U)} \}$ \textbf{else} $\{\kappa_{\Gamma}^{(V)},\kappa_{\Psi}^{(V)}, \kappa_{\Psi}^{(U)} \}$
\State $\hat{A}_{L}^n \leftarrow \big( \Upsilon_{(2)}^{(V)}, \Phi_{(2),L}^{(V)}, \tilde{Y}_{(2),L}^n \big)$
\Comment by using \gls*{sc} decoding
\State $\hat{T}_{(2),L}^n \leftarrow \big( \Upsilon_{(2)}^{(U)}, \Phi_{(2),L}^{(U)}, \hat{A}^n_{L}G_n, \tilde{Y}_{(2),L}^n \big)$
\Comment by using \gls*{sc} decoding
\State $\hat{\Lambda}_{1:L-1}^{(V)} \leftarrow \hat{A}_L\big[\mathcal{R}_{\Lambda}^{(n)} \big]$
\State \textbf{if} $k=1$ \textbf{then} $\hat{\Lambda}_{1:L-1}^{(U)} \leftarrow \hat{T}_{(2),L}\big[\mathcal{F}_{2}^{(n)} \big]$ \textbf{else} $\hat{\Lambda}_{1:L-1}^{(U)} \leftarrow \hat{T}_{(2),L}\big[\mathcal{Q}_{2}^{(n)} \big]$
\For{$i = L$ \text{to} $2$} 
\State $\hat{\Pi}_{(2),i-1}^{(V)} \leftarrow \hat{A}_{i}[ \mathcal{R}_{\text{S}}^{(n)}]$
\If{$k=1$}
\State $\hat{\Theta}_{i}^{(V)} \leftarrow \hat{A}_i [\mathcal{C}_{1}^{(n)}]$ \textbf{and} $\hat{\Gamma}_{i}^{(V)} \leftarrow \hat{A}_i[\mathcal{C}_{1,2}^{(n)}]$
\State $\hat{\bar{\Theta}}_{i}^{(V)} \leftarrow \hat{\Theta}_{i}^{(V)} \oplus \kappa_{\Theta}^{(V)}$ \textbf{and} $\hat{\bar{\Gamma}}_{i}^{(V)} \leftarrow \hat{\Gamma}_{i}^{(V)} \oplus \kappa_{\Gamma}^{(V)}$
\State $\hat{{\Psi}}_{1,i-1}^{(V)} \leftarrow \hat{A}_i[\mathcal{R}_{2}^{(n)}]$ \textbf{and} 
$\hat{{\Psi}}_{2,i-1}^{(V)} \leftarrow \hat{A}_i[\mathcal{R}_{1,2}^{\prime (n)}] \oplus \hat{\bar{\Theta}}_{2,i+1}^{(V)}$ 
\State $\hat{{\Gamma}}_{1,i-1}^{(V)} \leftarrow \hat{A}_i[\mathcal{R}_{1,2}^{(n)}]  \oplus \hat{\bar{\Gamma}}_{1,i+1}^{(V)}$ \textbf{and} $\hat{{\Gamma}}_{2,i-1}^{(V)} \leftarrow \hat{A}_i[\mathcal{R}_{2}^{\prime (n)}]$
\State $\Delta_{(2),i-1}^{(U)} \leftarrow \hat{T}_{(2),i}[\mathcal{M}_{2}^{(n)}]$
\Comment $\Delta_{(2),i-1}^{(U)}=\big[\hat{{\Psi}}_{3,i-1}^{(V)},\hat{{\Gamma}}_{3,i-1}^{(V)}\big]$
\State $\hat{\Upsilon}_{(1),i-1}^{\prime (V)} \leftarrow  \big(\hat{\bar{\Theta}}_{1,i}^{(V)} , \hat{\bar{\Gamma}}_{2,i}^{(V)} , \hat{\Pi}_{(2),i-1}^{(V)}, \hat{\Psi}_{i-1}^{(V)} ,\hat{\Gamma}_{i-1}^{(V)} , \hat{\Lambda}_{i-1}^{(V)} \big)$
\State $\hat{{\Psi}}_{(2),i-1}^{(U)} \leftarrow \hat{T}_{(2),i}[\mathcal{N}_{2}^{(n)}]$
\State $\hat{{O}}_{(2),i-1}^{(U)} \leftarrow \hat{T}_{(2),i}[\mathcal{O}_{2}^{(n)}] \oplus \kappa_{O}^{(U)}$
\State $\hat{\Upsilon}_{(2),i-1}^{\prime (U)} \leftarrow  \big( \hat{{\Psi}}_{(2),i-1}^{(U)}, \hat{\Lambda}_{(2),i-1}^{(U)}, \hat{{O}}_{(2),i-1}^{(U)} \big)$
\Else
\State $\hat{\Theta}_{i}^{(V)} \leftarrow \hat{A}_i [\mathcal{C}_{1}^{(n)}]$ \textbf{and} $\hat{\Gamma}_{i}^{(V)} \leftarrow \hat{A}_i[\mathcal{C}_{1,2}^{(n)}]$ 
\State $\hat{\bar{\Psi}}_{1,i-1}^{(V)} \leftarrow \hat{A}_i[\mathcal{R}_{2}^{(n)}]$ \textbf{and} 
$\hat{\bar{\Psi}}_{2,i-1}^{(V)} \leftarrow \hat{A}_i[\mathcal{R}_{1,2}^{\prime (n)}] \oplus \hat{\Theta}_{2,i+1}^{(V)}$ 
\State $\hat{\bar{\Gamma}}_{1,i-1}^{(V)} \leftarrow \hat{A}_i[\mathcal{R}_{1,2}^{(n)}]  \oplus \hat{\Gamma}_{1,i+1}^{(V)}$ \textbf{and} $\hat{\bar{\Gamma}}_{2,i-1}^{(V)} \leftarrow \hat{A}_i[\mathcal{R}_{2}^{\prime (n)}]$
\State $\Delta_{(2),i-1}^{(U)} \leftarrow \hat{T}_{(2),i}[\mathcal{L}_{2}^{(n)}]$
\Comment $\Delta_{(2),i-1}^{(U)}=\big[\hat{\bar{\Psi}}_{3,i-1}^{(V)},\hat{\bar{\Gamma}}_{3,i-1}^{(V)}\big]$
\State $\hat{\Psi}_{i-1}^{(V)} \leftarrow \big[\hat{\bar{\Psi}}_{1,i-1}^{(V)}, \hat{\bar{\Psi}}_{2,i-1}^{(V)}, \hat{\bar{\Psi}}_{3,i-1}^{(V)} \big] \oplus \kappa_{\Psi}^{(V)}$
\State $\hat{\Gamma}_{i-1}^{(V)} \leftarrow \big[\hat{\bar{\Gamma}}_{1,i-1}^{(V)},\hat{\bar{\Gamma}}_{2,i-1}^{(V)},\hat{\bar{\Gamma}}_{3,i-1}^{(V)}\big] \oplus \kappa_{\Gamma}^{(V)}$
\State $\hat{\Upsilon}_{(2),i-1}^{\prime (V)} \leftarrow  \big(\hat{\Theta}_{1,i}^{(V)} , \hat{\Gamma}_{2,i}^{(V)} ,\hat{\Pi}_{(2),i-1}^{(V)} , \hat{\Psi}_{i-1}^{(V)} ,\hat{\Gamma}_{i-1}^{(V)} , \hat{\Lambda}_{i-1}^{(V)} \big)$
\State $\hat{{\Psi}}_{(2),i-1}^{(U)} \leftarrow \hat{T}_{(2),i}[\mathcal{D}_{2}^{(n)}] \oplus \kappa_{\Psi}^{(U)}$
\State $\hat{\Upsilon}_{(2),i-1}^{\prime (U)} \leftarrow  \big( \hat{{\Psi}}_{(2),i-1}^{(U)}, \hat{\Lambda}_{(2),i-1}^{(U)} \big)$
\EndIf
\State $\hat{A}_{i-1}^n \leftarrow \big( \hat{\Upsilon}_{(2),i-1}^{\prime (V)}, \Phi_{(2),i-1}^{(V)}, \tilde{Y}_{(2),i-1}^n \big)$
\State $\hat{T}_{(2),i-1}^n \leftarrow \big( \hat{\Upsilon}_{(2),i-1}^{\prime (U)}, \Phi_{(2),i-1}^{(U)}, \hat{A}^n_{i-1}G_n, \tilde{Y}_{(2),i-1}^n \big)$
\EndFor \\
\textbf{Return} $\big(\hat{W}_{(2),1:L},\hat{S}_{(2),1:L} \big) \leftarrow \big(\hat{A}_{i-1}^n \hat{T}_{(2),i-1}^n\big)$
\end{algorithmic}
\end{algorithm}
\newpage

\begin{remark}
According to the previous decoding algorithms, Receiver~$k \in [1,2]$ decodes both $\tilde{A}_{i}^n$ and $\tilde{T}_{(k),i}^n$ from Block~$i\in [1,L]$ before moving to adjacent blocks (polar-based jointly decoding). Indeed, Receiver~1 needs to obtain first $\hat{T}_{(1),i}^n$ before decoding $\tilde{A}_{i+1}^n$ because $\big[\Pi_{(1),i+1}^{(V)}, \Delta_{(1),i+1}^{(V)} \big]$, which is required by this receiver to reliably estimate $\tilde{A}_{i+1}^n$, is repeated in $\tilde{T}_{(1),i}^n$. Similarly, Receiver~2 needs $\Delta_{(2),i-1}^{(V)}$ to reliably decode $\tilde{A}_{i-1}^n$, but it is repeated in $\hat{T}_{(2),i}^n$. 


Consider another decoding strategy for Receiver~$k\in [1,2]$ that obtains first $\hat{A}_{1:L}^n$, and then decodes the outer-layer $\tilde{T}_{(k),1:L}^n$. We refer to this decoding strategy as polar-based successive decoding. Clearly, $(R_{S_{(1)}}^{\star 2}, R_{S_{(2)}}^{\star 2}, R_{W_{(1)}}^{\star 2}, R_{W_{(2)}}^{\star 2}) \subset \mathfrak{R}_{\text{{MI-WTBC}}}^{(2)}$ is not achievable by using successive decoding because, according to the summary of the construction of $\tilde{A}_{1:L}\big[ \mathcal{G}^{(n)} \big]$ in the last part of Section~\ref{sec:PCSx1il}, in all cases $\tilde{T}_{(k),i}^n$ contains elements required by Receiver~1 to reliably decode $\tilde{A}_{i+1}^n$. Furthermore, for the same reason, all situations where $I(V;Y_{(k)}) < I(V;Z)$ for some $k \in [1,2]$ are not possible by using this strategy. Consequently, it is clear that joint decoding enlarges the inner-bound on the achievable region for a particular distribution\footnote{Although for a particular distribution the inner-bound is strictly larger with joint decoding, we cannot affirm that this decoding strategy enlarges $\mathfrak{R}_{\text{{MI-WTBC}}}$: rate points that are not achievable with successive decoding for this particular distribution may be achievable under another distribution.}.
\end{remark}

\section{Performance of the polar coding scheme}\label{sec:performancex}
The analysis of the polar coding scheme of Section~\ref{sec:PCSx1} leads to the following theorem.

\begin{theorem}\label{th:th1x}
Let $(\mathcal{X}, p_{Y_{(1)}Y_{(2)} Z|X}, \mathcal{Y}_{(1)} \times \mathcal{Y}_{(2)} \times \mathcal{Z})$  be an arbitrary \gls*{wtbc} where $\mathcal{X}~\in~\{0,1\}$. The \gls*{pcs} in Section~\ref{sec:PCSx1} achieves any corner point $(R_{S_{(1)}}^{\star k}, R_{S_{(2)}}^{\star k}, R_{W_{(1)}}^{\star k}, R_{W_{(2)}}^{\star k}) \subset \mathfrak{R}_{\text{{MI-WTBC}}}^{(k)}$ given in \eqref{eq:rsk}--\eqref{eq:rwbk} for any $k \in [1,2]$. 
%
\end{theorem}

\begin{corollary}\label{coro:th1x}
The \gls*{pcs} achieves any rate tuple of $\mathfrak{R}_{\text{{MI-WTBC}}}^{(k)}$ defined in Proposition~\ref{prop:MIB2x}.
\end{corollary}

The proof of Theorem~\ref{th:th1x} follows in four steps and is provided in the following subsections. In Section~\ref{sec:performance_ratesx} we show that the \gls*{pcs} approaches $(R_{S_{(1)}}^{\star k}, R_{S_{(2)}}^{\star k}, R_{W_{(1)}}^{\star k}, R_{W_{(2)}}^{\star k}) \subset \mathfrak{R}_{\text{{MI-WTBC}}}^{(k)}$ for any $k \in [1,2]$. In Section~\ref{sec:distributionDMSx} we prove that, for all $i \in [1,L]$, the joint distribution of $(\tilde{V}_i^n, \tilde{U}_{(1),i}^n, \tilde{U}_{(2),i}^n, \tilde{X}_i^n,\tilde{Y}_{(1),i}^n,\tilde{Y}_{(2),i}^n,\tilde{Z}_{i}^n)$ is asymptotically indistinguishable of the one of the original \gls*{dms} that is used for the polar code construction. Finally, in Section~\ref{sec:performance_reliabilityx} and Section~\ref{sec:performance_secrecyx} we show that the polar coding scheme satisfies the reliability and the secrecy conditions given in \eqref{eq:reliabilitycondx} and \eqref{eq:secrecycondx} respectively.

\subsection{Transmission rates}\label{sec:performance_ratesx}
We prove that the \gls*{pcs} described in Section~\ref{sec:PCSx1} approaches $(R_{S_{(1)}}^{\star k}, R_{S_{(2)}}^{\star k}, R_{W_{(1)}}^{\star k}, R_{W_{(2)}}^{\star k}) \subset \mathfrak{R}_{\text{{MI-WTBC}}}^{(k)}$ defined in \eqref{eq:rsk}--\eqref{eq:rwbk} for any $k \in [1,2]$. Also, we show that the overall length of secret keys $\kappa_{\Gamma}^{(V)}$, $\kappa_{\Upsilon\Phi_{(1)}}^{(V)}$, $\kappa_{\Upsilon\Phi_{(2)}}^{(V)}$, $\kappa_{O}^{(U)}$, $\kappa_{\Upsilon\Phi_{(1)}}^{(U)}$, $\kappa_{\Upsilon\Phi_{(2)}}^{(U)}$ is asymptotically negligible in terms of rate, and so is the overall length of $\kappa_{\Theta}^{(V)}$ and $\kappa_{\Theta}^{(U)}$ if the \gls*{pcs} operates to achieve the corner point of $\mathfrak{R}_{\text{{MI-WTBC}}}^{(1)}$, or that of $\kappa_{\Psi}^{(V)}$ and $\kappa_{\Psi}^{(U)}$ if the \gls*{pcs} operates to achieve the corner point of $\mathfrak{R}_{\text{{MI-WTBC}}}^{(2)}$. Moreover, we show that so is the amount of randomness required by the encoding (number of entries drawn by random \gls*{sc} encoding). 

\vspace{0.15cm}
\noindent \textbf{Private message rate}

\begin{enumerate}
\item Rate of $W_{(k)}$. For $i \in [1,L]$, we have $W_{(k),i} = \big[ W_{(k),i}^{(V)}, W_{(k),i}^{(U)} \big]$, where $W_{(k),i}^{(V)} = \tilde{A}_{i}\big[ \mathcal{C}^{(n)} \big]$ and $W_{(k),i}^{(U)} = \tilde{T}_{(k),i}\big[ \mathcal{J}_0^{(n)} \cup \mathcal{J}_k^{(n)} \big]$. Therefore, we obtain
\begin{align*}
\frac{1}{nL} \sum_{i=1}^L \big( \big| W_{(k),i}^{(V)} \big| + \big| W_{(k),i}^{(U)} \big| \big) & \stackrel{(a)}{=} \frac{1}{n} \Big( \Big| \mathcal{H}_{V}^{(n)} \setminus \mathcal{H}_{V|Z}^{(n)} \Big| + \Big| \mathcal{H}_{U_{(k)}|V}^{(n)} \setminus \mathcal{H}_{U_{(k)}|VZ}^{(n)} \Big| \Big) \\
& \stackrel{(b)}{=} \frac{1}{n} \big( \big| \mathcal{H}_{V}^{(n)} \big| - \big| \mathcal{H}_{V|Z}^{(n)} \big| + \big| \mathcal{H}_{U_{(k)}|V}^{(n)} \big| - \big| \mathcal{H}_{U_{(k)}|VZ}^{(n)} \big| \big) \\
& \xrightarrow{n \rightarrow \infty} H(V) - H(V|Z) +  H(U_{(k)}|V) - H(U_{(k)}|VZ),
\end{align*}
where $(a)$ holds because $\mathcal{C}^{(n)} = \mathcal{H}_{V|Z}^{(n)}$ and by the definition of $\mathcal{J}_0^{(n)}$ and $\mathcal{J}_k^{(n)}$ in \eqref{eq:j0} and \eqref{eq:jk} respectively; $(b)$ follows from the fact that $\mathcal{H}_{V}^{(n)} \supseteq \mathcal{H}_{V|Z}^{(n)}$ and $\mathcal{H}_{U_{(k)}|V}^{(n)} \supseteq \mathcal{H}_{U_{(k)}|VZ}^{(n)}$; and the limit holds by the source polarization theorem \cite{arikan2010source}.

\item Rate of $W_{(\bar{k})}$. For $i \in [1,L]$, all private information $W_{(\bar{k}),i}$ is carried in layer $\tilde{T}_{(\bar{k}),i}^n$. Specifically, we have $W_{(\bar{k}),i}^{(U)} = \tilde{T}_{(\bar{k}),i}\big[ \mathcal{B}_0^{(n)} \cup \mathcal{B}_k^{(n)} \big]$. Hence, we obtain
\begin{align*}
\frac{1}{nL} \sum_{i=1}^L  \big| W_{(\bar{k}),i}^{(U)} \big| & \stackrel{(a)}{=} \frac{1}{n} \Big| \mathcal{H}_{U_{(\bar{k})}|VU_{(k)}}^{(n)} \setminus \mathcal{H}_{U_{(\bar{k})}|VU_{(k)}Z}^{(n)} \Big|  \\
& \stackrel{(b)}{=} \frac{1}{n} \big| \mathcal{H}_{U_{(\bar{k})}|VU_{(k)}}^{(n)} \big| - \big| \mathcal{H}_{U_{(\bar{k})}|VU_{(k)}Z}^{(n)} \big| \\
& \xrightarrow{n \rightarrow \infty} H(U_{(\bar{k})}| V U_{(k)}) - H(U_{(\bar{k})}| V U_{(k)} Z),
\end{align*}
where $(a)$ holds by definition of $\mathcal{B}_0^{(n)}$ and $\mathcal{B}_{\bar{k}}^{(n)}$ in \eqref{eq:b0} and \eqref{eq:bk} respectively; $(b)$ holds because $\mathcal{H}_{U_{(\bar{k})}|VU_{(k)}}^{(n)} \supseteq \mathcal{H}_{U_{(\bar{k})}|VU_{(k)}Z}^{(n)}$; and the limit holds by the source polarization theorem~\cite{arikan2010source}.
\end{enumerate}

Therefore, the \gls*{pcs} attains $R_{W_{(k)}}^{\star k}$ and $R_{W_{(\bar{k})}}^{\star k}$ defined in \eqref{eq:rwk} and \eqref{eq:rwbk} respectively.

\vspace{0.15cm}
\noindent \textbf{Confidential message rate}

According to the \gls*{pcs} described in Section~\ref{sec:PCSx1}, to approach $(R_{S_{(1)}}^{\star k}, R_{S_{(2)}}^{\star k}, R_{W_{(1)}}^{\star k}, R_{W_{(2)}}^{\star k}) \subset \mathfrak{R}_{\text{{MI-WTBC}}}^{(k)}$, the inner-layer $\tilde{A}_{1:L}^n$ and the outer-layer $\tilde{T}_{(k),1:L}^n$ carry confidential information $S_{(k)}$ intended for Receiver~$k$, while the outer-layer $\tilde{T}_{(\bar{k}),1:L}^n$ carries confidential information $S_{(\bar{k})}$ intended for Receiver~$\bar{k}$. 

\begin{enumerate}

\item Rate of $S_{(k)}$. 
First, consider the confidential information $S_{(k)}^{(V)}$ that is carried in the inner-layer $\tilde{A}_{1:L}^n$. From Section~\ref{sec:PCSx1il}, in all cases we have $S_{(k),1}^{(V)} = \tilde{A}_1 \big[ \mathcal{I}^{(n)} \cup  \mathcal{G}^{(n)}_{1} \cup \mathcal{G}^{(n)}_{1,2} \big]$; for $i \in [2,L-1]$, we have $S_{(k),i}^{(V)} = \tilde{A}_{i} \big[ \mathcal{I}^{(n)} \big]$; and $S_{(k),L}^{(V)}=\tilde{A}_L \big[\mathcal{I}^{(n)} \cup  \mathcal{G}^{(n)}_2 \big]$. Recall that the definition of the set $\mathcal{I}^{(n)}$ depends on whether the \gls*{pcs} must achieve the corner point of regions $\mathfrak{R}_{\text{MI-WTBC}}^{(1)}$ or $\mathfrak{R}_{\text{MI-WTBC}}^{(2)}$. 

If the \gls*{pcs} operates to achieve the corner point of $\mathfrak{R}_{\text{MI-WTBC}}^{(1)}$, we have
\begin{align}
\big| \mathcal{I}^{(n)} \big| & \stackrel{(a)}{=}  \big| \mathcal{G}^{(n)}_0 \big|   +   \big| \mathcal{G}^{(n)}_{2} \big|  -  \big| \mathcal{R}^{(n)}_{1,2} \big| -  \big| \mathcal{R}^{\prime (n)}_{1,2} \big|  - \big| \mathcal{R}^{(n)}_{1} \big| -  \big| \mathcal{R}^{\prime (n)}_{1} \big|  \nonumber \\
& \stackrel{(b)}{=} \Big\{ \big| \mathcal{G}^{(n)}_0 \big|  +  \big| \mathcal{G}^{(n)}_{2} \big| -   \big| \mathcal{C}^{(n)}_{1} \big|  -  \big| \mathcal{C}^{(n)}_{1,2} \big| \Big\}^{+} \nonumber \\
& \stackrel{(c)}{=} \Big\{ \Big| \mathcal{H}^{(n)}_{V|Z} \cap \mathcal{L}^{(n)}_{V|Y_{(1)}} \Big| - \Big| \big( \mathcal{H}^{(n)}_{V|Z} \big)^{\text{C}} \cap \big( \mathcal{L}^{(n)}_{V|Y_{(1)}} \big)^{\text{C}} \Big| \Big\}^{+} \nonumber \\
& = \Big\{ \big| \mathcal{H}^{(n)}_{V|Z}  \big| - \big| \big( \mathcal{L}^{(n)}_{V|Y_{(1)}} \big)^{\text{C}} \big| \Big\}^{+},  \label{eq:setia} 
\end{align}
where $(a)$ holds by the definition of $\mathcal{I}^{(n)}$ in \eqref{eq:ASetIxk}; $(b)$ holds because, in all cases when $I(V;Y_{(1)}) < I(V;Z)$, we have $\big| \mathcal{G}^{(n)}_0 \big|   +   \big| \mathcal{G}^{(n)}_{2} \big|  =  \big| \mathcal{R}^{(n)}_{1,2} \big| +  \big| \mathcal{R}^{\prime (n)}_{1,2} \big| + \big| \mathcal{R}^{(n)}_{1} \big| +  \big| \mathcal{R}^{\prime (n)}_{1} \big|$ and $\big| \mathcal{G}^{(n)}_0 \big|  +  \big| \mathcal{G}^{(n)}_{2} \big| < \big| \mathcal{C}^{(n)}_{1} \big|  -  \big| \mathcal{C}^{(n)}_{1,2} \big|$ (see conditions in \eqref{eq:assumpRate1Impl2x2} and \eqref{eq:assumpRate1Impl2x3}) and, otherwise, we have $\big| \mathcal{R}^{(n)}_{1,2} \big| + \big| \mathcal{R}^{\prime (n)}_{1} \big| = \big|\mathcal{C}_{1,2}^{(n)}\big|$, $\big|\mathcal{R}^{(n)}_{1}\big| + \big| \mathcal{R}^{\prime (n)}_{1,2} \big| = \big|\mathcal{C}_{1}^{(n)}\big|$ and $\big| \mathcal{G}^{(n)}_0 \big|  +  \big| \mathcal{G}^{(n)}_{2} \big| \geq \big| \mathcal{C}^{(n)}_{1} \big|  -  \big| \mathcal{C}^{(n)}_{1,2} \big|$ (see condition in \eqref{eq:assumpRate1Impl2x1}); and $(c)$ follows from the partition of $\mathcal{H}_{V}^{(n)}$ defined in \eqref{eq:sG0x}--\eqref{eq:sC12x}.

Similarly, if the \gls*{pcs} operates to achieve the corner point of $\mathfrak{R}_{\text{MI-WTBC}}^{(2)}$, we have
\begin{align}
\big| \mathcal{I}^{(n)} \big| & \stackrel{(a)}{=}  \big| \mathcal{G}^{(n)}_0 \big|   +   \big| \mathcal{G}^{(n)}_{1} \big|  +  \big| \mathcal{G}^{(n)}_{2} \big|  -  \big| \mathcal{R}^{(n)}_{1,2} \big| -  \big| \mathcal{R}^{\prime (n)}_{1,2} \big| \nonumber  \\
& \quad - \big| \mathcal{R}^{(n)}_{1} \big| -  \big| \mathcal{R}^{\prime (n)}_{1} \big| - \big| \mathcal{R}^{(n)}_{2} \big| -  \big| \mathcal{R}^{\prime (n)}_{2} \big| - \big| \mathcal{R}^{(n)}_{\text{S}} \big| \nonumber \\
& \stackrel{(b)}{=}  \big| \mathcal{G}^{(n)}_0 \big|   +   \big| \mathcal{G}^{(n)}_{1} \big|   -  \big| \mathcal{R}^{(n)}_{1,2} \big| -  \big| \mathcal{R}^{\prime (n)}_{1,2} \big| - \big| \mathcal{R}^{(n)}_{2} \big| -  \big| \mathcal{R}^{\prime (n)}_{2} \big| \nonumber \\
& \stackrel{(c)}{=} \Big\{ \big| \mathcal{G}^{(n)}_0 \big|  +  \big| \mathcal{G}^{(n)}_{1} \big| -   \big| \mathcal{C}^{(n)}_{2} \big|  -  \big| \mathcal{C}^{(n)}_{1,2} \big| \Big\}^{+} \nonumber \\
& \stackrel{(d)}{=} \Big\{ \Big| \mathcal{H}^{(n)}_{V|Z} \cap \mathcal{L}^{(n)}_{V|Y_{(2)}} \Big| - \Big| \big( \mathcal{H}^{(n)}_{V|Z} \big)^{\text{C}} \cap \big( \mathcal{L}^{(n)}_{V|Y_{(2)}} \big)^{\text{C}} \Big| \Big\}^{+} \nonumber \\
& = \Big\{ \big| \mathcal{H}^{(n)}_{V|Z}  \big| - \big| \big( \mathcal{L}^{(n)}_{V|Y_{(2)}} \big)^{\text{C}} \big| \Big\}^{+}, \label{eq:setib}
\end{align}
where $(a)$ holds by the definition of $\mathcal{I}^{(n)}$ in \eqref{eq:ASetIxbk}; $(b)$ follows from \eqref{eq:ASetRsx} because the set $\mathcal{R}^{(n)}_{\text{S}}$ has size $\big| \mathcal{G}^{(n)}_2 \big| - \big| \mathcal{R}^{(n)}_1 \big| - \big| \mathcal{R}^{\prime (n)}_1 \big|$; $(c)$ holds because, in all cases contemplated in Section~\ref{sec:PCSx1il} when $I(V;Y_{(2)}) < I(V;Z)$, we have $\big| \mathcal{G}^{(n)}_0 \big|   +   \big| \mathcal{G}^{(n)}_{1} \big|   = \big| \mathcal{R}^{(n)}_{1,2} \big| + \big| \mathcal{R}^{\prime (n)}_{1,2} \big| + \big| \mathcal{R}^{(n)}_{2} \big| +  \big| \mathcal{R}^{\prime (n)}_{2} \big|$ and $\big| \mathcal{G}^{(n)}_0 \big|  +  \big| \mathcal{G}^{(n)}_{1} \big| < \big| \mathcal{C}^{(n)}_{2} \big|  -  \big| \mathcal{C}^{(n)}_{1,2} \big|$ (see condition in \eqref{eq:assumpRate1Impl2x3}) and, otherwise, we have $\big|\mathcal{R}^{(n)}_{2}\big| + \big| \mathcal{R}^{\prime (n)}_{1,2} \big| = \big|\mathcal{C}_{2}^{(n)}\big|$ and $\big| \mathcal{R}^{(n)}_{1,2} \big| + \big| \mathcal{R}^{\prime (n)}_{2} \big| = \big|\mathcal{C}_{1,2}^{(n)}\big|$ (see conditions in \eqref{eq:assumpRate1Impl2x1} and \eqref{eq:assumpRate1Impl2x2}); and $(d)$ follows from the partition of $\mathcal{H}_{V}^{(n)}$ defined in \eqref{eq:sG0x}--\eqref{eq:sC12x}.

Consequently, the rate of $S_{(k)}^{(V)}$, which is carried by the inner-layer $\tilde{A}_{1:L}^n$, is
\begin{align}
\frac{1}{nL} \sum_{i=1}^L  \big| S_{(k),i}^{(V)} \big| &  = \frac{( L  -  2  )}{nL} \big| \mathcal{I}^{(n)} \big|  +  \frac{1}{nL} \Big( \big| \mathcal{I}^{(n)} \cup \mathcal{G}^{(n)}_{1} \cup \mathcal{G}^{(n)}_{1,2} \big| +  \big| \mathcal{I}^{(n)} \cup \mathcal{G}^{(n)}_{2} \big| \Big) \nonumber \\
&  =  \frac{1}{n} \big| \mathcal{I}^{(n)} \big| + \frac{1}{nL} \Big( \big| \mathcal{G}^{(n)}_{1} \big| + \big| \mathcal{G}^{(n)}_{2} \big| + \big| \mathcal{G}^{(n)}_{1,2} \big| \Big) \nonumber \\
&  \stackrel{(a)}{=} \frac{1}{n} \big| \mathcal{I}^{(n)} \big| + \frac{1}{nL} \Big| \mathcal{H}^{(n)}_{V|Z} \cap \Big( \mathcal{L}^{(n)}_{V|Y_{(1)}}\cap   \mathcal{L}^{(n)}_{V|Y_{(2)}}\Big)^{\text{C}}\Big| \nonumber \\ 
&  \stackrel{(b)}{\geq} \frac{1}{n} \big| \mathcal{I}^{(n)} \big| + \frac{1}{nL}  \Big( \big| \mathcal{H}^{(n)}_{V|Z} \big| - \Big| \big( \mathcal{L}^{(n)}_{V|Y_{(1)}} \big)^{\text{C}} \Big| \Big) \nonumber \\
&   \stackrel{(c)}{=} \frac{1}{n} \Big\{ \big| \mathcal{H}^{(n)}_{V|Z} \big| -  \big|  \big( \mathcal{L}^{(n)}_{V|Y_{(k)}} \big)^{\text{C}} \big| \Big\}^{+} + \frac{1}{nL} \Big( \big| \mathcal{H}^{(n)}_{V|Z} \big| -  \big| \big( \mathcal{L}^{(n)}_{V|Y_{(1)}} \big)^{\text{C}} \big| \Big) \nonumber \\
&  \xrightarrow{n \rightarrow \infty} \big\{ H(V|Z) -  H(V|Y_{(k)}) \big\}^{+} +  \frac{1}{L} \big( H(V|Z)  - H(V|Y_{(1)}) \big) \nonumber \\
&   \xrightarrow{L \rightarrow \infty} \big\{ H(V|Z) -  H(V|Y_{(k)}) \big\}^{+}, \label{eq:ratesv}
\end{align}
where $(a)$ holds by the partition of $\mathcal{H}_{V}^{(n)}$ in \eqref{eq:sG0x}--\eqref{eq:sC12x}; $(b)$ follows from applying elementary set operations and because $\big| \mathcal{L}^{(n)}_{V|Y_{(1)}}\cap   \mathcal{L}^{(n)}_{V|Y_{(2)}}\big| \leq \big| \mathcal{L}^{(n)}_{V|Y_{(\ell)}} \big|$ for any $\ell \in [1,2]$; $(c)$ follows from \eqref{eq:setia} and \eqref{eq:setib}; and the limit when $n$ goes to infinity follows from applying source polarization~\cite{arikan2010source}. 

Now, consider the confidential information $S_{(k)}^{(U)}$ that is carried in the outer-layer $\tilde{T}_{(k),1:L}^n$. According to Section~\ref{sec:PCSx1ol}, if $i \in [2,L-1]$, we have $S_{(k),i}^{(U)} = \tilde{T}_{(k),i} \big[ \mathcal{F}_0^{(n)} \setminus  \big(\mathcal{D}_k^{(n)} \cup \mathcal{L}_k^{(n)} \big) \big]$. At Block~1, if $k=1$ then we have $S_{(1),1}^{(U)} = \tilde{T}_{(1),1} \big[ \big( \mathcal{F}_0^{(n)} \cup \mathcal{F}_1^{(n)} \big) \setminus  \big(\mathcal{D}_1^{(n)} \cup \mathcal{L}_1^{(n)} \big) \big]$ and, otherwise, we have $S_{(2),1}^{(U)} = \tilde{T}_{(2),1} \big[ \mathcal{F}_0^{(n)} \cup \mathcal{F}_2^{(n)} \big]$. Finally, at Block~$L$, if $k=1$ then we have $S_{(1),L}^{(U)} = \tilde{T}_{(1),L} \big[  \mathcal{F}_0^{(n)} \big]$, while if $k=2$ then $S_{(2),L}^{(U)} = \tilde{T}_{(2),L} \big[ \mathcal{F}_0^{(n)} \setminus  \big(\mathcal{D}_2^{(n)} \cup \mathcal{L}_2^{(n)} \big) \big]$. Consequently, the rate of $S_{(k)}^{(U)}$, which is carried by the outer-layer $\tilde{T}_{(k),1:L}^n$, is
\begin{align}
\frac{1}{nL} \sum_{i=1}^L  \big| S_{(k),i}^{(U)} \big|  & = \frac{1}{n} \big| \mathcal{F}_0^{(n)} \setminus  \big(\mathcal{D}_k^{(n)} \cup \mathcal{L}_k^{(n)} \big) \big|  +  \frac{1}{nL} \big| \mathcal{D}_{k}^{(n)} \cup \mathcal{L}^{(n)}_{k} \cup \mathcal{F}^{(n)}_{k} \big| \nonumber  \\
&  \stackrel{(a)}{\geq} \frac{1}{n} \Big( \big| \mathcal{H}_{U_{(k)}|VZ}^{(n)}  \big| - \big|  \big( \mathcal{L}_{U_{(k)}|VY_{(k)}}^{(n)} \big)^{\text{C}} \big| - \big\{ \big| \mathcal{C}_{k}^{(n)} \big| + \big| \mathcal{C}_{1,2}^{(n)} \big| - \big| \mathcal{G}_0^{(n)} \big| - \big| \mathcal{G}_{\bar{k}}^{(n)} \big| \big\}^{+} \Big) \nonumber \\
& \quad +  \frac{1}{nL} \Big( \big|  \mathcal{H}_{U_{(k)}|VY_{(k)}}^{(n)} \big| + \Big\{ \big| \mathcal{C}_{k}^{(n)} \big| + \big| \mathcal{C}_{1,2}^{(n)} \big| - \big| \mathcal{G}_0^{(n)} \big| - \big| \mathcal{G}_{\bar{k}}^{(n)} \big| \Big\}^{+} \Big)  \nonumber \\
&  \xrightarrow{n \rightarrow \infty} H(U_{(k)}|VZ) -  H(U_{(k)}|VY_{(k)})  - \big\{ H(V|Y_{(k)}) - H(V|Z) \big\}^{+} \nonumber \\
& \quad +  \frac{1}{L} \big( H(U_{(k)}|VY_{(k)}) +  \big\{ H(V|Y_{(k)}) - H(V|Z) \big\}^{+} \big) \nonumber \\
&   \xrightarrow{L \rightarrow \infty} H(U_{(k)}|VZ) -  H(U_{(k)}|VY_{(k)})  - \big\{ H(V|Y_{(k)}) - H(V|Z) \big\}^{+}, \label{eq:ratesu}
\end{align}
where $(a)$ holds by \eqref{eq:f0}--\eqref{eq:xnose} and recall that $\big| \mathcal{D}_k^{(n)} \big| + \big| \mathcal{F}_k^{(n)} \big| = \big| \mathcal{H}_{U_{(k)}|V}^{(n)} \setminus \mathcal{L}_{U_{(k)}|VY_{(k)}}^{(n)}\big|$, which is greater or equal to $\big| \mathcal{H}_{U_{(k)}|VY_{(k)}}^{(n)} \big|$; and the limit when $n$ goes to infinity follows from applying the source polarization theorem~\cite{arikan2010source}, where we have used similar reasoning as in \eqref{eq:setia} and \eqref{eq:setib} to obtain $\big\{ \big| \mathcal{C}_{k}^{(n)} \big| + \big| \mathcal{C}_{1,2}^{(n)} \big| - \big| \mathcal{G}_0^{(n)} \big| - \big| \mathcal{G}_{\bar{k}}^{(n)} \big| \big\}^{+} \xrightarrow{n \rightarrow \infty} \big\{H(V|Y_{(k)}) - H(V|Z)  \big\}^{+}$.

Finally, by combining \eqref{eq:ratesv} and \eqref{eq:ratesu}, we obtain that the rate of $S_{(k)}$ is
 \begin{align*}
\frac{1}{nL} \sum_{i=1}^L  \big| S_{(k),i} \big| = \frac{1}{nL} \sum_{i=1}^L  \big( \big| S_{(k),i}^{(V)} \big| + \big| S_{(k),i}^{(U)} \big| \big)  \xrightarrow{n \rightarrow \infty } H(VU_{(k)}|Z) - H(VU_{(k)}|Y_{(k)}),
\end{align*}
which is equal to the rate $R_{S_{(k)}}^{\star k}$ defined in \eqref{eq:rsk}.

\item Rate of $S_{(\bar{k})}$. 
The confidential message $S_{(\bar{k})}$ is carried entirely in $\tilde{T}_{(\bar{k}),1:L}^n$. According to Section~\ref{sec:PCSx1ol}, if $i \in [2,L-1]$, we have $S_{(\bar{k}),i}^{(U)} = \tilde{T}_{(\bar{k}),i} \big[ \mathcal{Q}_0^{(n)} \setminus  \big(\mathcal{O}_{\bar{k}}^{(n)} \cup \mathcal{N}_{\bar{k}}^{(n)} \cup \mathcal{M}_{\bar{k}}^{(n)} \big) \big]$. At Block~1, if $\bar{k}=1$ then $S_{(1),1}^{(U)} = \tilde{T}_{(1),1} \big[ \big( \mathcal{Q}_0^{(n)} \cup \mathcal{Q}_1^{(n)} \big) \setminus  \big(\mathcal{O}_1^{(n)} \cup\mathcal{N}_1^{(n)} \cup \mathcal{M}_1^{(n)} \big) \big]$ and, otherwise, $S_{(2),1}^{(U)} = \tilde{T}_{(2),1} \big[ \mathcal{Q}_0^{(n)} \cup \mathcal{Q}_2^{(n)} \big]$. At Block~$L$, if $\bar{k}=1$ then we have $S_{(1),L}^{(U)} = \tilde{T}_{(1),L} \big[  \mathcal{Q}_0^{(n)} \big]$, while if $\bar{k}=2$ then we have that $S_{(2),L}^{(U)} = \tilde{T}_{(2),L} \big[ \mathcal{Q}_0^{(n)} \setminus  \big(\mathcal{O}_2^{(n)} \cup \mathcal{N}_2^{(n)} \cup \mathcal{M}_2^{(n)} \big) \big]$. Consequently, we obtain
\begin{align}
& \frac{1}{nL} \sum_{i=1}^L  \big| S_{(\bar{k}),i}^{(U)} \big|  \nonumber \\
& \quad = \frac{1}{n} \big| \mathcal{Q}_0^{(n)} \setminus  \big(\mathcal{O}_{\bar{k}}^{(n)}  \cup \mathcal{N}_{\bar{k}}^{(n)} \cup \mathcal{M}_{\bar{k}}^{(n)} \big) \big|  +  \frac{1}{nL} \big| \mathcal{O}_{\bar{k}}^{(n)} \cup \mathcal{N}_{\bar{k}}^{(n)} \cup \mathcal{M}^{(n)}_{\bar{k}} \cup \mathcal{Q}^{(n)}_{\bar{k}} \big| \nonumber  \\
& \quad \stackrel{(a)}{=} \frac{1}{n} \Big( \Big| \mathcal{H}_{U_{(\bar{k})}|VU_{(k)}Z}^{(n)} \cap \mathcal{L}_{U_{(\bar{k})}|VY_{(\bar{k})}}^{(n)}  \Big| - \Big| \big( \mathcal{H}_{U_{(\bar{k})}|VU_{(k)}}^{(n)} \big)^{\text{C}} \cap \mathcal{H}_{U_{(\bar{k})}|V}^{(n)} \setminus \mathcal{L}_{U_{(\bar{k})}|VY_{(\bar{k})}}^{(n)}  \Big| \nonumber \\
& \qquad - \Big|  \mathcal{H}_{U_{(\bar{k})}|VU_{(k)}}^{(n)} \cap \big(\mathcal{H}_{U_{(\bar{k})}|VU_{(k)}Z}^{(n)} \big)^{\text{C}} \setminus \mathcal{L}_{U_{(\bar{k})}|V Y_{(\bar{k})}}^{(n)} \Big| - \big| \mathcal{M}_{\bar{k}}^{(n)} \big| \Big) \nonumber \\
& \qquad +  \frac{1}{nL} \big|\mathcal{O}_{\bar{k}}^{(n)}  \cup  \mathcal{N}_{\bar{k}}^{(n)} \cup \mathcal{M}^{(n)}_{\bar{k}} \cup \mathcal{Q}^{(n)}_{\bar{k}} \big|  \nonumber \\
& \quad \stackrel{(b)}{\geq} \frac{1}{n} \Big( \Big| \mathcal{H}_{U_{(\bar{k})}|VU_{(k)}Z}^{(n)} \cap \mathcal{L}_{U_{(\bar{k})}|VY_{(\bar{k})}}^{(n)}  \Big| - \Big|  \big(\mathcal{H}_{U_{(\bar{k})}|VU_{(k)}Z}^{(n)} \big)^{\text{C}} \setminus \mathcal{L}_{U_{(\bar{k})}|V Y_{(\bar{k})}}^{(n)} \Big| - \big| \mathcal{M}_{\bar{k}}^{(n)} \big| \Big) \nonumber \\
& \qquad +  \frac{1}{nL} \Big( \Big| \mathcal{H}_{U_{(\bar{k})}|V Y_{(\bar{k})}}^{(n)}  \Big| + \big| \mathcal{M}_{\bar{k}}^{(n)} \big| \Big)  \nonumber \\
& \quad = \frac{1}{n} \Big( \Big| \mathcal{H}_{U_{(\bar{k})}|VU_{(k)}Z}^{(n)} \Big| - \Big|  \big( \mathcal{L}_{U_{(\bar{k})}|V Y_{(\bar{k})}}^{(n)} \big)^{\text{C}} \Big| - \big| \mathcal{M}_{\bar{k}}^{(n)} \big| \Big)  + \frac{1}{nL} \Big( \Big|  \mathcal{H}_{U_{(\bar{k})}|V Y_{(\bar{k})}}^{(n)} \Big| + \big| \mathcal{M}_{\bar{k}}^{(n)} \big| \Big)  \nonumber \\
&  \xrightarrow{n \rightarrow \infty} H(U_{(\bar{k})}|VU_{(k)}Z) -  H(U_{(\bar{k})}|VY_{(\bar{k})})  -  \big( H(V|Y_{(\bar{k})}) - \min\{ H(V|Y_{(2)}), H(V|Z) \} \big) \nonumber \\
& \qquad \qquad + \frac{1}{L} \big( H(U_{(\bar{k})}|VY_{(\bar{k})}) - (H(V|Y_{(\bar{k})}) - \min\{ H(V|Y_{(2)}), H(V|Z) \}) \big) \nonumber \\
&   \xrightarrow{L \rightarrow \infty} H(U_{(\bar{k})}|VU_{(k)}Z) -  H(U_{(\bar{k})}|VY_{(\bar{k})})  -  \big( H(V|Y_{(\bar{k})}) - \min\{H(V|Y_{(2)}), H(V|Z) \} \big), \nonumber
\end{align}
where $(a)$ holds by the definition of sets in \eqref{eq:w0}--\eqref{eq:qk}; $(b)$ follows from the fact that sets $\big( \mathcal{H}_{U_{(\bar{k})}|VU_{(k)}}^{(n)} \big)^{\text{C}} \cap \mathcal{H}_{U_{(\bar{k})}|V}^{(n)}$ and $\mathcal{H}_{U_{(\bar{k})}|VU_{(k)}}^{(n)} \cap \big(\mathcal{H}_{U_{(\bar{k})}|VU_{(k)}Z}^{(n)} \big)^{\text{C}}$ are disjoint and subsets of $\big( \mathcal{H}_{U_{(\bar{k})}|VU_{(k)}Z}^{(n)} \big)^{\text{C}}$, and because $\big| \mathcal{O}_{\bar{k}}^{(n)} \big| + \big| \mathcal{N}_{\bar{k}}^{(n)} \big| + \big| \mathcal{Q}_{\bar{k}}^{(n)} \big| = \big| \mathcal{H}_{U_{(\bar{k})}|V}^{(n)} \setminus \mathcal{L}_{U_{(\bar{k})}|VY_{(\bar{k})}}^{(n)}\big|$, which is greater or equal to $\big| \mathcal{H}_{U_{(\bar{k})}|VY_{(\bar{k})}}^{(n)} \big|$; and the limit when $n$ goes to infinity follows from applying source polarization~\cite{arikan2010source} and the definition of $\mathcal{M}_{\bar{k}}^{(n)}$. If $\bar{k}=1$, according to \eqref{eq:m1} we have $\big| \mathcal{M}_{1}^{(n)} \big| = \big| \mathcal{G}_{1}^{(n)} \big| + \big| \mathcal{C}_{1}^{(n)} \big| - \big| \mathcal{G}_2^{(n)} \big| - \big| \mathcal{C}_{2}^{(n)} \big|$ if $I(V;Z) \leq I(V;Y_{(2)})$, whereas $\big| \mathcal{M}_{1}^{(n)} \big| = \big| \mathcal{C}_{1}^{(n)} \big| + \big| \mathcal{C}_{1,2}^{(n)} \big| - \big| \mathcal{G}_0^{(n)} \big| - \big| \mathcal{G}_{2}^{(n)} \big|$ if $I(V;Z) >  I(V;Y_{(2)})$, and we have
\begin{align*}
\big| \mathcal{C}_{1}^{(n)} \big| + \big| \mathcal{C}_{1,2}^{(n)} \big| - \big| \mathcal{G}_0^{(n)} \big| - \big| \mathcal{G}_{2}^{(n)} \big| & \xrightarrow{n \rightarrow \infty}  H(V|Y_{(1)}) - H(V|Z)),
\end{align*}
which follows from \eqref{eq:setia}, while from the partition of $\mathcal{H}_{V}^{(n)}$ in \eqref{eq:sG0x}--\eqref{eq:sC12x} we obtain
\begin{align}
\big| \mathcal{G}_{1}^{(n)} \big| + \big| \mathcal{C}_{1}^{(n)} \big| - \big| \mathcal{G}_2^{(n)} \big| - \big| \mathcal{C}_{2}^{(n)} \big| & = \big| \big( \mathcal{L}^{(n)}_{V|Y_{(1)}} \big)^{\text{C}} \cap \mathcal{L}^{(n)}_{V|Y_{(2)}} \big| -  \big| \big( \mathcal{L}^{(n)}_{V|Y_{(2)}} \big)^{\text{C}} \cap \mathcal{L}^{(n)}_{V|Y_{(1)}}  \big|  \nonumber \\
& = \big| \big( \mathcal{L}^{(n)}_{V|Y_{(1)}} \big)^{\text{C}} \big| - \big| \big( \mathcal{L}^{(n)}_{V|Y_{(2)}} \big)^{\text{C}} \big|, \nonumber \\
& \xrightarrow{n \rightarrow \infty} H(V|Y_{(1)}) - H(V|Y_{(2)}).  \nonumber
\end{align}
Otherwise, if $\bar{k}=2$, from \eqref{eq:m2} we have $\big| \mathcal{M}_{2}^{(n)} \big| \! =  \! \big\{ \big| \mathcal{C}_{2}^{(n)} \big| + \big| \mathcal{C}_{1,2}^{(n)} \big| - \big| \mathcal{G}_0^{(n)} \big| - \big| \mathcal{G}_{1}^{(n)} \big| \}^{\! +}$ and
\begin{align*}
\big\{ \big| \mathcal{C}_{2}^{(n)} \big| + \big| \mathcal{C}_{1,2}^{(n)} \big| - \big| \mathcal{G}_0^{(n)} \big| - \big| \mathcal{G}_{1}^{(n)} \big| \big\}^{+} & \xrightarrow{n \rightarrow \infty} \big\{H(V|Y_{(2)}) - H(V|Z)  \big\}^{+} \\
&  = H(V|Y_{(2)}) - \min\{H(V|Y_{(2)}), H(V|Z) \}.
\end{align*}
Therefore, the \gls*{pcs} attains $R_{S_{(\bar{k})}}^{\star k}$ defined in \eqref{eq:rsbk}.

\end{enumerate}

\vspace{0.15cm}
\noindent \textbf{Private-shared sequence rate}

First, from \cite{arikan2010source} (Section~V.A.4), we have that
\begin{align*}
\frac{1}{nL} \Big( \big|\kappa_{\Theta}^{(V)}\big| + \big|\kappa_{\Gamma}^{(V)}\big| + \sum_{k=1}^2 \big| \kappa_{\Upsilon\Phi_{(k)}}^{(V)}\big| \Big) \xrightarrow{n,L \rightarrow \infty} 0.
\end{align*}
If we substitute $\big|\kappa_{\Theta}^{(V)}\big|$ by $\big|\kappa_{\Psi}^{(V)}\big|$, it is very easy to prove, by applying similar reasoning, that the overall length is negligible in terms of rate as well.

If $k=1$, we have $\big| \kappa_{\Theta}^{(U)} \big| = \big| \mathcal{J}_{1}^{(n)} \big|$ and $\big| \kappa_{\Psi}^{(U)} \big| = 0$, whereas $\big| \kappa_{\Theta}^{(U)} \big| = 0$ and $\big| \kappa_{\Psi}^{(U)} \big| = \big| \mathcal{B}_{2}^{(n)} \big|$ if $k=2$. From the definition of $\mathcal{J}_{1}^{(n)}$ and $\mathcal{B}_{2}^{(n)}$ in \eqref{eq:jk} and \eqref{eq:bk} respectively, we obtain
\begin{align*}
\frac{1}{nL} \big| \kappa_{\Theta}^{(U)} \big| & \leq 
\big| \big( \mathcal{L}_{U_{(1)}|VY_{(1)}}^{(n)} \big)^{\text{C}} \big| \xrightarrow{n \rightarrow \infty}  \frac{1}{L} H(U_{(1)}|VY_{(1)}) \xrightarrow{L \rightarrow \infty}  0, \\
\frac{1}{nL} \big| \kappa_{\Psi}^{(U)} \big| & \leq 
\big| \big( \mathcal{L}_{U_{(2)}|VY_{(2)}}^{(n)} \big)^{\text{C}} \big| \xrightarrow{n \rightarrow \infty}  \frac{1}{L} H(U_{(2)}|VY_{(2)}) \xrightarrow{L \rightarrow \infty}  0.
\end{align*}
where the limit when $n$ goes to infinity follows from applying the source polarization theorem~\cite{arikan2010source}.

Finally, we have
\begin{align*}
& \frac{1}{nL} \sum_{\ell=1}^2 \big| \kappa_{\Upsilon\Phi_{(\ell)}}^{(U)}\big| + \frac{1}{nL} \big| \kappa_{O}^{(U)} \big| \nonumber \\
& \quad = \frac{1}{nL}  \Big(
L  \big| \big( \mathcal{H}_{U_{(k)}|V}^{(n)}  \big)^{\text{C}} \setminus  \mathcal{L}_{U_{(k)}|VY_{(k)}}^{(n)}  \big| +  \big| \mathcal{H}_{U_{(k)}|V}^{(n)} \setminus \mathcal{L}_{U_{(k)}|VY_{(k)}}^{(n)}  \big| \Big) \\
& \qquad + \frac{1}{nL} \Big( L  \big| \big( \mathcal{H}_{U_{(\bar{k})}|V}^{(n)}  \big)^{\text{C}} \setminus  \mathcal{L}_{U_{(\bar{k})}|VY_{(\bar{k})}}^{(n)}  \big| +  \big| \mathcal{H}_{U_{(\bar{k})}|V}^{(n)} \setminus \mathcal{L}_{U_{(\bar{k})}|VY_{(\bar{k})}}^{(n)}  \big| \Big) \\
& \qquad + \frac{1}{nL} \Big( \big| \big(\mathcal{H}_{U_{(\bar{k})}|VU_{(k)}}^{(n)}  \big)^{\text{C}} \cap\mathcal{H}_{U_{(\bar{k})}|V}^{(n)} \setminus \mathcal{L}_{U_{(\bar{k})}|VY_{(\bar{k})}}^{(n)} \big| \Big) \\
& \quad \leq \frac{1}{n}  \sum_{\ell=1}^2 \big| \big( \mathcal{H}_{U_{(\ell)}| V }^{(n)}  \big)^{\text{C}} \setminus  \mathcal{L}_{U_{(\ell)}|VY_{(\ell)}}^{(n)}  \big| + \frac{1}{nL} \big| \big( \mathcal{L}_{U_{(k)}|VY_{(k)}}^{(n)} \big)^{\text{C}} \big| + \frac{2}{nL} \big| \big( \mathcal{L}_{U_{(\bar{k})}|VY_{(\bar{k})}}^{(n)} \big)^{\text{C}} \big| \\
& \quad \stackrel{(a)}{\leq} \frac{1}{n} \sum_{\ell=1}^2 \big| \big( \mathcal{H}_{U_{(\ell)}| V Y_{(\ell)}}^{(n)}  \big)^{\text{C}} \setminus  \mathcal{L}_{U_{(\ell)}|VY_{(\ell)}}^{(n)}  \big| + \frac{1}{nL} \big| \big( \mathcal{L}_{U_{(k)}|VY_{(k)}}^{(n)} \big)^{\text{C}} \big| + \frac{2}{nL} \big| \big( \mathcal{L}_{U_{(\bar{k})}|VY_{(\bar{k})}}^{(n)} \big)^{\text{C}} \big| \\
& \quad \xrightarrow{n \rightarrow \infty}  \frac{1}{L} ( H(U_{(k)}| V Y_{(k)}) + 2 H(U_{(\bar{k})}| V Y_{(\bar{k})}) ) \xrightarrow{L \rightarrow \infty}  0,
\end{align*}
where $(a)$ holds because $\mathcal{H}_{U_{(\ell)}| V Y_{(\ell)}}^{(n)} \supset \mathcal{H}_{U_{(\ell)}| V}^{(n)}$ for any $\ell \in [1,2]$; and the limit when $n$ goes to infinity follows from applying the source polarization theorem~\cite{arikan2010source}.

Therefore, the amount of private-shared information between transmitter and legitimate receivers is negligible in terms of rate, and so is the rate of the additional transmissions.

\vspace{0.15cm}
\noindent \textbf{Rate of the additional randomness}

For $i \in [1,L]$, the encoder randomly draws (by \gls*{sc} encoding) the elements $\tilde{A}_i\big[\big(\mathcal{H}_V^{(n)} \big)^{\text{C}}\setminus \mathcal{L}_V^{(n)}\big]$, $\tilde{T}_{(k),i}\big[ \big( \mathcal{H}_{U_{(k)}|V}^{(n)} \big)^{\text{C}} \setminus \mathcal{L}_{U_{(k)}|V}^{(n)} \big]$ and $\tilde{T}_{(\bar{k}),i}\big[ \big( \mathcal{H}_{U_{(\bar{k})}|VU_{(k)}}^{(n)} \big)^{\text{C}} \setminus \mathcal{L}_{U_{(\bar{k})}|VU_{(k)}}^{(n)}  \big]$. Nevertheless, we have  
\begin{align*}
& \frac{1}{n} \Big( \big| \big(\mathcal{H}_V^{(n)} \big)^{\text{C}}\setminus \mathcal{L}_V^{(n)}\big| +  \Big| \big( \mathcal{H}_{U_{(k)}|V}^{(n)} \big)^{\text{C}} \setminus \mathcal{L}_{U_{(k)}|V}^{(n)} \Big|
+ \Big| \big( \mathcal{H}_{U_{(\bar{k})}|VU_{(k)}}^{(n)} \big)^{\text{C}} \setminus \mathcal{L}_{U_{(\bar{k})}|VU_{(k)}}^{(n)}  \Big| \Big)  \xrightarrow{n \rightarrow \infty} 0,
\end{align*}
where the limit as $n$ approaches infinity follows from applying source polarization~\cite{arikan2010source}.

\subsection{Distribution of the DMS after the polar encoding}\label{sec:distributionDMSx}

For $i \in [1,L]$, let $\tilde{q}_{A_i^n T_{(1),i}^n T_{(2),i}^n}$ denote the distribution of $(\tilde{A}_i^n,\tilde{T}_{(1),i}^n, \tilde{T}_{(2),i}^n)$ after the encoding. The following lemma, which is based on the results in \cite{7447169}, proves that $\tilde{q}_{A_i^n T_{(1),i}^n T_{(2),i}^n}$ and the marginal distribution $p_{A^n T_{(1)}^n T_{(2)}^n}$ of the original \gls*{dms} are nearly statistically indistinguishable for sufficiently large $n$ and, hence, so are $\tilde{q}_{V_i^n  T_{(1),i}^n T_{(2),i}^n X_i^n Y_{(1),i}^n Y_{(2),i}^n Z_i^n }$ and $p_{V^n T_{(1)}^n T_{(2)}^n X^n Y_{(1)}^n Y_{(2)}^n Z^n }$. This result is crucial for the reliability and secrecy performance of the polar coding scheme.

\begin{lemma}
\label{lemma:distDMSx}
For any $i \in  [1,L]$, we obtain 
\begin{IEEEeqnarray}{rCl}
\mathbb{V}(\tilde{q}_{A_i^n T_{(1),i}^n T_{(2),i}^n},p_{A^n T_{(1)}^n T_{(2)}^n}) & \leq & \delta^{(*)}_n, \nonumber \\
\mathbb{V}(\tilde{q}_{V_i^n  T_{(1),i}^n T_{(2),i}^n X_i^n Y_{(1),i}^n Y_{(2),i}^n Z_i^n }, p_{V^n T_{(1)}^n T_{(2)}^n X^n Y_{(1)}^n Y_{(2)}^n Z^n }) & \leq & \delta^{(*)}_n, \nonumber 
\end{IEEEeqnarray}
where $\delta^{(*)}_n \triangleq n \displaystyle 3 \sqrt{ 2  \sqrt{\ell n \delta_n 2 \ln 2} \big( \ell n - \log  \sqrt{\ell n \delta_n  2 \ln 2} \big) + \delta_n} + \sqrt{3} \sqrt{n \delta_n  2 \ln 2}$.
\end{lemma}
\begin{proof}
For the first claim, see \cite{alos2019polar} (Lemma~5) taking $M \triangleq 3$. The second holds because, for all $i \in [1,L]$, sequence $\tilde{X}_i^n$ and $X_i^n$ are deterministic functions of $(\tilde{A}_i^n,\tilde{T}_{(1),i}^n, \tilde{T}_{(2),i}^n)$ and $(A^n,T_{(1)}^n, T_{(2)}^n)$ respectively, and $\tilde{q}_{V_i^n  T_{(1),i}^n T_{(2),i}^n X_i^n Y_{(1),i}^n Y_{(2),i}^n Z_i^n } \equiv \tilde{q}_{X_i^n}p_{Y_{(1)}^n Y_{(2)}^n Z^n | X^n}$.
\end{proof}

\begin{remark}
Consider the \gls*{pcs} operating to achieve the corner point of $\mathfrak{R}_{\text{{MI-WTBC}}}^{(k)}$ for some $k \in [1,2]$. In this case, $\smash{O_{(\bar{k}),i}^{(U)} \subset \tilde{T}_{(\bar{k}),i}\big[ \big( \mathcal{H}_{U_{(\bar{k})}|VU_{(k)}}^{(n)} \big)^{\emph{\text{C}}} \cap \mathcal{H}_{U_{(\bar{k})}|V}^{(n)} \big]}$ is repeated, for $i \in [2,L]$, in $\tilde{T}_{(\bar{k}),i-1}^n$ (if $\bar{k}=1$) or, for $i \in [1,L-1]$, in $\tilde{T}_{(\bar{k}),i+1}^n$ (if $\bar{k}=2$). In both situations, $O_{(\bar{k}),i}^{(U)}$ is drawn by performing \gls*{sc} encoding and is repeated in some of the elements of the corresponding adjacent block whose indices correspond to $\mathcal{H}_{U_{(\bar{k})}|VU_{(k)}Z}^{(n)}$. Nevertheless, $O_{(\bar{k}),i}^{(U)}$ is not repeated directly, but the encoder copies $\bar{O}_{(\bar{k}),i}^{(U)} = O_{(\bar{k}),i}^{(U)} \oplus \kappa_{O}^{(U)}$ (see Figure~\ref{fig:enc_outerR1_2} and Figure~\ref{fig:enc_outerR2_1}). Hence, notice that $\kappa_{O}^{(U)}$ ensures that $\bar{O}_{(\bar{k}),i}^{(U)}$ is uniformly distributed.
\end{remark}

\subsection{Reliability analysis}\label{sec:performance_reliabilityx}
Consider that the \gls*{pcs} must achieve $(R_{S_{(1)}}^{\star k}, R_{S_{(2)}}^{\star k}, R_{W_{(1)}}^{\star k}, R_{W_{(2)}}^{\star k}) \subset \mathfrak{R}_{\text{{MI-WTBC}}}^{(k)}$. In this section we prove that Receiver~$k$ is able to reconstruct $(W_{(k)},S_{(k)})$ with arbitrary small error probability, while Receiver~$\bar{k}$ is able to reconstruct $(W_{(\bar{k})},S_{(\bar{k})})$. Recall that the inner-layer $\tilde{A}_{1:L}^n$ and the outer-layer $\tilde{T}_{(k),1:L}^n$ carry $(W_{(k)},S_{(k)})$, and the outer-layer $\tilde{T}_{(\bar{k}),1:L}^n$ carries $(W_{(\bar{k})}, S_{(\bar{k})})$. Although $\tilde{A}_{1:L}^n$ only contains information intended for Receiver~$k$, the other receiver must reliably reconstruct them in order to be able to decode $\tilde{T}_{(\bar{k}),1:L}^n$.

Consider the probability of incorrectly decoding $\big(\tilde{A}_{1:L}^n, \tilde{T}_{(1),1:L}^n \big)$ at Receiver~$\ell \in [1,2]$. For $i \in [1,L]$, let $\tilde{q}_{V_i^n T_{(\ell),i}^n Y_{(\ell),i}^n}$ and $p_{V^n T_{(\ell)}^n Y_{(\ell)}^n}$ be marginals of $\tilde{q}_{V_i^n T_{(1),i}^n T_{(2),i}^n X_i^n Y_{(1),i}^n Y_{(2),i}^n Z_i^n }$ and $p_{V^n T_{(1)}^n T_{(2)}^n X^n Y_{(1)}^n Y_{(2)}^n Z^n }$ respectively, and define an optimal coupling \cite{levin2009markov} (Proposition~4.7) \mbox{between} $\tilde{q}_{V_i^n T_{(\ell),i}^n Y_{(\ell),i}^n}$ and $p_{V^n T_{(\ell)}^n Y_{(\ell)}^n}$ such that $\mathbb{P} \big[ \mathcal{E}_{V_i^n T_{(\ell),i}^n Y_{(\ell),i}^n} \big] = \mathbb{V} \big( \tilde{q}_{V_i^n T_{(\ell),i}^n Y_{(\ell),i}^n} , p_{V^n T_{(\ell)}^n Y_{(\ell)}^n} \big)$,
where $\mathcal{E}_{V_i^n T_{(\ell),i}^n Y_{(\ell),i}^n} \triangleq \big\{ \big( \tilde{V}_i^n, T_{(\ell)}^n, \tilde{Y}_{(\ell),i}^n \big) \neq \big( V^n, T_{(\ell)}^n ,Y_{(\ell)}^n \big) \big\}$. Additionally, define
\begin{IEEEeqnarray}{c}
\Scale[0.95]{\mathcal{E}_{(\ell),i} \! \triangleq \! \Big\{ \Big( \hat{A}_{(\ell),i} \big[ \big( \mathcal{L}_{V|Y_{(\ell)}}^{(n)} \big)^{\text{C}} \big], \hat{T}_{(\ell),i} \big[ \big( \mathcal{L}_{U_{(\ell)}|VY_{(\ell)}}^{(n)} \big)^{\text{C}} \big] \Big) \! \neq \!  \Big( \tilde{A}_{i} \big[ \big( \mathcal{L}_{V|Y_{(\ell)}}^{(n)} \big)^{\text{C}} \big], \tilde{T}_{(\ell),i} \big[ \big( \mathcal{L}_{U_{(\ell)}|VY_{(\ell)}}^{(n)} \big)^{\text{C}} \big] \Big)  \Big\}.} \nonumber 
\end{IEEEeqnarray}
Recall that $(\Upsilon_{(\ell)}^{(V)},\Phi_{(\ell),1:L}^{(V)})$ and $(\Upsilon_{(\ell)}^{(U)},\Phi_{(\ell),1:L}^{(U)})$ is available to legitimate Receiver~$\ell$. Thus, $\mathbb{P} [ \mathcal{E}_{(1),1} ] = \mathbb{P} [ \mathcal{E}_{(2),L} ]  = 0$
because given $(\Upsilon_{(1)}^{(V)},\Phi_{(1),1:L}^{(V)})$ and $(\Upsilon_{(1)}^{(U)},\Phi_{(1),1:L}^{(U)})$ Receiver~1 knows $\tilde{A}_{1} \big[\big( \mathcal{L}_{V|Y_{(1)}}^{(n)} \big)^{\text{C}} \big]$ and $\tilde{T}_{(1),1} \big[\big( \mathcal{L}_{U_{(1)}|VY_{(1)}}^{(n)} \big)^{\text{C}} \big]$, while given $(\Upsilon_{(2)}^{(V)},\Phi_{(2),1:L}^{(V)})$ and $(\Upsilon_{(2)}^{(U)},\Phi_{(2),1:L}^{(U)})$ legitimate Receiver~2 knows $\tilde{A}_{L} \big[ \big( \mathcal{L}_{V|Y_{(2)}}^{(n)} \big)^{\text{C}} \big]$ and $\tilde{T}_{(2),L} \big[\big( \mathcal{L}_{U_{(2)}|VY_{(2)}}^{(n)} \big)^{\text{C}}  \big]$. Furthermore, due to the chaining structure, recall that we have $\big( \tilde{A}_{i} \big[ \big( \mathcal{L}_{V|Y_{(1)}}^{(n)} \big)^{\text{C}} \big], \tilde{T}_{(1),i} \big[ \big( \mathcal{L}_{U_{(1)}|VY_{(1)}}^{(n)} \big)^{\text{C}} \big] \big) \subset \big( \tilde{A}_{i-1}^n, \tilde{T}_{(1),i-1}^n \big)$ for $i \in [2,L]$. Therefore, at legitimate Receiver~1, for $i \in [2,L]$ we have
\begin{IEEEeqnarray}{c}
\mathbb{P} [ \mathcal{E}_{(1),i} ] \leq \mathbb{P} \big[ \big( \tilde{A}_{i-1}^n, \tilde{T}_{(1),i-1}^n \big) \big].  \label{eq:e1ix} 
\end{IEEEeqnarray}
Similarly, we have seen that $\big( \tilde{A}_{i} \big[ \big( \mathcal{L}_{V|Y_{(2)}}^{(n)} \big)^{\text{C}} \big], \tilde{T}_{(2),i} \big[ \big( \mathcal{L}_{U_{(2)}|VY_{(2)}}^{(n)} \big)^{\text{C}} \big] \big) \subset \big( \tilde{A}_{i+1}^n, \tilde{T}_{(2),i+1}^n \big)$ for $i \in [1,L-1]$. Thus, at legitimate Receiver~2, for $i \in [1,L-1]$ we obtain
\begin{IEEEeqnarray}{c}
\mathbb{P} [ \mathcal{E}_{(2),i} ] \leq \mathbb{P} \big[ \big( \tilde{A}_{i+1}^n, \tilde{T}_{(2),i+1}^n \big) \big]. \IEEEeqnarraynumspace \label{eq:e2ix} 
\end{IEEEeqnarray}
Hence, the probability of incorrectly decoding $\big( \tilde{A}_{i}^n, \tilde{T}_{(\ell),i}^n \big)$ at the Receiver~$\ell \in[1,2]$ is
\begin{align*}
& \mathbb{P} \big[ \big(\hat{A}_{(\ell),i}^n, \hat{T}_{(\ell),i}^n \big) \neq \big( \tilde{A}_i^n, \tilde{T}_{(\ell),i}^n \big) \big]  \\
& \quad = \mathbb{P} \Big[ \big(\hat{A}_{(\ell),i}^n, \hat{T}_{(\ell),i}^n \big) \neq \big( \tilde{A}_i^n, \tilde{T}_{(\ell),i}^n \big) \big| \mathcal{E}_{V_i^n U_{(\ell)} Y_{(\ell),i}^n}^{\text{C}} \cap  \mathcal{E}_{(\ell),i}^{\text{C}} \Big] \mathbb{P} \Big[ \mathcal{E}_{V_i^n U_{(\ell)} Y_{(\ell),i}^n}^{\text{C}} \cap  \mathcal{E}_{(\ell),i}^{\text{C}} \Big] \\
& \qquad + \mathbb{P} \Big[ \big(\hat{A}_{(\ell),i}^n, \hat{T}_{(\ell),i}^n \big) \neq \big( \tilde{A}_i^n, \tilde{T}_{(\ell),i}^n \big) \big| \mathcal{E}_{V_i^n U_{(\ell)} Y_{(\ell),i}^n} \cup  \mathcal{E}_{(\ell),i} \Big] \mathbb{P} \Big[\mathcal{E}_{V_i^n U_{(\ell)} Y_{(\ell),i}^n} \cup  \mathcal{E}_{(\ell),i} \Big] \\
& \quad \leq \mathbb{P} \Big[ \big(\hat{A}_{(\ell),i}^n, \hat{T}_{(\ell),i}^n \big) \neq \big( \tilde{A}_i^n, \tilde{T}_{(\ell),i}^n \big) \big| \mathcal{E}_{V_i^n U_{(\ell)} Y_{(\ell),i}^n}^{\text{C}} \cap  \mathcal{E}_{(\ell),i}^{\text{C}} \Big] +  \mathbb{P} \Big[\mathcal{E}_{V_i^n U_{(\ell)} Y_{(\ell),i}^n} \Big] +  \mathbb{P} \Big[ \mathcal{E}_{(\ell),i} \Big] \\
& \quad \stackrel{(a)}{\leq} 2 \delta_n +  \mathbb{P} \big[\mathcal{E}_{V_i^n U_{(\ell)} Y_{(\ell),i}^n} \big] +  \mathbb{P} \big[ \mathcal{E}_{(\ell),i} \big] \\
& \quad \stackrel{(b)}{\leq} 2 \delta_n +  \delta_n^{(*)} + \mathbb{P} \big[ \mathcal{E}_{(\ell),i} \big] \\
& \quad \stackrel{(c)}{\leq}  i \big( 2 \delta_n +  \delta_n^{(*)} \big)
\end{align*}
where $(a)$ holds by \cite{arikan2010source} (Theorem~2); $(b)$ follows from the optimal coupling and Lemma~\ref{lemma:distDMSx}; and $(c)$ holds by induction and \eqref{eq:e1ix}--\eqref{eq:e2ix}. Therefore, by the union bound, we obtain 
\begin{align*}
\mathbb{P} \big[ (W_{(\ell),1:L}, S_{(\ell),1:L}) \neq (\hat{W}_{(\ell),1:L} , \hat{S}_{(\ell),1:L}) \big] & \leq \sum_{i=1}^L \mathbb{P} \big[ \big(\hat{A}_{(\ell),i}^n, \hat{T}_{(\ell),i}^n \big) \neq \big( \tilde{A}_i^n, \tilde{T}_{(\ell),i}^n \big) \big]  \\
& \leq \frac{L (L+1)}{2} \big( 2 n \delta_n + \delta_{n}^{(*)} \big),
\end{align*}
and, consequently, for sufficiently large $n$ the \gls*{pcs} satisfies the reliability condition in~\eqref{eq:reliabilitycondx}.

\subsection{Secrecy analysis}\label{sec:performance_secrecyx}
Since the encoding of Section~\ref{sec:PCSx1} takes place over $L$ blocks of size $n$, we need to prove that
\begin{IEEEeqnarray}{c}
\lim_{n \rightarrow \infty} I \left( S_{(1),1:L} S_{(2),1:L} ; \tilde{Z}_{1:L}^n \right) = 0. \nonumber
\end{IEEEeqnarray}

Consider that the \gls*{pcs} operates to achieve $(R_{S_{(1)}}^{\star k}, R_{S_{(2)}}^{\star k}, R_{W_{(1)}}^{\star k}, R_{W_{(2)}}^{\star k}) \subset \mathfrak{R}_{\text{{MI-WTBC}}}^{(k)}$, where $k \in [1,2]$. For any $i \in [1,L]$, the confidential message $S_{(k),i}$ is stored in $\tilde{A}_i\big[ \mathcal{H}_{V|Z}^{(n)} \big]$ and $\tilde{T}_{(k),i}\big[  \mathcal{H}_{U_{(k)}|VZ}^{(n)} \big]$, and $S_{(\bar{k}),i}$ is stored in $\tilde{T}_{(\bar{k}),i}\big[  \mathcal{H}_{U_{(\bar{k})}|VU_{(\bar{k})}Z}^{(n)} \big]$. Hence, the following lemma shows that strong secrecy holds for any Block~$i \in [1,L]$.

\begin{lemma}
\label{lemma:secrecy1x}
For any $i \in [1,L]$ and sufficiently large $n$, we have
\begin{IEEEeqnarray}{c}
I \big(\tilde{A}_i\big[ \mathcal{H}_{V|Z}^{(n)} \big] \tilde{T}_{(k),i}\big[  \mathcal{H}_{U_{(k)}|VZ}^{(n)} \big] \tilde{T}_{(\bar{k}),i}\big[  \mathcal{H}_{U_{(\bar{k})}|VU_{(k)}Z}^{(n)} \big]; \tilde{Z}_i^n \big)  \leq \delta_n^{(\text{\emph{S}})}, \nonumber
\end{IEEEeqnarray}
where $\delta_n^{(\text{\emph{S}})} \triangleq 3 n \delta_n  + 2 \delta_n^{(*)} \big(3 n  -  \log \delta_n^{(*)} \big)$ and $\delta_n^{(*)}$ is defined as in Lemma~\ref{lemma:distDMSx}.
\end{lemma}
\begin{proof}
See Appendix~\ref{app:secrecy1x}.
\end{proof}

The following step is to prove asymptotically statistically independence between eavesdropper's observations from Blocks 1 to $L$. We address this part slightly differently depending on whether the \gls*{pcs} must achieve the corner point of $\mathfrak{R}_{\text{{MI-WTBC}}}^{(1)}$ or $\mathfrak{R}_{\text{{MI-WTBC}}}^{(2)}$. 

\vspace{0.15cm}
\noindent \textbf{Secrecy analysis when polar code operates to achieve the corner point of} $\bm{\mathfrak{R}_{\text{{MI-WTBC}}}^{(1)}}$

For convenience and with slight abuse of notation, for $i \in [1,L]$ let $\tilde{R}_{(1),i}^n \triangleq \big(\tilde{A}_i^n, \tilde{T}_{(1),i}^n \big)$, which carries $W_{(1),i} \triangleq \big[ W^{(V)}_{(1),i}, W^{(U)}_{(1),i}\big]$ and $S_{(1),i} \triangleq \big[ S^{(V)}_{(1),i}, S^{(U)}_{(1),i}\big]$. According to the previous encoding, we have $W^{(V)}_{(1),i} = \tilde{A}_i \big[ \mathcal{C}^{(n)} \big]$ and $W^{(U)}_{(1),i} = \tilde{T}_{(1),i} \big[ \mathcal{J}^{(n)}_{0} \cup \mathcal{J}^{(n)}_{1} \big]$. Therefore, we define $W_{(1),i} \triangleq \big[ W^{\prime}_{(1),i}, W^{\prime \prime}_{(1),i}\big]$, where $W^{\prime}_{(1),i} \triangleq \big[ W^{\prime (V)}_{(1),i}, W^{\prime (U)}_{(1),i}\big]$, being $W^{\prime (V)}_{(1),i} = \tilde{A}_i\big[\mathcal{C}^{(n)}_{1} \cup \mathcal{C}^{(n)}_{1,2} \big]$ and $W^{\prime (U)}_{(1),i} = \tilde{T}_{(1),i}\big[\mathcal{J}^{(n)}_{1} \big]$. Then, $W_{(1),i}^{\prime\prime} \triangleq \big[ W^{\prime \prime (V)}_{(1),i}, W^{\prime \prime (U)}_{(1),i}\big]$, being $W^{\prime \prime (V)}_{(1),i} = \tilde{A}_i\big[\mathcal{C}^{(n)}_{0} \cup \mathcal{C}^{(n)}_{2} \big]$ and $W^{\prime \prime (U)}_{(1),i} = \tilde{T}_{(1),i}\big[\mathcal{J}^{(n)}_{0} \big]$. Recall that, for $i \in [1,L-1]$, an \emph{encrypted} version of $W^{\prime}_{(1),i+1}$, namely $\bar{\Omega}_{(1),i+1} \triangleq \big[ \bar{\Theta}_{i+1}^{(V)}, \bar{\Gamma}_{i+1}^{(V)}, \bar{\Theta}_{(1),{i+1}}^{(U)} \big]$, is repeated in $\tilde{R}_{(1),i}^n$. In fact, from $\tilde{A}_{i+1}^n$, recall that sequence $\Delta_{(1),i+1}^{(V)} = \big[\bar{\Theta}_{3,i+1}^{(V)},\bar{\Gamma}_{3,i+1}^{(V)}\big] \subseteq \bar{\Omega}_{(1),i+1}$ is repeated in $\tilde{T}_{(1),i}^n$ but now this dependency appears implicitly. Furthermore, for $i \in [1,L]$ define $\Xi_{(1),i} \triangleq \big[ \Psi_i^{(V)}, \Gamma_i^{(V)}, \Pi_{(2),i}^{(V)}, \Lambda_{i}, \Lambda_{(1),i}^{(U)} \big]$, which denotes the entire sequence depending on $\tilde{R}_{(1),i}^n$ that is repeated in $\tilde{R}_{(1),i+1}^n$ if $i \in [1,L-1]$. 

\setstretch{1.20} 
Finally, for $i \in [1,L]$ we have $\tilde{T}_{(2),i}^n$ that carries $W_{(2),i}$ and $S_{(2),i}$. For convenience, define  $\Xi_{(2),i} \triangleq \big[ \Psi_{(2),i}^{(U)}, \Lambda_{(2),i}^{(U)} \big]$, which, together with $\bar{O}^{(U)}_{(2),i}$, will be repeated in $\tilde{T}_{(2),i+1}^n$ if $i \in [1,L-1]$.
 
According to these previous definitions and setting $\kappa_{\Omega} \triangleq [\kappa_{\Theta}^{(V)}, \kappa_{\Gamma}^{(V)}, \kappa_{\Theta}^{(U)}]$, notice that Figure~\ref{fig:dependenciesR1x} represents a Bayesian graph that describes the dependencies between the variables involved in the \gls*{pcs} of Section~\ref{sec:PCSx1} when it operates to achieve the corner point of $\mathfrak{R}_{\text{{MI-WTBC}}}^{(1)}$.

Although we have seen that the \gls*{pcs} introduces bidirectional dependencies, we have reformulated the encoding to obtain that they take place forward only. To do so, additionally we regard $\bar{\Omega}_{(1),i}$ as an independent random sequence generated at Block~$i-1$ and properly stored in $\tilde{R}_{i-1}^n$. Then, by using $\kappa_{\Omega}$, the encoder obtains $W^{\prime}_{(1),i}$ that is repeated in Block~$i$. 

\begin{figure}[h]
\centering
\includegraphics[width=0.97\linewidth]{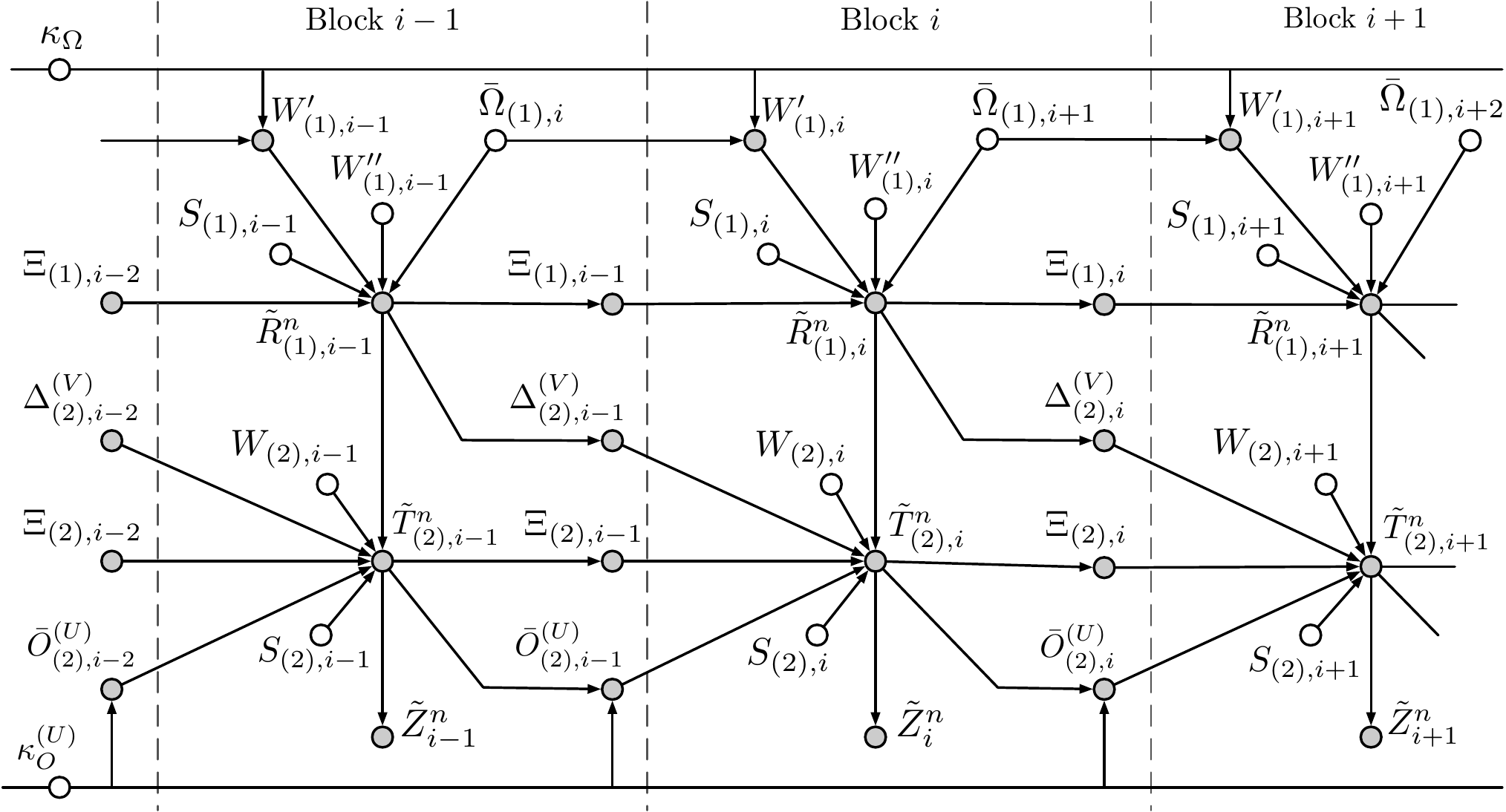}
\caption{\setstretch{1.2}
Graphical representation (Bayesian graph) of the dependencies between random variables involved in the \gls*{pcs} when it operates to achieve the corner point of $\smash{\mathfrak{R}_{\text{{MI-WTBC}}}^{(1)}}$. Independent random variables are indicated by white nodes, whereas those that are dependent are indicated by gray nodes.
}\label{fig:dependenciesR1x} 
\end{figure}

The following lemma shows that eavesdropper observations $\tilde{Z}_i^n$ are asymptotically statistically independent of observations $\tilde{Z}_{1:i-1}^n$ from previous blocks.
\begin{lemma}
\label{lemma:secrecy2x}
For any~$i \in [2,L]$ and sufficiently large $n$, we have
\vspace*{-0.14cm}
\begin{align*}
I \big(S_{(1),1:L} S_{(2),1:L} \tilde{Z}_{1:i-1}^n; \tilde{Z}_{i}^n \big) \leq \delta_n^{(\text{\emph{S}})}, 
\end{align*}
where $\delta_n^{(\text{\emph{S}})}$ is defined as in Lemma~\ref{lemma:secrecy1x}.
\end{lemma}
\begin{proof}
See Appendix~\ref{app:secrecy2x}.
\end{proof}

Therefore, we obtain
\vspace*{-0.1cm}
\begin{align*}
I \big(S_{(1),1:L} S_{(2),1:L} ; \tilde{Z}_{1:L}^n \big)  & =  I \big(S_{(1),1:L} S_{(2),1:L} ; \tilde{Z}_{1}^n \big) + \sum_{i=2}^L I \big(S_{(1),1:L} S_{(2),1:L} ; \tilde{Z}_{i}^n \big| \tilde{Z}_{1:i-1}^n \big) \\[-0.5em]
& \stackrel{(a)}{\leq} I \big(S_{(1),1:L} S_{(2),1:L} ; \tilde{Z}_{1}^n \big) + (L-1) \delta_n^{(\text{S})} \\
& \stackrel{(b)}{\leq}  L \delta_n^{(\text{S})},
\end{align*}
where $(a)$ holds by Lemma~\ref{lemma:secrecy2x}; and $(b)$ holds by independence between $S_{2:L}$ and any random variable from Block~$1$, and from applying Lemma~\ref{lemma:secrecy1x} to bound $I(S_{(1),1}S_{(2),1};\tilde{Z}_1^n)$.

Thus, for sufficiently large $n$, the \gls*{pcs} satisfies the strong secrecy condition in~\eqref{eq:secrecycondx}.


\vspace{0.15cm}
\noindent \textbf{Secrecy analysis when polar code operates to achieve the corner point of} $\bm{\mathfrak{R}_{\text{{MI-WTBC}}}^{(2)}}$
\vspace{0.1cm}

Now, for convenience and with slight abuse of notation, for $i \in [1,L]$ let $\tilde{R}_{(2),i}^n \triangleq \big(\tilde{A}_i^n, \tilde{T}_{(2),i}^n \big)$, which carries $W_{(2),i} \triangleq \big[ W^{(V)}_{(2),i}, W^{(U)}_{(2),i}\big]$ and $S_{(2),i} \triangleq \big[ S^{(V)}_{(2),i}, S^{(U)}_{(2),i}\big]$. According to the previous encoding, we have $W^{(V)}_{(2),i} = \tilde{A}_i \big[ \mathcal{C}^{(n)} \big]$ and $W^{(U)}_{(2),i} = \tilde{T}_{(2),i} \big[ \mathcal{J}^{(n)}_{0} \cup \mathcal{J}^{(n)}_{2} \big]$. Therefore, we define $W_{(2),i} \triangleq \big[ W^{\prime}_{(2),i}, W^{\prime \prime}_{(2),i}\big]$, where $W^{\prime}_{(2),i} \triangleq \big[ W^{\prime (V)}_{(2),i}, W^{\prime (U)}_{(2),i}\big]$, being $W^{\prime (V)}_{(2),i} = \tilde{A}_i\big[\mathcal{C}^{(n)}_{2} \cup \mathcal{C}^{(n)}_{1,2} \big]$ and $W^{\prime (U)}_{(2),i} = \tilde{T}_{(2),i}\big[\mathcal{J}^{(n)}_{2} \big]$. Then, $W_{(2),i}^{\prime\prime} \triangleq \big[ W^{\prime \prime (V)}_{(2),i}, W^{\prime \prime (U)}_{(2),i}\big]$, being $W^{\prime \prime (V)}_{(2),i} = \tilde{A}_i\big[\mathcal{C}^{(n)}_{0} \cup \mathcal{C}^{(n)}_{2} \big]$ and $W^{\prime \prime (U)}_{(2),i} = \tilde{T}_{(2),i}\big[\mathcal{J}^{(n)}_{0} \big]$. Recall that, for $i \in [2,L]$, now an \emph{encrypted} version of $W^{\prime}_{(2),i-1}$, namely $\bar{\Omega}_{(2),i-1} \triangleq \big[ \bar{\Psi}_{i-1}^{(V)}, \bar{\Gamma}_{i-1}^{(V)}, \bar{\Psi}_{(2),{i-1}}^{(U)} \big]$, is repeated in $\tilde{R}_{(2),i}^n$. Indeed, $\Delta_{(2),i-1}^{(V)}  =  \big[\bar{\Psi}_{3,i-1}^{(V)},\bar{\Gamma}_{3,i-1}^{(V)}\big] \subseteq   \bar{\Omega}_{(2),i-1}$ is repeated in $\tilde{T}_{(2),i}^n$, and now this dependency appears implicitly.

\setstretch{1.20} 
Recall that, for $i \in [1,L-1]$, $\Pi_{(2),i}^{(V)} = \tilde{A}_{i} \big[ \mathcal{I}^{(n)} \cap \mathcal{G}_2^{(n)} \big]$, which contains part of $S_{(2),i}^{(V)}$, is repeated in $\tilde{A}_{i+1} \big[ \mathcal{R}_{\text{S}}^{(n)} \big]$, while $\Lambda_{1}^{(V)} = \tilde{A}_{1} \big[ \mathcal{R}_{\Lambda}^{(n)} \big]$, which contains part of $S_{(2),1}^{(V)}$, is replicated in $\tilde{A}_{2:L} \big[ \mathcal{R}_{\Lambda}^{(n)} \big]$. For convenience, now we would like to have backward dependencies only. Therefore, we can consider that $S_{(2),1}^{\prime (V)} = \tilde{A}_1 \big[ \big(\mathcal{I}^{(n)} \cup \mathcal{G}_1^{(n)} \big)  \setminus \big( \mathcal{R}_{\Lambda}^{(n)} \cup \mathcal{R}_{\text{S}}^{(n)} \big) \big]$, for $i \in [2,L-1]$ then $S_{(2),i}^{\prime (V)} = \tilde{A}_i \big[ \big(\mathcal{I}^{(n)} \setminus \mathcal{G}_2^{(n)} \big) \cup \mathcal{R}_{\text{S}}^{(n)} \big]$, and $S_{(2),L}^{\prime (V)}= \tilde{A}_i \big[ \mathcal{I}^{(n)} \cup \mathcal{R}_{\text{S}}^{(n)} \cup \mathcal{R}_{\Lambda}^{(n)} \big]$. Then, for $i \in [2,L]$, we have that $\Pi_{(2),i}^{\prime (V)} = \tilde{A}_{i} \big[ \mathcal{I}^{(n)} \cap \mathcal{R}_{\text{S}}^{(n)} \big]$ and $\Lambda_{i}^{\prime (V)} = \tilde{A}_{i} \big[ \mathcal{R}_{\Lambda}^{(n)} \big]$ are repeated in $\tilde{A}_{i-1} \big[ \big( \mathcal{I}^{(n)} \cap \mathcal{G}_2^{(n)}\big) \cup \mathcal{R}_{\Lambda}^{(n)} \big]$. Similarly, we can regard $S_{(2),1}^{\prime (U)} = \tilde{T}_{(2),1} \big[ \mathcal{F}_0^{(n)} \big]$, for $i \in [2,L-1]$ then $S_{(2),i}^{\prime (U)} = S_{(2),i}^{(U)} =\tilde{T}_{(2),i}\big[ \mathcal{F}_0^{(n)}  \setminus \big( \mathcal{D}_2^{(n)} \cup \mathcal{L}_2^{(n)} \cup \mathcal{O}_2^{(n)} \big) \big]$, and $S_{(2),L}^{\prime (U)} = \tilde{T}_{(2),L} \big[ \mathcal{F}_0^{(n)} \cup \mathcal{F}_2^{(n)} \big]$. Then, for $i \in [2,L]$, we can consider that $\Lambda_{(2),i}^{\prime (U)} = \tilde{T}_{(2),i} \big[ \mathcal{F}_2^{(n)} \big]$ is repeated in $\tilde{T}_{(2),i-1} \big[ \mathcal{F}_2^{(n)} \big]$. Therefore, for $i \in [1,L-1]$ we define $\Xi_{(2),i+1} \triangleq \big[ \Theta_{i+1}^{(V)}, \Gamma_{i+1}^{(V)}, \Pi_{(2),i+1}^{\prime (V)}, \Lambda_{i+1}^{\prime (V)}, \Lambda_{(2),i+1}^{\prime (U)} \big]$, which denotes the entire sequence depending on $\tilde{R}_{(2),i+1}^n$ that is repeated in $\tilde{R}_{(2),i}^n$.

Finally, we have $\tilde{T}_{(1),i}^n$ that carries $W_{(1),i}$ and $S_{(1),i}$. Now, for $i \in [1,L-1]$, we regard $S_{(1),i}^{\prime (U)} = \tilde{T}_{(1),i} \big[ \mathcal{Q}_0^{(n)}  \setminus \big( \mathcal{N}_1^{(n)} \cup \mathcal{M}_1^{(n)} \cup \mathcal{O}_1^{(n)} \big) \big]$, and $S_{(1),L}^{\prime (U)} = \tilde{T}_{(1),L} \big[ \mathcal{Q}_0^{(n)} \cup \mathcal{Q}_1^{(n)} \big]$. Then, for $i \in [2,L]$ we consider that $\Lambda_{i}^{\prime (U)} = \tilde{T}_{(1),i} \big[ \mathcal{Q}_1^{(n)} \big]$ is repeated in $\tilde{T}_{(1),i-1} \big[ \mathcal{Q}_1^{(n)} \big]$. Thus, we define $\Xi_{(1),i} \triangleq \big[ \Theta_{(1),i}^{(U)}, \Lambda_{(1),i}^{\prime (U)} \big]$, which, together with $\bar{O}^{(U)}_{(1),i}$, will be repeated in $\tilde{T}_{(1),i-1}^n$ if $i \in [2,L]$. Also, define ${\ominus}_{(1),i}^{(V)} \triangleq \big[\Delta_{(1),i}^{(V)}, \Pi_{(1),i}^{(V)} \big]$, which denotes the part of $\tilde{A}_i^n$ that is repeated in $\tilde{T}_{(1),i-1}^n$.

According to these previous definitions and setting $\kappa_{\Omega} \triangleq [\kappa_{\Psi}^{(V)}, \kappa_{\Gamma}^{(V)}, \kappa_{\Psi}^{(U)}]$, notice that Figure~\ref{fig:dependenciesR2x} represents a Bayesian graph that describes the dependencies between the variables involved in the \gls*{pcs} of Section~\ref{sec:PCSx1} when it operate to achieve the corner point of $\mathfrak{R}_{\text{{MI-WTBC}}}^{(2)}$.

In order to obtain that dependencies take place backward only, we regard $\bar{\Omega}_{(2),i}$, for $i \in [2,L]$, as an independent random sequence that is generated at Block~$i-1$ and is properly stored in $\tilde{R}_{(2),i-1}^n$. Then, by using $\kappa_{\Omega}$, the encoder obtains $W^{\prime}_{(2),i}$ that is repeated in Block~$i$.

\begin{figure}[h]
\centering
\includegraphics[width=0.97\linewidth]{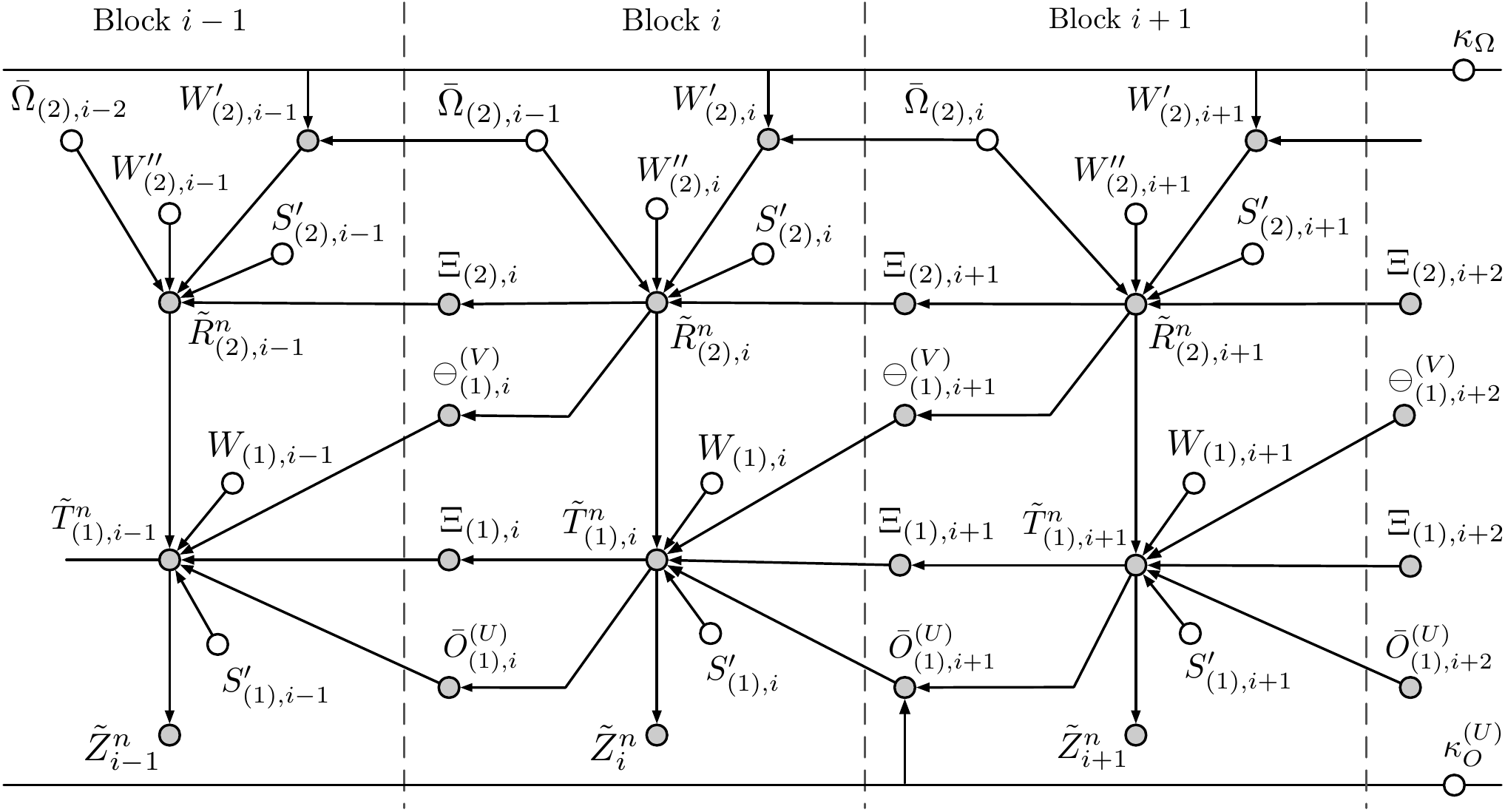}
\caption{\setstretch{1.2}
Graphical representation (Bayesian graph) of the dependencies between random variables involved in the \gls*{pcs} when it operates to achieve the corner point of $\smash{\mathfrak{R}_{\text{{MI-WTBC}}}^{(2)}}$. Independent random variables are indicated by white nodes, whereas those that are dependent are indicated by gray nodes.
}\label{fig:dependenciesR2x} 
\end{figure}

The following lemma shows that eavesdropper observations $\tilde{Z}_i^n$ are asymptotically statistically independent of observations $\tilde{Z}_{i+1:L}^n$ from previous blocks.
\begin{lemma}
\label{lemma:secrecy3x}
For any~$i \in [1,L-1]$ and sufficiently large $n$, we have
\begin{align*}
I \big(S_{(1),1:L} S_{(2),1:L} \tilde{Z}_{i+1:L}^n ; \tilde{Z}_{i}^n \big) \leq \delta_n^{(\text{\emph{S}})}, \\[-2.4em]
\end{align*}
where $\delta_n^{(\text{\emph{S}})}$ is defined as in Lemma~\ref{lemma:secrecy1x}.
\end{lemma}
\begin{proof}
See Appendix~\ref{app:secrecy3x}.
\end{proof}

Therefore, we obtain
\begin{align*}
& I \big(S_{(1),1:L} S_{(2),1:L} ; \tilde{Z}_{1:L}^n \big)  \\[-0.7em]
& =  I \big(S_{(1),1:L} S_{(2),1:L} ; \tilde{Z}_{L}^n \big) + \sum_{i^{\prime}=1}^{L-1} I \big(S_{(1),1:L} S_{(2),1:L} ; \tilde{Z}_{L- i^{\prime}}^n \big| \tilde{Z}_{L-i^{\prime}+1:L}^n \big) \\[-0.7em]
& \stackrel{(a)}{\leq} I \big(S_{(1),1:L} S_{(2),1:L} ; \tilde{Z}_{L}^n \big) + (L-1) \delta_n^{(\text{S})} \\
& \stackrel{(b)}{\leq}  L \delta_n^{(\text{S})}  \\[-2em]
\end{align*}
where $(a)$ holds by Lemma~\ref{lemma:secrecy3x}; and $(b)$ holds by independence between $S_{1:L-1}$ and any random variable from Block~$L$, and from applying Lemma~\ref{lemma:secrecy1x} to bound $I(S_{(1),L}S_{(2),L};\tilde{Z}_L^n)$. 

Thus, for sufficiently large $n$, the \gls*{pcs} satisfies the strong secrecy condition in~\eqref{eq:secrecycondx}.

\begin{remark}
For $i \in [1,L]$, recall that $O_{(1),i}^{(U)}$ is not generated independently, but drawn by using \gls*{sc} encoding. Hence, if the \gls*{pcs} operates to achieve the corner point of $\mathfrak{R}_{\text{{MI-WTBC}}}^{(2)}$, we only can obtain a causal Bayesian graph by reformulating the encoding so that dependencies between blocks take place place backward only.
\end{remark}

\begin{remark}
We conjecture that the use $\kappa_{\Omega}^{(V)}$ is not needed for the \gls*{pcs} to satisfy the strong secrecy condition when operates to achieve any of the corner points. However, the key is required in order to prove this condition by means of analyzing a causal Bayesian graph.
\end{remark}

\section{Concluding remarks}\label{sec:conclusionx}
A strongly secure \gls*{pcs} has been proposed for the \gls*{wtbc} with two legitimate receivers and one eavesdropper. We have compared two inner-bounds on the achievable region of the \gls*{miwtbc} model, where a transmitter wants to send different information (private and confidential) intended for each receiver. Then, we have provided a polar code that achieves the inner-bound that is strictly larger for a particular input distribution. The only difference between the random coding techniques used to characterize the two bounds is the decoding strategy: \emph{joint decoding} in the stronger inner-bound, and \emph{successive decoding} in the other.

Our scheme uses polar-based Marton's coding, which requires three encoding layers: one inner-layer that must be reliably decoded by both receivers, and two outer-layers associated to each legitimate receiver. Due to the non-degradedness assumption of the channel, the encoder builds a chaining construction that induces bidirectional dependencies between adjacent blocks, which need to be taken carefully into account in the secrecy analysis.

In order to achieve the larger inner-bound for a particular distribution, the chaining construction must repeat some elements from the inner-layer to the outer-layers of adjacent blocks, and turns out that this cross-dependency between encoding layers makes the use of \emph{polar-based joint decoding} crucial. As in \cite{alos2019polar}, the use of a negligible secret-key is required to prove that eavesdropper's observations for different blocks are statistically independent of one another, which is necessary to show that the polar code satisfies the strong secrecy condition. Furthermore, now the \gls*{pcs} needs to use another secret-key that also becomes negligible in terms of rate as the number of blocks grows indefinitely. This key is required to randomize a non-negligible set of elements of one outer-layer that are drawn by means of \gls*{sc} encoding and are needed by the corresponding receiver. In this way, the chaining construction can repeat these elements in adjacent blocks without causing a significant distortion on the input distribution.


\begin{appendices}

\section{Proof of Lemma \ref{lemma:secrecy1x}}\label{app:secrecy1x}
For $n$ sufficiently large, we have
\begin{align*}
& I \big(\tilde{A}_i\big[ \mathcal{H}_{V|Z}^{(n)} \big] \tilde{T}_{(k),i}\big[  \mathcal{H}_{U_{(k)}|VZ}^{(n)} \big] \tilde{T}_{(\bar{k}),i}\big[  \mathcal{H}_{U_{(\bar{k})}|VU_{(k)}Z}^{(n)} \big]; \tilde{Z}_i^n \big)   \\
& \quad \leq \big| \mathcal{H}_{V|Z}^{(n)} \big|  + \big|  \mathcal{H}_{U_{(k)}|VZ}^{(n)} \big| +  \big|  \mathcal{H}_{U_{(\bar{k})}|VU_{(k)}Z}^{(n)} \big| \\
& \qquad - H \big( \tilde{A}_i\big[ \mathcal{H}_{V|Z}^{(n)} \big] \tilde{T}_{(k),i}\big[  \mathcal{H}_{U_{(k)}|VZ}^{(n)} \big] \tilde{T}_{(\bar{k}),i}\big[  \mathcal{H}_{U_{(\bar{k})}|VU_{(k)}Z}^{(n)} \big] \big| \tilde{Z}^n_i \big) \\
& \quad \stackrel{(a)}{\leq} \big| \mathcal{H}_{V|Z}^{(n)} \big|  + \big|  \mathcal{H}_{U_{(k)}|VZ}^{(n)} \big| + \big|  \mathcal{H}_{U_{(\bar{k})}|VU_{(k)}Z}^{(n)} \big| \\
& \qquad - H \big( A\big[ \mathcal{H}_{V|Z}^{(n)} \big] T_{(k)}\big[\mathcal{H}_{U_{(k)}|VZ}^{(n)} \big] T_{(\bar{k})}\big[  \mathcal{H}_{U_{(\bar{k})}|VU_{(k)}Z}^{(n)} \big] \big| Z^n \big) + 6 n \delta_n^{(*)}  -  2 \delta_n^{(*)} \log \delta_n^{(*)} \\
& \quad \stackrel{(b)}{\leq} 3 n \delta_n  + 6 n \delta_n^{(*)}  -  2 \delta_n^{(*)} \log \delta_n^{(*)}
\end{align*}
where $(a)$ holds by \cite{alos2019polar} (Lemma~6) (where $M \triangleq 3$ and $O \triangleq 1$) and Lemma~\ref{lemma:distDMSx}; and $(b)$ because
\begin{align*}
& H \big( A\big[ \mathcal{H}_{V|Z}^{(n)} \big] T_{(k)}\big[\mathcal{H}_{U_{(k)}|VZ}^{(n)} \big] T_{(\bar{k})}\big[  \mathcal{H}_{U_{(\bar{k})}|VU_{(k)}Z}^{(n)} \big] \big| Z^n \big) \\ 
&  \quad \geq  H \big( A \big[ \mathcal{H}_{V|Z}^{(n)} \big] \big| Z^n \big) + H \big( T_{(k)}\big[\mathcal{H}_{U_{(k)}|VZ}^{(n)} \big] \big| V^n Z^n \big)  + H \big( T_{(\bar{k})}\big[  \mathcal{H}_{U_{(\bar{k})}|VU_{(k)}Z}^{(n)} \big] \big| V^n U_{(k)}^n Z^n \big) \\
&  \quad \geq  \sum_{j \in \mathcal{H}_{V|Z}^{(n)}} H \big(A (j) \big| A^{1:j-1} Z^n \big) 
+  \sum_{j \in \mathcal{H}_{T_{(k)}|VZ}^{(n)}} H \big(T_{(k)} (j) \big| T_{(k)}^{1:j-1} V^n Z^n \big) \\
& \qquad +  \sum_{j \in \mathcal{H}_{T_{(\bar{k})}|VU_{(k)}Z}^{(n)}} H \big(T_{(\bar{k})} (j) \big| T_{(\bar{k})}^{1:j-1} V^n T_{(k)}^n Z^n \big) \nonumber \\
&  \quad \geq \big| \mathcal{H}_{V|Z}^{(n)} \big| (1 - \delta_n) + \big| \mathcal{H}_{T_{(k)}|VZ}^{(n)} \big| (1 - \delta_n) + \big| \mathcal{H}_{T_{(\bar{k})}|VT_{(\bar{k})}Z}^{(n)} \big| (1 - \delta_n)
\end{align*}
where we have used the fact that conditioning does not increase entropy, the invertibility of $G_n$, and the definition of $\mathcal{H}_{V|Z}^{(n)}$, $\mathcal{H}_{T_{(k)}|VZ}^{(n)}$ and $\mathcal{H}_{T_{(\bar{k})}|VT_{(\bar{k})}Z}^{(n)}$ in \eqref{eq:HUZx}, \eqref{eq:HUVZx} and \eqref{eq:HUVUZx} respectively.

\section{Proof of Lemma~\ref{lemma:secrecy2x}}\label{app:secrecy2x}
For any $i \in [2,L]$ and sufficiently large $n$, we have
\begin{align*}
& I \big(S_{(1),1:L} S_{(2),1:L} \tilde{Z}_{1:i-1}^n; \tilde{Z}_{i}^n \big) \\
& \quad =  I \big(S_{(1),1:i} S_{(2),1:i} \tilde{Z}_{1:i-1}^n ;\tilde{Z}_{i}^n \big) + I \big(S_{(1),i+1:L} S_{(2),i+i:L} ; \tilde{Z}_{i}^n \big|  S_{(1),1:i-1} S_{(2),1:i-1} \tilde{Z}_{1:i-1}^n \big)  \\
& \quad \stackrel{(a)}{=} I \big(S_{(1),1:i} S_{(2),1:i} \tilde{Z}_{1:i-1}^n ;\tilde{Z}_{i}^n \big)  \\
& \quad \leq  I \big(S_{(1),1:i} S_{(2),1:i} \tilde{Z}_{1:i-1}^n \Xi_{(1),i-1} \Delta^{(V)}_{(2),i-1}  \Xi_{(2),i-1}  \bar{O}_{(2),i-1}^{(U)} ; \tilde{Z}_{i}^n  \big)  \\
& \quad =  I \big(S_{(1),i} S_{(2),i} \Xi_{(1),i-1} \Delta^{(V)}_{(2),i-1}  \Xi_{(2),i-1}  \bar{O}_{(2),i-1}^{(U)}; \tilde{Z}_{i}^n  \big)  \\
& \qquad +  I \big(S_{(1),1:i-1} S_{(2),1:i-1}\tilde{Z}_{1:i-1}^n ; \tilde{Z}_{i}^n \big| S_{(1),i} S_{(2),i} \tilde{Z}_{1:i-1}^n \Xi_{(1),i-1} \Delta^{(V)}_{(2),i-1}  \Xi_{(2),i-1}  \bar{O}_{(2),i-1}^{(U)} \big)  \\
& \quad \stackrel{(b)}{\leq} \delta_n^{(\text{S})}  + I \big(S_{(1),1:i-1} S_{(2),1:i-1}\tilde{Z}_{1:i-1}^n ; \tilde{Z}_{i}^n \big| S_{(1),i} S_{(2),i} \tilde{Z}_{1:i-1}^n \Xi_{(1),i-1} \Delta^{(V)}_{(2),i-1}  \Xi_{(2),i-1}  \bar{O}_{(2),i-1}^{(U)} \big)  \\
& \quad \stackrel{(c)}{=} \delta_n^{(\text{S})} + I \big(S_{(1),1:i-1} S_{(2),1:i-1}\tilde{Z}_{1:i-1}^n ; \tilde{Z}_{i}^n \big| \mathsf{B}_{i-1} \big)  \\
& \quad \leq \delta_n^{(\text{S})} + I \big(S_{(1),1:i-1} S_{(2),1:i-1}\tilde{Z}_{1:i-1}^n ; \tilde{Z}_{i}^n W^{\prime}_{(1),i} \big| \mathsf{B}_{i-1} \big)  \\
& \quad =  \delta_n^{(\text{S})}  +  I \big(S_{(1),1:i-1} S_{(2),1:i-1} \tilde{Z}_{1:i-1}^n ;  W^{\prime}_{(1),i} \big| \mathsf{B}_{i-1} \big) \\
& \qquad + I \big(S_{(1),1:i-1} S_{(2),1:i-1}\tilde{Z}_{1:i-1}^n ; \tilde{Z}_{i}^n  \big| \mathsf{B}_{i-1} W^{\prime}_{(1),i} \big)  \\
& \quad \stackrel{(d)}{=} \delta_n^{(\text{S})}  +  I \big(S_{(1),1:i-1} S_{(2),1:i-1} \tilde{Z}_{1:i-1}^n ;  W^{\prime}_{(1),i} \big| \mathsf{B}_{i-1} \big) \\
& \quad \leq  \delta_n^{(\text{S})}  +  I \big(\tilde{R}_{(1),1:i-1} S_{(2),1:i-1} ;  \tilde{Z}_{1:i-1}^n W^{\prime}_{(1),i} \big| \mathsf{B}_{i-1} \big) \\
& \quad =  \delta_n^{(\text{S})}  +  I \big(\tilde{R}_{(1),1:i-1} S_{(2),1:i-1} ;  W^{\prime}_{(1),i} \big| \mathsf{B}_{i-1} \big) +  I \big(\tilde{Z}_{1:i-1}^n ;  W^{\prime}_{(1),i} \big| \mathsf{B}_{i-1} \tilde{R}_{(1),1:i-1} S_{(2),1:i-1} \big) \\
& \quad \stackrel{(e)}{=}  \delta_n^{(\text{S})}  +  I \big(\tilde{R}_{(1),1:i-1} S_{(2),1:i-1} ;  W^{\prime}_{(1),i} \big| \mathsf{B}_{i-1} \big) \\
& \quad =  \delta_n^{(\text{S})}  +  I \big(\tilde{R}_{(1),1:i-1} S_{(2),1:i-1} ;  \bar{\Omega}_{(1),i} \oplus \kappa_{\Omega}^{(V)} \big| \mathsf{B}_{i-1} \big) \\
& \quad \stackrel{(f)}{=}  \delta_n^{(\text{S})} 
\end{align*}
where $(a)$ holds by independence between $(S_{(1),i+1:L}, S_{(2),i+1:L})$ and any random variable from Blocks~$1$~to~$i$; $(b)$ holds by Lemma~\ref{lemma:secrecy1x} because, according to Section~\ref{sec:PCSx1}, we have
\begin{align*}
& \big[S_{(1),i} S_{(2),i} \Xi_{(1),i-1} \Delta^{(V)}_{(2),i-1}  \Xi_{(2),i-1}  \bar{O}_{(2),i-1}^{(U)} \big]  \\
& \quad = \big[ \tilde{A}_i\big[ \mathcal{H}_{V|Z}^{(n)} \big] \tilde{T}_{(1),i}\big[  \mathcal{H}_{U_{(k)}|VZ}^{(n)} \big] \tilde{T}_{(2),i}\big[  \mathcal{H}_{U_{(\bar{k})}|VU_{(k)}Z}^{(n)} \big] \big];
\end{align*}
in $(c)$ we have defined $\mathsf{B}_{i-1} \triangleq \big[S_{(1),i} S_{(2),i} \Xi_{(1),i-1} \Delta^{(V)}_{(2),i-1}  \Xi_{(2),i-1}  \bar{O}_{(2),i-1}^{(U)} \big]$; $(d)$ follows from applying d-separation \cite{pearl2009causality} over the Bayesian graph in Figure~\ref{fig:dependenciesR1x} to obtain that $\tilde{Z}_{i}^n$ and $(S_{(1),1:i-1}S_{(2),1:i-1}\tilde{Z}_{1:i-1}^n)$ are conditionally independent given $\big(\mathsf{B}_{i-1},W^{\prime}_{(1),i}\big)$; $(e)$ also follows from applying \emph{d-separation} to obtain that $W^{\prime}_{(1),i}$ and $\tilde{Z}_{1:i-1}^n$ are conditionally independent given $(\mathsf{B}_{i-1},\tilde{R}_{(1),1:i-1}, S_{(2),1:i-1})$; and $(f)$ holds because $\bar{\Omega}_{(1),i}^{(V)}$ is independent of $\mathsf{B}_{i-1}$, $S_{(2),1:i-1}$ and any random variable from Block~$1$~to~$(i-2)$, and because from applying crypto-lemma \cite{crypto1241} we obtain that $\bar{\Omega}_{(1),i}^{(V)} \oplus \kappa_{\Omega}^{(V)}$ is independent of $\tilde{R}_{(1),i-1}^n$.

\section{Proof of Lemma~\ref{lemma:secrecy3x}}\label{app:secrecy3x}
For any $i \in [2,L]$ and sufficiently large $n$, we have 
\begin{align*}
& I \big(S_{(1),1:L} S_{(2),1:L} \tilde{Z}_{i+1:L}^n ; \tilde{Z}_{i}^n \big) \\
& \quad = I \big(S^{\prime}_{(1),1:L} S^{\prime}_{(2),1:L} \tilde{Z}_{i+1:L}^n ; \tilde{Z}_{i}^n \big) \\
& \quad =  I \big(S^{\prime}_{(1),i:L} S^{\prime}_{(2),i:L} \tilde{Z}_{i+1:L}^n ;\tilde{Z}_{i}^n \big) + I \big(S^{\prime}_{(1),1:i-1} S^{\prime}_{(2),1:i-1} ; \tilde{Z}_{i}^n \big|  S^{\prime}_{(1),i:L} S^{\prime}_{(2),i:L} \tilde{Z}_{i+1:L}^n \big)  \\
& \quad \stackrel{(a)}{=} I \big(S^{\prime}_{(1),i:L} S^{\prime}_{(2),i:L} \tilde{Z}_{i+1:L}^n ;\tilde{Z}_{i}^n \big) \\
& \quad \leq  I \big( S^{\prime}_{(1),i:L} S^{\prime}_{(2),i:L} \tilde{Z}_{i+1:L}^n \Xi_{(2),i+1} \ominus_{(1),i+1}, \Xi_{(1),i+1} \bar{O}_{(1),i+1}^{(U)}  ;\tilde{Z}_{i}^n  \big)  \\
& \quad = I \big( S^{\prime}_{(1),i} S^{\prime}_{(2),i} \Xi_{(2),i+1} \ominus_{(1),i+1}, \Xi_{(1),i+1} \bar{O}_{(1),i+1}^{(U)}  ;\tilde{Z}_{i}^n  \big)  \\
& \qquad +  I \big(S^{\prime}_{(1),i+1:L} S^{\prime}_{(2),i+1:L} \tilde{Z}_{i+1:L}^n  ;\tilde{Z}_{i}^n \big| S^{\prime}_{(1),i} S^{\prime}_{(2),i} \Xi_{(2),i+1} \ominus_{(1),i+1}, \Xi_{(1),i+1} \bar{O}_{(1),i+1}^{(U)} \big)  \\
& \quad \stackrel{(b)}{\leq} \delta_n^{(\text{S})}  + I \big(S^{\prime}_{(1),i+1:L} S^{\prime}_{(2),i+1:L} \tilde{Z}_{i+1:L}^n  ;\tilde{Z}_{i}^n \big| S^{\prime}_{(1),i} S^{\prime}_{(2),i} \Xi_{(2),i+1} \ominus_{(1),i+1}, \Xi_{(1),i+1} \bar{O}_{(1),i+1}^{(U)} \big)  \\
& \quad \stackrel{(c)}{=} \delta_n^{(\text{S})}  + I \big(S^{\prime}_{(1),i+1:L} S^{\prime}_{(2),i+1:L} \tilde{Z}_{i+1:L}^n  ;\tilde{Z}_{i}^n \big| \mathsf{B}_{i+1} \big)  \\
& \quad \leq \delta_n^{(\text{S})}  + I \big(S^{\prime}_{(1),i+1:L} S^{\prime}_{(2),i+1:L} \tilde{Z}_{i+1:L}^n  ; \tilde{Z}_{i}^n W^{\prime}_{(2),i}  \big| \mathsf{B}_{i+1} \big)  \\
& \quad =  \delta_n^{(\text{S})}  + I \big(S^{\prime}_{(1),i+1:L} S^{\prime}_{(2),i+1:L} \tilde{Z}_{i+1:L}^n  ; W^{\prime}_{(2),i}  \big| \mathsf{B}_{i+1} \big) \\
& \qquad + I \big(S^{\prime}_{(1),i+1:L} S^{\prime}_{(2),i+1:L} \tilde{Z}_{i+1:L}^n  ; \tilde{Z}_{i}^n  \big| \mathsf{B}_{i+1} W^{\prime}_{(2),i}  \big) \\
& \quad \stackrel{(d)}{=} \delta_n^{(\text{S})}  + I \big(S^{\prime}_{(1),i+1:L} S^{\prime}_{(2),i+1:L} \tilde{Z}_{i+1:L}^n  ; W^{\prime}_{(2),i}  \big| \mathsf{B}_{i+1} \big) \\
& \quad \leq  \delta_n^{(\text{S})}  +  I \big(\tilde{R}_{(2),i+1:L} S^{\prime}_{(1),i+1:L} \tilde{Z}_{i+1:L}^n  ; W^{\prime}_{(2),i}  \big| \mathsf{B}_{i+1} \big) \\
& \quad =  \delta_n^{(\text{S})}  +  I \big(\tilde{R}_{(2),i+1:L} S^{\prime}_{(1),i+1:L} ; W^{\prime}_{(2),i}  \big| \mathsf{B}_{i+1} \big) +  I \big( \tilde{Z}_{i+1:L}^n  ; W^{\prime}_{(2),i}  \big| \mathsf{B}_{i+1} \tilde{R}_{(2),i+1:L} S^{\prime}_{(1),i+1:L}  \big)  \\
& \quad \stackrel{(e)}{=}  \delta_n^{(\text{S})}  +  I \big(\tilde{R}_{(2),i+1:L} S^{\prime}_{(1),i+1:L} ; W^{\prime}_{(2),i}  \big| \mathsf{B}_{i+1} \big) \\
& \quad = \delta_n^{(\text{S})}  +  I \big(\tilde{R}_{(2),i+1:L} S^{\prime}_{(1),i+1:L} ; \bar{\Omega}_{(2),i} \oplus \kappa_{\Omega} \big| \mathsf{B}_{i+1} \big) \\
& \quad \stackrel{(f)}{=}  \delta_n^{(\text{S})} 
\end{align*}
where $(a)$ holds by independence between $(S^{\prime}_{(1),1:i-1}, S^{\prime}_{(2),1:i-1})$ and any random variable from Blocks~$i+1$~to~$L$; $(b)$ holds by Lemma~\ref{lemma:secrecy1x} because
\begin{align*}
& \big[S^{\prime}_{(1),i} S^{\prime}_{(2),i} \Xi_{(2),i+1} \ominus_{(1),i+1}, \Xi_{(1),i+1} \bar{O}_{(1),i+1}^{(U)} \big]  \\
& \quad = \big[ \tilde{A}_i\big[ \mathcal{H}_{V|Z}^{(n)} \big] \tilde{T}_{(2),i}\big[  \mathcal{H}_{U_{(k)}|VZ}^{(n)} \big] \tilde{T}_{(1),i}\big[  \mathcal{H}_{U_{(\bar{k})}|VU_{(k)}Z}^{(n)} \big] \big];
\end{align*}
in $(c)$ we have defined $\mathsf{B}_{i+1} \triangleq \big[S^{\prime}_{(1),i} S^{\prime}_{(2),i} \Xi_{(2),i+1} \ominus_{(1),i+1}, \Xi_{(1),i+1} \bar{O}_{(1),i+1}^{(U)} \big]$; $(d)$ follows from applying d-separation \cite{pearl2009causality} over the Bayesian graph in Figure~\ref{fig:dependenciesR2x} to obtain that $\tilde{Z}_{i}^n$ and $(S^{\prime}_{(1),i+1:L}S^{\prime}_{(2),i+1:L}\tilde{Z}_{i+1:L}^n)$ are conditionally independent given $\big(\mathsf{B}_{i+1},W^{\prime}_{(2),i} \big)$; $(e)$ also follows from applying \emph{d-separation} to obtain that $(W^{\prime}_{(2),i}$ and $\tilde{Z}_{i+1:L}^n$ are conditionally independent given $(\mathsf{B}_{i+1},\tilde{R}_{(2),i+1:L}, S^{\prime}_{(1),i+1:L})$; and $(f)$ holds because $\bar{\Omega}_{(2),i}^{(V)}$ is independent of $\mathsf{B}_{i+1}$, $S_{(1),1:i-1}$ and any random variable from Block~$i+1$~to~$L$, and because from applying crypto-lemma \cite{crypto1241} we obtain that $\bar{\Omega}_{(2),i}^{(V)} \oplus \kappa_{\Omega}^{(V)}$ is independent of $\tilde{R}_{(2),i+1}^n$.

\end{appendices}

\bibliographystyle{ieeetr}
\bibliography{bibliography}

\end{document}